\documentclass[12pt]{ociamthesis}  

\usepackage{mathtools}
\usepackage{amsthm}
\usepackage{amssymb}
\usepackage{amsfonts}
\usepackage[english]{babel}
\usepackage{hyperref}
\usepackage[T1]{fontenc}
\usepackage[utf8]{inputenc}
\usepackage{braket}
\usepackage{float}
\usepackage{stmaryrd}
\usepackage{csquotes}
\usepackage[backend=biber,style=numeric]{biblatex}

\usepackage[svgnames]{xcolor}
\usepackage{tikz}
\usepackage[final]{changes}

\numberwithin{equation}{section}

\theoremstyle{plain}
\newtheorem{theorem}{Theorem}[section]
\newtheorem{lemma}[theorem]{Lemma}
\newtheorem{proposition}[theorem]{Proposition}

\newtheorem{corollary}[theorem]{Corollary}

\theoremstyle{definition}
\newtheorem{definition}[theorem]{Definition}
\newtheorem{example}[theorem]{Example}

\theoremstyle{remark}
\newtheorem{remark}[theorem]{Remark}
\newtheorem{prob}[theorem]{Problem}

\newcommand{\derive}{\Longrightarrow}

\newcommand{\multiednce}{\textbf{MultiEdNCE}}
\newcommand{\ednce}{\textbf{edNCE}}

\usetikzlibrary{decorations.pathmorphing}
\usetikzlibrary{decorations.markings}
\usetikzlibrary{decorations.pathreplacing}
\usetikzlibrary{arrows}
\usetikzlibrary{shapes}

\pgfdeclarelayer{edgelayer}
\pgfdeclarelayer{nodelayer}
\pgfsetlayers{edgelayer,nodelayer,main}

\tikzstyle{braceedge}=[decorate,decoration={brace,amplitude=10pt}]
\tikzstyle{square box}=[rectangle,fill=white,draw=black,minimum height=6mm,minimum width=6mm,yshift=0.7mm]
\tikzstyle{wire label}=[font=\footnotesize, auto,swap]

\tikzstyle{none}=[inner sep=0pt]
\tikzstyle{gn}=[circle,fill=Lime,draw=Black,line width=0.8 pt]
\tikzstyle{rn}=[circle,fill=Red,draw=Black, line width=0.8 pt]
\tikzstyle{H}=[rectangle,fill=Yellow,draw=Black]
\tikzstyle{line}=[scalar,fill=White,draw=Black]
\tikzstyle{io}=[rectangle,fill=White,draw=Black]
\tikzstyle{block}=[rectangle,fill=Orange,draw=Black]
\tikzstyle{graph}=[circle,fill=White,draw=Black]
\tikzstyle{empty}=[rectangle,fill=none,draw=none]
\tikzstyle{box}=[rectangle,fill=White,draw=Black]
\tikzstyle{dot}=[circle,fill=Black,draw=Black,inner sep=0pt,minimum size=1pt]
\tikzstyle{small dot}=[circle,fill=Black,draw=Black,inner sep=0pt,minimum size=1pt]
\tikzstyle{Dot}=[circle,fill=Black,draw=Black,inner sep=0pt,minimum size=3pt]
\tikzstyle{diam}=[rectangle,fill=Black,draw,yscale=1.2,rotate=45]
\tikzstyle{gangle}=[rectangle,fill=Lime,draw=Black]
\tikzstyle{rangle}=[rectangle,fill=Red,draw=Black]
\tikzstyle{circ}=[circle,fill=none,draw=Black,scale=1.3]
\tikzstyle{ellip}=[ellipse,fill=none,draw=Black,scale=1.3,minimum width =1.3cm]
\tikzstyle{ellip2}=[ellipse,fill=White,draw=Black,scale=1.3,minimum width =3cm]
\tikzstyle{bbox}=[rectangle,fill=Blue,draw=Blue,scale=0.6]
\tikzstyle{gg}=[shape=rectangle,fill=White,draw=Black,dashed]

\tikzstyle{nodev}=[circle,fill=none,draw=Black,scale=1]
\tikzstyle{greynode}=[circle,fill=Grey,draw=Black,scale=1]
\tikzstyle{blacknode}=[circle,fill=Black,draw=Black,scale=1]
\tikzstyle{wirev}=[circle,fill=Black,draw=Black,inner sep=0pt,minimum size=3pt]
\tikzstyle{wirevred}=[circle,fill=Red,draw=Black,inner sep=0pt,minimum size=3pt]

\tikzstyle{simple}=[-,draw=Black]
\tikzstyle{directed}=[->,draw=Black]
\tikzstyle{bdirected}=[<->,draw=Black]
\tikzstyle{bothdirs}=[bdirected,draw=Black]
\tikzstyle{bothdirsred}=[bdirected,draw=Red]
\tikzstyle{blue}=[-,draw=Blue]
\tikzstyle{redd}=[directed,draw=Red]
\tikzstyle{redu}=[-,draw=Red]
\tikzstyle{blued}=[directed,draw=Blue]
\tikzstyle{dash}=[dashed,draw=Black]
\tikzstyle{dashedred}=[dashed,draw=Red]

\tikzstyle{dotpic}=[scale=0.5]

\tikzstyle{every picture}=[baseline=-0.25em]

\newcommand{\stikz}[2][1]{\scalebox{#1}{
\InputIfFileExists{#2}{}{\input{./tikz/#2}}
}}
\newcommand{\cstikz}[2][1]{\begin{center}\stikz[#1]{#2}\end{center}}

\title{Rewriting Context-free Families of String Diagrams}

\author{Vladimir Nikolaev Zamdzhiev}
\college{Oriel College}

\degree{Doctor of Philosophy}     
\degreedate{Trinity 2016}

\begin{document}

\baselineskip=18pt plus1pt

\setcounter{secnumdepth}{3}
\setcounter{tocdepth}{3}

\maketitle                  

\begin{acknowledgementslong}
  I would like to thank my supervisors Prof. Bob Coecke and Prof. Samson Abramsky
for giving me the opportunity to do a DPhil and also for the invaluable advice
they have provided me over the course of my DPhil studies. I also want to thank
Dr. Aleks Kissinger who was my primary supervisor during my DPhil. He spent
a lot of his time on my supervision and has always been very helpful. I cannot
imagine doing my DPhil without his help.

I also want to thank my family for their support during my studies. My
mother, Simka Zamdzhieva, my father, Nikolay Zamdzhiev, and my sister, Desislava
Zamdzhieva, have always been there to provide emotional support and advice.
For this, I will always be very grateful.

I also gratefully acknowledge financial support from the Scatcherd European
Scholarship and the EPSRC. I would not have been able to undertake this
study without their financial support.

Thanks also go to many of my friends that I have met over the years. There
are, of course, too many people to mention here, but special thanks go to
Ventsislav Chonev, Nikola Vlahov, Alasdair Campbell, Joelle Grogan, 
Claire Deligny and Lucy Auton with whom I have shared many of my joys and
worries during my DPhil studies. I also very much enjoyed my three year
volunteer role as bar manager at Oriel College MCR, where I have met many
friends. I would like to thank the Oriel MCR community for this excellent
experience.
\end{acknowledgementslong}

\begin{abstract}
  String diagrams provide a convenient graphical framework which may be used for
equational reasoning about morphisms of monoidal categories. However, unlike
term rewriting, which is the standard way of reasoning about the morphisms of
monoidal categories, rewriting string diagrams results in shorter equational
proofs, because the string diagrammatic representation allows us to formally
establish equalities modulo any rewrite steps which follow from the monoidal
structure.

Manipulating string diagrams by hand is a time-consuming and error-prone
process, especially for large string diagrams. This can be ameliorated by using
software proof assistants, such as Quantomatic.

However, reasoning about concrete string diagrams may be limiting and in some
scenarios it is necessary to reason about entire (infinite) families of string
diagrams. When doing so, we face the same problems as for manipulating concrete
string diagrams, but in addition, we risk making further mistakes if we are not
precise enough about the way we represent (infinite) families of string
diagrams.

The primary goal of this thesis is to design a mathematical framework for
equational reasoning about infinite families of string diagrams which is
amenable to computer automation. We will be working with context-free families
of string diagrams and we will represent them using context-free graph
grammars. We will model equations between infinite families of diagrams using
rewrite rules between context-free grammars. Our framework represents
equational reasoning about concrete string diagrams and context-free families
of string diagrams using double-pushout rewriting on graphs and context-free
graph grammars respectively. We will prove that our representation is sound by
showing that it respects the concrete semantics of string diagrammatic
reasoning and we will show that our framework is appropriate for software
implementation by proving important decidability properties. 
\end{abstract}

\begin{romanpages}          
\tableofcontents            
\end{romanpages}            

\pagestyle{plain}

\chapter{Introduction}

A \emph{monoidal category} is a category $\mathbf{C}$ which is equipped with a
special object $I \in \mathbf{C}$, called the monoidal unit, and also a
bifunctor $\otimes : \mathbf{C} \times \mathbf{C} \to \mathbf{C}$, which is
associative and unital, called the \emph{tensor product}. Like any other
category, two morphisms $f$ and $g$ may be composed in the usual way $f \circ
g$, provided they are compatible, but they may also be composed in another,
orthogonal way, using the tensor product: $f \otimes g$.

Monoidal categories are defined at a high level of abstraction and they have
found applications in many different fields. For example, in linguistics they
can be used to provide the compositional semantics for
sentences~\cite{LambekvsLambek}, in concurrency theory, monoidal categories
describe the structure of Petri nets with boundary~\cite{petri_shit}, in
control theory they can be used to study signal-flow
graphs~\cite{Bonchi2015,Bonchi2014,Baez2014a} and in categorical quantum
mechanics~\cite{cqm}, monoidal categories provide the basic framework for all
developments in the field and have been used to design diagrammatic calculi
for quantum computation, such as the ZX-calculus \cite{bigmainzx, zx_small}.

In any monoidal category, we can prove:
\begin{equation}\label{eq:shit-compose}
(f \otimes id) \circ (id \otimes g)= (id \otimes g) \circ (f \otimes id)
\end{equation}
and the proof involves a few simple steps:
\begin{align*}
(f \otimes id) \circ (id \otimes g) &= (f \circ id) \otimes (id \circ g)
&\text{(interchange law)}\\
&= (id \circ f) \otimes (g \circ id) &\text{(identity)}\\
&= (id \otimes g) \circ (f \otimes id) &\text{(interchange law)}
\end{align*}
where we have used the interchange law, which is not an axiom of monoidal
categories, but a theorem which follows from its axioms
\added{(assuming the axiomatisation given in \cite{mclean}, also cf. proof of
Equation~\eqref{eq:interchange-shit}).}
However, instead of
using terms like we did above, the same fact can be established in an
intuitively simpler way if we use \emph{string diagrams}:
\begin{equation}\label{eq:glorious-compose}
\stikz{intro-laino.tikz}
\end{equation}

String diagrams \cite{joyal_street} are two-dimensional graph-like structures
which can be used to formally reason about monoidal categories in a sound and
complete way -- an equation between two morphisms in a monoidal category
follows from the monoidal data iff the two string diagrams which represent
these morphisms are isotropic, that is, one can be continuously deformed into
the other. In the example above, we can see that by sliding the boxes up or
down the wires, we can obtain one diagram from the other.

String diagrams may be used as an alternative to the standard term-based
approach for equational reasoning in monoidal categories. One of the advantages
provided by doing so is that equational reasoning may be done \emph{modulo} any
equality which follows from the monoidal axioms. As a result, by using string
diagrams, equational proofs are shorter compared to proofs using
term rewriting, because rewrite steps which follow from the monoidal structure
are absorbed into the diagrammatic formalism. This also allows us to focus on
the \emph{additional} structure provided by the model category in which we are
working in (that is, any structure on top of the monoidal one).

A particular example of a string diagrammatic theory is the ZX-calculus.
It is a diagrammatic calculus which may be used to reason about quantum
computation and information. Unlike the standard language for quantum
computing which is based on the Hilbert space formalism, the ZX-calculus
is described via string diagrams. The underlying monoidal category is
\textbf{FdHilb}, the category of finite dimensional Hilbert spaces
and linear maps.
The standard presentation of the syntax of the
ZX-calculus~\cite{bigmainzx} includes generators of the form:
\begin{equation}\label{eq:zx-syntax}
\stikz{zx-spiders.tikz}
\end{equation}
where $\alpha \in [0,2\pi).$
\added{These generators, when interpreted in \textbf{FdHilb},
define the following linear maps:}

\[
\left \llbracket \stikz{green-spider-int.tikz} \right \rrbracket
=
\begin{cases}
\ket{0^m} \mapsto \ket{0^n}\\
\ket{1^m} \mapsto e^{i\alpha} \ket{1^n}\\
\text{others } \mapsto 0
\end{cases}
\mbox{}\quad
\left \llbracket \stikz{red-spider-int.tikz} \right \rrbracket
=
\begin{cases}
\ket{+^m} \mapsto \ket{+^n}\\
\ket{-^m} \mapsto e^{i\alpha} \ket{-^n}\\
\text{others } \mapsto 0
\end{cases}
\]

\added{where the green generators are defined over the
$\{\ket 0, \ket 1 \}^m$
basis and the red generators are defined over the
$\{\ket +, \ket - \}^m$
basis. The equational rules of the ZX-calculus then describe how these two
families of linear maps interact. However, unlike the usual way of reasoning
about quantum computing, the relationships are described in a purely
diagrammatic manner, so reasoning in the ZX-calculus can be done by diagram
rewriting instead of using linear algebra.}

If the labels $\alpha$ of the generators are
allowed to range over $[0,2\pi),$ then the ZX-calculus is \emph{universal} for
quantum computation, meaning it can exactly represent any morphism in
\textbf{FdHilb}. However, in that case it is also
\emph{incomplete}~\cite{zx-incomplete}, meaning that using its equational rules
we cannot prove all true equalities in \textbf{FdHilb}. However, if we restrict
the labels $\alpha$ to the integer multiples $k \pi/4$ in $[0,2\pi),$ then
the calculus is approximately universal and complete for an important segment
of quantum computing, called Stabilizer Quantum Mechanics~\cite{zx_complete} and
has many applications for Fault-Tolerant Quantum
Computation~\cite{fault-tolerant}. We will use this version of the ZX-calculus
as motivation for some of the constructions in this thesis and we will often
provide examples using it in order to illustrate more abstract concepts that we
develop.

Equational reasoning in the ZX-calculus is done
by manipulating diagrams, instead of using linear algebra like in the
traditional language for quantum computing.
For example, a well-known fact in quantum computing
is that the composition of two CNOT gates is the identity. The standard
way of proving this in the traditional language is by evaluating the matrix
representation of the composition, which in this case consists of multiplying
two $4 \times 4$ matrices:
\[
\begin{bmatrix}
    1 & 0 & 0 & 0 \\
    0 & 1 & 0 & 0 \\
    0 & 0 & 0 & 1 \\
    0 & 0 & 1 & 0 
\end{bmatrix}
\begin{bmatrix}
    1 & 0 & 0 & 0 \\
    0 & 1 & 0 & 0 \\
    0 & 0 & 0 & 1 \\
    0 & 0 & 1 & 0 
\end{bmatrix}
=
\begin{bmatrix}
    1 & 0 & 0 & 0 \\
    0 & 1 & 0 & 0 \\
    0 & 0 & 1 & 0 \\
    0 & 0 & 0 & 1 
\end{bmatrix}
\]
to get the identity matrix. The same fact can be formally established in
the ZX-calculus by rewriting string diagrams:
\begin{equation}\label{eq:cnot-intro}
\stikz[0.98]{cnots-laina.tikz}
\end{equation}
with respect to the axioms of the ZX-calculus. The monoidal structure of
\textbf{FdHilb} is captured by the string diagrammatic nature of the calculus
and the additional structure of \textbf{FdHilb} is captured by the equational
axioms of the ZX-calculus. Rewriting diagrams in the ZX-calculus, with respect
to its axiomatic rules, then corresponds to equational reasoning in
\textbf{FdHilb} \emph{modulo} any equational step which follows from the
monoidal structure (such as the interchange law from above).

The shorter equational proofs is not the only advantage of using string
diagrams for reasoning about monoidal categories. Composition of morphisms
in a monoidal category is inherently two-dimensional -- morphisms can be
composed using the standard categorical composition $(- \circ -)$ or
the tensor product $(- \otimes -).$ However, the term-based syntax for
writing down morphisms is one-dimensional -- we write our compositions on
a line and we have to use brackets in order to correctly specify the order in
which the composition operations need to be performed. Another advantage of
using string diagrams is that their syntax captures the compositional relations
of monoidal categories in a more intuitively clear way. We can see this even
for simple cases -- compare the term-based syntax of
\eqref{eq:shit-compose} with the string diagrammatic syntax of
\eqref{eq:glorious-compose}. The difference becomes even more
pronounced when considering larger terms. For example, for morphisms
$f: A \to A \otimes C$, $g: B \otimes C \to B \otimes D$, $h:A \otimes C \to
C$, $k: D \to B$, compare:
\[
(id_C \otimes k) \circ 
(h \otimes id_D) \circ 
(f \otimes id_B \otimes id_D) \circ
(id_A \otimes g)
\]
with:
\cstikz{intro-golqmoto-laino.tikz}

Equational reasoning with string diagrams is done via subdiagram substitution.
Given an equational rule between two string diagrams, one side of the rule
is matched onto a subdiagram of some target diagram which we wish to rewrite,
and the matched part is then replaced with the other side of the equational
rule. For example, in rewrite step (4) of \eqref{eq:cnot-intro} above, the
rule which is applied is simply:
\cstikz{tapa-aksioma.tikz}

However, rewriting large diagrams in such a way may involve many non-trivial
rewrite steps \cite{zamdzhiev_msc}. Doing this by hand has several
disadvantages: the process could take a long time, manually creating the
diagrams at each step is tedious and it is possible to commit errors when a
rewrite rule is not applied correctly. To avoid this, it is preferable to use a
software proof assistant like \emph{Quantomatic}~\cite{quanto-cade} which can
assist with the reasoning process. Internally, Quantomatic uses a discrete
representation of string diagrams, called \emph{string graphs}
\cite{kissinger_dphil}, which represents string diagrams as special kinds of
labelled graphs and then string diagram rewriting is represented by the
well-established graph transformation approach called \emph{double-pushout
rewriting} \cite{dpo_book,dpo_original}.

Quantomatic is a diagrammatic proof assistant. It has a graphical user
interface which allows users to define their own theories based on string
diagrams and also to create and edit string diagrams and string diagram rewrite
rules. Quantomatic supports two ways of rewriting string diagrams --
user-guided, where users may apply rewrite rules at each step of the
derivation, or fully automated rewriting, where multiple rewrite steps are
\replaced{performed}{perfromed}
by Quantomatic in accordance to user-defined tactics. Moreover,
users may mix both approaches and the software allows them to explore multiple
rewrite sequences. There are also useful export options -- diagrams, rewrite
rules and entire derivations may be exported as Tikz code ready to be embedded
in a LaTeX document and they may also be easily exported as interactive HTML5
elements within a webpage. All of these features make Quantomatic a useful
tool for equational reasoning about string diagrams as it reduces the amount
of time needed to manipulate them and greatly reduces the possibilities of
errors being introduced during the derivation process.

Aside from rewriting string diagrams one at a time, in many scenarios it is
useful, and sometimes necessary, to rewrite entire \emph{families} of string
diagrams. By rewriting a single string diagram using some equational rule, we
establish an equality between two string diagrams which therefore represents an
equality between two morphisms in the category in which we are reasoning.
However, in this thesis we will show that we can do equational reasoning
for string diagrams at a higher level. A \emph{family} of string diagrams is an
(infinite) set of string diagrams which are related in some way. When
rewriting a family of diagrams, we establish an \emph{equational schema}, that
is we establish (infinitely) many different equalities which relate pairs
of string diagrams and thus morphisms. Reasoning in this way is strictly
more general and then specific equalities between string diagrams follow
as special cases of an equational schema.

We will give examples from the ZX-calculus to illustrate these ideas. To begin
with, the standard presentation of the syntax of the
ZX-calculus~\cite{bigmainzx} includes generators which are described as
families of diagrams, like in~\eqref{eq:zx-syntax}. This means that our green
and red nodes may have any number of inputs or outputs.
The standard axiomatisation
contains concrete equational rules between concrete string diagrams, like the
following one:
\begin{equation}\label{eq:zx-copy}
\stikz{zx-concrete-rule.tikz}
\end{equation}
but it also contains equational schemas between families of ZX-diagrams which
are introduced as axioms, like the following one:
\begin{equation}\label{eq:zx-family-rule}
\stikz{zx-family-rule.tikz}
\end{equation}
The open-ended wires at the bottom are called \emph{inputs} and the open-ended
wires at the top are called \emph{outputs}. The intended meaning of this
equational schema is that for any number of inputs or outputs, the two green
nodes may be merged and their labels added together while preserving the inputs
and outputs, as long as the two nodes are connected by a wire. This axiom of
the ZX-calculus describes infinitely many equalities between pairs of concrete
diagrams. For example, if we require that both nodes have one input and one
output and $\alpha=\beta=0$, then we can get a concrete \emph{instance} of the
equational schema:
\begin{equation}\label{eq:zx-family-instance}
\stikz{zx-family-instance.tikz}
\end{equation}
where we use the convention not to depict numbers $\alpha$ if they are equal
to zero. An instance of an equational schema is simply an equational rule,
which may be used to rewrite string diagrams. For example the instance
\eqref{eq:zx-family-instance} is used to
\replaced{perform}{perfrom}
rewrite step (1) of
the rewrite sequence~\eqref{eq:cnot-intro}. But, we may also use an
equational schema to rewrite an entire family of string diagrams, thereby
establishing a new equational schema.
For example, by inductively applying rewrite rules
\eqref{eq:zx-copy} and \eqref{eq:zx-family-rule}
we can easily prove that the following equational schema is also true:
\begin{equation}\label{eq:zx-family-derive}
\stikz{zx-family-derive.tikz}
\end{equation}
We can then apply this equational schema to the family of string diagrams
given below:
\cstikz{zx-simple-family-laina.tikz}
to establish the new equational schema below:
\cstikz{zx-simple-family-rewrite2.tikz}

Equational reasoning on the level of families of string diagrams and equational
schemas subsumes reasoning on the level of string diagrams and equational
rules. Moreover, reasoning on this level is sometimes necessary. For example,
quantum algorithms and protocols are usually described in terms which allow
for input of arbitrary size. Therefore, if we wish to reason about them
using the ZX-calculus, then we have to be able to talk about families
of ZX-diagrams. The primary focus of this thesis is the study of infinite
families of string diagrams, and in particular, the formal methods which could
be used to do equational reasoning with them. We will not restrict ourselves
to the ZX-calculus or any other specific theory. Instead, we will study this
problem in generality -- we only assume that our string diagrams can be
labelled over finite alphabets of node and wire labels, and we will show
how we can do equational reasoning for certain infinite families of string
diagrams.

When working with large families of string diagrams (in terms of the number of
nodes and edges required to depict them) we face the same problems as with
working with large concrete string diagrams. Therefore, ideally, we would wish
to be able to perform the reasoning process using software support. However,
when describing families of string diagrams so far, we have been rather
informal. We described these families using the $(\cdots)$ notation. While this
notation is intuitive, it is not precise enough for computer implementation.
Therefore, if we wish to work with infinite families of string diagrams using a
proof assistant, then we need to be able to describe these families using a
formal notation.

Quantomatic does support reasoning with certain infinite families of string
diagrams. These families are described by the theory of \emph{!-graphs},
pronounced, \emph{bang graphs}. !-graphs, as the name suggests, are special
kinds of graphs. They are strictly more general compared to string graphs.
Instead of using the $(\cdots)$ notation, families of string diagrams are
denoted by marking subgraphs of a string graph with \emph{!-boxes}, which are
graphically depicted by drawing a blue box around the required subgraph. The
subgraphs which are marked by !-boxes are allowed to be copied an arbitrary
number of times while preserving the connection relations with the rest of the
graph. This is the mechanism used by !-graphs in order to represent infinite
families of string diagrams.

For example, in Quantomatic we can depict a !-graph in the following way:
\cstikz{bang-intro-shit.tikz}
and it is interpreted as the family of diagrams on the left-hand side of
\eqref{eq:zx-family-derive}. More formally, it is the set:
\[
\left\llbracket \stikz{bang-intro-shit.tikz} \right\rrbracket =
\stikz{bang-intro-shit2.tikz}
\]
Then, an equational schema between two families of string diagrams is
represented by a pair of !-graphs with a bijective correspondence between
their !-boxes. For example, the equational schema \eqref{eq:zx-family-derive}
is represented by:
\cstikz{bang-intro-laina.tikz}
Finally, equational reasoning between entire (infinite) families of string
diagrams is represented by double-pushout (DPO) rewriting of !-graphs.
Quantomatic fully supports this kind of rewriting and provides features for
fully automated and user-guided rewriting of families of diagrams.

However, the !-graph formalism is limited in terms of its expressive power.
!-graphs can only represent families of string diagrams which are of bounded
diameter, meaning that there is a fixed upper bound on the shortest distance
between any pair of nodes in its diagrams. In the context of the ZX-calculus,
this is very
\replaced{limiting}{limitting},
because it means that any protocol or algorithm whose
time-complexity is not constant cannot be described via !-graphs. Another
limitation is that !-graphs can only represent families of string diagrams
which are finitely colourable. For example, this means that !-graphs cannot
represent any family of string diagrams which contain a clique of arbitrary
size. In the context of the ZX-calculus this means that we cannot represent
the local complementation rule of \cite{euler_necessity},
\added{given by:}
\begin{equation}\label{eq:local-shit}
\stikz{local-complement-final.tikz}
\end{equation}
\added{where $K_{n-1}$ denotes the totally connected graph on $n-1$ green vertices
connected to each other via Hadamard ($H$) gates.}
This rule is very
important as it is used in the only known
\replaced{decision}{decission}
procedure for equality
of stabilizer operations in the ZX-calculus \cite{zx_complete}. In
fact, one of the initial design goals for the !-graph formalism was precisely
this~--~the ability to represent the local complementation rule of the
ZX-calculus.\footnote{Ross Duncan. Personal communication.}

Our primary motivation in this thesis is to develop an alternative to the
!-graph formalism which avoids some of its limitations in terms of expressive
power, while retaining as many of its useful features as possible. In
particular, we wish to be able to do equational reasoning on infinite families
of string diagrams in a formal way which can be implemented in Quantomatic, or
other software proof-assistants. Of course, as with any other language
generating device, there is a tradeoff between expressive power on one side and
decidability, complexity and structural properties on the other.

The alternative which we propose is based on (slightly extended)
\emph{context-free graph grammars}~\cite{c-ednce}, which avoid both
of the limitations of !-graphs described above.
These grammars are a
generalisation of the standard context-free grammars on strings. However,
context-free graph grammars (CFGGs) generate languages of graphs, not strings,
and they are strictly more powerful than the standard context-free string
grammars. Both CFGGs and !-graphs can generate languages of graphs, but as
we shall show, the classes of languages which they induce are incomparable --
there are !-graph languages which cannot be represented using CFGGs and
vice-versa.

In particular, the alternative grammars to !-graphs which we shall study are
called \emph{B-ESG grammars}, which is a
shorthand for Boundary Encoded String Graph grammars. B-ESG grammars are
a special kind of CFGGs with a simple extension which allows us to encode
some additional graph structure in specially labelled edges.
\added{An example of a B-ESG grammar which represents the local complementation
rule of Equation~\eqref{eq:local-shit} is given by:}
\cstikz[0.7]{local-comp-B-ESG.tikz}
\added{This is a compact representation of a very powerful and complicated
derived rule of the ZX-calculus which !-graphs cannot express. In later
chapters, we will describe in detail how B-ESG grammars work and how they
can represent equational rules of the ZX-calculus.}

We will show that our proposed
alternative is strictly more expressive than an important class of !-graphs,
called \emph{!-graphs with trivial overlap}. The author does not know of any
!-graph languages of interest for practical applications outside of that class,
so the proposed alternative is suitable in terms of expressive power for
current known uses. We also show how B-ESG grammars can be used to represent
equational schemas between infinite context-free families of string diagrams.
Our proposed framework also supports rewriting infinite families of string
diagrams using already established equational schemas. All of our rewrites are
kept sound in the sense that they respect the concrete semantics of the
equational schemas and families of diagrams used in the rewrites. Moreover,
because we are also interested in implementing software support for our
framework, we keep all of our constructions decidable.

The rest of the thesis is structured as follows. In
Chapter~\ref{ch:background}, we introduce all of the relevant background
theory. This includes an introduction to \emph{adhesive} and \emph{partially
adhesive} categories which show how to do DPO rewriting. DPO rewriting is
the graph transformation mechanism which we use in order to model equational
reasoning, both for string diagrams and also for families of string diagrams.
We also describe string diagrams and their discrete representation, string
graphs, in sufficient detail. We then give a more detailed introduction to
families of string diagrams and !-graphs. We conclude the background chapter by
introducing the context-free graph grammars (CFGGs) which we shall be using.
In particular, we shall be working with B-edNCE graph grammars. This chapter
does not contain any original work, except for a few simple propositions.

The original work is presented in the rest of the chapters.
In Chapter~\ref{ch:context-free}, we consider the expressive power of
CFGGs
\added{on languages of string graphs. In particular, we show that the two
dominant classes of context-free graph grammars -- Vertex Replacement (VR) and
Hyperedge Replacement (HR) grammars -- have the same expressive power on
languages of string graphs, whereas on general graphs VR grammars are strictly
more expressive compared to HR grammars.}
We will show that CFGGs are strictly more expressive than a class
of !-graphs  called \emph{!-graphs with no overlap}. This is an important
class of !-graphs, but there are other !-graph languages of interest which
are not included in it. We
\replaced{also}{then}
showcase some limitations of
both !-graph languages and CFGG languages and use that as justification to
consider a slightly more expressive language generating device, which is
introduced in the next chapter.

In Chapter~\ref{ch:besg}, we introduce a simple extension to B-edNCE grammars,
\added{which are special kinds of VR grammars.}
\replaced{This extension}{which}
provides them with a little bit more expressive power compared to
standard B-edNCE grammars. We call these grammars \emph{encoded} B-edNCE
grammars.
\added{The extension formalises the idea that specially labelled
edges can be thought of as fixed graphs. This idea increases the expressive
power of B-edNCE grammars, however, it does not increase the expressive power
of HR grammars. This simple extension allows us to cover languages of interest,
however it is not clear how HR grammars can be extended to cover the same
languages and for this reason we base our investigations on encoded B-edNCE
grammars.}
We then provide sufficient and easily decidable conditions for \added{encoded}
B-edNCE grammars which guarantee that they generate only languages of string
graphs \added{and call the grammars which satisfy these conditions
\emph{B-ESG} grammars.}
Therefore by restricting ourselves to B-ESG grammars, we can represent
families of string diagrams. We also show that the B-ESG grammars are the
largest class of encoded B-edNCE grammars which generate languages of
string graphs. We also prove that B-ESG grammars are strictly more expressive
than !-graphs with trivial overlap. We conclude by showing that B-ESG grammars
have important decidability properties for the operation of a software
proof-assistant. In particular, we show that the membership problem and the
matching enumeration problems are decidable.

In Chapter~\ref{ch:rewriting}, we begin by showing that DPO rewriting is
well-behaved for B-edNCE grammars. We then show how to rewrite B-edNCE grammars
using DPO rewriting, such that the rewrites are sound with respect to the
concrete semantics, that is, any concrete instance of our grammar rewrites
corresponds to a sequence of concrete rewrite rules from the equational
schemas we are representing. These techniques are then extended to B-ESG
grammars by showing that we can rewrite B-ESG grammars using B-ESG rewrite
rules in a sound way. In terms of string diagrams, this shows that we can
do equational reasoning on context-free families of string diagrams using
equational schemas where both sides of the schema are also a context-free
family of diagrams.

Finally, in Chapter~\ref{ch:conclude}, we provide some concluding remarks and
discuss future work.

\chapter{\label{ch:background}Background}
In Section~\ref{sec:adhesive}, we introduce \emph{adhesive} categories, which
are categories suitable for doing double-pushout (DPO) rewriting. Specific
instances of adhesive categories which are presented are $\mathbf{Set}$ and
$\mathbf{MultiGraph}$ -- the categories of sets and multigraphs respectively.

However, in this thesis, we will also be doing DPO rewriting in categories
which are not adhesive. In Section~\ref{sec:partial_adhesive}, we will describe
\emph{partially adhesive} categories, which are categories where DPO rewriting
behaves exactly like DPO rewriting in their ambient adhesive category, provided
that the rewrite rules and matchings satisfy additional conditions.
$\mathbf{Graph}$, the category of graphs where parallel edges are not allowed,
is an example of a partially adhesive category (with ambient adhesive category
$\mathbf{MultiGraph})$ which is discussed.

In Section~\ref{sec:string}, we introduce \emph{string diagrams}, which are the
primary objects which we wish to reason about. We describe their relationship
to (traced symmetric) monoidal categories and then introduce a discrete
representation
of string diagrams, called \emph{string graphs}, which we shall use as it
is more straightforward to build a proof assistant on top of this theory.
We also describe the partially adhesive structure of string graphs.

We proceed by introducing \emph{families} of string diagrams in
Section~\ref{sec:family} and show how we can reason about infinitely many
string diagrams at the same time, instead of just concrete string diagrams. We
also describe \emph{!-graphs} which provide for a formal and finite way to
represent certain infinite families of string graphs (and thus string
diagrams).

Finally, in Section~\ref{sec:graph-grammars}, we provide an introduction to
\emph{context-free graph grammars}. This is the main mechanism which we use in
order to represent families of string graphs, as an alternative to !-graphs, in
later chapters and most of the original results are stated for these grammars.

\section{Adhesive Categories}\label{sec:adhesive}

Adhesive categories were introduced by Lack and Sobocinski in
\cite{adhesive_categories}. Adhesive categories establish a categorical
framework which generalises the standard method of doing double-pushout (DPO)
rewriting on graphs. Because of the
\replaced{high level}{high-level}
of abstraction on which adhesive
categories are defined, DPO rewriting can be performed on any category which is
shown to satisfy their axioms and the rewrites retain important
properties enjoyed by DPO rewriting on graphs, such as the Local Church-Rosser
Theorem and the Concurrency Theorem.

Adhesive categories fit in very nicely in the context of this thesis. We will
be doing DPO rewriting on both graphs and edNCE grammars both of which will be
based on two different graph models. DPO rewriting over graphs is well-known,
but it has not been studied in the context of edNCE grammars.  So, instead of
showing how DPO rewriting works in each of those cases, we will recall how DPO
rewriting can be performed in adhesive categories where it has already been
demonstrated. To make use of these results,
\deleted{we}
we will show that the specific
categories
\deleted{in which}
we are interested in are adhesive (or partially adhesive,
see Section \ref{sec:partial_adhesive}) and then DPO rewriting along with
several useful lemmas follow as a result in each of our model categories.

In this section, we will introduce all of the definitions and propositions
related to adhesive categories and DPO
\replaced{rewriting}{rerwriting}
that we will need in the rest
of the thesis. We begin by providing the formal definition for an adhesive
category.

\begin{definition}[Adhesive Category \cite{adhesive_categories}]
\label{def:adhesive_category}
A category $\mathbf{C}$ is said to be \emph{adhesive} if

\begin{enumerate}
\item $\mathbf{C}$ has pushouts along monomorphisms
\item $\mathbf{C}$ has pullbacks
\item pushouts along monomorphisms are van Kampen squares
\end{enumerate}
\end{definition}

Van Kampen squares are crucial for establishing many of the properties of
adhesive categories. However, in this thesis, we do not make direct use of them
in any definitions or propositions. Thus, it is not necessary to understand
what a van Kampen square is and for this reason we will not provide a formal
definition.

We will also not make any further references to pullbacks either, so the only
important part of the definition (in the context of this thesis) is that
pushouts along monomorphisms always exist. To show that the categories we are
interested in are adhesive, we will
\replaced{consider}{cosider}
a few well-known examples of
adhesive categories and then we will make use of two lemmas which show that
adhesive categories are closed under certain categorical constructions.
These lemmas and one of the examples are provided below.

\begin{example}[\cite{adhesive_categories}]\label{ex:set}
The category $\mathbf{Set}$ is adhesive.
\end{example}

\begin{lemma}[\cite{adhesive_categories}]\label{lem:adhesive_slice}
If $\mathbf{C}$ is adhesive then so are $\mathbf{C}/C$ and $C/\mathbf{C}$
for any object $C$ of $\mathbf{C}$.
\end{lemma}

\begin{lemma}[\cite{adhesive_categories}]\label{lem:adhesive_functor}
If $\mathbf{C}$ is adhesive then so is any functor category $[\mathbf{X},
\mathbf{C}]$.
\end{lemma}

These two lemmas are powerful and they allow us to easily prove adhesivity
of some of the categories we are interested in. We shall illustrate this now
by proving that the category of labelled multigraphs is adhesive, which
is well-known and pointed out in \cite{adhesive_categories}.

\begin{definition}[Unlabelled Multigraphs]
The category of \emph{unlabelled multigraphs} is $\mathbf{UMultiGraph}$. This
category is defined as the functor category $[\mathbb G,
\mathbf{Set}]$, where $\mathbb G$ is the two object category given by:
\cstikz{unlabelled_category_graph.tikz}
For a multigraph $H \in \mathbf{UMultiGraph},$ we shall denote with $V_H$ its
set of vertices and with $E_H$ its set of edges.
\end{definition}

This definition is standard in the graph transformation literature.
An object of $\mathbf{UMultiGraph}$ consists of a set of vertices $(V)$,
a set of edges $(E)$ and two functions $s$ (source) and $t$ (target) which
assign the source and target vertices to edges. Note, that the graphs
in this category may have self-loops and parallel edges (also frequently called
multiple edges). For this reason, we will refer to them as unlabelled
\emph{multigraphs}. We will later introduce \emph{graphs} which are just
labelled multigraphs which do not have any self-loops and any pair of parallel
edges needs to have different labels.

Now we can easily prove that this category is adhesive.

\begin{lemma}
$\mathbf{UMultiGraph}$ is an adhesive category.
\end{lemma}
\begin{proof}
Combining Lemma~\ref{lem:adhesive_functor} with Example~\ref{ex:set}, it
follows immediately that the category $\mathbf{UMultiGraph}$ is adhesive.
\end{proof}

We can extend our multigraphs by introducing labels. A \emph{labelled
multigraph} is a multigraph where there are two labelling functions which
assign labels to the vertices and edges of the multigraph. 

\begin{definition}[Multigraphs]\label{def:multigraph}
The category of \emph{labelled multigraphs} over an alphabet of vertex labels
$\Sigma$ and an alphabet of edge labels $\Gamma,$ is the category
$\mathbf{MultiGraph}_{\Sigma, \Gamma}$ whose objects are pairs $(H, l)$
with $H \in \mathbf{UMultiGraph}$ an unlabelled
multigraph and $l = (l_V, l_E)$ a
pair of labelling functions:
\begin{align*}
l_V     &: V_H   \to \Sigma                             &\text{(the vertex
     labelling function)}\\
l_E     &: E_H   \to \Gamma                             &\text{(the edge
     labelling function)}
\end{align*}
A morphism between two multigraphs
$f : (G, l) \to (H, l')$ is a morphism $f: G \to H$ of
$\mathbf{UMultiGraph}$, which in addition respects the labelling, that is
the following diagrams commute:
\cstikz{respect_labelling.tikz}
\end{definition}

Labelled multigraphs are strictly more general than unlabelled multigraphs --
by choosing singleton sets for the vertex and edge label alphabets $\Sigma$ and
$\Gamma$, we get $\mathbf{UMultiGraph} \cong \mathbf{MultiGraph}_{\Sigma,
\Gamma}.$ We will be working with labelled multigraphs from now on, so for
brevity we will simply refer to them as multigraphs. If the alphabets
$\Sigma$ and $\Gamma$ are clear from the context, then we will simply write
$\mathbf{MultiGraph}$ for the category of labelled multigraphs over $\Sigma$
and $\Gamma$.
If $H \in \mathbf{MultiGraph}$, then we shall denote with $[H]$ the set of
all multigraphs isomorphic to $H$.

\begin{lemma}\label{lem:multigraph-adhesive}
$\mathbf{MultiGraph}_{\Sigma, \Gamma}$ is adhesive.
\end{lemma}
\begin{proof}
Consider the category of \emph{typed multigraphs} $\mathbf{UMultiGraph}/T,$
where $T \in \mathbf{UMultiGraph}$ is the multigraph given by:
\cstikz{typing_graph.tikz}
The typing graph $T$ in the above definition has vertices given by the
\replaced{vertex}{vertx}
label alphabet $\Sigma$ and for each pair of vertices $(u,v)$ and each edge
label $\gamma \in \Gamma$, there is an edge from $u$ to $v$ with label $\gamma$.
Thus, an object in the slice category $\mathbf{UMultiGraph}/T$ assigns
a vertex label to each vertex and an edge label to each edge with no
further restrictions. It is well-known (and easy to see) that
$\mathbf{MultiGraph}_{\Sigma, \Gamma}$
is isomorphic to the slice category $\mathbf{UMultiGraph}/T$
\cite{typed_vs_labelled}. Therefore $\mathbf{MultiGraph}_{\Sigma, \Gamma}$ is
adhesive by Lemma~\ref{lem:adhesive_slice}.
\end{proof}

In a similar way, by making use of these two lemmas, we will later show that
the category of multi-edNCE grammars is adhesive.

Adhesive categories can be used in order to understand rewriting for certain
kinds of structures. One of the most basic notions related to that is the
notion of a rewrite rule. It can be defined in an arbitrary category, not
necessarily an adhesive one. Intuitively, a rewrite rule consists of a
left-hand side object, a right-hand side object and a third object which
appears as a subobject in both of them.

\begin{definition}[Rewrite rule]
A \emph{rewrite rule} $t$ in an arbitrary category $\mathbf{C}$ is a span of
monomorphisms $L \xleftarrow{l} I \xrightarrow{r} R.$
\end{definition}

We shall refer to $L$ as the \emph{left-hand side} of the rule, $R$ as the
\emph{right-hand side} and $I$ as the \emph{interface} of the rule.
In the literature, rewrite rules (or productions in DPO grammars) are sometimes
defined simply as a span. However, some of the important theorems and lemmas
about adhesive categories only hold for spans over monomorphisms. In addition,
we will only consider rewrite rules consisting of a pair of monomorphisms and
for this reason we choose the above definition. We will use rewrite rules in
appropriate categories in order to encode equations (equational schemas)
between string diagrams (families of string diagrams).

Next, we consider a lemma which relates the monomorphisms of a span
to their counterparts in a pushout square.

\begin{lemma}[\cite{adhesive_categories}]\label{lem:mono_stable}
In any adhesive category, monomorphisms are stable under pushout.
\end{lemma}

Therefore, in any adhesive category, given a span of monomorphisms, the pushout
is guaranteed to exist and will consist of four monomorphisms. All of the
pushouts that we will consider in the main body of the thesis will be over
monomorphisms and this lemma will help us to establish that.

Next, we will introduce the notion of pushout complement. Intuitively, it
can be thought of as an operation where we subtract a part $L$ from
another object $H$ while preserving a third part $I$ which is shared by
both objects.

\begin{definition}[Pushout complement \cite{adhesive_categories}]
Let $l: I \to L$ and $m: L \to H$, be two morphisms in an arbitrary category.
The \emph{pushout complement} of the pair $(l,m)$ consists of morphisms
$k : I \to K$ and $s: K \to H$ for which the resulting square:
\cstikz{adhesive-pushout-shit.tikz}
commutes
and is a pushout. 
\end{definition}

Unlike the pushout, a pushout complement is not a universal construction
in category theory. A pushout complement is not necessarily
unique and in certain categories there can be two or more
\replaced{non-isomorphic}{non-isomorhpic}
objects that form a pushout complement. However, adhesive categories
ensure that pushout complements are indeed unique (up to isomorphism)
when they are computed over monomorphisms. We will only consider pushout
complements and pushouts where all morphisms are monomorphisms and
therefore the following lemma is of great importance.

\begin{lemma}[\cite{adhesive_categories}]\label{lem:unique_pushout_complement}
In any adhesive category, pushout complements of monos (if they exist) are
unique up to isomorphism.
\end{lemma}

The significance of this lemma is that it establishes the uniqueness of DPO
rewriting in adhesive categories. As already mentioned, a DPO rewrite consists
of two operations -- first, the pushout complement is computed and then the
pushout of the newly established span. The pushout is always unique and the
above lemma tells us that the entire DPO rewrite is then unique as well.

The next lemma will be useful in showing that all morphisms in our DPO
diagrams are mono.

\begin{lemma}\label{lem:pushout_complement_monos}
In any adhesive category, if $H \xleftarrow{m} L \xleftarrow{l} I$ are both
monomorphisms and their pushout complement exists:
\cstikz{pushout_complement_simple.tikz}
then $s$ and $k$ are also monomorphisms.
\end{lemma}
\begin{proof}
In any adhesive category, monomorphisms are stable under pushout
(Lemma~\ref{lem:mono_stable}). Therefore, $s$ is a monomorphism.

The square commutes and therefore we know $m \circ l = s \circ k$. Both
$m$ and $l$ are
\replaced{monomorphisms}{monomorhpisms}
and therefore $m \circ l$ and thus
$s \circ k$ are also monos. The latter then implies $k$ is a mono as well.
\end{proof}

Adhesive categories ensure that pushout complements are unique, but we still
don't know under what conditions a pushout complement exists. However, the
answer is not provided by the framework of adhesive categories.  Instead, the
conditions need to be determined in every model category in which we wish to do
rewriting separately. The following definition formalizes this idea.

\begin{definition}[Matching conditions]
In any category $\mathbf{C}$, given a rewrite rule $t := L \xleftarrow{l} I
\xrightarrow{r} R,$ a morphism $m: L \to H$ satisfies the \emph{matching
conditions} with respect to $t$ precisely when there exists a pushout
complement of $(l,m)$. A \emph{matching} is a monomorphism $m: L \to H$
which satisfies the matching conditions.
\end{definition}

The matching conditions are also called the \emph{gluing} conditions in
the literature. 
Note, that in the above definition a matching has to be a monomorphism. In the
literature, this is not always the case, but in this thesis we will only
consider injective matchings, so we build this requirement into the definition
in order to simplify the rest of the presentation.

The matching conditions depend on the category $\mathbf{C}$ in which we wish to
compute the pushout complement. For example, the matching conditions in the
category $\mathbf{Set}$ are trivial -- they are satisfied by any
mono. The matching conditions in the category $\mathbf{MultiGraph}$ are
given below.

\begin{example}[\cite{dpo_book}]\label{ex:graph_matching}
In the category $\mathbf{MultiGraph}$, given a rewrite rule $L \xleftarrow{l} I
\xrightarrow{r} R$ and a monomorphism $m: L \to H$, the pushout complement
of $(l,m)$ exists iff the following condition is satisfied:
\begin{description}
\item[No dangling edges:] No edge $e \in E_H - m(E_L)$ is incident to any vertex
$v \in m(V_L - l(V_I))$
\end{description}
\end{example}

The matching conditions for graphs have been well-known for decades and
clearly they are easy to decide by a computer. We shall see that when we
consider rewriting edNCE grammars, the matching conditions will be
straightforward generalisations of the above matching conditions for graphs.

The final and most central notion that we will introduce in this section
is that of a DPO rewrite.

\begin{definition}[DPO Rewrite \cite{adhesive_categories}]
In any category $\mathbf{C}$, given a rewrite rule $t := L \xleftarrow{l} I
\xrightarrow{r} R,$ and a monomorphism $m: L \to H,$ we say $H$ \emph{rewrites}
to $M$ using $t$ over $m$, and denote it with
$H \leadsto_{t,m} M$
if the following diagram exists: 
\cstikz{dpo_adhesive.tikz}
and both squares are pushouts.
\end{definition}

We will also refer to DPO rewrites simply as \emph{rewrites} because this
will be the only kind of rewriting that we will consider. A rewrite consists
of applying a rewrite rule at a specified mono to a given object. Then,
a pushout complement (left square) is computed (if it exists) which may or may
not be unique. In an adhesive category, we already know that it is unique.
Following that, the second pushout (right square) is computed which may or may
not exist, however in an adhesive category it is guaranteed to exist. This can
be made precise via the following theorem.

\begin{theorem}\label{thm:adhesive_rewrite}
In any adhesive category $\mathbf{C}$, given a rewrite rule $t:= L
\xleftarrow{l} I \xrightarrow{r} R$ and a matching $m:L \to H$,
then $H \leadsto_{t,m} M,$ where $M$
is uniquely determined (up to isomorphism) by the following DPO diagram:
\cstikz{dpo_adhesive.tikz}
Moreover, all morphisms in the above diagram are mono.
\end{theorem}
\begin{proof}
$m$ satisfies the matching conditions therefore a pushout complement
of $(l,m)$ exists. From Lemma~\ref{lem:unique_pushout_complement} the
pushout complement is unique (up to isomorphism) and from
Lemma~\ref{lem:pushout_complement_monos} it follows that $k$ and $s$ are mono.
Then, by the first defining property of adhesive categories, the right pushout
square exists and therefore $M$ is uniquely determined (up to isomorphism).
Finally, Lemma~\ref{lem:mono_stable} shows that $f$ and $g$ are
also mono.
\end{proof}

Therefore, in any adhesive category, a rewrite rule and a matching
ensure that a rewrite can be performed and it is unique.
Next, we will consider two concrete examples of DPO rewriting in adhesive
categories -- one in the category of sets and one in the category of
multigraphs.

\subsection{Example: DPO rewriting in $\mathbf{Set}$}

As a concrete example, let's see how DPO rewriting works in the category of
sets and total functions $\mathbf{Set}$. In full generality, in $\mathbf{Set},$
the pushout of the two monos $l,k$ in the square below:
\cstikz{adhesive-pushout-shit.tikz}
is given (up to iso) by the set $H := L \sqcup K/\sim,$ where $(- \sqcup -)$ is
the disjoint union, $\sim$ is the finest equivalence relation which
relates $j_1\circ l(i) \sim j_2\circ k(i)$ for every $i \in I$, where $j_1,
j_2$ are the inclusions of the disjoint union. Therefore, $H$ is the quotient
of the disjoint union of $L$ and $K$ under the equivalence relation $\sim$.
Then, $m(x) = [j_1(x)]_{\sim}$ and $s(x) = [j_2(x)]_{\sim}.$

Given a pair of monos $H \xleftarrow m L \xleftarrow l I,$ then the pushout
complement of $(l,m)$ always exists and is given (up to iso) by the set
$K := H - m(L-l(I)),$ the subset inclusion $s: K \hookrightarrow H$ and the
function $k: I \to K$, given by $k := l$.

Let's look at a couple of specific examples.
Consider the following rewrite
rule:
\cstikz{set_rewrite_rule_example.tikz}
where both morphisms are just the subset inclusions. The elements of the
interface set we will call \emph{interface elements} or \emph{boundary
elements}. The elements of the LHS set which are not in the image of the
interface are called \emph{interior} elements.

In $\mathbf{Set}$, the matching conditions are trivial -- they are always
satisfied by any monomorphism. Therefore, this rewrite rule can be applied to
any set with cardinality three or higher. What the rule does is to preserve two
elements of the target set $(a,b)$, delete one other element $(c)$ and then add
two additional elements to it $(f,g)$. An application of the rule to the target
set $\{a,b,c,d,e\}$ is shown below:
\cstikz{set_rewrite_dpo.tikz}
where again the morphisms are the subset inclusions. Computing the pushout
complement can be done by removing the interior elements from the target
set. In this case, this amounts to removing the element $c$. Once that is
done, the right pushout is computed, which in this case is just the union
of the two sets. 

Of course, the morphisms do not have to be subset inclusions because DPO
rewriting is defined up to isomorphism. So, let's consider another case --
we will apply the same rewrite rule to the target set $\{a',b',c',f,g\}.$
The result is illustrated by the following DPO diagram:
\cstikz{set_rewrite_dpo2.tikz}
where the horizontal morphisms are subset inclusions and the vertical morphisms
all map $x \mapsto x'.$

In the rest of the thesis, we will be doing DPO rewriting on graphs and edNCE
grammars. As we have already pointed out, DPO rewriting on edNCE grammars is a
straightforward generalisation from DPO rewriting on graphs. The latter is, in
turn, a generalisation of DPO rewriting on sets. Indeed, sets can be seen as
totally disconnected graphs (graphs with no edges) and vice versa. In this
sense DPO rewriting on graphs fully generalises DPO rewriting on sets.
Nevertheless, DPO rewriting on sets acts in the same way as DPO rewriting
does (component-wise) on the set of vertices and set of edges of a graph
and also in the same way as DPO rewriting on the different components
of edNCE grammars as we shall see. Therefore, understanding DPO
rewriting on sets is crucial for the rest of the work as it is the most
fundamental mechanism for rewrites.
 
\subsection{Example: DPO rewriting in $\mathbf{MultiGraph}$}
\label{subsec:dpo-graph}

Next, let's consider DPO rewriting in the category of multigraphs.
In $\mathbf{MultiGraph}$, the pushout and
the pushout complement
are computed component-wise in the same way as in $\mathbf{Set}.$
The pushout of the two monos $g,k$ in the square below:
\cstikz{adhesive-pushout-shit2.tikz}
is given by the multigraph $H$ with components (vertices and edges) $V_H$ and
$E_H$
given by the pushouts in $\mathbf{Set}:$
\cstikz{sets-pushout-fuck.tikz}
and assigning functions (source and target) $s$ and $t$ given by:
\[
s_H(e) =
\begin{cases}
m_V\circ s_L(e') & \text{ if } e=m_E(e') \text{ for some } e' \in E_L\\
f_V\circ s_K(e') & \text{ if } e=f_E(e') \text{ for some } e' \in E_K
\end{cases}
\]
\[
t_H(e) =
\begin{cases}
m_V\circ t_L(e') & \text{ if } e=m_E(e') \text{ for some } e' \in E_L\\
f_V\circ t_K(e') & \text{ if } e=f_E(e') \text{ for some } e' \in E_K
\end{cases}
\]
and with labelling functions $l_V^H, l_E^H$ given by:
\[
l_V^H(v) =
\begin{cases}
l_V^L(v')& \text{ if } v=m_V(v') \text{ for some } v' \in V_L\\
l_V^K(v')& \text{ if } v=f_V(v') \text{ for some } v' \in V_K
\end{cases}
\]
\[
l_E^H(e) =
\begin{cases}
l_E^L(e')& \text{ if } e=m_E(e') \text{ for some } e' \in E_L\\
l_E^K(e')& \text{ if } e=f_E(e') \text{ for some } e' \in E_K
\end{cases}
\]
Given a pair of monos $H \xleftarrow m L \xleftarrow g I,$ the pushout
complement of $(g,m)$, exists iff $m$ satisfies the no dangling edges
condition (cf. Example~\ref{ex:graph_matching}). If that is the case, the pushout
complement is given (up to iso) by the full subgraph $K$ of $H$ with components
$X_K := X_H - m_X(X_L-g_X(X_I)),$ for $X \in \{V,E\},$ the subgraph inclusion
$f: K \hookrightarrow H$ and the graph homomorphism $k: I \to K$, where $k
= g$ set-theoretically.

Let's look at a couple of specific examples.
Consider
the following rewrite rule:
\cstikz{graph_rewrite_rule.tikz}
The effect of this rewrite rule is to remove one vertex $(v_3)$ and the
only other incident edge to it from a target multigraph and also to reverse
the direction of one other edge, provided a suitable matching can be
found. For example, we may apply this rule to the target multigraph:
\cstikz{graph_rewrite_target.tikz}
If we choose the matching which maps $v_i \mapsto v_i$ from the different
multigraphs, then we get:
\cstikz{graph_rewrite_dpo.tikz}
where again all
\replaced{morphisms}{morhpisms}
in the above diagram map $v_i \mapsto v_i$ and
the mapping of the edges is then uniquely determined.
However, the matching conditions in
$\mathbf{MultiGraph}$ are not trivial, as already pointed out.
In the example above they are satisfied, so
the above rewrite rule
may be applied to the specified target multigraph and then the result is the
multigraph in the bottom right corner of the diagram, which again is only
defined up to isomorphism.

In $\mathbf{Set}$ we cannot provide an example where a mono may fail
to produce a rewrite. However, we can do so in $\mathbf{MultiGraph}$. For
example,
consider the same rewrite rule, but this time let's take the target multigraph
$H$ given by:
\cstikz{graph_rewrite_target2.tikz}
where the matching maps $v_i \mapsto v_i$. Then, consider the following
diagram:
\cstikz{graph_rewrite_fail_dpo.tikz}
The no dangling edges condition is violated, because the target multigraph has
an edge outside of the image of $m$ which is incident to an
interior vertex $(v_3).$ Therefore, there exists no multigraph $K$ such that
the above square is a pushout. To understand why, recognize that
the offending edge must be in the image of $s: K \to H$ and therefore
both vertices $v_3$ and $v_4$ must be in the image of $s$ as well. However,
then vertex $v_3$ in $H$ must be in the images of both $m$ and $s$ while
it has no pre-image in $I$ and this clearly violates the pushout construction
on vertices.

DPO rewriting on multigraphs, and in particular, string graphs, is how we
formalize equational reasoning on string diagrams. String graphs represent
string diagrams, rewrite rules on string graphs represent rewrite rules between
string diagrams and a DPO rewrite on a string graph represents a string
diagram rewrite. By making use of (partially) adhesive categories, these
concepts can be generalised from graphs to edNCE grammars and
in this way we may model equational reasoning on entire families of
string diagrams and not just concrete diagrams.

\section{Partially Adhesive Categories}\label{sec:partial_adhesive}

In this section we will describe Partially Adhesive Categories. They were first
introduced by Kissinger and Duncan in \cite{partial_adhesive} and later in
\cite{kissinger_dphil} where they were slightly modified. Intuitively, a
partially adhesive category $\mathbf{C}$ is a category which embeds fully and
faithfully into an adhesive category $\mathbf{D}$, such that DPO rewriting in
$\mathbf{C}$ behaves in the same way as it does in $\mathbf{D}$, provided
that certain additional conditions are satisfied.

\added{As we have seen in the previous section, multigraphs form an adhesive
category. However, graphs, which are simply multigraphs without parallel edges
or self-loops, do not form an adhesive category, but a partially adhesive one.
Of course, DPO rewriting in the category of graphs behaves in exactly the same
way as it does in the category of multigraphs, provided that we avoid pushouts
which may establish parallel edges. This example is discussed in detail in
Subsection~\ref{subsec:simple_graphs}.}

In an adhesive category,
pushouts always exist over a span of monomorphisms. However, in partially
adhesive categories, this may not necessarily be the case. Instead, pushouts
over monomorphisms may exist provided that the span satisfies further
requirements. Similarly, the matching conditions in a partially adhesive
category
may be stricter compared to the matching conditions in the ambient category.

In the main chapters of this thesis, we will mostly be performing DPO rewriting
on categories which are not adhesive, but partially adhesive. In particular,
the graph model assumed by edNCE grammars does not form an adhesive category,
but a partially adhesive one. Moreover, edNCE grammars, which can be seen
as a generalisation of these graphs, also form a partially adhesive category
and not an adhesive one. Thus, partially adhesive categories provide
us with a nice framework where we can describe how DPO rewriting works and
how it relates to DPO rewriting in adhesive categories. Then, if we wish to use
DPO rewriting in our model categories, all we have to do in addition is to
identify the matching conditions and the conditions under which pushouts along
monomorphisms exist. Moreover, under these conditions, DPO rewriting
behaves in the same way as it does in the ambient adhesive category and we
can therefore make use of the already established results related to
adhesive categories.

We begin by introducing the formal definition for a partially adhesive category
and afterwards we shall introduce the rest of the notions which are needed to
understand rewrites within it. We will be following the presentation in
\cite{kissinger_dphil}.

\begin{definition}[Partially Adhesive Category]
A \emph{partially adhesive category} is a category $\mathbf{C}$ for which there
exists a full and faithful functor $\mathcal{S}: \mathbf C \to \mathbf D$,
where $\mathbf D$ is an adhesive category and $\mathcal{S}$
preserves monomorphisms.
\end{definition}

\begin{remark}
The above definition is slightly more general than the one presented
in \cite{kissinger_dphil} because we do not require the category
$\mathbf{C}$ to be a full subcategory of $\mathbf{D}$. The rest of the
definitions and proofs presented there are fully compatible with this
slight generalisation.
\end{remark}

The rest of the definitions in this section all assume that we are given
a partially adhesive category $\mathbf{C}$ and $\mathcal S: \mathbf{C}
\to \mathbf D$ which is the embedding functor into the adhesive category
$\mathbf D$.

\begin{definition}[$\mathcal{S}$-span and $\mathcal{S}$-pushout
\cite{kissinger_dphil}]
A span $L \xleftarrow{l} I \xrightarrow r R$ in $\mathbf C$ is called an
$\mathcal S$-\emph{span} if it has a pushout and that pushout is preserved by
$\mathcal S$. Such pushouts are called $\mathcal S$-\emph{pushouts}.
\end{definition}

So $\mathcal S$-pushouts are exactly those pushouts in $\mathbf C$ which are
preserved by $\mathcal S$. When doing DPO rewriting in partially adhesive
categories, we will limit ourselves to only these kinds of spans and pushouts.
We shall see that the specific partially adhesive categories we are interested
in
(those of graphs and edNCE grammars) have simple and easily decidable
conditions which characterise their $\mathcal S$-spans (and therefore $\mathcal
S$-pushouts) and the rest of the $\mathcal S$-diagrams which we will introduce
in this section.

\begin{definition}[$\mathcal{S}$-pushout complement \cite{kissinger_dphil}]
An $\mathcal S$-\emph{pushout complement} for a pair of arrows $(l,m)$ in
$\mathbf C$ is a pushout complement, where the following diagram is
an $\mathcal S$-pushout:
\cstikz{s_pushout.tikz}
\end{definition}

Similar in spirit to the previous definition, $\mathcal S$-pushout complements
are those pushout complements which are preserved by $\mathcal S$. Pushout
complements in $\mathbf C$ are not necessarily unique, but if we restrict
ourselves to $\mathcal S$-pushout complements, then they are unique as the
next lemma tells us.

\begin{lemma}\label{lem:partial_adhesive_unique}
If a pair of arrows $(l,m)$ in $\mathbf C$, where $l$ and $m$ are mono, have an
$\mathcal S$-pushout complement, then it is unique up to isomorphism.
Moreover, all morphisms in the pushout square are mono.
\end{lemma}
\begin{proof}
Assume the pair of arrows have two $\mathcal S$-pushout complements given by
arrows $(k,s)$ and $(k',s')$ as in the diagrams below:
\cstikz{s_pushouts_proof.tikz}
Then, by definition, the following two squares are pushouts in $\mathbf D$:
\cstikz{s_pushouts_proof2.tikz}
and moreover both squares are a pushout complement for the pair of
morphisms $(\mathcal S(l),\mathcal S(m)).$ $\mathcal S$ preserves monos and
therefore, by Lemma~\ref{lem:unique_pushout_complement} there exists an
isomorphism $i: \mathcal S(K)\to \mathcal S(K')$ making the following diagram
commute:
\cstikz{s_pushouts_proof3.tikz}
$\mathcal S$ is full and therefore, there exists an isomorphism
$i': K \to K'$ such that $\mathcal S(i') = i$. Then, by making use of the
fact that $\mathcal S$ is also faithful, we get that the following
diagram commutes:
\cstikz{s_pushouts_proof4.tikz}
which completes the proof for the first part of the lemma. The second
part follows from Lemma~\ref{lem:pushout_complement_monos} and the fact
that full and faithful functors reflect monomorphisms.
\end{proof}

This lemma illustrates nicely how we may use important properties of
adhesive categories when we suitably restrict the spans and the matchings
in our partially adhesive category. The additional restrictions which
may need to be imposed on the matching morphisms are formalised
(but not specified) in the next definition.

\begin{definition}[$\mathcal{S}$-matching \cite{kissinger_dphil}]
For a rewrite rule $L \xleftarrow l I \xrightarrow r R$ a
monomorphism $m: L \to H$ is called an $\mathcal S$-matching if
$(l,m)$ has an $\mathcal S$-pushout complement.
\end{definition}

Every $\mathcal S$-matching $m$ in $\mathbf C$ satisfies the matching
conditions and also $\mathcal S(m)$ is a matching in $\mathbf{D}.$
However, $m$ may need to
satisfy more conditions compared to those required by the matching conditions
in
$\mathbf C$, because we are interested in specific kinds of pushout
complements.

The next definition combines all of the notions introduced so far in
order to formalise a suitable notion of a rewrite in a partially adhesive
category.

\begin{definition}[$\mathcal{S}$-rewrite \cite{kissinger_dphil}]
\label{def:s-rewrite}
Let $t:= L \xleftarrow l I \xrightarrow r R$ be a rewrite rule and
$m: L \to H$ be an $\mathcal S$-matching. Let $K$ be the $\mathcal S$-pushout
complement of $(l,m)$. Then, if the right pushout square in the following
diagram exists and is an $\mathcal S$-pushout:
\cstikz{dpo_adhesive.tikz}
we say that $M$ is the $\mathcal S$-\emph{rewrite} of $t$ at $m$. We
will use the same notation as the one for a (DPO) rewrite
and write this as $H \leadsto_{t,m} M$ and it will be clear from context
which kind of rewrite we are referring to.
\end{definition}

Clearly, every $\mathcal S$-rewrite in a partially adhesive category is also
a rewrite, but not vice versa. In this thesis, we are only interested in
$\mathcal S$-rewrites, because they behave in the same way as they do in
their ambient adhesive category and we can therefore make use of the
established results. This is made precise by the next theorem.

\begin{theorem}\label{thm:partial_adhesive_rewrite}
Given an $\mathcal S$-rewrite as in Definition~\ref{def:s-rewrite}, then $M$ is
uniquely determined (up to isomorphism). Moreover, all morphisms in the
DPO diagram are mono and also the DPO diagram is preserved by $\mathcal S$
in the ambient adhesive category $\mathbf{D}$, thus $\mathcal S(H)
\leadsto_{\mathcal S(t), \mathcal S(m)} \mathcal S(M).$
\end{theorem}
\begin{proof}
Both pushout squares exist by definition of $\mathcal S$-rewrite. Because
they are both $\mathcal S$-pushouts, then clearly the DPO diagram is
preserved by $\mathcal S$ and therefore $\mathcal S(H)
\leadsto_{\mathcal S(t), \mathcal S(m)} \mathcal S(M).$ From
Theorem~\ref{thm:adhesive_rewrite} and the fact that full and faithful
functors reflect monomorphisms it follows that all the morphisms are mono.
Uniqueness of the rewrite in $\mathbf C$ follows from Lemma~\ref{lem:partial_adhesive_unique} and the fact that pushouts are also unique up to iso.
\end{proof}

Therefore, by restricting ourselves to $\mathcal S$-rewrites in a partially
adhesive category, many of the useful properties of adhesive categories
may be reflected into the partially adhesive one. For example, because
$\mathcal S$ is fully faithful, it will reflect all colimits and therefore
pushouts in particular. In this thesis, the most important feature which
gets reflected is the uniqueness of the pushout complement. However, the
framework of partially adhesive categories may be useful in determining whether
other important properties of adhesive categories also carry over, such as
the Local Church-Rosser Theorem and the Concurrency Theorem, for suitable
$\mathcal S$-diagrams. We will leave this question open for future work.

\subsection{Example: DPO rewriting in $\mathbf{Graph}$}
\label{subsec:simple_graphs}

Now, let's consider a concrete example of a partially adhesive category.
We already know that $\mathbf{MultiGraph}$ is adhesive, however if we restrict
ourselves to graphs without self-loops or
parallel edges with the same label, then we get a category which is not
adhesive, but partially adhesive.
These will be the kinds of graphs which we will be using throughout most of the
thesis, so we will simply call them \emph{graphs} for brevity and in order to
distinguish them from multigraphs. They may be defined in a more compact way
compared to multigraphs.

\begin{definition}[Graph \cite{c-ednce}]\label{def:graph}
A \emph{graph} over an alphabet of vertex labels $\Sigma$ and an alphabet of
edge labels $\Gamma$ is a tuple $H = (V, E, \lambda)$, where $V$ is a finite
set of nodes, $E \subseteq \{(v, \gamma, w) | v, w \in V, v \not= w, \gamma \in
\Gamma\}$ is the set of edges and $\lambda : V \to \Sigma$ is the vertex
labelling function.
\end{definition}

This is the notion of graph that we will be using when working with
context-free graph grammars. This definition is more compact compared to the
definition of multigraphs, because edges are uniquely identified by their
source
\replaced{vertex}{vertx},
target vertex and edge label, so the entire edge data may be
described as a subset of $V \times \Gamma \times V$. Multigraphs are strictly
more general and this is clearly impossible to do for them. The requirement
that our graphs do not have self-loops or parallel edges with the same label is
crucial for establishing some of the properties for edNCE grammars which we
will be using in later chapters.

Next, we define the category of graphs and graph homomorphisms and show that
it is partially adhesive, where its ambient adhesive category is
$\mathbf{MultiGraph}$.

\begin{definition}[Graph homomorphism]
Given two graphs $H, K$ over vertex and edge label alphabets $\Sigma$ and
$\Gamma$ respectively,
a \emph{graph homomorphism}
from $H$ to $K$ is a function $f: V_H \to V_K$, such that
if $(v, \gamma, w) \in E_H$, then $(f(v), \gamma, f(w)) \in E_K$
and for all
$v \in V_H,$ we have $\lambda_K(f(v)) = \lambda_H(v)$.
\end{definition}

\begin{definition}[Category of Graphs]
Given a vertex label alphabet $\Sigma$ and an edge label alphabet $\Gamma$,
then we will denote with $\mathbf{Graph}_{\Sigma, \Gamma}$ the category of
graphs and graph homomorphisms over $\Sigma$ and $\Gamma$. If the labelling
alphabets are clear from the context, then we will simply refer to it as
$\mathbf{Graph}$.
\end{definition}

\begin{theorem}
$\mathbf{Graph}$ is a partially adhesive category with
embedding functor $\mathcal S: \mathbf{Graph} \to \mathbf{MultiGraph}.$
\end{theorem}
\begin{proof}
The definition of $\mathcal S$ and the theorem 
follow as a special case of Theorem~\ref{thm:edNCEvsMultiEdNCE}.
\end{proof}

However, while $\mathbf{MultiGraph}$ is
adhesive, $\mathbf{Graph}$ is not. The reason is because
$\mathbf{Graph}$ does not have unique pushouts complements over monos,
which violates Lemma~\ref{lem:unique_pushout_complement}. In particular,
consider the following two squares:
\cstikz[0.85]{simple-graph-pushout.tikz}
where the morphisms are just subgraph inclusions. The left square
is an $\mathcal S$-pushout in $\mathbf{Graph}$ and therefore it is
a pushout in both $\mathbf{Graph}$ and $\mathbf{MultiGraph}$. However,
the right square is not a pushout in $\mathbf{MultiGraph}$, but it is a
pushout in $\mathbf{Graph}$. Therefore, the pair of morphisms
$(m, id)$ has two non-isomorphic pushout complements in $\mathbf{Graph}$
and thus, the category is not adhesive. For completeness, the pushout
of the right span in $\mathbf{MultiGraph}$ is given by:
\cstikz[0.85]{simple-graph-proper-pushout.tikz}
that is, it establishes a pair of parallel edges (with the same label), which
is not allowed
in $\mathbf{Graph}$. Therefore, $\mathcal S$-rewrites in
$\mathbf{Graph}$ are well-defined, however, general rewrites are not
because pushout complements and thus rewrites are not necessarily unique.

Next, we characterise the $\mathcal S$-spans (and therefore the $\mathcal
S$-pushouts).
\begin{lemma}\label{lem:s-span-simple-graph}
A span of monos $K \xleftarrow k I \xrightarrow r R$ in $\mathbf{Graph}$
is an $\mathcal S$-span iff the following condition holds:
\begin{description}
\item[ParEdges:] For any vertices $v,w \in I,$  if there exists an
edge $(k(v), \alpha, k(w)) \in K$ and an edge
$(r(v), \alpha, r(w))\in R,$ then there exists
an edge $(v, \alpha, w) \in I$.
\end{description}
\end{lemma}
\begin{proof}
This is a special case of Lemma~\ref{lem:s-spans} from the original body
of work.
\end{proof}

The main idea is that the above condition ensures that a span of monos
cannot produce a pushout in $\mathbf{MultiGraph}$ which contains a pair
of parallel edges with the same label. We proceed by characterising $\mathcal
S$-matchings
(and therefore $\mathcal S$-pushout complements).

\begin{lemma}\label{lem:s-matching-simple-graph}
Given a pair of monos $H \xleftarrow m L \xleftarrow l I$ in
$\mathbf{Graph}$, $m$ is an $\mathcal S$-matching iff $\mathcal S(m)$
satisfies the matching conditions (in $\mathbf{MultiGraph}$). Moreover,
if the $\mathcal S$-pushout complement exists, then it is given
(up to isomorphism) by the full subgraph of $H$ with components
$X_H - m(X_L-l(X_I))$, for $X \in \{V,E\}.$
\end{lemma}
\begin{proof}
Again, this is a special case of Lemma~\ref{lem:s-pushout_complement}
from the original body of work.
\end{proof}

So, this lemma tells us that in order to construct an $\mathcal S$-pushout
complement, we simply have to follow the same procedure as in
$\mathbf{MultiGraph}$.
Finally, the next theorem consolidates all of the requirements for doing
$\mathcal S$-rewrites.

\begin{theorem}\label{thm:dpo-simple-graphs}
In the category \textbf{Graph}, given a rewrite rule $t := L
\xleftarrow{l} I \xrightarrow{r} R$ and an $\mathcal{S}$-matching $m: L
\to H$, then the $\mathcal S$-rewrite induced by $m$ and $t$ exists iff
the following conditions are satisfied:
\begin{description}
\item[Edges:] For any two vertices $v,w \in I$, if there exist edges
$(m \circ l (v), \alpha, m \circ l (w)) \in H$ and
$(r(v), \alpha, r(w)) \in R,$ then there must be an edge
$(l(v), \alpha, l(w)) \in L.$
\end{description}
\end{theorem}
\begin{proof}
Again, special case of Theorem~\ref{thm:dpo-ednce} from the main
body of work.
\end{proof}

In summary, if we wish to
\replaced{perform}{perfrom}
DPO rewriting on graphs, such
that it acts in the same way as it does on multigraphs, then we simply
have to verify that the $\mathbf{Edges}$ condition is satisfied (in addition
to the matching conditions). Therefore, all the examples presented in
subsection~\ref{subsec:dpo-graph} also hold in $\mathbf{Graph}$ as they do
not violate the $\mathbf{Edges}$ condition.

DPO rewriting on graphs is just a special case of doing DPO rewriting
on edNCE grammars. The reason that edNCE grammars form a partially adhesive
category and not an adhesive one is precisely because their underlying
graph model are graphs and not multigraphs. If we define edNCE grammars using
multigraphs as the underlying graph model, then edNCE grammars form an
adhesive category. From the perspective of DPO rewriting this is preferable,
however this also has a negative effect on the language-theoretic properties of
edNCE grammars and for this reason we will use the standard definition
of edNCE grammars, which forms a partially adhesive category.

\section{String Diagrams and String Graphs}\label{sec:string}

In this section we begin by providing a short introduction to the theory
of string diagrams and how they relate to (traced symmetric) monoidal
categories.
This is described in Subsection~\ref{sub:string-diagrams}. Next, in
Subsection~\ref{sub:string-graphs} we describe string graphs, which are a
discrete representation of string diagrams amenable to automation in software.

\subsection{String Diagrams}\label{sub:string-diagrams}
String Diagrams were introduced by Roger Penrose in
\cite{penrose_tensor} as a way to conveniently represent Abstract Tensor
Systems.
A string diagram consists of a collection of labelled nodes which are drawn in
a plane. Every node has a certain number of inputs and outputs which are drawn
as wires. Nodes can be connected to each other (and even themselves) by
connecting their inputs and outputs. A single node with its inputs and outputs
represents an abstract tensor in Penrose's paper. Also, more
generally, nodes can
\deleted{also}
be used to represent arrows in monoidal categories
\cite{joyal_street}. A string diagram consisting of several nodes that are
connected in some way represents tensor contraction in Penrose's paper and
more generally, composition of morphisms in a monoidal category. It is
possible for a string diagram to contain wires that are not connected to a node
at some end and also wires that form circles. The length and shape of the
wires is irrelevant -- two string diagrams are equal if they are isotropic,
that is, one diagram can be obtained from the other via a continuous
deformation.
An example of two isotropic
string diagrams is provided below:
\begin{equation}\label{eq:string-isotropic}
\stikz{isotropic-laina.tikz}
\end{equation}
The introduction of string diagrams by Penrose was limited to their
applications for abstract tensor systems. The
first more general treatment of string diagrams was provided by
Andre Joyal and Ross Street in the middle of the 1980s. Perhaps the most
seminal paper related to string diagrams is \cite{joyal_street} where the
same authors show that reasoning with string diagrams is sound and complete
for (traced symmetric) monoidal categories. In particular, if an equality
between
two morphisms in some monoidal category follows from the axioms of
a monoidal category, then the string diagram representations of both
morphisms are isotropic. As a result, equational reasoning with string
diagrams provides shorter proofs compared to the standard term-based approach
used in monoidal category theory, because the canonical morphisms get
absorbed in the graphical notation and may be ignored.

For example, in any monoidal category $\mathbf{C}$, the interchange law:
\begin{equation}\label{eq:interchange-shit}
(g \circ f) \otimes (j \circ h) = (g \otimes j) \circ (f \otimes h)
\end{equation}
is provable
from the axioms of a monoidal category:
\begin{align*}
(g \circ f) \otimes (j \circ h) &= \otimes(g \circ f, j \circ h) &(\text{infix
notation})\\
&= \otimes ((g, j) \circ (f,h)) &(\text{composition in } \mathbf C \times
\mathbf C)\\
&= (\otimes (g, j)) \circ (\otimes (f,h)) &(\text{functoriality of } \otimes) \\
&= (g \otimes j) \circ (f \otimes h) &(\text{infix notation})
\end{align*}
However, using string diagrams, the interchange law is trivial:
\cstikz{interchange-govno.tikz}
because both sides of the equation are represented using the same (isotropic)
string diagram. Thus, reasoning for monoidal categories using string diagrams
can be done modulo the interchange law, and in general, modulo any equation
which follows from the axioms of monoidal categories. In this way, we
can produce shorter equational proofs compared to the traditional term-based
rewriting approach.

A single morphism $f: A_1\otimes \cdots \otimes A_n \to
B_1 \otimes \cdots \otimes B_m$ is represented as a string diagram: 
\cstikz{string-diagram-kur.tikz}
where the open-ended wires on the bottom are its \emph{inputs} and the
open-ended wires on the top are its \emph{outputs}.
The tensor product of $f$ and $g$ is represented by horizontal composition
of diagrams:
\cstikz{string-diagram-hor.tikz}
and the composition $g \circ f$ is represented by plugging the outputs
of $f$ into the inputs of $g$ (provided the two morphisms can be composed):
\cstikz{string-diagram-ver.tikz}
As a specific example, let's consider monoids.
In
a monoidal category $(\mathbf{C}, \otimes, I)$, a \emph{monoid} is a triple
$(A, m, e),$
where:
\begin{enumerate}
\item $A \in Obj(\mathbf{C})$ is an object of $C$
\item $m: A \otimes A \to A$ is a morphism of $C$
\item $e: I \to A$ is a morphism of $C$
\end{enumerate}
such that:
\begin{align*}
m \circ (id_A \otimes m) = m \circ (m \otimes id_A) \\
m \circ (e \otimes id_A) = id_A = m \circ (id_A \otimes e)
\end{align*}
Instead of reasoning about monoids in the usual notation using terms, we
may use string diagrams to do so. In particular, setting:
\cstikz{monoid_diagrams_example.tikz}
then the monoid axioms become:
\cstikz{monoid_diagrams_example2.tikz}
Equational reasoning for string diagrams is performed via subdiagram
substitution. As an example, given the following string diagram:
\begin{equation}\label{eq:string-diagram-host}
\stikz{string-govna-pisvane.tikz}
\end{equation}
we can apply one of the unitality axioms by matching one side of the
axiom into the diagram, cutting it out along the input/output wires of the
subdiagram and then replacing it with the other side of the unitality axiom
(the replaced subdiagram is indicated by the dashed box):
\begin{equation}\label{eq:string-replacing-crap}
\stikz{string-replacing-crap.tikz}
\end{equation}
In general, by using string diagrams in this way, we can reason equationally
about (traced symmetric) monoidal categories without having to perform any
rewrite
steps which follow from the monoidal structure of the category.

However, Joyal and Street describe string diagrams in terms of non-discrete
topological notions. One of our primary interests related to string diagrams
is the ability to perform computer-assisted proofs using them.
The formalisation provided by Joyal and Street is difficult to implement
in a theorem prover and for this reason we will consider a discrete
representation of string diagrams which makes use of graphs. This
representation is called String Graphs and it is introduced in the next
subsection.

\subsection{String Graphs}\label{sub:string-graphs}

String Graphs were originally introduced under the name "Open Graphs" in
\cite{open_graphs}, \cite{open_graphs1}. Since then, the theory has been
further developed and the term was changed to "String Graphs" in
\cite{kissinger_dphil}. String graphs were introduced in order to represent
string diagrams in a way that allows for efficient automation by computer
systems.

String graphs are a special kind of labelled directed graphs. Every vertex
of a string graph is either a wire-vertex or a node-vertex. Constructing a
string graph from a string diagram (and vice versa) is straightforward -- for
every node in the diagram we create a node-vertex and we replace the wires with
a sequence of wire-vertices which are connected by edges. Node-vertices should
be thought of as the main vertices of the graph, whereas wire-vertices are
only there to assist in the formalization of string diagrams as actual graphs.
Wires which are not connected at some end to a node, must end in a wire-vertex
in the string graph. The notion that two string diagrams are equal under
topological changes translates to wire-homeomorphism in string graphs~-- two
string graphs are equal if we can get one from the other by increasing and
decreasing the number of wire-vertices on certain wires. As an example, the
following two string graphs are equal and represent the string diagrams from
Equation~\eqref{eq:string-isotropic}:
\begin{center}
$
\begin{aligned}
\begin{tikzpicture}
	\node [style=nodev] (f) at (0, 0) {$f$};
	\node [style=wirev] (fh) at (0.3, 0.7) {};
	\node [style=wirev] (hg) at (1.65, 0.7) {};
	\node [style=nodev] (g) at (2, 0) {$g$};
	\node [style=nodev] (h) at (1, 1) {$h$};
	\node [style=wirev] (kh) at (1, 2) {};
	\node [style=nodev] (k) at (1, 3) {$k$};
	\node [style=wirev] (i1) at (0, -1) {};
	\node [style=wirev] (i2) at (2, -1) {};
	\node [style=wirev] (o) at (1, 4) {};
	\node [style=wirev] (c) at (3, 2) {};

	\draw [<-] (f) to (fh);
	\draw [<-] (fh) to (h);
	\draw [<-] (g)  to (hg);
	\draw [<-] (hg) to (h);
	\draw [<-] (h)  to (kh);
	\draw [<-] (kh) to (k);
	\draw [<-] (i1) to (f);
	\draw [<-] (i2) to (g);
	\draw [->] (o)  to (k);
	\draw [->] (3,1) to[out=0,in=0] (c);
	\draw [-] (c) to[out=180,in=180] (3,1);
\end{tikzpicture}
\end{aligned}
	\mbox{}\quad\quad=\quad\quad
\begin{aligned}
\begin{tikzpicture}
	\node [style=nodev] (f) at (0, 0) {$f$};
	\node [style=wirev] (fh) at (0.2, 0.6) {};
	\node [style=wirev] (fh1) at (0.4, 0.8) {};
	\node [style=wirev] (hg) at (1.65, 0.7) {};
	\node [style=nodev] (g) at (2, 0) {$g$};
	\node [style=nodev] (h) at (1, 1) {$h$};
	\node [style=wirev] (kh1) at (1, 2.33) {};
	\node [style=wirev] (kh2) at (1, 1.66) {};
	\node [style=nodev] (k) at (1, 3) {$k$};
	\node [style=wirev] (i1) at (0, -0.7) {};
	\node [style=wirev] (i11) at (0, -1.3) {};
	\node [style=wirev] (i111) at (0, -1.8) {};
	\node [style=wirev] (i2) at (2, -1.8) {};
	\node [style=wirev] (o) at (1, 4) {};
	\node [style=wirev] (c) at (3, 2) {};
	\node [style=wirev] (cl) at (2.65, 1.5) {};
	\node [style=wirev] (cr) at (3.3, 1.5) {};

	\draw [<-] (fh) to (fh1);
	\draw [<-] (fh1) to (h);
	\draw [<-] (f) to (fh);
	\draw [<-] (g)  to (hg);
	\draw [<-] (hg) to (h);
	\draw [<-] (kh1) to (k);
	\draw [<-] (kh2) to (kh1);
	\draw [->] (kh2) to (h);
	\draw [<-] (i1) to (f);
	\draw [->] (i1) to (i11);
	\draw [->] (i11) to (i111);
	\draw [<-] (i2) to (g);
	\draw [->] (o)  to (k);
	\draw [->, bend right] (cr) to (c);
	\draw [->, bend right] (c) to (cl);
	\draw [->, bend right] (cl) to (cr);
\end{tikzpicture}
\end{aligned}
$
\end{center}

We now proceed to give formal definitions for all notions related to
string graphs as they will be the fundamental objects in which we are
interested in.
\added{Recall that so far we have introduced two different graph models --
multigraphs (cf. Definition~\ref{def:multigraph}) and graphs (cf.
Definition~\ref{def:graph}). In this section we will define string graphs
assuming the multigraph model, because this is how they have been originally
introduced and because we wish to stay close to the background theory. In
addition, the multigraph model allows us to define self-loops on wire-vertices,
whereas the other model does not. This is discussed at the end of the section
where we explain why string graphs may be alternatively defined using graphs
(in the sense of Definition~\ref{def:graph}) without losing any expressivity
for practical purposes.}

We begin by introducing the labelling alphabets which we
shall use.

\begin{definition}[String Graph Alphabets]\label{def:terminal_alphabets}
Throughout this section and the rest of the thesis, we will be working with
edge and vertex labelled graphs. There are finitely many labels which are
split into several different alphabets. This is made precise by the
following definition:
\begin{description}
\item[1.] $\Delta$ is the \emph{alphabet of vertex labels}.
\item[2.] $\mathcal{N} 
\subseteq \Delta$ is the \emph{alphabet of node-vertex
labels} and $\mathcal{W} = \Delta - \mathcal{N}$ is the alphabet of
\emph{wire-vertex labels}.
\item[3.] $\Gamma$ is the \emph{alphabet of all edge labels}.
\end{description}
\end{definition}

Thus, each vertex label is either a node-vertex label or a wire-vertex
label. With this in place, we may now define string graphs.


\begin{definition}[String Graph \cite{icgt}]\label{def:string-graph}
A \textit{string graph} over an alphabet of vertex labels $\Delta$
and an alphabet of edge labels $\Gamma$ is a directed labelled multigraph whose
vertices are labelled by the set $\Delta$ and whose edges are labelled by the
set $\Gamma$, where vertices with labels in $\mathcal N$ are called
\textit{node-vertices} and vertices with labels in $\mathcal W$ are called
\textit{wire-vertices}, and the following conditions hold:
\begin{description}
\item[1.] there are no edges directly connecting two node-vertices
\item[2.] the in-degree of every wire-vertex is at most one and
\item[3.] the out-degree of every wire-vertex is at most one
\end{description}
The category of string graphs over $\Delta$ and $\Gamma$ is denoted by
$\mathbf{SGraph}_{\Delta, \Gamma}$ and it is the full subcategory of
$\mathbf{MultiGraph}_{\Delta, \Gamma}$ whose objects are string graphs. When
$\Delta$ and $\Gamma$ are clear from the context, we will simply write
$\mathbf{SGraph}.$
\end{definition}

The above definition talks about directed string graphs, but we can easily
modify it in order to define undirected string graphs. All of our proofs,
propositions and definitions in this thesis will be about directed graphs,
but they can easily be modified to cover undirected graphs as well. However,
for brevity, we will only stick to the directed case, which is also more
complicated. Still, we shall sometimes provide examples involving undirected
graphs (and graph grammars) because they may be simpler to illustrate certain
ideas.

We depict string graphs in the following way. Node-vertices and edges are
depicted in the same way we depicted vertices and edges for graphs.
Wire-vertices will be depicted using
black nodes which are smaller compared to node-vertices. While
we do allow wire-vertices to be labelled, we can usually ignore their labels
as they usually do not carry semantic meaning. For this reason, we will often
not depict the label of a wire-vertex.

\begin{example}\label{ex:string_graph_general}
 An example of a string graph with three node-vertices is provided below:
\cstikz{string_graph_example.tikz}
\end{example}

The category $\mathbf{SGraph}$ is not adhesive, because it does not have unique
pushout complements. The reason is the same as for the category of
graphs -- parallel edges are not allowed. In fact, the same counter-example
from Subsection~\ref{subsec:simple_graphs} may be used to show that
$\mathbf{SGraph}$ is not adhesive, where we use wire-vertices in particular.
However, $\mathbf{SGraph}$ is partially adhesive as shown by the next
proposition.

\begin{proposition}[\cite{kissinger_dphil}]
The category $\mathbf{SGraph}$ is partially adhesive, where $\mathcal I:
\mathbf{SGraph} \to \mathbf{MultiGraph}$ is the inclusion functor.
\end{proposition}
\begin{proof}
The inclusion functor is obviously full and faithful. Monomorphisms in
$\mathbf{SGraph}$ are just injective graph homomorphisms and therefore they
are preserved by $\mathcal I$.
\end{proof}

String diagrams may have inputs and outputs which can be used to compose
string diagrams between each other. The same notion also exists for string
graphs and is defined next.

\begin{definition}[Inputs, Outputs and Isolated Vertices
\cite{kissinger_dphil}]
A wire-vertex of a string graph $H$ is called an \emph{input} if it has no
incoming edges. A wire-vertex with no outgoing edges is called an \emph{output}.
If a wire-vertex has no incident edges, then it is both an input and an output
and it is called an \emph{isolated} wire-vertex. The set of inputs of $H$ is
denoted as $In_H$, the set of outputs as $Out_H$ and the \emph{boundary} of $H$
as $Bound_H := In_H \cup Out_H.$ We denote with $\emph{In}(H)$, $\emph{Out}(H)$
and $\emph{Bound}(H)$ the string graphs with no edges whose vertices are
respectively the inputs, outputs and boundary wire-vertices of $H$. 
\end{definition}

In the string graph from Example~\ref{ex:string_graph_general}, the rightmost
vertex is an isolated wire-vertex.
The rest of the inputs are the two top vertices and
the rest of the outputs are the two bottom vertices.

With this in place, we can now characterise the $\mathcal I$-spans
and thus $\mathcal I$-pushouts in $\mathbf{SGraph}$.

\begin{lemma}[\cite{kissinger_dphil}]\label{lem:string_pushout}
A span of monos $K \xleftarrow k I \xrightarrow r R$ in $\mathbf{SGraph}$
is an $\mathcal I$-span iff the following conditions hold:
\begin{description}
\item[1.] for all $v \in In_I$ at least one of $k(v)$ and $r(v)$ is an input
\item[2.] for all $v \in Out_I$ at least one of $k(v)$ and $r(v)$ is an output
\end{description}
\end{lemma}

When working with string graphs, we shall restrict ourselves to special
kinds of rewrite rules which allow us to correctly simulate string diagram
rewriting. They are presented in the next definition.

\begin{definition}[String Graph Rewrite Rule \cite{kissinger_dphil}]
\label{def:string-graph-rewrite-rule}
  A \textit{String Graph Rewrite Rule} is a span of monomorphisms
  $L \stackrel{l}{\longleftarrow} I \stackrel{r}\longrightarrow R$
  in $\mathbf{SGraph},$
  with the following properties:
  \begin{description}
    \item[P1] $L$ and $R$ do not have any isolated wire-vertices
    \item[P2] $In(L) \cong In(R)$ and $Out(L) \cong Out(R)$
    \item[P3] $I \cong In(L) + Out(L) \cong In(R) + Out(R)$
    \item[P4] The following diagram commutes :
  \end{description}
  \begin{center}
    \begin{tikzpicture}
  \node (L)  at (-4, 0) {$L$};
  \node (I)  at (0, 0) {$I$};
  \node (R)  at (4, 0) {$R$};
  \node (IL) at (-2, 2) {$In(L)$};
  \node (IR) at (2, 2) {$In(R)$};
  \node (OL) at (-2, -2) {$Out(L)$};
  \node (OR) at (2, -2) {$Out(R)$};

  \draw [->, dashed] (I) to node[above] {$l$} (L);
  \draw [->, dashed] (I) to node[above] {$r$} (R);
  \draw [->] (IL) to node[above] {$i$} (I);
  \draw [->] (IR) to node[above] {$i'$} (I);
  \draw [->] (OL) to node[above] {$j$} (I);
  \draw [->] (OR) to node[above] {$j'$} (I);
  \draw [<->] (IL) to node[above] {$\sim$} (IR);
  \draw [<->] (OL) to node[above] {$\sim$} (OR);
  \draw [left hook->] (IL) to (L);
  \draw [right hook->] (IR) to (R);
  \draw [left hook->] (OL) to (L);
  \draw [right hook->] (OR) to (R);
\end{tikzpicture}
  \end{center}
  where $i, j, i', j'$ are the coproduct inclusions. 
\end{definition}

In other words, a string graph rewrite rule 
$L \stackrel{l}{\longleftarrow} I \stackrel{r}\longrightarrow R$
is such that $L$ and $R$ have no isolated wire-vertices, $I$ consists
entirely of isolated wire-vertices and the inputs/outputs in $L$ and $R$ are
in bijective correspondence which is respected by the span monos.

\begin{lemma}\label{lem:string_complement}
Given a string graph rewrite rule $L \xleftarrow l I \xrightarrow r R,$ then
a monomorphism $m: L \to H$ satisfies the no dangling edges condition
iff $m$ is an $\mathcal I$-matching. Moreover, the
$\mathcal I$-pushout complement is given (up to isomorphism) by the full
subgraph of
$H$ with components $X_H - m(X_L-l(X_I))$, where $X \in \{V,E\}.$
\end{lemma}
\begin{proof}
$m$ satisfies the no dangling edges condition iff the pushout
complement of $H \xleftarrow m L \xleftarrow l I$ in $\mathbf{MultiGraph}$
exists (and is unique up to isomorphism). Let it be given by the following
square:
\cstikz{pushout_tashak.tikz}
The pushout complement in $\mathbf{MultiGraph}$ is the full subgraph
$K\subseteq H$ whose components are given by $X_H - m(X_L-l(X_I))$, where $X
\in \{V,E\},$ as shown in Subsection~\ref{subsec:dpo-graph}. Thus, to complete
the proof, we have to show that $K$ is a string graph. $H$ is a string graph
and $K$ is obtained from $H$ by removing some edges and vertices. This cannot
violate any of the three conditions from Definition~\ref{def:string-graph} and
therefore $K$ must be a string graph.
\end{proof}

\begin{theorem}
In the category $\mathbf{SGraph}$, given a string graph rewrite rule $t := L
\xleftarrow{l} I \xrightarrow{r} R$ and a monomorphism $m: L
\to H$, which satisfies the no dangling edges condition,
then the $\mathcal I$-rewrite induced by $m$ and $t$ exists
and is given by $H \leadsto_{t,m} M,$ where $M$ is uniquely determined
(up to isomorphism) by the following DPO diagram:
\cstikz{dpo_adhesive.tikz}
\end{theorem}
\begin{proof}
From Lemma~\ref{lem:string_complement}, we know that the pushout complement $K$
exists and is unique. Following Lemma~\ref{lem:string_pushout}, we have to show
that for an arbitrary wire-vertex $v \in I$, then $r(v)$ is an input in $R$ or
$k(v)$ is an input in $K$ (the case for outputs follows by the same arguments).
Assume the opposite, that is, neither $k(v)$, nor $r(v)$ is an input. Because
$r(v)$ is not an input, this means that $r(v)$ must be an output in $R$,
according to the definition of string graph rewrite rule. Again from the same
definition, it follows $l(v)$ must be an output and moreover $l(v)$ has
in-degree exactly one, because isolated wire-vertices are not allowed in $L$.
Now, observe that the incident edge of $l(v)$ in $L$ will be removed when
computing the pushout complement $K$. Thus, $k(v)$ is an input in $K$ and we
get a contradiction.
\end{proof}

This theorem shows that DPO rewriting for string graphs is well-defined and
behaves exactly as it does for multigraphs as long as we restrict ourselves
to string graph rewrite rules. However, this alone is not enough to represent
string diagrammatic rewriting. The last piece which is missing is the
notion of wire-homeomorphism. We will first define wires and then
we shall define what we mean by wire-homeomorphism.

\begin{definition}[Wire \cite{icgt}]\label{def:wire}
A \textit{wire} is a maximal connected subgraph of a string graph consisting
of only wire-vertices and at least one edge. There are three cases: (a) it
forms a simple directed cycle, which is called a \textit{circle}, (b) it is a
chain where one or both endpoints are connected to node-vertices, which is
called an \textit{attached wire}, or (c) it is a chain not connected to any
node-vertices, which is called a \textit{bare wire}.
\end{definition}

\begin{example}
One example for each of the three different kinds of wires:
\cstikz{wires.tikz}
\end{example}

\begin{definition}[Wire-homeomorphic string graphs \cite{icgt}]
\label{def:wire-homeo}
Two string graphs $H$ and $H'$ are called \textit{wire-homeomorphic}, written
$H \sim H'$ if $H'$ can be obtained from $H$ by either merging two adjacent
wire-vertices (top) or by splitting a wire-vertex into two adjacent
wire-vertices (bottom) any number of times, while preserving the number of
wires:
    \[ \stikz{two-wires.tikz} \ \ \mapsto \ \ \stikz{one-wire.tikz}\]
    \[\stikz{one-wire.tikz} \ \ \mapsto \ \ \stikz{two-wires.tikz} \]
\end{definition}

Therefore, two wire-homeomorphic string graphs differ only in the length
of their wires.
It's easy to see that wire-homeomorphism is an equivalence relation and we
shall denote the wire-homeomorphism class of a graph $H$ (modulo graph
isomorphism) by $[H]_{\sim}$. Recall, that for a graph $H$, $[H]$ denotes
the set of all graphs isomorphic to $H$. Thus, $[H] \subseteq [H]_{\sim}.$
Also, note that any wire-homeomorphism class $[H]_{\sim}$ has a minimal
representative -- it is given by the string graph $H' \in [H]_{\sim}$, such
that no two wire-vertices in $H'$ may be merged without decreasing the number
of wires.

\begin{example}
The following string graphs are wire-homeomorphic:
\cstikz{wire-homeo-example.tikz}
and the left one is the minimal representative of its wire-homeomorphism
class.
\end{example}

The main result regarding string graphs in \cite{kissinger_dphil} is that the
category of framed cospans of string graphs modulo wire-homeomorphism is the
same as the free traced symmetric monoidal category. This justifies the fact
that DPO
rewriting on string graphs modulo wire-homeomorphism correctly represents
string diagram rewriting. To understand why DPO rewriting without
wire-homeomorphism is not sufficient to represent string diagram rewriting,
consider the next example.

\begin{example}
We can now show how to represent the string diagram rewrite from
Equation~\eqref{eq:string-replacing-crap}. First, we represent the
unitality axiom which we use for the rewrite as a string graph rewrite rule:
\cstikz{unitality-string-rewrite.tikz}
Notice, that the interface of the rewrite rule is uniquely determined by
the LHS and the RHS of the string diagram equation, as it just consists
of the inputs and outputs of the diagram, represented as wire-vertices.
Next, we represent the string diagram \eqref{eq:string-diagram-host} as
a string graph:
\cstikz{string-diagram-host.tikz}
Then, replacing the subdiagram in \eqref{eq:string-replacing-crap} is done
via a DPO rewrite:
\cstikz{string-diagram-dpo-model.tikz}
Notice, that even though the string graph rewrite rule and the host string
graph are all the minimal representatives of their wire-homeomorphism classes,
the result of the rewrite is not.
\end{example}

So, in order to correctly represent string diagram rewriting using DPO
rewriting on string graphs, we have to consider two string graphs to be
equal when they are wire-homeomorphic. In particular, when matching a string
graph $L$ onto another string graph $H$, then we might have to grow or shrink
some wires in $H$, before we can perform a rewrite.
The next two definitions make this precise.

\begin{definition}[String Graph Matching \cite{icgt}]
Let $L \xleftarrow l I \xrightarrow r R$ be a string graph rewrite rule. Then a
\textit{string graph matching} of $L$ onto a string graph $H$ is a monomorphism
$m : L \to \widetilde H$ where $H \sim \widetilde H$ and where $m$ satisfies
the no dangling edges condition.
\end{definition}

\begin{definition}[String Graph Rewrite \cite{icgt}]
A \textit{string graph rewrite} of a string graph $H$ by a
string graph rewrite rule $L \xleftarrow l I
\xrightarrow r R$ using a string graph matching $m : L \to \widetilde H$ (for
$\widetilde H
\sim H$) consists of the following DPO diagram in
\textbf{SGraph}:
  \cstikz{dpo-squares.tikz}
\end{definition}

\begin{example}
In the example below, the embedding on the left fails, however the embedding on
the right of the same graph succeeds:
\cstikz{growing-crap.tikz}
while both graphs $H$ and $H'$ represent the same string diagram. Therefore,
it is necessary to grow the wire of the host graph $H$ if we wish to do a DPO
rewrite involving the string graph $L$ as the LHS of some string graph rewrite
rule.
\added{Note, that in this case a non-injective matching would also succeed,
however, the theory of string graphs and its correctness has only been
described over injective matchings. For an alternative representation of
string diagrams where this is not necessary, see
Section~\ref{sec:related-work-back}.}
\end{example}

Before we conclude this section, we point out that
we can alternatively define string graphs using Definition~\ref{def:graph}
instead of using multigraphs.
However, there is a small caveat. Observe, that any string graph may not have
parallel edges -- there can be no edges between a pair of node-vertices and the
in-degree (out-degree) of any wire-vertex is at most one. Node-vertices
cannot have self-loops, but a wire-vertex may have a self-loop, like
in the following example:
\cstikz{self-loop.tikz}
However, any string graph which contains a self-loop is wire-homeomorphic to
a string graph which contains no self-loops and is therefore a graph in the
sense of Definition~\ref{def:graph}. This can be done by simply removing any
wire-vertex with a self-loop and introducing a circle wire:
\cstikz{closed-wire-vertices.tikz}
String Graph rewriting is performed modulo wire-homeomorphism,
so this
limitation of the graph model is insignificant.
When working with graph grammars, we will be using this notion of graphs
and string graphs.
\added{This is because edNCE graph grammars do not allow for self-loops and
therefore we need to consider string graphs, not as multigraphs, but as graphs
in the sense of Definition~\ref{def:graph}.}

\section{Families of String Diagrams and !-graphs}\label{sec:family}

In this section we will introduce how to rewrite not just singular
string diagrams, but entire sets, or \emph{families}, of string diagrams. This
is described in Subsection~\ref{sub:string-family}. Next, in
Subsection~\ref{sub:bang-graphs}, we describe how we can formally and
finitely represent certain families of string diagrams using a generalisation
of string graphs, called \emph{!-graphs}.

\subsection{Infinite families of string diagrams}\label{sub:string-family}

In the previous section we explained how string diagrams can be used to
represent morphisms in monoidal categories. In particular, a single string
diagram represents a single morphism and by rewriting a string diagram we
establish an equality between two morphisms from the category. However, this
process can be
\replaced{limiting}{limitting},
as it does not immediately extend to infinite sets
of morphisms, which is necessary for some practical applications. For example,
(quantum) protocols and algorithms are usually described in terms which allow
for input of arbitrary size and therefore a complete description in terms of
string diagrams requires the ability to express \emph{infinite families}
(or sets) of string diagrams.

Let's consider a simple example of an algebraic structure where the ability
to reason about infinite families of morphisms can be beneficial. We saw
how we can represent monoids using string diagrams in the previous section.
We also showed how equational reasoning can be done for concrete morphisms.
Using simple inductive arguments we can prove propositions which establish
infinitely many equalities between pairs of morphisms which are related
in a certain way.

For example by using associativity of the monoid operation, we can prove that
for any $n \in
\mathbb N, n \geq 2$, an $n$-ary left-associative application of the monoid
operation $m$ is equal to an $n$-ary right-associative application of $m$.
Using terms this can be expressed as:
\[m \circ (m \circ (m \circ \cdots \circ (m \otimes id_A) \otimes id_A) \otimes
\cdots \otimes id_A) \otimes id_A\]
\[=\]
\[m \circ (id_A \otimes (m \circ (id_A \otimes (\cdots
(id_A \otimes m)\cdots))))
\]
Using string diagrams, the same proposition can be expressed more
elegantly:
\cstikz[0.88]{monoid_family_trees.tikz}
In both cases, we use the familiar $(\cdots)$ notation in order to indicate
that we are representing an infinite family of morphisms. In this
way, we establish an \emph{equational schema}, which represents infinitely
many equalities between pairs of morphisms -- for every specific choice
of $n$, we know that an $n$-ary left-associative application of $m$ is equal
to an $n$-ary right-associative application of $m$. For example, if $n=2$,
then we simply get the associativity axiom. If $n=3$, we get:
\cstikz{monoid_shit.tikz}
The usefulness of this sort of reasoning should be clear -- once we have
established an equational schema between two families of morphisms
then we can easily establish equalities between \emph{concrete}
morphisms by choosing appropriate instantiations (in this case, choosing
the same number of applications on both sides).

Our primary motivation is to design a mathematical framework which supports
equational reasoning for infinite families of string diagrams in a formal way
which can be automated in software tools. We have already seen that we can use
string graphs as a discrete representation of string diagrams. However, the
intuitive $(\cdots)$ notation is not precise enough for direct
implementation in software. Towards this end, we will
provide a short introduction to the theory of !-graphs which
\replaced{allows}{allow}
us to
formally represent and do equational reasoning about certain infinite classes
of string graphs.

\subsection{!-graphs}\label{sub:bang-graphs}

!-graphs (pronounced \textit{bang graphs}) were introduced in
\cite{pattern_graphs} under the name \textit{pattern graphs} and were later
renamed to their current name in \cite{merry_dphil} which contains the most
complete and detailed description of !-graphs.

A !-graph is a generalised string graph which allows us to represent infinite
families of string graphs in a formal way. In addition to wire and node
vertices, !-graphs also have specially marked
\deleted{full}
subgraphs called
\emph{!-boxes}. These subgraphs may be copied infinitely many times (while
preserving connection relations), thus allowing us to represent an infinite set
of string graphs in a finite way.

The motivation behind !-graphs is to remove some of the informalities in
expressing infinite families of rewrite rules and infinite families of graphs
using the $(\cdots)$ notation in order to ease the
development of software proof-assistants.

In this subsection, we will provide a short introduction to the theory of
!-graphs which will allow the reader to follow the rest of the presentation.
We will first define !-graphs as special kinds of graphs in
Subsection~\ref{sub:bang-graph-graph}, then in
Subsection~\ref{sub:bang-graph-language} we will show how !-graphs induce
infinite languages of string graphs.

\subsubsection{Graph-theoretic notions of !-graphs}\label{sub:bang-graph-graph}

We start by describing the labelling alphabets which we shall use.

\begin{definition}[!-Graph Alphabets]\label{def:bang_alphabets}
We will use the following labelling alphabets when working with !-graphs.
\begin{description}
\item[1.] $\Sigma$ is the \emph{alphabet of vertex labels}.
\item[2.] $\mathcal{N} 
\subseteq \Sigma$ is the \emph{alphabet of node-vertex
labels}
\item[3.] $\mathcal{W} \subseteq \Sigma$ is the alphabet of \emph{wire-vertex
labels}.
\item[4.] $\Sigma = \mathcal N \cup \mathcal W \cup \{!\},$ where
$!$ is a special label which we will use to label !-vertices.
\item[5.] $\Gamma$ is the \emph{alphabet of all edge labels}.
\end{description}
\end{definition}

For the rest of the section, all of our constructions will be over an
arbitrary !-graph alphabet as described in the previous definition. Any
vertex with label "!" will be called a \emph{!-vertex}. For a multigraph
$G$, the full subgraph of $G$ consisting only of !-vertices will be
denoted as $\beta(G)$.

So, !-graphs use the same labelling symbols as string graphs, except for
the newly introduced label "!", which marks !-vertices. We will depict
!-graphs in the same way as string graphs, with the exception that !-vertices
will be coloured in blue and will have a small square shape. In addition,
edges incident to !-vertices will also be coloured in blue.

\begin{example}\label{ex:bang-shit-notation}
A !-graph with two node-vertices, one wire-vertex and a single !-vertex.
\cstikz{bang_graph_tashatsi.tikz}
\end{example}

Next, we introduce the notion of \emph{open subgraph} which specifies
the subgraphs which may form a !-box in a sound way.

\begin{definition}[Open Subgraph \cite{gam}]
A subgraph $O$ of a string graph $H$ is said to be \textit{open} if it is a
full subgraph of $H$ and furthermore $\textit{In}(H - O) \subseteq
\textit{In}(H)$ and $\textit{Out}(H - O) \subseteq \textit{Out}(H)$.
\end{definition}

In other words, a subgraph $O$ of a string graph $H$ is open if it is not
adjacent to any wire-vertex in $H - O.$ Therefore, if we create any number
of copies of such a subgraph while preserving the edges connecting it to
the rest of $H$, then the result will again be a string graph, because we
cannot violate the conditions about the in-degree and out-degree of
wire-vertices.

\begin{example}\label{ex:open-shit}
Consider the following string graph, which is the same as
Example~\ref{ex:bang-shit-notation}, but with the !-vertex removed:
\cstikz{string_graph_open_shit.tikz}
There are five different open subgraphs. Two of them are trivial -- the empty
graph and the entire string graph itself. The remaining open subgraphs
are the full subgraphs with vertex sets given by $\{v_2\}, \{v_1,v_2\}$ or
$\{v_2,v_3\}.$ The subgraphs with vertex sets $\{v_1\}$ or $\{v_3\}$ are
not open, because they are adjacent to the wire-vertex $v_2$.
\end{example}

In !-graphs, the full subgraph consisting of only !-vertices must form
a partial order, or more precisely, it must form a posetal graph.

\begin{definition}[Posetal multigraph \cite{merry_dphil}]
A multigraph is called \emph{posetal}, if it contains at most one edge between
any two vertices and, when considered as a relation, forms a partial order.
\end{definition}

\begin{example}
The following graph is posetal:
\cstikz{posetal.tikz}
\end{example}

The next definition is useful as just a simple notational convenience.

\begin{definition}[Forgetful mapping \cite{merry_dphil}]
For every multigraph $G$, we define $U(G)$ to be the graph $G-\beta(G),$ that
is, the full subgraph of $G$ which contains no !-vertices.
\end{definition}

So, if $G$ is the graph from Example~\ref{ex:bang-shit-notation} and
$H$ is the graph from Example~\ref{ex:open-shit}, then $U(G) = H$. !-graphs
can be seen as generalised string graphs and then the mapping $U$ can
be seen as a forgetful functor from the category of !-graphs to the category
of string graphs. This should become clear after we formally define
!-graphs.

!-graphs represent infinite families of string graphs by allowing certain
subgraphs to be copied infinitely many times. For a given !-graph, the
specific subgraphs which may be copied in such a way are given by its !-boxes,
which are introduced next.

\begin{definition}[!-box \cite{merry_dphil}]
Given a
\replaced{multigraph}{mutligraph}
$G$ and a !-vertex $b$, its \emph{!-box}, denoted $B(b)$ is
the full subgraph of $G$ consisting of $b$ and all of its successors (that is,
all vertices which have an in-edge with source $b$).
For a given !-vertex $b$ we will also refer to the subgraph $B(b)$ as the
\emph{contents} of $b$.
\end{definition}

The next definition introduces !-graphs, which are the main
objects of study in this section.

\begin{definition}[!-graph \cite{merry_dphil}]
A \emph{!-graph} is a multigraph $G \in \mathbf{MultiGraph}_{\Sigma,
\Gamma},$ such that:
\begin{enumerate}
\item If $e \in G$ is such that $t(e)$ is a !-vertex, then $s(e)$ is
also a !-vertex.
\item $U(G)$ is a string graph.
\item $\beta(G)$ is posetal.
\item For every !-vertex $b\in G$, $U(B(b))$ is an open subgraph of $U(G).$
\item For any two !-vertices $b, b'\in G$, if $b' \in B(b)$ then
$B(b') \subseteq B(b).$
\end{enumerate}
\end{definition}

\begin{remark}
The above definition is equivalent to the one presented in \cite{merry_dphil},
where all node-vertex types are allowed to have arbitrary arity. In
particular, the first condition is not explicitly stated there, but it follows
from the slice construction which the author uses.
\end{remark}

Condition 1 from the definition above implies that there are no out-edges
from wire-vertices or node-vertices to !-vertices. Condition 2 then justifies
the fact that !-graphs can be seen as generalised string graphs. Condition
3 imposes a partial order on the !-vertices. Condition 5 then requires
that a !-box $B(b')$ must be entirely contained in any other
!-box $B(b)$ where $b$ is greater than $b'$ with respect to the partial order.
Finally, condition 4 implies that copying a !-box any number of times, while
preserving the connection relations with the rest of the graph, would result in
a valid !-graph.

The definition of a !-graph suggests a more compact graphical presentation
which we will be using from now on. We will depict the !-vertices as before,
but instead of explicitly depicting the outgoing edges of a !-vertex $b$, we
will simply draw a blue box around all of the vertices which are its
\replaced{successors}{sucessor}.
Thus the insides of such a blue box is simply the !-box $B(b)$. The !-vertex
itself will be depicted as one of the corners of the blue box.

\begin{example}\label{ex:better-notation}
Using this notation, Example~\ref{ex:bang-shit-notation} becomes:
\cstikz{bang-shit-better.tikz}
This presentation is easier to work with especially when we have multiple
!-boxes which intersect on some part of the !-graph:
\cstikz{bang-shit-better2.tikz}
\end{example}

\subsubsection{Languages of !-graphs}\label{sub:bang-graph-language}

Now that we have described the graph structure of !-graphs in sufficient
detail, we can show how each !-graph induces a language of \emph{concrete}
string graphs.

\begin{definition}[Concrete Graph \cite{gam}]
A !-graph with no !-vertices is called a \textit{concrete graph} or simply a
string graph.
\end{definition}

Given a !-graph $G$ with !-vertex $b$, we can produce new !-graphs from $G$ by
applying one of two operations (any number of times) to $b$. They are described
in the next definition.

\begin{definition}[!-box Operations]
Given a !-graph $G$ and a !-vertex $b$, the two !-box operations are
defined as follows:
\begin{itemize}
  \item $\textrm{KILL}_b(G) := G - B(b)$
  \item $\textrm{EXPAND}_b(G)$ is defined by the pushout of the following
  subgraph inclusions:
\end{itemize}
   \cstikz{expand-operation.tikz}
\end{definition}

\begin{theorem}
Let $G$ be a !-graph and $b\in G$ be a !-vertex. Then, $\textrm{EXPAND}_b(G)$
and $\textrm{KILL}_b(G)$ are also !-graphs.
\end{theorem}
\begin{proof}
This follows from Theorem 4.3.4 of \cite{merry_dphil} after recognising that an
EXPAND operation is simply a COPY operation followed by a DROP operation, as
defined in that work.
\end{proof}

From an operational point of view, applying a KILL operation to a !-box has
the effect that the entire !-box is removed from the graph.

\begin{example}
Consider the !-graphs from Example~\ref{ex:better-notation}.
Applying a KILL operation to the first one yields the following result:
\cstikz{bang-kill1.tikz}
The second !-graph in that example has two different !-boxes and there
are different ways we can apply a KILL operation:
\cstikz{bang-kill2.tikz}
\end{example}

Applying an EXPAND operation to a !-box creates a copy of the !-graph
in its contents and connects it to the rest of the graph in the same way
as the original.

\begin{example}
Consider the !-graphs from Example~\ref{ex:better-notation}.
Applying an EXPAND operation to the first one yields the following result:
\cstikz{bang-expand1.tikz}
The second !-graph in that example has two different !-boxes and there
are different ways we can apply an EXPAND operation:
\cstikz[0.95]{bang-expand2.tikz}
\end{example}

When working with !-graphs, we are always interested in the languages which
they induce. Therefore, the following definition is crucial.

\begin{definition}[!-graph Language \cite{gam}]
The \textit{language} of a !-graph $G$ is the set of all concrete string
graphs obtained by applying a sequence of !-box operations on $G$.
\end{definition}

\begin{example}\label{ex:bang-language}
The languages of the two !-graphs from the previous examples are given by:
\[
\left\llbracket \stikz{bang-govno.tikz} \right\rrbracket =
\stikz{bang_graph_semantics.tikz}
\]
and
\[
\left\llbracket \stikz[0.8]{bang-govno2.tikz} \right\rrbracket =
\stikz[0.8]{bang_graph_semantics2.tikz}
\]
Therefore, the first !-graph represents the following family of string
diagrams:
\cstikz{diagram-family-laino2.tikz}
whereas the second represents:
\cstikz{diagram-family-laino.tikz}
\end{example}

In later chapters, we will compare !-graphs with context-free graph grammars
and B-ESG grammars in terms of the expressiveness of their languages. In order
to provide a more detailed picture, we will introduce two subclasses of
!-graphs, which are characterised by the relationships between their !-boxes.

\begin{definition}[Nested !-boxes]
Given a !-graph $G$ and !-vertices $b$ and $b'$, if $b' \in B(b)$ then we say
that the !-box $B(b')$ is \emph{nested} inside the !-box
\replaced{$B(b)$}{$b$}.
\end{definition}

Note, that if $b'$ is nested inside $B(b),$ the definition of !-graph requires
that the contents of $b$ contain the contents of $b'$, that is, $B(b')
\subseteq B(b).$ The second !-graph from Example~\ref{ex:bang-language} has a
pair of nested !-boxes -- the !-box $B(b_2)$ is nested inside $B(b_1).$

\begin{definition}[Overlap and Trivial Overlap \cite{gam}]
Given a pair of non-nested !-boxes $b_1$ and $b_2$, we say that $b_1$ and $b_2$ are
\textit{overlapping} if $B(b_1) \cap B(b_2) \not =
\emptyset$.
  $b_1$ and $b_2$ \textit{overlap trivially} if $B(b_1) \cap B(b_2)$ consists
of only the interior of zero or more closed wires, where one endpoint is a
node-vertex only in $B(b_1)$ and the other is a node-vertex only in $B(b_2)$.
\end{definition}

\begin{example}[\cite{gam}]
We illustrate the different kinds of relationships between !-boxes using
this example:
\cstikz{bang_relationship.tikz}
$b_1$ and $b_2$ overlap trivially, because they only overlap on the interior of
a wire connecting node-vertices in different !-boxes. $b_3$ and $b_4$ overlap
non-trivially, because their intersection contains a node-vertex.
$b_5$ and $b_6$ overlap non-trivially because their intersection consists of
a wire-vertex which is not part of a wire whose endpoints are in distinct
!-boxes. For any other pair of !-boxes, their intersection is empty and
therefore they overlap trivially as well.

Also, there is no pair of nested !-boxes. In particular, observe that the
contents of $b_6$ are just the wire-vertex and the !-vertex $b_6$ itself.
However, $b_6 \not \in B(b_5),$ because $b_6$ is depicted as being outside
of the !-box $B(b_5)$.
\end{example}

\begin{definition}[\textbf{BGTO} \cite{gam}]\label{def:bgto}
The set of all languages which are induced by !-graphs is denoted \textbf{BG}.
A !-graph where any two non-nested !-boxes do not overlap is called a
\textit{!-graph with no overlap}. The set of all languages which are induced
by !-graphs with no overlap is denoted \textbf{BGNO}.  A !-graph where any two
non-nested !-boxes overlap trivially is called a \textit{!-graph with trivial
overlap}. The set of all languages induced by these !-graphs is denoted
\textbf{BGTO}.
\end{definition}

It is easy to see that $\mathbf{BGNO} \subsetneq \mathbf{BGTO} \subsetneq
\mathbf{BG}.$ Before we conclude this section, we point out a subtle, yet
very important, difference in the operational semantics between nested
and non-nested overlapping !-boxes. Consider the following two
!-graphs:
\cstikz{bang-pisna-mi.tikz}
The language of $H_1$ is given by:
\[L(H_1) = \stikz[0.9]{bang-ezik.tikz}\]
The language of $H_2$ is given by:
\[L(H_2) = \stikz[0.9]{bang-ezik2.tikz}\]

That is, if we view these string graph languages as the string diagram
languages they represent, then $L(H_1)$ is the language of all trees of depth
at most two, whereas $L(H_2)$ is the language of all balanced trees of depth
at most two. For every string graph $H \in L(H_1)$, every node-vertex connected
to the root via a wire has an arbitrary number of children. However, for every
string graph $H \in L(H_2)$, every two node-vertices connected to the root via
a wire have the same number of children. This is because for every
expansion of $b_1$ ($b_2$), a copy of the string graph in their intersection is
produced which is added to the contents of $b_2$ ($b_1$):
\cstikz{pisna-mi-jivota-veche.tikz}
As a result, if one of the !-boxes is expanded $n$ times and the other
$m$ times, then there would be $nm$ copies of the string graph in their
intersection.

\section{Graph Grammars}\label{sec:graph-grammars}
Graph Grammars have been developed as a generalization of the well-known
context-free grammars for strings. A graph grammar can be thought of as a
finite collection of productions which specify instructions on how to generate
(infinite) sets of graphs. In this section we will provide an introduction
to the theory of context-free graph grammars.  They can be split up into two
main groups -- Vertex Replacement (VR)~\cite{c-ednce} and
\replaced{Hyperedge}{Hypergraph}
Replacement (HR)~\cite{hr-grammars}.

\added{The VR grammars represent the \emph{connecting} approach to graph
transformation, whereas the HR grammars represent the \emph{gluing} approach to
graph transformation. The connecting approach is usually described in
set-theoretic terms and the gluing approach is usually described in algebraic
(or categorical) terms. In later chapters, we will model equational reasoning
using DPO rewriting which is an example of the gluing approach. As we have
seen, DPO rewrites may be understood categorically. At first glance, this
consideration makes HR grammars a more attractive candidate to represent
families of string graphs. As we show in Chapter~\ref{ch:context-free}, HR
grammars and VR grammars have the same expressive power on string graphs.
However, we argue that this expressive power is insufficient and we need to
introduce an extension in order to be able to represent important languages of
string graphs. VR grammars are strictly more expressive on (unrestricted)
graphs and there is a simple extension which allows them to represent the
languages we are interested in. However, the author does not know how to extend
HR grammars to serve the same purpose and for this reason we will be using VR
grammars for the rest of the thesis.}

\deleted{We will focus on VR grammars because of
their strictly greater expressive power and because there is a natural subclass
of VR grammars that has exactly the same expressive power as HR grammars on
graphs.}

We will follow the presentation in \cite{c-ednce} as it is the
standard and most comprehensive reference on VR grammars. In particular,
we will be working with confluent edNCE grammars (C-edNCE), which
are the largest class of deterministically confluent graph grammars.

\subsection{edNCE grammars}

Throughout this section and the rest of the thesis, we will be working with
edge and vertex labelled graphs. There are finitely many labels which are
split into several different alphabets. This is made precise by the
following definition and we will use the same alphabets throughout this work
unless otherwise noted.

\begin{definition}[Alphabets]\label{def:alphabets}
We will use the following alphabets:
\begin{description}
\item[1.] $\Sigma$ is the \emph{alphabet of all vertex labels}.
\item[2.] $\Delta \subseteq \Sigma$ is the \emph{alphabet of terminal
vertex labels}.
\item[3.] $\mathcal{N} 
\subseteq \Delta$ is the \emph{alphabet of node-vertex
labels} and $\mathcal{W} = \Delta - \mathcal{N}$ is the alphabet of
\emph{wire-vertex labels}.
\item[4.] $\Gamma$ is the \emph{alphabet of all edge labels}.
\item[5.] $\Omega \subseteq \Gamma$ is the \emph{alphabet of all final edge
labels}.
\end{description}
\end{definition}

In most of this thesis, we will have $\Gamma = \Omega$, so all edge labels
will be final. However, we will always need nonterminal vertex labels
for our grammars in order to
\replaced{perform}{perfrom}
derivations. In addition, because we
will be working with string graphs, the terminal vertex labels are
partitioned into two sets -- node-vertices and wire-vertices. In this
section this distinction between node-vertices and wire-vertices can be
ignored, but this difference will be significant in later chapters when we
start working with string graphs.

edNCE graph grammars are not defined over multigraphs, but over
graphs in the sense of Definition~\ref{def:graph}. So, this is the graph
model which we are working with in this section.
Next, we formally
introduce the notion of a graph language.

\begin{definition}[Graph Language \cite{c-ednce}]
  The set of all graphs over $\Sigma$ and $\Gamma$ is denoted by $GR_{\Sigma,
  \Gamma}$. The set of all graphs modulo graph isomorphism is denoted by
  $[GR_{\Sigma, \Gamma}]$. A \textit{graph language} is a subset of
  $[GR_{\Sigma, \Gamma}]$.
\end{definition}

So, a graph language is a potentially infinite set of graphs, where isomorphic
graphs are identified together. The graph grammars that we are interested in
will be generating graph languages. For simplicity, we shall simply refer to
graph grammars as \emph{grammars} and to graph languages as \emph{languages}
and it should be clear from the context when we are talking about graph
languages or string languages generated by string grammars.

The next concept we introduce is an extension to our notion of graph. An
\emph{extended graph} provides the necessary information on how a
specific graph can be used to replace a nonterminal vertex and connect it to
the local neighbourhood of the nonterminal vertex that is to be replaced.
\begin{definition}[Extended Graph \cite{c-ednce}]
\label{def:extended-graph}
  An \textit{Extended Graph} over $\Sigma$ and $\Gamma$ is a pair $(H,
  C),$ where $H \in GR_{\Sigma,\Gamma}$ is a graph and $C \subseteq \Sigma
  \times \Gamma \times \Gamma \times V_H \times \{in, out\}$. $C$ is called
  a \textit{connection relation} and its elements $(\sigma, \beta, \gamma, x,
  d)$ are called \textit{connection instructions}.
  The set of all extended graphs over $\Sigma$ and $\Gamma$ is denoted
  by $EGR_{\Sigma, \Gamma}$.
\end{definition}

\begin{remark}
In the literature, extended graphs are commonly referred to as \emph{graphs
with embedding}. However, in the next chapters we will be introducing the
notion of grammar embedding which is defined in terms of extended graph
embeddings. In order to avoid confusion between graphs with embedding and
embeddings of graphs or embeddings of graphs with embeddings we have simply
renamed the term. Other than this, the definition is the same as the one in
\cite{c-ednce}.
\end{remark}

The name for an extended graph is well-justified.
If an extended graph has no connection instructions (i.e. $C=\emptyset$),
then we can think of it as just an ordinary graph.
Before we explain the significance of the connection instructions, we will
first introduce a graphical notation for depicting extended graphs (and
their connection instructions). The notation is essentially the same as in
\cite{c-ednce} but it only differs in some cosmetic details.
The graph part is drawn as before. Around the graph, we will draw a
rectangular box. A connection instruction $(\sigma, \beta, \gamma, x, in)$ will
be depicted in the following way: we place a $\sigma$-labelled vertex outside of
the
box; the label $\beta$ is put next to an arc which we create from the
$\sigma$-labelled vertex to
the boundary of the box; $\gamma$ is then placed next to another arc which
goes from the endpoint of the previous edge to the vertex $x$ of the graph.
A connection instruction $(\sigma, \beta, \gamma, x, out)$ is drawn in
the same way, except that we reverse the direction of both arcs.

\begin{example}\label{ex:extended_graph}
Let's consider the following alphabets which we will use for the next few
examples. Let $\Delta = \{\sigma_1, \sigma_2, \sigma_3\},$ and $\Sigma =
\Delta \cup \{X,Y\}.$ That is, the nonterminal vertex labels are
$X$ and $Y$ and the rest of the labels are terminal. Let
$\Gamma = \Omega = \{\alpha, \beta, \gamma\}.$ So, we have three kinds
of edge labels, all of which are final. $\sigma_1$-labelled vertices
will be coloured in white, $\sigma_2$-labelled vertices will be coloured in
grey and $\sigma_3$-labelled vertices will be coloured in black.

Then,
the following extended graph $(D, C_D)$, where $V_D = \{v_1, v_2, v_3\}$,
$\lambda_D(v_1) = \sigma_1,$
$\lambda_D(v_2) = X,$
$\lambda_D(v_3) = \sigma_3,$
$E_D = \{(v_1, \beta ,v_3), (v_1, \beta, v_2), (v_3, \gamma, v_2)\}$ and
$C_D = \{
(\sigma_1, \alpha, \beta, v_1, in),
(\sigma_1, \alpha, \alpha, v_2, in),
(\sigma_1, \gamma, \alpha, v_2, out),
\}$
is depicted in the following way:
\cstikz{gg_graph_embedding.tikz}
However, when working with graph grammars, we are only interested in generating
graphs (and graph languages) up to isomorphism, so the names of the vertices
will often be irrelevant and we may omit them:
\cstikz{gg_graph_embedding2.tikz}
When there are two or more connection instructions of the form 
$(\sigma, \beta, \gamma_i, x_i, d),$ where only the $\gamma_i$ and
$x_i$ may vary, then we will sometimes depict them as having the same
source outside of the box with the arcs branching out after they cross
the perimeter of the box:
\cstikz{gg_graph_embedding3.tikz}
\end{example}

We can think of an extended graph as just a normal graph with some
additional information which describes how the graph should be substituted into
another graph while replacing a nonterminal vertex.
If we are given a mother graph $(H, C_H)$ with a nonterminal vertex
$v \in H$ and a daughter graph $(D, C_D)$, then if we replace the vertex
$v$ with $(D,C_D),$ we get a new extended graph obtained in the following way.
Every connection instruction $(\sigma, \beta, \gamma,x,in) \in C_D$ in
the daughter graph
means that for every $\sigma$-labelled vertex $w$ in the
mother graph for which there is a $\beta$-labelled edge going \textbf{in}to the
nonterminal vertex $v$ of the mother graph, then the
substitution process will establish a $\gamma$-labelled edge from $w$ to $x$.
This
should become more clear after referring to example \ref{ex:subst} which is
presented in a graphical form as well. The meaning for $(\sigma, \beta,
\gamma,x,out)$ is analogous.
Next, we provide a formal definition for the substitution operation, which
is the fundamental mechanism underlying derivations in edNCE grammars.

\begin{definition}[Graph Substitution \cite{c-ednce}]
  Let $(H, C_H), (D, C_D) \in EGR_{\Sigma,\Gamma}$ be two extended graphs,
  where $H$ and $D$ are disjoint. Let $v \in V_H$ be a vertex of $H$.
  The \textit{substitution} of $(D, C_D)$ for $v$ in $(H, C_H)$ is denoted by
  $(H, C_H)[v/(D,C_D)]$ and is given by the extended graph whose
  components are:
  \begin{align*}
    V \quad=\quad &(V_H - \{v\} )\cup V_D\\
    E \quad=\quad &\{(x,\gamma,y)\in E_H | x \not = v, y\not = v\} \cup E_D\\
    &\cup \{(w,\gamma,x)\ |\ \exists \beta \in \Gamma : (w,\beta,v) \in E_H
      ,(\lambda_H(w), \beta, \gamma, x,in)\in C_D \}\\
    &\cup \{(x,\gamma,w)\ |\ \exists \beta \in \Gamma : (v,\beta,w) \in E_H
      ,(\lambda_H(w), \beta, \gamma, x,out)\in C_D \}\\
    \lambda(x) \quad=\quad
      &\begin{cases}
	\lambda_H(x) & \text{if } x \in (V_H - \{v\})\\
	\lambda_D(x) & \text{if } x \in V_D
      \end{cases}\\
    C \quad=\quad &\{(\sigma,\beta, \gamma,x,d) \in C_H\ |\ x \not = v\}\\
    &\cup \{(\sigma,\beta, \delta,x,d)\ |\ \exists \gamma \in \Gamma :
    (\sigma,\beta, \gamma,v,d) \in C_H,(\sigma, \gamma, \delta,x,d)\in C_D\}
  \end{align*}
\end{definition}

So, the substitution process may use the connection instructions in order
to create new edges between the mother graph and the daughter graph. These
new edges will be called \emph{bridges}.

\begin{definition}[Bridges \cite{c-ednce}]
The edges of $(H, C_H)[v/(D,C_D)]$ that are established by the substitution
process, i.e., that are not in $E_H$ or $E_D$, are called \emph{bridges}.
\end{definition}

We will shortly introduce edNCE grammars, which
\replaced{perform}{perfrom}
derivations by
substituting nonterminal vertices with extended graphs. Therefore, it
is crucial to understand the substitution process and we will consider several
examples.

\begin{example}\label{ex:subst}
Let's use the same alphabets from Example~\ref{ex:extended_graph} and
the same extended graph $(D,C_D)$. We will use this
graph as a daughter graph, that is, we will be substituting it into a mother
graph, which we will call $(H,C_H).$
  $(H,C_H)$ is shown on the left and $(D,C_D)$ is on the right:
  \cstikz[0.95]{gg_subst_graph.tikz}

  Substituting $(D,C_D)$ for the nonterminal
  vertex $y$ with label $Y$ in the above graph yields the following result:
  \begin{center}
    \begin{tikzpicture}
	\begin{pgfonlayer}{nodelayer}
		\node [style=none] (0) at (-2.5, 1.25) {};
		\node [style=none] (1) at (-2.5, -3) {};
		\node [style=none] (2) at (5.25, -3) {};
		\node [style=none] (3) at (5.25, 1.25) {};
		\node [style=greynode] (4) at (-1.75, -2) {};
		\node [style=nodev] (5) at (-1.75, -1) {};
		\node [style=blacknode] (6) at (-1.75, 0.75) {};
		\node [style=nodev] (7) at (2.5, 0) {};
		\node [style=none] (8) at (2, -1) {$\beta$};
		\node [style=none] (9) at (-0.5, -2.25) {$\alpha$};
		\node [style=box] (10) at (0.5, -2) {$X$};
		\node [style=none] (11) at (4, -1) {$\beta$};
		\node [style=none] (12) at (2.5, -2.25) {$\gamma$};
		\node [style=blacknode] (13) at (4.5, -2) {};
		\node [style=none] (14) at (-1.5, -0.25) {$\beta$};
		\node [style=none] (15) at (-4.75, -1.25) {$(H,C_H)[y/(D,C_D)]:$};
	\end{pgfonlayer}
	\begin{pgfonlayer}{edgelayer}
		\draw [style=simple] (0.center) to (3.center);
		\draw [style=simple] (3.center) to (2.center);
		\draw [style=simple] (2.center) to (1.center);
		\draw [style=simple] (1.center) to (0.center);
		\draw [style=->] (5) to (6);
		\draw [style=->] (7) to (10);
		\draw [style=->] (7) to (13);
		\draw [style=->] (13) to (10);
		\draw [style=->] (10) to (4);
	\end{pgfonlayer}
\end{tikzpicture}
  \end{center}
  Note that, only one connection instruction is used to establish bridges.
  The other two are not used, because the edge connecting to the nonterminal
  vertex in $(H,C_H)$ is not of the appropriate type. In the following
  several examples, the graph on the left is the mother graph and the
  graph on the right is the result of substituting in $(D,C_D)$ as before:
  \cstikz[0.9]{big_subst.tikz}
  In the above example, we see that each connection instruction in
  $C_D$ is used to establish one new bridge, while the nonterminal vertex
  has only two incident edges.
  \cstikz[0.9]{big_subst2.tikz}
  In the above example, a single connection instruction is used to establish
  two bridges.
  \cstikz[0.9]{big_subst3.tikz}
  In this example, we see that connection instructions in the mother graph
  may also be used to establish new connection instructions in the
  resulting graph in the same way that edges are used.
\end{example}

An important property of the substitution operation is that it is associative.

\begin{lemma}[\cite{c-ednce}]
Let $K,H,D$ be three mutually disjoint extended graphs. Let $w$ be
a vertex of $K$ and $v$ a vertex of $H$. Then, $K[w/H][v/D] = K[w/H[v/D]].$
\end{lemma}

Associativity of graph substitution, together with confluence
ensure that the derivations which we will be
interested in may be adequately described using derivation trees in the
usual sense.

Next, we define the concept of an edNCE Graph Grammar. edNCE is an
abbreviation for \textbf{N}eighbourhood \textbf{C}ontrolled \textbf{E}mbedding
for \textbf{d}irected graphs with dynamic \textbf{e}dge relabelling. The
justification for this name is the following: derivations in an edNCE grammar
consist of performing graph substitution by replacing a nonterminal vertex
with an extended graph as specified by some production of the grammar. As
we have already seen, embedding an (extended) graph depends only on the
neighbourhood of the nonterminal vertex which is being replaced, hence
the \textbf{NCE} part of the abbreviation. Clearly, we are working with
directed graphs (hence the \textbf{d} letter in the abbreviation). Finally,
the \textbf{e} in the name stands for dynamic edge relabelling which means
that edges connected to a nonterminal vertex may get their labels changed
after the vertex is replaced with another graph.

\begin{definition}[edNCE Graph Grammar \cite{c-ednce}]
  An \textit{edNCE Graph Grammar} is a tuple $G = (\Sigma, \Delta, \Gamma,
  \Omega, P, S)$, where 
  \begin{itemize}
    \item $\Sigma$ is the alphabet of vertex labels
    \item $\Delta \subseteq \Sigma$ is the alphabet of terminal vertex labels
    \item $\Gamma$ is the alphabet of edge labels
    \item $\Omega \subseteq \Gamma$ is the alphabet of final edge labels
    \item $P$ is a finite set of productions
    \item $S \in \Sigma - \Delta$ is the initial nonterminal label
  \end{itemize}
  Productions are of the form $X \rightarrow (D, C)$, where $X \in \Sigma -
  \Delta$ is a nonterminal label and $(D, C) \in EGR_{\Sigma, \Gamma}$ is an
  extended graph.
For a production $p:= X \to (D,C)$, we shall say that the \emph{left-hand side}
of $p$ is $X$ and denote it with $lhs(p)$. The \emph{right-hand side} of $p$ is
the extended graph $(D,C)$ and we denote it with $rhs(p)$. Vertices
which have a label from $\Delta$ are called \emph{terminal vertices} and
vertices with labels from $\Sigma - \Delta$ are called \emph{nonterminal
vertices}. Edges with labels from $\Omega$ are called \emph{final} edges and
edges with labels from $\Gamma - \Omega$ are called \emph{non-final} edges. An
(extended) graph is called terminal if all of its vertices are
terminal.
\end{definition}

\begin{remark}
We will always work with grammars where all edges are final. This does
not result in a loss of expressive power.
We shall also
refer to $rhs(p)$ as the \emph{body} of the production and we shall refer to
$lhs(p)$ as the \emph{label} of the production.
\end{remark}

Instead of presenting grammars using set-theoretic notation, we will often
present them graphically as it is more compact and intuitive. We will use the
same notation as in \cite{c-ednce}. A grammar is simply a set of productions
with a designated starting nonterminal label. The graphical presentation
simply depicts each production of the grammar. Productions are extended graphs,
which we already know how to depict, together with a nonterminal production
label. We depict a production by just drawing its associated extended graph and
placing its production label at the top-left corner of the bounding frame.

\begin{example}\label{ex:complete_grammar}
Let's consider an example edNCE grammar. 
The grammar below generates the set of all undirected complete graphs $K_n$ and
will be referred to multiple times in this thesis. The edges and terminal
vertices all have the same type and so we don't depict them for simplicity.
\cstikz{gg_kn_grammar.tikz}
The above grammar has three productions, given by $S \to H_1$, 
$X \to H_2$ and $X \to H_3,$ where $H_1$ is the leftmost extended graph,
$H_2$ is the middle extended graph and $H_3$ is the rightmost extended
graph. This example is used to demonstrate the graphical depiction of
edNCE grammars, but we will soon turn our attention to their languages.
\end{example}

Before we introduce the notion of derivation in an edNCE grammar, we first
have to define production copies. In turn, they depend on the notion
of extended graph homomorphism which we define next.

\begin{definition}[Extended Graph homomorphism]
Given two extended graphs $(H,C_H), (K,C_K) \in EGR_{\Sigma,\Gamma}$, an
\emph{extended graph homomorphism} between $(H, C_H)$ and $(K,C_K)$ is a
function $f: V_H \to V_K$, such that $f$ is a graph homomorphism from
$H$ to $K$ and if  $(\sigma, \beta, \gamma, x, d) \in C_H$ then
$(\sigma, \beta, \gamma, f(x), d) \in C_K$. If, in addition, $f$ is a graph
isomorphism and $C_k = \{(\sigma, \beta, \gamma, f(x), d)\ |\ (\sigma,
\beta,\gamma,x,d) \in C_H\},$
then we
say that $f$ is an \emph{extended graph isomorphism}.
\end{definition}

Graph substitution is only defined for disjoint pairs of graphs. However,
edNCE grammars have finitely many productions which specify what graphs need
to be substituted in for nonterminal vertices. We are only
interested in graph languages defined up to isomorphism and for this reason
it is necessary to consider isomorphic copies of productions if we wish to have
arbitrarily long derivation sequences.

\begin{definition}[Production copy \cite{c-ednce}]
Two productions $X_1 \to (D_1,C_1)$ and $X_2 \to (D_2, C_2)$ are isomorphic if
$X_1 = X_2$ and $(D_1,C_1)$ is isomorphic to $(D_2,C_2)$ as an extended graph.
For an edNCE grammar  $G = (\Sigma, \Delta, \Gamma, \Omega, P, S)$,
we denote with $\emph{copy}(P)$ the infinite set of all productions that are
isomorphic to some production in $P$. An element of $\emph{copy}(P)$ is called
a \emph{production copy} of $G$.
\end{definition}

The next notion is crucial for understanding the operation of edNCE grammars.
It formalizes which kinds of graph substitutions are legal for a given
edNCE grammar. Like their context-free string counterparts, derivations for
graph grammars consist of a sequence of expanding nonterminals and replacing
them with the body of an applicable production.

\begin{definition}[Derivation \cite{c-ednce}]
For a graph grammar $G=(\Sigma,\Delta,\Gamma,\Omega,P,S)$ and extended graphs
$H, H' \in EGR_{\Sigma,\Gamma},$ let $v \in V_H$ be a vertex
in $H$ and $p: X\to(D,C) \in \emph{copy}(P)$ be a production copy of the
grammar, such that $H$ and $D$ are disjoint. We say $H \Longrightarrow_{v,p}
H'$ is
a \textit{derivation step} if $\lambda_H(v) = X$ and $H'=H[v/(D,C)]$. If $v$
and $p$ are clear from the context, then we write $H \Longrightarrow H'$. A
sequence of derivation steps $H_0 \Longrightarrow_{v_1,p_1} H_1
\Longrightarrow_{v_2,p_2} \cdots \Longrightarrow_{v_n,p_n} H_n$ is called a
\textit{derivation}. We write $H \Longrightarrow_* H'$ if there
exists a derivation from $H$ to $H'$. A derivation $H \Longrightarrow_*
H'$ is \emph{concrete}
if $H'$ is terminal.
\end{definition}

\begin{example}
  A derivation in the grammar from Example~\ref{ex:complete_grammar} for
  the complete graph on four vertices $K_4$ is given by:
  \begin{center}
    \begin{tikzpicture}
	\begin{pgfonlayer}{nodelayer}
		\node [style=box] (0) at (-11, -0.25) {$S$};
		\node [style=none] (1) at (-10, -0.25) {$\Longrightarrow$};
		\node [style=box] (2) at (-6.5, -0.25) {$X$};
		\node [style=circ] (3) at (-8.75, -0.25) {};
		\node [style=none] (4) at (-5.25, -0.25) {$\Longrightarrow$};
		\node [style=box] (5) at (-2, -0.75) {$X$};
		\node [style=circ] (6) at (-4.25, -0.75) {};
		\node [style=circ] (7) at (-4.25, -3.25) {};
		\node [style=circ] (8) at (-4.25, -2) {};
		\node [style=none] (9) at (-5.25, -2.75) {$\Longrightarrow$};
		\node [style=box] (10) at (-2, -3.25) {$X$};
		\node [style=circ] (11) at (-2, -2) {};
		\node [style=circ] (12) at (-4.25, -5.5) {};
		\node [style=circ] (13) at (-2, -4.25) {};
		\node [style=none] (14) at (-5.25, -5) {$\Longrightarrow$};
		\node [style=circ] (15) at (-2, -5.5) {};
		\node [style=circ] (16) at (-4.25, -4.25) {};
		\node [style=circ] (17) at (-4.25, 0.5) {};
	\end{pgfonlayer}
	\begin{pgfonlayer}{edgelayer}
		\draw [style=simple] (3) to (2);
		\draw [style=redu] (17) to (5);
		\draw [style=redu] (6) to (17);
		\draw [style=simple] (8) to (7);
		\draw [style=redu] (8) to (11);
		\draw [style=redu] (11) to (7);
		\draw [style=simple] (6) to (5);
		\draw [style=simple] (11) to (10);
		\draw [style=redu] (7) to (10);
		\draw [style=redu] (10) to (8);
		\draw [style=simple] (16) to (13);
		\draw [style=simple] (12) to (16);
		\draw [style=simple] (12) to (13);
		\draw [style=redu] (13) to (15);
		\draw [style=redu] (15) to (16);
		\draw [style=redu] (12) to (15);
	\end{pgfonlayer}
\end{tikzpicture}
  \end{center}
  where the bridges are coloured in red.
\end{example}

Now that we have defined derivations of an edNCE grammar, we are ready
to define its language. This is formalized by the next three definitions.
For a given edNCE grammar, a \emph{starting graph} is simply a graph with
no connection instructions, no edges and just a single vertex labelled with
the initial nonterminal label. A \emph{sentential form}, similar to the
context-free string case, is any graph (up to isomorphism) which can be derived
from the starting graph. Finally, the language of a grammar consists of all
terminal sentential forms modulo graph isomorphism.

\begin{definition}[Starting Graph \cite{c-ednce}]
  We say that $sn(S,z) \in EGR_{\Sigma,\Gamma}$ is a \textit{starting graph} if
  it has only one vertex given by $z$, its label is $S$, the graph has no edges 
  and no connection instructions.
\end{definition}

\begin{definition}[Sentential Form \cite{c-ednce}]
A \emph{sentential form} of an edNCE grammar
$G=(\Sigma,\Delta,\Gamma,\Omega,P,S)$ is a graph $H$ such that
$Sn(S,z) \Longrightarrow_* H$ for some $z$.
\end{definition}

\begin{definition}[Graph Grammar Language \cite{c-ednce}]
The graph language induced by a graph grammar
$G=(\Sigma,\Delta,\Gamma,\Omega,P,S)$ is given by:
\[L(G) := \{[H]\ |\ H \in GR_{\Delta,\Omega} \text{ and } sn(S,z)
\Longrightarrow_* H \text{ for some } z\}\]
where $[H]$ denotes the equivalence class of all graphs which are isomorphic to
$H$.
\end{definition}

\begin{example}
  The language of the grammar from Example~\ref{ex:complete_grammar} consists
  of all undirected complete graphs with $n \geq 2$ vertices.
\end{example}

We consider graph languages up to isomorphism. Therefore, in the definition
above, for the starting graph $sn(S,z)$, the choice of vertex name $z$ does not
matter. Only the initial nonterminal label is relevant. Because of this, we
will often refer to derivations in a grammar as $S \Longrightarrow_* H$.

\subsection{Subclasses of edNCE grammars}\label{sub:grammar-subclasses}

The edNCE graph grammars as introduced so far are not necessarily confluent.
The order of application of productions matters as it may produce different
results. This makes reasoning about their derivations (and thus languages)
considerably harder. Therefore, a natural subclass of edNCE grammars are those
grammars which are confluent. Indeed, confluent edNCE grammars have nicer
structural, decidability and complexity properties, their derivations may be
described via derivation trees and there is also a powerful logical
characterization of their languages. These grammars are called \emph{C-edNCE}
grammars, where the $\emph{C}$ stands for \emph{confluent}. C-edNCE grammars
are the largest class of deterministically context-free graph grammars which
have been studied. For this reason, this class is also commonly known as VR
(Vertex Replacement) graph grammars as it is the most general class of
context-free graph grammars which utilise vertex replacement (as opposed to
edge/hyperedge replacement).

The next definition formally introduces C-edNCE grammars by providing static
conditions which guarantee confluence. The conditions are static in the sense
that they can be easily checked (and decided) by examining the productions
of a grammar. We provide the definition for completeness, but it is not crucial
for understanding the rest of the thesis, because we will be working with
specific subclasses of C-edNCE grammars which enjoy even nicer properties.

\begin{definition}[C-edNCE grammar \cite{c-ednce}]
An edNCE grammar $G=(\Sigma,\Delta,\Gamma,\Omega,P,S)$ is \emph{confluent}
or a \emph{C-edNCE} grammar, if for all productions $X_1 \to (D_1, C_1)$ and
$X_2 \to (D_2, C_2)$ in $P$, all vertices $x_1 \in V_{D_1}$ and $x_2 \in
V_{D_2}$, and all edge labels $\alpha, \delta \in \Gamma,$ the following
equivalence holds:
\[\exists \beta \in \Gamma: (X_2, \alpha, \beta, x_1, out) \in C_1
\text{ and } (\lambda_{D_1}(x_1), \beta, \delta, x_2, in) \in C_2\]
\[\iff\]
\[\exists \gamma \in \Gamma: (X_1, \alpha, \gamma, x_2, in) \in C_2
\text{ and } (\lambda_{D_2}(x_2), \gamma, \delta, x_1, out) \in C_1\]
\end{definition}

This static definition guarantees the following dynamic property which we shall
refer to as \emph{dynamic confluence}. In general, dynamic properties will
describe some behaviour of the sentential forms of C-edNCE grammars while
static properties will describe their productions in decidable terms. In this
sense, the next proposition states that dynamic confluence is equivalent to
static confluence for edNCE grammars and we can therefore decide if any
edNCE grammar is confluent in both senses.

\begin{proposition}[Confluence \cite{c-ednce}]
  An edNCE grammar $G=(\Sigma,\Delta,\Gamma,\Omega,P,S)$ is
  confluent iff the following holds for every graph $H \in EGR_{\Sigma,
  \Gamma}$:
  if $H \Longrightarrow_{u_1,p_1} H_1 \Longrightarrow_{u_2,p_2} H_{12}$ and
  $H \Longrightarrow_{u_2,p_2} H_2 \Longrightarrow_{u_1,p_1} H_{21}$ are
  derivations of $G$ with $u_1 \not = u_2$, then $H_{12} = H_{21}$.
\end{proposition}

C-edNCE grammars are confluent, however when performing derivations, the
context around a nonterminal vertex may change. Because of this, C-edNCE
grammars in general do not have a normal form called \emph{context-consistency}
which significantly simplifies multiple proofs and constructions in this work.
For this reason, we will work with the largest subclass of C-edNCE grammars
which are known to have this normal form  (defined next) and we will leave the
problem of determining whether our results hold for C-edNCE grammars in general
for future work.

\begin{definition}[B-edNCE grammar \cite{c-ednce}]
\label{def:b-ednce}
An edNCE grammar $G=(\Sigma,\Delta,\Gamma,\Omega,P,S)$ is
\emph{boundary}, or a B-edNCE grammar, if, for every production
$X \to (D,C)$
\begin{description}
\item[1.] $D$ does not contain edges between nonterminal vertices
\item[2.] $C$ does not contain connection instructions of the form
$(\sigma, \beta, \gamma, x, d)$ where $\sigma$ is a nonterminal label
\end{description}
\end{definition}

Therefore, in any B-edNCE grammar, the sentential forms are such that
nonterminal vertices are never adjacent. It is easily seen that
B-edNCE grammars satisfy the defining conditions for C-edNCE grammars, but
not the other way around. Thus, they are confluent, but in terms of generative
power they are strictly less expressive compared to C-edNCE grammars.

An interesting proper subclass of B-edNCE grammars are linear edNCE grammars.
Like their context-free string
\replaced{counterparts}{coutnerparts},
they can have at most one
nonterminal in the body of each production. As a result, their derivation
trees are simply lines and their derivations are usually easy to understand.

\begin{definition}[LIN-edNCE grammar \cite{c-ednce}]
An edNCE grammar $G=(\Sigma,\Delta,\Gamma,\Omega,P,S)$ is \emph{linear},
or a LIN-edNCE grammar, if for every production $X \to (D,C)$, $D$ has at most
one nonterminal vertex.
\end{definition}

The sentential forms of B-edNCE and LIN-edNCE grammars may have arbitrarily
many terminal vertices which are connected to a nonterminal vertex. A natural
and proper subclass of B-edNCE where this is not possible is presented in the
next definition. The main idea for that subclass is that each neighbour of a
nonterminal vertex may be uniquely identified via the combination of its
vertex label and the label and direction of its edge connecting it to the
nonterminal.

\begin{definition}[Bnd-edNCE grammar \cite{c-ednce}]
A B-edNCE grammar $G=(\Sigma,\Delta,\Gamma,\Omega,P,S)$ is
\emph{nonterminal neighbourhood deterministic}, or a Bnd-edNCE grammar, if,
for every production $X \to (D,C)$, every nonterminal vertex $x \in V_D$,
all $\gamma \in \Gamma$ and all $\sigma \in \Delta:$
\begin{description}
\item[1.] $\{y \in V_D\ |\ (y,\gamma,x) \in E_D \text{ and } \lambda_D(y)=
\sigma\} \cup \{\beta \in \Gamma\ |\ (\sigma, \beta, \gamma, x, in) \in C\}$
is a singleton or empty
\item[2.] $\{y \in V_D\ |\ (x,\gamma,y) \in E_D \text{ and } \lambda_D(y)=
\sigma\} \cup \{\beta \in \Gamma\ |\ (\sigma, \beta, \gamma, x, out) \in C\}$
is a singleton or empty
\end{description}
\end{definition}

A proper subclass of Bnd-edNCE grammar are the \emph{apex} grammars, where
connection instructions are only allowed to terminal vertices. As a result,
any vertex connected to a nonterminal will have its final neighbourhood
determined after expanding each nonterminal it is connected to. These
grammars can generate only languages of bounded degree.

\begin{definition}[A-edNCE grammar \cite{c-ednce}]
An edNCE grammar $G=(\Sigma,\Delta,\Gamma,\Omega,P,S)$ is called \emph{apex}
or an A-edNCE grammar, if for every production $X \to (D,C)$ and every
connection instruction $(\sigma, \beta, \gamma, x, d) \in C$, $x$
and $\sigma$ are terminal.
\end{definition}

If an apex grammar is also linear, then we shall call it a LIN-A-edNCE
grammar. These grammars are considered to form the simplest class of
VR grammars in terms of structural and expressive properties.

We will denote the class of languages which can be generated by an X-edNCE
grammar by $\textbf{X-edNCE}$. Figure~\ref{fig:inclusion} describes the
relationship between the different subclasses of \textbf{VR} languages. All
inclusions are proper. \textbf{HR} stands for the class of languages
expressible by Hyperedge Replacement grammars and the result is true up to
encoding, because HR and VR grammars generate different types of (hyper)graphs.
Most constructions in this thesis will be built around B-edNCE grammars.

\begin{figure}[h]
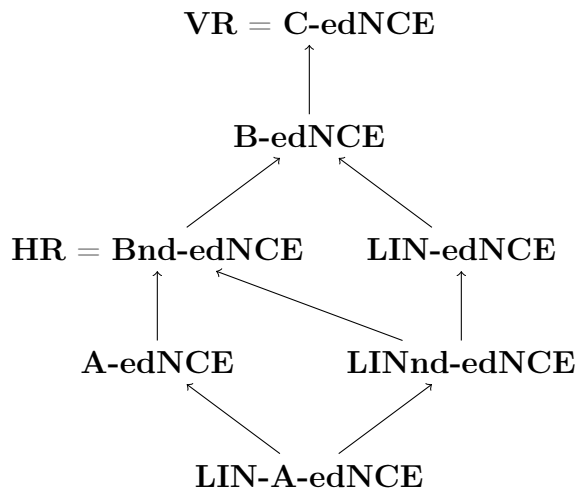

\cstikz{class_relationship.tikz}
\caption{Inclusion diagram for subclasses of VR languages}
\label{fig:inclusion}
\end{figure}

\subsection{Derivation Trees}\label{sub:derivation-tree}

Like in the case for context-free string grammars, the order in which
productions are applied in C-edNCE grammars does not matter. The result of two
derivation sequences where some of the productions have been permuted will be
the same graph. Thus, derivation trees for C-edNCE grammars are useful in
the same way that derivation trees are useful for context-free string
grammars -- they provide a convenient representation which allows us to
ignore the unimportant order of production applications and the trees also
better capture the syntactic structure of the objects (graphs or strings) with
respect to the grammar.

For the next two definitions we shall assume that the vertices in our
edNCE grammars have a linear order. This order may be assigned to the vertices
in an arbitrary way. It is only used in order to relate specific nonterminal
vertices from the productions of our grammars to specific nodes in a
derivation tree. We shall first provide the necessary definitions and then
elaborate more on them and provide examples.

\begin{definition}[p-labelled derivation tree]
Let $G$ be an edNCE grammar. A \emph{p-labelled derivation tree} of G is a
rooted, ordered tree $T$ of which the nodes are labelled by productions of $G$,
such that if node $u$ of $T$ has label $p$, and $rhs(p)$ has nonterminal
vertices $x_1,..., x_k$, in that order, then $u$ has children $v_1,...,v_k$, in
that order, and the left-hand side of the label of $v_i$ in $T$ equals the
label of $x_i$ in $rhs(p).$
\end{definition}

\begin{definition}[Yield of derivation tree]
The yield of a p-labelled derivation tree $T$ is defined recursively by
$yield(T) := rhs(p)[x_1/yield(T_1)]...[x_k/yield(T_k)],$ where $p$ is the label
of the root of $T$, $x_1,...,x_k$ are the nonterminal vertices (in that order)
in $rhs(p)$ and $T_1,...,T_k$ are the direct subtrees of $T$ (in that order),
where in addition, we take isomorphic production copies such that the
substitution is well-defined.
\end{definition}

\begin{remark}
The above two definitions are essentially the same as the ones from
\cite{c-ednce}. However, the authors there define p-labelled derivation
trees in terms of c-labelled derivation trees. We will not be
using c-labelled derivation trees and for this reason we provide the direct
definition above which is equivalent to the one in \cite{c-ednce}. From now
on, we shall simply refer to p-labelled derivation trees as \emph{derivation
trees}.
\end{remark}

From any derivation tree we can obtain a concrete derivation sequence and
vice-versa. Given a derivation tree, any tree traversal where each node
of the tree is visited only after its parent had been visited first describes
a concrete derivation sequence -- the first production to be applied is the
one identified by the root node and each following one is given by applying the
production identified by the next node in the traversal where the nonterminal
being replaced is the one identified by the connecting edge. In the other
direction, given a concrete derivation sequence, a derivation tree can be
constructed in a unique way -- the root node of the tree is labelled with the
name of the first production in the sequence, then the children of any node $n$
in the tree are labelled by the productions of the sequence which expand the
relevant nonterminal vertices in $rhs(p),$ where the label of $n$ is $p$.

\begin{example}
Consider the following A-edNCE (and thus B-edNCE) grammar which generates star
graphs (complete bipartite graphs $K_{1,n}$):
\cstikz{star_unoptimal.tikz}
where the names of the productions are shown on the left. Note, this
grammar is not optimal and the same language can be generated by a LIN-edNCE
grammar where the $Y$ nonterminal vertices and productions are removed.
However, this grammar is useful for illustrating derivation trees.
Let's assume that the nonterminal vertices $X$ have lower order than the
nonterminal vertices $Y$ in all productions of the grammar. Thus, when
expanding nonterminal vertices labelled $X$ they will correspond to the left
child in the derivation tree of a parent node.
For example, consider the following derivation tree:
\cstikz{derivation_tree_example.tikz}
The yield of this tree is the graph $K_{1,3}$:
\cstikz{star_example.tikz}
which may be computed by performing any of the derivation sequences induced by
the derivation tree. For example, a left-right depth-first traversal of this
derivation tree yields the following derivation sequence:
\cstikz[0.9]{star_derive.tikz}
Because this grammar is confluent, then any other derivation sequence
induced by this derivation tree will produce the same result, which is
the yield of the derivation tree.
\end{example}

Note, that in general, the yield of a derivation tree for an edNCE grammar
is not necessarily the same as the result of a derivation sequence induced
by the tree.
This is because
edNCE grammars in general are not confluent. However, we have restricted
ourselves to confluent grammars only and because of that the order of
production applications (and hence tree traversals) does not matter and they
will always produce the same result.

\subsection{Normal Forms}\label{sub:normal_forms}
In this subsection we present several normal forms for B-edNCE grammars,
\replaced{all}{most}
of
which are well known in the literature.
\replaced{Our contribution is to}{In addition we} point out that all of
them can be combined into a single normal form. Working with normal forms
with B-edNCE grammars allows for simpler reasoning and for more
concise definitions which we will need later on in this thesis. For example, in
Chapter~\ref{ch:besg}, we provide necessary conditions on normal form B-edNCE
grammars for certain properties. Stating all those conditions for arbitrary
B-edNCE grammars would be much more challenging and certainly less elegant.
However, the normal forms which we consider here do not restrict the expressive
power of B-edNCE grammars, and moreover, they can all be constructed
effectively by a computer.

\added{Our first normal form is about nonfinal edges. In an edNCE grammar
nonfinal edges are used in order to block subsequent derivations -- once a
nonfinal edge is established between a pair of terminal vertices, it cannot be
modified any further and thus any subsequent sentential forms are not
considered part of the language. This might sometimes help with grammar design,
but it also makes grammars much more difficult to reason about, so we shall
only consider grammars where nonfinal edges are not allowed at all. As the next
lemma (and its corollary) show, this does not decrease our expressive power.}

\begin{definition}[Nonblocking grammar \cite{c-ednce}]
An edNCE grammar $G$ is \textit{nonblocking} if every terminal
sentential form $H$ of $G$ has final edges only (i.e., $H \in L(G)$).
\end{definition}

\begin{lemma}[\cite{nonblocking, c-ednce}]\label{lem:nonblocking}
For every B-edNCE grammar an equivalent nonblocking B-edNCE grammar can be
constructed.
\end{lemma}

This was first proved in \cite{nonblocking} for C-edNCE grammars.
Another proof is presented in \cite{c-ednce} where, in addition, the authors
point out that the result also holds for B-edNCE grammars.

\begin{corollary}\label{cor:nonblocking}
For every B-edNCE grammar an equivalent B-edNCE grammar can be constructed
which doesn't use any non-final edges in any sentential form.
\end{corollary}

\begin{proof}
After applying the construction in Theorem \ref{lem:nonblocking}, simply
declare all edge labels to be final. This has no effect on the generated
language as all terminal sentential forms do not contain non-final edges.
\end{proof}

Next, we introduce the \emph{context} of a nonterminal vertex $v$. This
information describes the label data of the neighbourhood of $v$.

\begin{definition}[Context \cite{c-ednce}]\label{def:context}
For a vertex $v$ of an extended graph $H$, we define the \emph{context} of $v$
in $H$
  by:
\begin{align*}
cont_H(v) := \{(\lambda_H(w), \gamma, out)\ |\ (v,\gamma,w) \in E_H\} &\cup\\
  \{(\lambda_H(w), \gamma, in)\ |\ (w,\gamma,v) \in E_H\} &\cup\\
  \{(\sigma, \gamma, out)\ |\ (\sigma,\beta,\gamma,v,out)
                                 \in C_H\} &\cup\\
  \{(\sigma, \gamma, in)\ |\ (\sigma,\beta,\gamma,v,in)
                                 \in C_H\} &
\end{align*}
\end{definition}

\added{The next normal form which we will introduce is about context consistent
grammars. In a context consistent grammar, the context around every nonterminal
vertex with a given label is always the same. This will be helpful for us in
later chapters, because it will allows us to reason locally about the
productions of grammars when showing some of the admissibility results that we
need.}

\begin{definition}[Context consistent grammar
\cite{c-ednce}]\label{def:context-consistent}
We say that an edNCE grammar $G$ is \emph{context consistent} if there is a
mapping $\eta$ from $\Sigma - \Delta$ to the set of all possible contexts,
such that for every nonterminal vertex $v$ of every sentential form $H$ of $G$,
$cont_H(v) = \eta(\lambda_H(v)).$
\end{definition}

Thus, in a context consistent grammar, the context of a nonterminal vertex
is determined by its label.

\added{The next normal form which we will need is about neighbourhood
preserving grammars.
These grammars do not allow edges incident to nonterminals
to be dropped during the embedding process. Any vertex incident to a
nonterminal vertex in the mother graph will therefore be incident to at least
one new vertex created by the daughter graph. This normal form will be
useful in later chapters when we need to relate the inputs and outputs
of a pair (or triple) of B-ESG grammars. In particular, it guarantees that
any wire-vertex in any sentential form of such a grammar cannot
subsequently become an input or output, if it wasn't one in the past.}

\begin{definition}[Neighbourhood preserving grammar
\cite{c-ednce}]\label{def:neighbour}

We say that an edNCE grammar $G$ is \emph{neighbourhood preserving} if, for
every sentential form $H$ of $G$, every nonterminal vertex $v$ of $H$, and
every production $p$ with $lhs(p) = \lambda_H(v)$, the following holds:
if $(w,\beta, v) \in E_H\ (\text{or } (v, \beta, w) \in E_H)$, then
$rhs(p)$ has a connection instruction of the form
$(\lambda_H(w), \beta, \gamma, x, d)$ with $d=in$ ($d=out$, respectively).
\end{definition}

\begin{lemma}[\cite{compare_bgg}]\label{lem:context-neighbour}
For every B-edNCE grammar an equivalent B-edNCE grammar can be constructed
which is both context consistent and neighbourhood preserving.
\end{lemma}

\added{The next two normal forms have a direct analogy in the standard
context-free grammars on strings. Any derivation step in a graph grammar (or
string grammar) which contains no chain productions or empty productions will
always add at least one terminal vertex (or terminal character) to its
sentential forms. As a result, the number of terminal vertices (characters)
in sentential forms always increases which makes it easy to decide some
properties by using a bounding argument.}

\begin{definition}[Chain production\cite{c-ednce}]\label{def:chain_production}
Given an edNCE grammar, we say that a production $X \to (D,C)$ is a
\emph{chain} production, if $D$ consists of a single nonterminal vertex.
\end{definition}

\begin{definition}[Empty production\cite{c-ednce}]\label{def:empty_production}
Given an edNCE grammar, we say that a production $X \to (D,C)$ is an
\emph{empty} production, if $D$ is the empty graph.
\end{definition}

In the literature, there's a well-known normal form for C-edNCE grammars where
all empty productions can be removed (see \cite{c-ednce}). Obviously, if the
empty graph is part of the language of a grammar $G$,
then we cannot obtain a normal form for $G$ which contains no empty
productions. For this reason, we will consider a normal form where all empty
productions, except for possibly an initial one, are removed.

\added{The next normal form also translates immediately into the standard
string grammars. A reduced grammar is one where every production may be
executed at least once in some derivation. That is, there exists no production
which is not reachable from an initial one. Such unreachable productions have
no effect on the generated language and may simply be ignored. This normal
form will be useful for us when we characterise some necessary conditions for
our grammars that ensure they generate only languages of string graphs.}

\begin{definition}[Reduced grammar\cite{c-ednce}]\label{def:reduced}
An edNCE grammar is called \emph{reduced} if every production of the grammar
participates in at least one concrete derivation.
\end{definition}

In order to determine whether an edNCE grammar is reduced, we simply have to 
check if every production of the grammar is reachable from an initial
production.

\added{The final normal form which we introduce is about grammars which contain
no useless connection instructions.  As the name suggests, useless connection
instructions can simply be removed from any edNCE grammar. Of course, this
requires identifying those connection instructions first. It's not immediately
obvious how this can be done for arbitrary B-edNCE grammars, but we shall see
that this is easy under some of the conditions of the already mentioned normal
forms. B-ESG grammars which contain no useless connection instructions are
easier to reason about when it comes to relating inputs and outputs in B-ESG
rewrite rules.}

\begin{definition}[Useless connection
instruction\cite{c-ednce}]\label{def:useless}
For a given edNCE grammar, we say that a connection instruction of some
production is \emph{useless}, if for all possible concrete derivations,
the connection instruction is never used by the substitution process in order
to establish a bridge.
\end{definition}

The central theorem of this section simply states that all of the definitions
concerning B-edNCE grammars provided here can be combined into one normal form.

\begin{theorem}\label{thm:normal-form}
For every B-edNCE grammar $G$, we can effectively construct a B-edNCE grammar
$G'$ with $L(G) = L(G'),$ such that $G'$:
\begin{itemize}
  \item[1.] contains no non-final edges
  \item[2.] is context consistent
  \item[3.] is neighbourhood preserving
  \item[4.] contains no empty productions
  \item[5.] contains no chain productions
  \item[6.] is reduced
  \item[7.] contains no useless connection instructions
\end{itemize}
\end{theorem}
\begin{proof}
Condition 1 follows from Corollary~\ref{cor:nonblocking}. The proof of
Lemma~\ref{lem:context-neighbour} doesn't introduce any new non-final edges,
therefore conditions 1-3 can be combined.

An equivalent grammar without empty productions or chain productions can be
effectively constructed. For a proof, see \cite{c-ednce}. The construction for
removing empty productions is a straightforward generalisation to the one for
context-free string grammars -- in the productions of the grammar substitute
the empty graph for some nonterminal vertices that can generate the empty graph
and finally delete all empty productions, except for possibly an initial one.
This doesn't violate any of the conditions 1-3, because the replaced
nonterminal vertices must have empty context. The construction for removing
chain
productions is also analogous to the one for context-free string grammars -- it
consists of composing several productions of the original grammar into one and
then adding that production to the grammar. This composition preserves
properties 1-4.

Creating an equivalent grammar which is reduced simply involves removing all
productions which are unreachable from an initial production. This doesn't
have an effect on the generated language, because these productions can
obviously never be used in any concrete derivation. Therefore conditions 1-6
can be combined.

Finally, it should be clear that we can remove useless connection instructions
without changing the language of the grammar, because by definition they
do not affect any derivation. Identifying the useless connection instructions
for a grammar which satisfies conditions 2-3 is particularly easy -- for a
production $X \to (D,C)$ the useless connection instructions are of the form
$(\sigma, \beta,\gamma, x, d) \in C$, where $(\sigma, \beta, d) \not \in \eta(X).$
Obviously, doing this doesn't violate any of the previously established
properties 1-6, which completes the proof.
\end{proof}

We will be using this normal form for several proofs, so we give it a name.

\begin{definition}[Combined normal form]\label{def:combined}
We say that a B-edNCE grammar $G$ is in \emph{combined normal form} (CNF), if
the seven conditions of Theorem~\ref{thm:normal-form} are satisfied.
\end{definition}

\subsection{Monadic Second Order Logic for C-edNCE grammars}\label{sub:mso}
Monadic Second Order (MSO) logic provides a useful and convenient language for
describing graphs and graph properties. C-edNCE languages are fully
characterised by MSOL in a grammar-independent way. In particular, the class of
C-edNCE languages is the class of languages obtained by applying a MSO
definable function to the set of all trees over a ranked alphabet. These
results have proven to be very useful and they have led to several important
decidability and closure results.

We will provide a brief introduction to MSO logic and
state some properties which are relevant for C-edNCE grammars. For our
purposes, only a few of the relevant notions are required in order to be able
to understand the proofs in the rest of the chapters. The presentation
of MSO logic presented below follows that of \cite{c-ednce}.

For fixed alphabets $\Sigma$ and $\Gamma$ of vertex labels and edge labels
respectively, we will define the (infinite) language of MSO formulas and
denote it with $\emph{MSOL}(\Sigma, \Gamma)$. MSO formulas have two types of
variables. The first type are vertex variables, which we will denote with lower
Latin letters such as $u, v, w,$ etc. For a given graph $H \in
GR_{\Sigma,\Gamma}$, the vertex variables will range over the vertices $V_H$ of
$H$. The second type of variables are vertex-set variables, which we will
denote with capital Latin letters, such as $U, V, W,$ etc. For a given graph $H
\in GR_{\Sigma,\Gamma}$, the vertex-set variables will range over subsets of
$V_H$. Edge variables or edge-set variables are not allowed. Each MSOL
formula expresses some property of a graph $H \in GR_{\Sigma, \Gamma}$.

There are four different types of atomic formulas in MSOL, where the
meaning is given in parenthesis:
\begin{align*}
&\emph{lab}_{\sigma}(u) &( \text{vertex } u \text{ has label } \sigma, \text{
where } \sigma \in \Sigma) \\
&\emph{edge}_{\gamma}(u,v) &( \text{there exists an edge } (u, \gamma, v),
\text{ where } \gamma \in \Gamma) \\
&u=v &( \text{vertex } u \text{ is the same as vertex } v) \\
&u \in U &( \text{vertex } u \text{ is in the set } U) \\
\end{align*}

The rest of the formulas in $\emph{MSOL}(\Sigma, \Gamma)$ are built
from the four different types of atomic formulas by using
the standard propositional connectives $\land, \lor, \lnot, \implies, \iff$
and the standard quantifiers $\forall, \exists$ in the obvious way with
the usual meaning.
In particular, the logic is called \emph{Monadic Second
Order}, because we allow quantification over not just vertex variables, but
also over vertex-set variables.

A formula $\phi \in \emph{MSOL}(\Sigma, \Gamma)$ is \emph{closed} if it has
no free variables. A graph $H \in GR_{\Sigma, \Gamma}$ satisfies a
closed formula $\phi \in \emph{MSOL}(\Sigma, \Gamma)$, if $\phi$ is true
for $H$ and we denote this by $H \models \phi.$ If a formula
$\phi \in \emph{MSOL}(\Sigma, \Gamma)$ has free variables $u_1,...,u_n,
U_1,...,U_k$, then we will denote that as $\phi(u_1,...,u_n, U_1,...,U_k).$

\begin{example}
We can construct an MSO formula $\emph{edge}(u,v)$ which expresses the fact
that there exists an edge from vertex $u$ to vertex $v$. It is given by:
\[\emph{edge}(u,v) := \bigvee_{\gamma \in \Gamma} \emph{edge}_{\gamma}(u,v)\]
Using the above formula, we may construct a closed MSO formula
$\emph{edges}$:
\[ \emph{edges} := \forall u,v.\ u \not= v \implies \emph{edge}(u,v)\]
Then, $H \models \emph{edges}$ if the graph $H$ has an edge between any pair of
distinct vertices.
\end{example}

Next, we introduce the notion of an MSO definable language, which will
play a crucial role in some of the proofs in later chapters. In particular,
we shall see that for a fixed string graph, its wire-homeomorphism class
is an MSO definable language.

\begin{definition}[MSO definable language \cite{c-ednce}]
A graph language $L \subseteq [GR_{\Sigma, \Gamma}]$ is \emph{MSO definable}
if there is a closed formula $\phi \in \emph{MSOL}(\Sigma, \Gamma)$ such
that $L = \{[H] \in [GR_{\Sigma, \Gamma}]\ |\ H \models \phi \}$.
\end{definition}

So, any closed MSO formula $\phi$ gives rise to a MSO definable language.
These languages are the graph grammar analogue of regular languages
for standard context-free string grammars in the sense that they provide
the analogous closure property. Context-free string languages are closed under
intersection with regular languages and as the next theorem shows, C-edNCE
and B-edNCE languages are closed under intersection with MSO definable
languages.

\begin{theorem}[\cite{c-ednce}]\label{thm:mso_closed}
\textbf{B-edNCE} is closed under intersection with MSO definable
languages.
\end{theorem}

The above theorem also holds for C-edNCE languages, but it is also true for
B-edNCE languages as well, which is the only languages we will be working with.
This theorem then can be immediately used in order to obtain some important
decidability properties which we shall make use of in the next chapters.

\subsection{Decidability and Complexity results}\label{subsec:decide}

In this subsection, we will list some notable decidability and complexity
results for VR grammars which we will make use of in the following
chapters. In general, B-edNCE grammars are strictly more
powerful than context-free string grammars and therefore their decidability
and complexity properties are in general worse compared to their string
counterparts.

Nevertheless, C-edNCE and B-edNCE grammars enjoy some important decidability
properties which we will make use of. One of the most important ones is
the decidability of the membership problem.

\begin{theorem}[\cite{c-ednce}]
Given an arbitrary graph $H \in GR_{\Sigma, \Gamma}$ and an edNCE grammar
$G=(\Sigma,\Delta,\Gamma,\Omega,P,S)$, it is decidable whether or not
$[H] \in L(G).$
\end{theorem}

However, the problem is NP-complete in general, even for simple subclasses of
edNCE grammars, such as the LIN-A-edNCE grammars. Still, there are known
efficient algorithms for many special cases.

The next decision problem concerns emptiness and finiteness of grammars
which will be used in several proofs.

\begin{proposition}[\cite{c-ednce}]\label{prop:emptiness}
It is decidable, for an arbitrary C-edNCE grammar $G$, whether or not
$L(G)$ is empty. Also, it is decidable whether or not $L(G)$ is finite.
Moreover, if $L(G)$ is finite, it can be constructed effectively.
\end{proposition}

This proposition leads to an interesting corollary. In particular, it shows
that we may decide whether the graphs in the language of some grammar all
satisfy a fixed MSO formula.

\begin{corollary}[\cite{c-ednce}]
Let $L \in$ \textbf{B-edNCE} be a language generated by a B-edNCE grammar
and let $R$ be an MSO definable language. Then, it is decidable
whether or not $L \cap R = \emptyset$. Also, it is decidable whether or not
$L \subseteq R$.
\end{corollary}
\begin{proof}
The first decidability result follows immediately by combining
Proposition~\ref{prop:emptiness} and Theorem~\ref{thm:mso_closed}.

For the second one, let $R := \{[H] \in [GR_{\Sigma, \Gamma}]\ |\ H \models
\phi \}$. Then, we can define $\lnot R :=
\{[H] \in [GR_{\Sigma, \Gamma}]\ |\ H \models \lnot \phi \}$.
By the first decidability result we can decide whether or not
$L \cap \lnot R = \emptyset.$ However, this is equivalent to deciding $L
\subseteq R.$
\end{proof}

The next proposition shows that B-edNCE grammars can simulate
context-free string grammars.

\begin{proposition}
For any context-free string grammar $G$, there exists a
B-edNCE grammar $G'$ which generates the same string language up
to encoding.
\end{proposition}
\begin{proof}
Any sentential form $x_1x_2...x_n$ of $G$ (which may include nonterminals)
will be encoded as the following graph:
\cstikz{simulation.tikz}
where the vertices between the $x_i$-vertices are dummy terminal vertices,
which may be thought of as wire-vertices and the whole graph can thus
be seen as a string graph as defined in Section~\ref{sec:string}.

Then, to simulate $G$, for each production $X \to x_1x_2...x_n$ of $G$,
we add the following production to $G'$:
\cstikz{simulation2.tikz}
Thanks to the introduction of the dummy wire-vertices, $G'$ is
boundary.
\end{proof}

\begin{corollary}
Given two B-edNCE grammars $G_1$ and $G_2$, it is undecidable whether
or not $L(G_1) = L(G_2).$
\end{corollary}
\begin{proof}
From the last proposition it follows that decidability of language equality
for B-edNCE grammars implies decidability of language equality for
context-free string grammars. However, that problem is undecidable
\cite{formal_languages_book}.
\end{proof}

\subsection{Parikh's Theorem}

In this subsection, we will introduce Parikh's theorem for C-edNCE languages.
Parikh's Theorem was first described by Rohit J. Parikh in \cite{parikh}
for context-free string grammars. It defines a mapping from the language
of a grammar to a set of vectors of nonnegative integers. The main result is
that this mapping results in a semilinear set, if the grammar is context-free.
Parikh's theorem may be used to prove that certain string languages are not
context-free. We will use the version of the theorem for graph languages
in order to show that certain languages cannot be described using C-edNCE
grammars.

In the definition below, we shall denote the set of all nonnegative
integers as $\mathbb N$ and the vector space of all $n$-tuples of nonnegative
integers as $\mathbb N^n$.

\begin{definition}[Semilinear set \cite{semilinear}]
We shall say that a set $L \subseteq \mathbb N^n$ is \emph{linear} if there
exists $c \in \mathbb N^n$ and a finite subset $P \subseteq \mathbb N^n$, such
that
\[L= \{x \in \mathbb N^n\ |\ x=
c + x_1 + \cdots + x_m, \text{ where } m \in \mathbb N, x_i \in P\}\]
We will say a set $L \subseteq \mathbb N^n$ is \emph{semilinear} if it is
a finite union of linear sets.
\end{definition}

Note, that in the definition of linear set, the vectors $x_i$ may appear any
finite number of times. So, for $P=\{x_1,x_2,\ldots x_k\},$ an equivalent way
to define a linear set is:
\[L= \{x \in \mathbb N^n\ |\ x= c + \sum_{i=1}^{k} t_i\cdot x_i, \text{ where }
t_i \in \mathbb N\}\]

\begin{example}
The set:
\[
\{(n ,k)\ |\ n \in \mathbb N, k=2n \text{ or } k=3n+1\}\\
\]
\[=\]
\[
\{(n ,2n)\ |\ n \in \mathbb N\} \cup \{(n ,3n+1)\ |\ n \in \mathbb N\}
\]
is semilinear, because it is the union of the two sets above, which are
both linear. However, the following set:
\[\{(n ,n^2)\ |\ n \in \mathbb N\}\]
is not semilinear.
\end{example}

For brevity, we will not present Parikh's theorem for context-free string
grammars. Instead, we present the version for context-free graph
grammars, which is given in \cite{parikh_graph}. 

\begin{definition}[Parikh mapping \cite{parikh_graph}]
For a graph language $L \subseteq [GR_{\Delta, \Gamma}]$ and a sequence
$\pi = (\Delta_1, \ldots, \Delta_k)$ of subsets of $\Delta$, we denote by
$Par_{\pi}(L)$ the set:
\[
Par_{\pi}(L) :=
\left\{
(n_1,\ldots,n_k)\ |\ H \in L,
n_i = \#\{v \in V_H\ |\ \lambda_H(v) \in \Delta_i\}
\right\}
\]
\end{definition}

In other words, the mapping sends a graph to a vector of nonnegative integers,
where each component is just the number of vertices with certain labels. Note,
that the mapping ignores edges. The main theorem of this subsection is given
next.

\begin{theorem}[\cite{parikh_graph}]\label{thm:parikh}
For every graph language $L$ in \textbf{C-edNCE}, $Par_{\pi}(L)$ is
semilinear.
\end{theorem}

So, Parikh's theorem establishes necessary conditions for C-edNCE languages
and we will use it in the next chapters to show that certain languages
are not context-free.

\begin{example}
The language of all square grids $L$:
\cstikz{square-grids-crap.tikz}
is not a C-edNCE language, because its Parikh mapping is
$Par_{\Sigma}(L) =  \{n^2\ |\ n \in \mathbb N\},$
which is not semilinear.
\end{example}

\section{Related Work}\label{sec:related-work-back}

\added{A similar approach to using partial adhesive categories (cf.
Section~\ref{sec:partial_adhesive}) is using \emph{quasiadhesive categories}
\cite{quasi}. The two theories are motivated by the same problem~--~adhesivity
is sometimes too strong for categories of interest and a generalisation is
needed in order to perform DPO rewriting. Quasiadhesive categories take the
approach of requiring that pushouts are stable under \emph{regular} monos,
instead of arbitrary monos. This generalisation still retains many of the
useful properties of adhesive categories. However, in this thesis we use partial
adhesive categories, because they are even more general and as we will show we
can fully characterise the conditions under which DPO rewriting is
well-behaved. This requires more work on our part, but it also allows us to use
a larger class of monos over which we can do rewrites.}

\added{In this work, we represent string diagrams using string graphs. An
alternative approach to reasoning on string diagrams is presented in
\cite{aleks-hypergraph}, where the authors use hypergraphs instead. In the
string graph representation, the nodes of a string diagram correspond to
vertices, whereas in the hypergraph approach, the nodes of a string diagram are
represented using hyperedges instead. The main benefits of the hypergraph
approach is that it is applicable to any symmetric monoidal category, whereas
the string graph approach requires traced structure in addition. In the
hypergraph approach the need for considering equality up to wire-homeomorphism
is avoided and the category used for rewriting is adhesive, instead of
partially adhesive. However, it is unclear how the hypergraph approach can be
lifted to higher-order rewriting, that is, rewriting on families of string
diagrams, which is our primary goal in this thesis. As we have seen, many of
the languages we are interested in involve families of string diagrams where
some of the nodes may have unbounded arity. This would correspond to languages
of hypergraphs where some hyperedges are allowed to have unbounded number of
tentacles. Hypergraph grammars~\cite{hr-grammars} do not allow for such
languages and the author does not know of any other grammars which support
this. Therefore, it is likely that a novel kind of hypergraph grammar would
have to be developed in order to lift that approach to rewriting on the grammar
level.}

\chapter{\label{ch:context-free}Context-free Graph Grammars for String Graphs}

Our primary goal in this thesis is to develop a framework for equational
reasoning with infinite families of string diagrams which is amenable to
computer automation. From previous work, we know how to efficiently represent
string diagrams using string graphs \cite{kissinger_dphil} and families of
string diagrams using !-graphs \cite{merry_dphil}.

In this chapter we will study the expressive power of !-graphs and context-free
graph grammars on string graphs. We begin by identifying important limitations
of the expressive power of !-graphs in Section~\ref{sec:bang-expressive-shit}.
We use this as a justification to consider an alternative to !-graphs with
increased expressive power.

Our first choice for such an alternative is context-free graph
grammars (CFGGs)~\cite{c-ednce}. These grammars have been studied extensively
and there is a rich literature of results around them, so this makes them an
obvious candidate to consider as an alternative to !-graphs. In
Section~\ref{sec:context-shit-power}, we study the expressive power of
CFGGs in relation to !-graphs. The main results are that an important subclass
of !-graph languages, namely $\mathbf{BGNO}$ (cf. Definition~\ref{def:bgto}) is
properly included in the class of context-free languages. We also show that
!-graph languages and context-free languages are incomparable in general, in
the sense that neither class includes the other.

Finally, in Section~\ref{sec:context-laina}, we briefly summarise all of
the results by illustrating the overall comparison between the expressive
power of context-free graph grammars and !-graphs.
We also identify the key limitations of
the expressive power of context-free grammars which we will use as a
justification to consider a simple extension to their expressive power in later
chapters.

Many of the results in this chapter have been previously published in
\cite{gam}.

\section{Limitations of !-graph expressiveness}\label{sec:bang-expressive-shit}

A very important language for the ZX-calculus is the family of undirected
complete string graphs $SK_n$. Each string graph in $SK_n$ has $n$
node-vertices each of which is connected to all other
node-vertices via a single wire containing a single wire-vertex. For example,
$SK_3$ is shown below:
\cstikz[0.8]{sk_3.tikz}
This language is important for the ZX-calculus as it is the fundamental
building block for the local complementation rule of the
ZX-calculus, which is very useful when working with the measurement-based
quantum computing paradigm~\cite{euler_necessity}. 
In addition, the local complementation rule is used in
the only
known
\replaced{decision}{decission}
procedure for deciding equality between ZX-diagrams in the
stabilizer fragment of quantum mechanics \cite{zx_complete}.

However, as we
shall see next, $SK_n$ cannot be described by
!-graphs and therefore this family of string graphs cannot be represented in
Quantomatic and cannot be used in any Quantomatic rewrite rules.
We shall use this as motivation to
consider a more powerful graph generating device.

The more general cause of this limitation
is that any !-graph language has bounded chromatic number where we
consider the underlying string diagram. This is made precise by the next
definition.

\begin{definition}[String Graph Chromatic Number]
Given a string graph $H$, define $C(H)$ to be the graph with the same
node-vertices as $H$, but where all wire-vertices on a wire between a pair of
node-vertices $(v,w)$ are removed and the wire is replaced by an edge $(v,
\alpha ,w)$, for some irrelevant edge label $\alpha$. Then the \emph{string
graph chromatic number} of $H$, denoted $\chi'(H)$ is:
\[\chi'(H) := \chi(C(H))\]
where $\chi(G)$ is the chromatic number of a graph $G$. A language of
string graphs $L$ is said to have a \emph{bounded} string graph chromatic
number if there exists $k \in \mathbb N$, such that for every string graph
$H \in L$, we have $\chi'(H) \leq k.$
\end{definition}

\begin{proposition}
Every !-graph language has bounded chromatic number.
\end{proposition}
\begin{proof}
Consider an arbitrary !-graph $H$. Let $H'$ be the underlying string graph of
$H$ where we disregard all !-boxes. Then, $C(H')$ can be $k$-coloured for some
$k \in \mathbb N$. Now, consider an arbitrary string graph $M \in L(H).$ Any
vertex $v$ of $C(M)$ is a copy of some vertex $v'$ from $H'$ (or is just the
same vertex). Thus, to get a $k$-colouring of $C(M)$, simply assign the same
colour to $v$ as the colour of $v'$ in $H'$. The copies of $v'$ may only be
connected to copies of the neighbours of $v'$ and therefore there can be no
colouring conflicts.
\end{proof}

\begin{corollary}
The language of graphs $SK_n$ is not a !-graph language.
\end{corollary}
\begin{proof}
For every string graph $H \in SK_n$ with $n$ node-vertices, we have $\chi'(H) =
n$. Therefore, $SK_n$ is of unbounded string graph chromatic number and the
language cannot be induced by a !-graph.
\end{proof}

We proceed by presenting a severe limitation of !-graph languages.
We will show that any !-graph language
has a fixed upper bound on the maximum shortest path between a pair of
vertices. This notion is usually referred to as \emph{graph
diameter} in the literature, but we make small modifications to the definition
in order to suit our purposes.

\begin{definition}[Maximum distance \cite{gam}]
  For a graph $H$ and vertices $v,u \in H$, the \textit{distance} between $u$
  and $v$ is the length of the shortest path connecting $u$ and $v$. If there
  is no path between $u$ and $v$ then we say that the distance is -1. The
  distance between a vertex and itself is 0. The \textit{maximum distance} for
  a graph $H$ is the largest distance among all pairs of vertices.
\end{definition}

\begin{definition}[Bounded maximum distance \cite{gam}]
  For a set of graphs $L$, we say that
  $L$ is of \textit{bounded maximum distance} if there exists an
  integer $n \in \mathbb{N}$, such that the maximum distance for every graph in
  $L$ is smaller than $n$. If such an $n$ does not exist, then we say
  that $L$ is of \textit{unbounded maximum distance}.
\end{definition}

\begin{proposition}[\cite{gam}]
  The language induced by any !-graph $H$ is of bounded maximum distance.
\end{proposition}

\begin{proof}
  Consider a !-box $b$ of $H$.
  Applying a KILL operation to $b$ cannot increase the maximum
  distance. Applying an EXPAND operation once could potentially increase it,
  however, applying an EXPAND more than once will not increase it any further.
  Thus, regardless of how many times an EXPAND operation is applied to $b$
  the maximum distance can only increase by a fixed amount. Also, because of
  the symmetric properties of the EXPAND map, any nested !-boxes within $b$
  or any overlapping !-boxes can increase the maximum distance with a fixed
  amount as well regardless of how many EXPAND operations are applied to any of
  them.

  By combining the above observation with the fact that $H$ has finitely many
  !-boxes we can conclude that the set of graphs induced by $H$ is of bounded
  maximum distance.
\end{proof}

In the context of the ZX-calculus, this limitation is severe because it shows
that !-graphs can only be used to represent families of quantum circuits with a
fixed depth. This corresponds to quantum algorithms and protocols with time
complexity $O(1)$.
Both of the limitations identified in this section provide a compelling reason
to
consider alternatives to the !-graph formalism, as we do in the next
section.

\section{Expressivity of context-free grammars on string graphs}
\label{sec:context-shit-power}

In this section we outline the relationship between !-graph languages and
context-free languages.
In Subsection~\ref{sub:overlapping-shit} we show that the classes of languages
induced by
(unrestricted) !-graphs and context-free grammars respectively are incomparable.
We also make the relationship between !-graphs and CFGGs more explicit by
considering some of their subclasses and comparing them.
In subsection \ref{sec:nested}, we show that the
language induced by any !-graph $H$ with no overlapping non-nested !-boxes can
be
described by a LIN-edNCE grammar, which can moreover be constructed effectively
from $H$.

\subsection{CFGGs and general !-graphs}\label{sub:overlapping-shit}

We begin by showing that the two most important classes of context-free
graph grammars have the same expressive power when restricted to string
graphs. For general graphs VR grammars, also known as C-edNCE grammars, are
strictly more expressive.

\begin{proposition}[\cite{gam}]\label{prop:generative}
  The generative power of C-edNCE grammars and Hyperedge Replacement grammars
  on string graphs is the same.
\end{proposition}

\begin{proof}
  The graph $K_{3,3}$ is not a subgraph of any string graph. Then, the
  proposition follows immediately from the main result in \cite{hrl_equal_vr}.
\end{proof}

As we are interested in string graph languages, then it doesn't matter
what grammar we choose.
\added{However}, we will use C-edNCE grammars (and in particular
B-edNCE grammars), because in later chapters we will propose a simple
extension which increases their expressive power.
\replaced{This}{However, that} extension
does not result in increased expressive power for hyperedge replacement
grammars \added{and the author does not see how HR grammars may otherwise be
extended in order to have satisfactory generative power for our purposes.
Because of this, we will use edNCE grammars throughout the thesis.}

Next, we show that there exists a very simple simple context-free language
which cannot be described by any !-graph.

\begin{proposition}[\cite{gam}]
  The language of the LIN-A-edNCE grammar $G$, given by:
  \cstikz{unbounded.tikz}
  is not induced by any !-graph.
\end{proposition}
\begin{proof}
The language generated by $G$ consists of node-vertices connected
in a line via a single wire. It is given by:
\[ L(G) = \left\{ \stikz[0.8]{lines.tikz} \right\}\]
Obviously, this language is of unbounded maximum distance and therefore it
cannot be generated by any !-graph.
\end{proof}

This proposition shows that $\mathbf{CF} \not\subseteq \mathbf{BG}$, where
$\mathbf{CF}$ is the set of all context-free string graph languages
and $\mathbf{BG}$ is the set of all !-graph languages.

Next, we show that there exists an entire class of !-graphs whose
languages are not context-free.

\begin{proposition}
Let $H$ be a !-graph which satisfies the following conditions:
\begin{enumerate}
\item There exist a pair of !-boxes $b$ and $b'$ which overlap on a vertex
$v$ with label $\alpha$
\item $b$ or $b'$ contain a vertex $w$ with label $\beta \not= \alpha$
\item $w$ is not in the overlap of $b$ and $b'$
\item $v$ and $w$ are not in any other !-boxes
\item There are no other vertices with labels $\alpha$ or $\beta$
\end{enumerate}
Then, the language induced by $H$ is not context-free.
\end{proposition}
\begin{proof}
We can prove this theorem by making use of Parikh's theorem for C-edNCE
languages. Assume that the language induced by $H$ is context-free, that is,
there exists a C-edNCE grammar $G$, such that $L(G) = L(H) =: L$.

Let $\pi = (\{\alpha\}, \{\beta\})$ and let's consider the Parikh mapping
of $L$. From Theorem~\ref{thm:parikh}, we know that
$Par_{\pi}(L)$ is semilinear. Without loss of generality we shall assume that
$w$ is in $b$, but not in $b'$.

Then:
\[Par_{\pi}(L) = \{\ (bb', b)\ |\ b,b' \in \mathbb N\}\]
Thus, for every expand
operation on the !-box $b'$, we get as many copies of the vertex $v$ as we
apply expand operations on the !-box $b$. However, this language is obviously
not semilinear and we get a contradiction.
\end{proof}

This proposition may be further generalised by relaxing some of the conditions,
but then the proof becomes very involved. In any case, the proposition is
enough to show that there is a large class of !-graph languages which
are not context-free. Therefore, this establishes that $\mathbf{BG} \not
\subseteq \mathbf{CF}.$

However, as the next proposition shows, there are !-graph languages with
non-nested, but overlapping !-boxes which are context-free.

\begin{proposition}
Let $H$ be a !-graph of the following form:
\cstikz{stupid_overlap.tikz}
where $K$ is an arbitrary string graph. $b_1$ and $b_2$ overlap on $K$ and
$b_1$ contains another (isomorphic) copy of $K$, no part of which appears in
$b_2$. Then, $L(H)$ is the language of a Lin-A-edNCE grammar.
\end{proposition}
\begin{proof}
If we expand the !-boxes $b_1$ and $b_2$, $m$ and $n$ times respectively, then
we would get $m+mn$ copies of $K$. Because we consider all possible sequences
of !-box operations when generating a !-graph language, this means that $L(H)$
simply consists of an arbitrary number of copies of $K$ (in other words, we can
fix $n$ to be zero).  This language can therefore be generated by the
Lin-A-edNCE grammar:
\cstikz{dumb_grammar.tikz}
\end{proof}

This proposition shows that overlapping !-boxes do not necessarily imply
that the language is not context-free.

Let's denote with $SK_{m,n}$ the language of directed complete bipartite
string graphs. The node-vertices in $SK_{m,n}$ are partitioned into two sets
$M$ and $N$ of cardinality $m$ and $n$ respectively. Each node-vertex from
$M$ is connected to all other node-vertices of $N$ via a single wire
consisting of a single wire-vertex.

\[SK_{m,n} := \left\{ \stikz[0.7]{k_m_n.tikz} \right\}\]

The next proposition shows that this language can be expressed by a
!-graph with trivially overlapping !-boxes, but it is not context-free.

\begin{proposition}
The family of complete bipartite string graphs $SK_{m,n}$ is not
in \textbf{C-edNCE}, but it is in \textbf{BGTO}.
\end{proposition}
\begin{proof}
The following !-graph:
\cstikz{trivial_overlap.tikz}
induces the language $SK_{m,n}$. To see that it is not context-free, consider
its Parikh mapping when setting $\pi = (\{\mathcal N\}, \{\mathcal W\})$.
So, the first component of the Parikh mapping will count the number of
node-vertices and the second one will count the number of wire-vertices.
Then, we have:
\[ Par_{\pi}(SK_{m,n}) = \{(m+n, mn)\ |\ m,n \in \mathbb N\} \]
This set is obviously not semilinear and from Theorem~\ref{thm:parikh} it
follows that the language is not context-free.
\end{proof}

Therefore, this proposition implies
$\mathbf{BGTO} \not\subseteq \mathbf{CF}$
(cf. Definition~\ref{def:bgto}).
Unfortunately, this means that graph languages of interest are not
context-free. The language $SK_{m,n}$ is very important for
traced symmetric monoidal categories which use commutative Frobenius algebras
in their internal language (such as the ZX-calculus),
as it is one side of the generalised bialgebra rule, which is a very powerful
and crucial distributive law. The rule is shown below (in string diagram form),
for completeness:
\cstikz{general-bialgebra.tikz}

\subsection{CFGGs and \textbf{BGNO}}\label{sec:nested}

The main result of this subsection is a theorem stating that the language
induced by any !-graph with no overlapping non-nested !-boxes can be directly
represented by a LIN-edNCE grammar, which can moreover be generated
effectively.

Before we present the main results in this subsection, we
prove a series of lemmas which are used in the proof of the main theorem.
Each lemma describes how to build bigger LIN-edNCE grammars out of smaller ones
which in turn describe the language induced by certain subgraphs of a given
!-graph. All of our constructions are effective in the sense that they can be
performed by a computer.

We introduce some conventions that are used throughout our proofs. Given a
!-graph $H$, its vertices will be $\{v_1,v_2,...,v_n\}$. To these vertices, we
will associate edge labels $\{\alpha_1,\alpha_2,...,\alpha_n\}$ which
will be used in the productions of our grammars. Informally, an edge labelled
with $\alpha_i$ will have as its source or target either the original vertex
$v_i$ or one of its copies. This is used in some productions of our grammars,
so that we can easily refer to all copies of such a vertex at once and connect
them to other vertices.

In order to prove the main result in this section, we will use edNCE grammars
satisfying the conditions detailed in the next definition. This allows us
to split the proof into several small lemmas.

\begin{definition}[!-linear form \cite{gam}]\label{def:!-linear}
Any LIN-edNCE grammar
$G=(\Sigma,\Delta,\Gamma,\Gamma,P,S)$ is said to be in \emph{!-linear form}
if it satisfies the following conditions:
\begin{enumerate}
  \item For every terminal vertex $v_i$ in a production of $G$, there is
  a unique edge label $\alpha_i \in \Gamma.$ The set of these $\alpha_i$
  is $\mathcal X \subset \Gamma$.
  \item There exists a function $\chi: \mathcal X \to \Delta$ which associates
  terminal vertex labels for every edge label $\alpha_i \in \mathcal X$.
  \item There is a single final production, that is, a production with no
  nonterminal vertices. The body of this production is simply the empty graph.
  \item Every production with nonterminal vertex $x$ is such that for
    every terminal vertex $v_i$, there exist edges $(v_i, \alpha_i, x)$ and
    $(x, \alpha_i, v_i)$, where $\alpha_i$ is a unique label associated with
    $v_i$. Graphically we will depict these edges using bidirectional arrows
    for compactness.
  \item Every production with nonterminal vertex $x$ has connection
  instructions $(\sigma_i, \alpha_i, \alpha_i, x, in)$ and
  $(\sigma_i, \alpha_i, \alpha_i, x, out),$ where the $\alpha_i$ range
  over $\mathcal X$ and $\sigma_i = \chi(\alpha_i)$.
  Graphically, we will depict that using bidirectional arrows as a shorthand
  notation for a connection instruction in each direction. Moreover, because
  $\sigma_i$ is determined by $\alpha_i$ we will not depict the label in
  the connection instructions for simplicity.
\end{enumerate}
\end{definition}

\begin{lemma}[\cite{gam}]
  Given a concrete string graph $H$, there exists a LIN-edNCE grammar
  $G$ which generates the language $\{H\}$. Moreover, this grammar
  can be effectively constructed and is in !-linear form.
  \label{lem:base}
\end{lemma}
\begin{proof}
  This is done by the following grammar:
  \cstikz{super_simple_grammar.tikz}
  where the set of vertices of $H$ is $V_H = \{v_1, v_2,..., v_n\}$
  and each edge with label $\alpha_i$ has as source or target $v_i$ and the
  non-terminal vertex labelled $F$.
\end{proof}

\begin{lemma}[\cite{gam}]
  Given a !-graph $H$ and a !-linear form grammar $G$ which generates the same
  language as $H$, there exists a
  grammar $G'$ which generates the same language as the
  following !-graph:
  \cstikz{simple.tikz}
  Moreover, $G'$ can be effectively constructed and is in !-linear form.
  \label{lem:nested}
\end{lemma}
\begin{proof}
  Let's assume that $S$ is the initial nonterminal label of $G$, $F$ is the
  final nonterminal label of $G$ and that $S'$ and $F'$ are not production
  labels used by $G$.

  First, we modify the production labelled $F$ to be the following:
  \cstikz{nested_modify.tikz}
  Finally, to the productions of $G$ we add the following
  productions, where $S'$ is the initial nonterminal label:
  \cstikz{nested_grammar.tikz}
  A derivation
  $S'\Rightarrow S \Rightarrow \cdots \Rightarrow F$
  creates a concrete string graph from the language of $H$, so it is
  simulating a single EXPAND operation applied to the top-level !-box,
  together with a concrete instantiation of the !-boxes in $H$. By
  construction, we can iterate this, thus allowing us to generate multiple
  disjoint concrete graphs, all of which are in the language of $H$.
  A derivation $S' \Rightarrow F'$ simulates the final
  KILL operation applied to the top-level !-box.
\end{proof}

\begin{lemma}[\cite{gam}]
  Given disjoint !-graphs $H,K$ and !-linear form grammars $G_1, G_2$
  which generate the same languages as $H$ and $K$ respectively, then there
  exists a grammar $G'$ which generates the same language as the
  following !-graph:
  \cstikz{disjoint.tikz}
  Moreover, $G'$ can be effectively constructed and is in !-linear form.
  \label{lem:disjoint}
\end{lemma}

\begin{proof}
  Let the vertices of $H$ be $\{v_1,v_2,\ldots,v_k\}$ and let the vertices of $K$
  be $\{v_{k+1},v_{k+2},\ldots,v_n\}$. Also, let $S_i$ and $F_i$ be the starting
  and final production labels respectively of $G_i$. First, we modify each
  non-final production of $G_1$, by adding connection
  instructions for edge labels $\alpha_{k+1},\ldots, \alpha_n$ in the following
  way: 
  \cstikz{disjoint_modify1.tikz}
  where the new additions are coloured in red. This doesn't change the language
  of $G_1$ and is done so that we can put the grammar in the
  required form.
  Similarly, modify all non-final productions of $G_2$ by
  adding to their connection instructions the missing edge labels $\alpha_1,
  \alpha_2, \ldots, \alpha_k$. 

  Finally, modify $F_1$ to be the production
  depicted below:
  \cstikz{disjoint_modify2.tikz}
  so that we can chain together the two grammars.
  
  The required grammar $G'$ has as its productions the modified
  productions of $G_1$ and $G_2$ with initial nonterminal label
  $S_1$. A derivation $S_1 \Rightarrow \cdots \Rightarrow F_1$ creates a
  concrete graph from the language of $H$ and a derivation $S_2 \Rightarrow
  \cdots \Rightarrow F_2$ creates a graph from the language of $K$. By
  chaining the two grammars, we simply generate two disjoint concrete graphs,
  one from the language of $H$ and one from the language of $K$, as required.
\end{proof}

\begin{lemma}[\cite{gam}]
  Given !-graph $H$, where $H$ contains a !-box $b$ and given a !-linear form
grammar $G$ which generates the same languages as $H$, there exist grammars
$G'$ and $G''$ which generate the same languages as:
  \[ \stikz{connected.tikz} \qquad \textrm{and} \qquad
  \stikz{connected2.tikz} \]
  respectively, where in both cases, the newly depicted edge (coloured in red)
is incident to the contents of $b$ in $H$ and the edge is also incident to a
node-vertex in $H$ which is not in any !-boxes. Moreover, these grammars can be
effectively generated and are in !-linear form.
  \label{lem:connect}
\end{lemma}
\begin{proof}
  In both cases, for the newly depicted edge, identify the wire-vertex as $v_i$
  and the node-vertex as $v_j$. To get $G'$, identify the unique production of
  $G$ which contains a vertex incident to an edge with
  label $\alpha_j$. Then add to its connection instructions a new edge in the
  following way:
  \cstikz{connect_modify.tikz}
  where the red-coloured edge is the new addition. To get the grammar
  $G''$, follow the same procedure, but with the directions reversed: 
  \cstikz{connect_modify2.tikz}
In each case,
this modification has the effect that we connect all copies of the wire vertex
(identified by the connection instruction labelled $\alpha_i$)
to the single node vertex (identified by the edge labelled $\alpha_j$), which
is the only change required compared to the concrete graphs of $H$.
\end{proof}

\begin{theorem}[\cite{gam}]
  Given a !-graph $H$ such that it doesn't have any non-nested overlapping
!-boxes,
  there exists a LIN-edNCE grammar $G$ which generates the same language as
  $H$. Moreover, this grammar can be effectively constructed and is in !-linear
  form.
  \label{thm:non-overlap}
\end{theorem}

\begin{proof}
  We present a proof by induction on the number of !-boxes of $H$.
  For the base case, if $H$ has no !-boxes, then lemma \ref{lem:base} completes
  the proof.
  For the step case, pick any top-level !-box $b$ and let's consider the full
  subgraph of $H$ which consists of the contents of $b$. Call this subgraph
  $K$. Any vertex $v$ of $H
  - K$ which is adjacent to $K$ must be a node-vertex, because
  otherwise this would violate the openness condition of !-boxes.
  Let $w \in K$ be a wire-vertex that $v$ is adjacent to and let $e$ be the
  edge connecting $v$ to $w$. 
  \cstikz{thm_no_overlap.tikz}
  If $v$ is in some !-box $b'$, then the openness condition of !-boxes
  implies that $w$ must also be in $b'$. However, we have assumed that
  $H$ does not contain overlapping !-boxes, so this is not possible and thus
  $v$ is not in any !-boxes. Therefore, we can use lemma \ref{lem:connect} to
  reduce the problem to showing that we can effectively construct a grammar for
  $H - e$. Similarly, by applying the same lemma multiple times, we
  can reduce the problem to showing that we can effectively construct a grammar
  for the !-graph consisting of the disjoint !-graphs $K$ and $H - K$.
  Applying lemma \ref{lem:disjoint} and the induction hypothesis then reduces
  the problem to showing the theorem for $K$.  Finally, we can apply lemma
  \ref{lem:nested} to $K$ and then the induction hypothesis to complete the
  proof.
\end{proof}

This theorem is very important as it shows that $\mathbf{BGNO} \subseteq
\mathbf{CF}.$ $\mathbf{BGNO}$ is a class of !-graphs of
crucial importance. All families of string diagrams which are expressible
using !-graphs and have known practical applications fall within this
class, with the exception of $SK_{n,m}.$

\section{Limitations of context-free graph grammars}\label{sec:context-laina}

We begin by showing a limitation of context-free grammars in terms
of their expressive power on string graphs.

\begin{proposition}
The family of graphs $SK_n$ is not context-free.
\end{proposition}
\begin{proof}
We can prove this using Parikh's theorem
by setting $\pi = (\{\mathcal N\}, \{\mathcal W\})$. So, the first
component of the Parikh mapping counts the number of node-vertices
and the second component counts the number of wire-vertices.
Then,
we get:
\[ Par_{\pi}(SK_n) = \left\{\left(n, \binom{n}{2}\right)\ |\ n \in \mathbb
N\right\} \]
which is not semilinear and therefore $SK_n$ is not a context-free
graph language.
\end{proof}

As we have shown before, !-graphs cannot express this language either,
nevertheless, this language is important and we would like to be able
to formally represent it.

The figure below summarises the results of this chapter:
\cstikz{inclusion-figure-bgno.tikz}
where all regions of the diagram are populated, except for the part between
\textbf{LIN-edNCE} and
$\mathbf{CF}$ inside $\mathbf{BG}$ indicated by the dashed red line, which
has not been proven, but it is conjectured.

In Section~\ref{sec:bang-expressive-shit} we identified two key limitations
of !-graphs -- they are of bounded diameter and the language $SK_n$ cannot
be expressed by them. We have shown that context-free graph grammars do
not suffer from the former limitation, but they do suffer from the latter.
In addition, context-free graph grammars cannot express the language
$SK_{m,n}$ which is also crucial. For these reasons, in the next chapter we
will introduce a simple extension to our context-free grammars which overcomes
these limitations, while retaining many of their structural and decidability
properties.

\chapter{\label{ch:besg}B-ESG grammars}
In this chapter we will introduce \emph{B-ESG} grammars, which are extended
context-free graph grammars.
B-ESG grammars are a shorthand for \emph{Boundary Encoded String Graph}
grammars.
They are based on B-edNCE grammars and the \emph{B} in the name has the
same intended meaning. We will also show that B-ESG grammars
do not suffer from any of the limitations in expressive power that
we identified in the previous chapter.

In Section~\ref{sec:encoded-b-ednce}, we begin by describing
\emph{encoded B-edNCE grammars}. An encoded B-edNCE grammar is a slightly
more general B-edNCE grammar. Its derivations consist in two parts -- first,
a graph is produced in the same way as for standard B-edNCE grammars
and then secondly, some of its edges are replaced with fixed graphs according
to some additional rules. There are two main results of this section. The
first one is that encoded B-edNCE grammars are strictly more expressive
compared to context-free graph grammars when we restrict ourselves to string
graphs. The second result is that encoded B-edNCE grammars properly include
$\mathbf{BGTO}$, which is the class of !-graph languages where non-nested
!-boxes may overlap trivially.

Encoded B-edNCE grammars generate languages of (general) graphs. They are
not restricted to languages of string graphs. However, our aim is to model
string diagrammatic reasoning by representing string graphs, so naturally we
are only interested in languages consisting solely of string graphs. In
Section~\ref{sec:b-esg-languages} we will introduce B-ESG grammars,
which are encoded B-edNCE grammars which satisfy conditions that ensure
their derivations result in string graphs.

In Section~\ref{sec:expres-shit} we show that we do not lose any expressive
power, on languages of string graphs, by restricting ourselves to B-ESG
grammars from encoded B-edNCE grammars. In particular, we prove that for
any encoded B-edNCE grammar which generates a language consisting of
string graphs, there exists a B-ESG grammar generating the same language.

Finally, in Section~\ref{sec:decide} we prove important decidability
properties for B-ESG grammars. In particular, we prove that the membership
and match-enumeration problems are decidable, which are crucial for the
operation of a software proof assistant.

Many of the results in this chapter are based on previously published
results in \cite{icgt}. However, since then the results have been
improved by the author by generalising the definition of B-ESG grammar.

\section{Encoded B-edNCE grammars}\label{sec:encoded-b-ednce}
We begin by presenting a refinement of the vertex and edge label alphabets we
will use. The new addition is a subset of edge labels $\mathcal E$ which will
denote \emph{encoding} edge labels. The next definition describes the alphabets
we will need for the rest of the chapter and all of our constructions will be
over these alphabets.

\begin{definition}[B-ESG alphabets]\label{def:encoding_alphabets}
We will use the following alphabets:
\begin{description}
\item[1.] $\Sigma$ is the \emph{alphabet of all vertex labels}.
\item[2.] $\Delta \subseteq \Sigma$ is the \emph{alphabet of terminal
vertex labels}.
\item[3.] $\mathcal{N} 
\subseteq \Delta$ is the \emph{alphabet of node-vertex
labels} and $\mathcal{W} = \Delta - \mathcal{N}$ is the alphabet of
\emph{wire-vertex labels}.
\item[4.] $\Gamma$ is the \emph{alphabet of all edge labels}. We will not
use any non-final edge labels (see Corollary~\ref{cor:nonblocking}).
\item[5.] $\mathcal E \subsetneq \Gamma$ is the alphabet of \emph{encoding}
edge labels.
\end{description}
\end{definition}

In the previous chapter we saw that C-edNCE grammars cannot represent
languages of interest to the ZX-calculus and also entire classes of languages
which are expressible using !-graphs while being relevant to the ZX-calculus
and other string diagrammatic theories which contain Frobenius algebras.
We propose a simple extension to the generative power of C-edNCE grammars
with the motivation of overcoming these limitations. The main idea is to use
B-edNCE grammars which generate \emph{encoded string graphs}. These
encoded string graphs are just like string graphs, but with the addition that
we allow edges to connect node-vertices if they carry encoding labels.
The encoding edges would then be replaced by fixed graphs according to a
simple set of rules. Thus, we would use B-edNCE grammars to generate languages
of encoded string graphs which are then decoded using a very simple set of DPO
rewrite rules to obtain languages of string graphs.

\begin{definition}[Encoded string graph \cite{icgt}]
An \textit{encoded string graph} is a string graph where we additionally allow
edges with labels $\alpha \in \mathcal E$ to connect pairs of node-vertices.
Edges labelled by some $\alpha \in \mathcal E$ will be called \emph{encoding}
edges.
\end{definition}

\begin{definition}[Decoding system \cite{icgt}]
A \textit{decoding system} $T$ is a set of DPO rewrite rules of the form:
\cstikz[0.8]{simple-dpo-def.tikz}
one for every triple $(\alpha, \sigma_1, \sigma_2) \in \mathcal E \times
\mathcal N \times \mathcal N,$
where the LHS consists of a single edge with encoding label $\alpha \in
\mathcal E$ connecting a $\sigma_1$-labelled node-vertex to a
$\sigma_2$-labelled node-vertex, and the RHS is a string graph which contains
the same two node-vertices and at least one additional vertex while containing
no inputs, outputs, or encoding labels.
\end{definition}

Instead of depicting the decoding rules as a span (as we did in the above
definition), we will introduce a more compact notation which we shall use
from now on when depicting decoding systems. A DPO rewrite
rule from a decoding system $T$:
\cstikz[0.8]{simple-dpo-def.tikz}
will be depicted as:
\cstikz{simple-dpo-def-short.tikz}
That is, we simply take its LHS and RHS. By doing this, we don't lose any
information, because the interface of the DPO rule is simply the two
node-vertices which are common in both sides.

\begin{proposition}
Any decoding system $T$ is confluent and terminating.
\end{proposition}
\begin{proof}
The left-hand side of each decoding rule contains an encoding edge, whereas
the right-hand side does not. Thus, any application of a rule from $T$
decreases the number of encoding edges by one and therefore $T$ is
a terminating rewrite system. For confluence, observe that each DPO
rewrite rule in $T$ contains an invariant part -- the two node-vertices
are the same in both sides. Thus, all rule applications of $T$ are independent
of each other and may all be applied at the same time in parallel.
\end{proof}

Given an (encoded string) graph, \emph{decoding} is the process of applying all
of the rules of $T$ to the graph. As the above proposition shows, this is a
very simple process which may even be done in a single step.
If $H$ is an encoded string graph, we shall say that $H'$ has been
\emph{decoded} from $H$, and denote this with $H \Longrightarrow_*^T H',$ if
the graph $H'$ is the result of applying all rules from $T$ to $H$, such that
$H'$ contains no encoding edges.
Next, we present a
lemma which establishes a relationship between encoded string graphs, the
process of decoding and string graphs.

\begin{lemma}\label{lem:decoding}
Given two graphs $H, H'$ with $H \Longrightarrow_*^T H'$, where $T$ is
a decoding system, then $H$ is an encoded string graph iff $H'$ is a string
graph.
\end{lemma}

\begin{proof}
$\left( \Longrightarrow \right)$ Decoding an encoded string graph $H$ consists
of replacing all of the encoding edges with some string graphs, according to
the rules of $T$. This can be broken into two steps -- in the first step the
encoding edges from $H$ are removed which results in a string graph by
definition. The second step is to then glue some string graphs over the
invariant node-vertices which again results in a string graph by definition.

$\left( \Longleftarrow \right)$ $H$ cannot contain a wire-vertex $v$ with
in-degree (out-degree) more than one, because otherwise so does $H'$ as
decoding will not add, nor remove any neighbours of $v$. If $H$ contains
an edge between a pair of node-vertices with label $\alpha \not \in \mathcal
E$, then so does $H'$ as $T$ will preserve this edge. Therefore, $H$ must
be an encoded string graph.
\end{proof}

\begin{definition}[Encoded B-edNCE grammar]\label{def:encoded-b-ednce}
An \emph{Encoded B-edNCE grammar} is a pair $B=(G,T)$, where $G$ is a B-edNCE
grammar and $T$ is a decoding system.
\end{definition}

\begin{definition}[Concrete derivation and language \cite{icgt}]
A \textit{concrete derivation} for an encoded B-edNCE grammar $B = (G,T)$ with
$S$ the initial nonterminal for $G$, consists of a concrete derivation
$sin(S,z)
\Longrightarrow_*^G H_1$ in $G$,
followed by a decoding $H_1 \Longrightarrow_*^T H_2$. We will denote such a
concrete derivation as $sin(S,z) \Longrightarrow_*^G H_1 \Longrightarrow_*^T
H_2$ or
simply with $sin(S,z) \Longrightarrow_*^B H_2$ if the graph $H_1$ is not
relevant for
the context.
The \textit{language} of $B$ is given by $L(B):=\{[H]\ |\ sin(S,z)
\Longrightarrow_*^{B}
H\}$.
\end{definition}

Note, that the above definition defines derivations and the language of an
encoded B-edNCE grammar up to isomorphism, which is compatible with the
definition of B-edNCE language. This extension to B-edNCE grammars can
be intuitively seen as simply generating graphs in \textbf{B-edNCE}, where we
think of the encoding edges as representing fixed graphs.

\begin{example}\label{ex:encoded-b-ednce-examples}
We will depict encoded B-edNCE grammars in the same way as we depict
B-edNCE grammars, with the addition of depicting all of the DPO rewrite
rules of the decoding system above the grammar. For example, an
encoded B-edNCE grammar which generates the language $SK_n$ of complete
string graphs is provided below:
\cstikz{complete-string-graphs.tikz}
Observe, that the B-edNCE grammar is essentially the same as the one from
Example~\ref{ex:complete_grammar}. It generates complete string graphs, where
all edges are labelled by encoding symbols. These edges are then simply replaced
by a wire consisting of a single wire-vertex, while preserving the
endpoints. A derivation of $SK_3$ is given by:
\cstikz{sk3-derive.tikz}
The horizontal derivation sequence is simply a concrete derivation in the
B-edNCE grammar. Once that is done, the vertical derivation step simply decodes
all of the encoding edges, as specified by the decoding system.
\end{example}

\begin{example}\label{ex:encoded-kmn}
An encoded B-edNCE grammar which generates the language $SK_{n,m}$ of complete
bipartite string graphs is provided below:
\cstikz{complete-bipartite-string-graphs.tikz}
Similarly to the previous example, the B-edNCE grammar generates complete
bipartite graphs, where all edges are labelled by encoding symbols. A
derivation of $SK_{2,2}$ is given by:
\cstikz{sk22-derive.tikz}
\end{example}

Recall that both languages $SK_{m,n}$ and $SK_{n}$ were considered in
the previous chapter and we highlighted their importance. However, we showed
that neither of them can be described by context-free graph grammars. Now we
see that the simple extension which we introduced (the addition of the decoding
system) allows us to represent them. Using this fact, the next proposition
follows immediately.

\begin{proposition}\label{prop:encoded-power}
Encoded B-edNCE grammars are strictly more expressive than B-edNCE grammars.
\end{proposition}
\begin{proof}
Obviously, any B-edNCE grammar $G$ may be simulated by an encoded B-edNCE
grammar $B=(G,\emptyset),$ that is, the alphabet of encoding labels is empty
and therefore so is the decoding system. Then, the proposition follows after
considering the language from Example~\ref{ex:encoded-b-ednce-examples}
or Example~\ref{ex:encoded-kmn}.
\end{proof}

Therefore, encoded B-edNCE grammars are indeed more powerful compared to
B-edNCE grammars in general. Also, recall that \textbf{B-edNCE} $\subsetneq$
\textbf{C-edNCE} = \textbf{CF}. Since we are only interested in string graph
languages, a natural question to ask is if encoded B-edNCE grammars are more
powerful compared to context-free graph grammars when we restrict ourselves to
string graphs. The next theorem shows that this is indeed the case.

\begin{theorem}
The generative power of encoded B-edNCE grammars is strictly greater
than the generate power of context-free graph grammars on string graphs.
\end{theorem}
\begin{proof}
Recall from Proposition~\ref{prop:generative}, that any context-free string
graph language may be generated by a Hyperedge Replacement grammar.
However, from \cite{c-ednce} (pp.57) we know that Hyperedge Replacement
grammars have the same generative power as Bnd-edNCE grammars, which are a
proper subclass of B-edNCE grammars (as shown in
Subsection~\ref{sub:grammar-subclasses}). Then the theorem follows after
considering Proposition~\ref{prop:encoded-power}.
\end{proof}

We can provide a more detailed picture of the expressive power of encoded
B-edNCE grammars by comparing them to !-graphs. Recall that context-free string
graph languages contain the class \textbf{BGNO}, but they do not contain
\textbf{BGTO}. However, we can show that encoded B-edNCE grammars properly
include the class \textbf{BGTO}, which in turn contains \textbf{BGNO}. The rest
of the section is devoted to proving this result.

We will show this result by building on the proof from
Subsection~\ref{sec:nested}. First, we generalise the definition of
!-linear form to encoded B-edNCE grammars.

\begin{definition}[!-encoded grammar]
We will say that an encoded B-edNCE grammar $B=(G,T)$ is \emph{!-encoded}
if $G$ is !-linear (cf. Definition~\ref{def:!-linear}).
\end{definition}

The next lemma is similar in spirit to the lemmas of
Subsection~\ref{sec:nested}. It shows how to construct a larger !-encoded
grammar from a smaller one, both of which simulate the languages of
!-graphs which differ by the addition of a wire in a trivially overlapping pair
of !-boxes.

\begin{lemma}
Given !-graph $H$ which contains non-nested !-boxes $b_1$ and $b_2$ and given a
!-encoded grammar $B=(G,T)$ which generates the same language as
$H$, there exists a !-encoded grammar $B' = (G', T')$ which generates the same
language, as the following !-graph:
  \cstikz{overlap.tikz}
where the new additions are coloured in red and the newly depicted
wire-vertices are in both $b_1$ and $b_2$ and no other !-box. Moreover, this
grammar can be effectively constructed.
  \label{lem:overlap}
\end{lemma}
\begin{proof}
  For the newly depicted edges whose source or target are the two
node-vertices, identify the source node-vertex as $v_i$ and
  the target node-vertex as $v_j$.  To get the desired grammar $G'$,
  identify the unique production $X$ of $G$, such that $X$
  contains a vertex incident to an edge with label $\alpha_j$. Then,
  add to its connection instructions a new edge with unique label $\beta_k$:
  \cstikz{overlap_final.tikz}
  where the new addition is coloured in red. Note, that with this
  construction, we are not creating the newly depicted wire-vertices, nor any
of their
  copies. We are connecting all copies of the node-vertex $v_i$ to all copies
  of the node-vertex $v_j$ directly with edges labelled with $\beta_p$.

  Therefore, the grammar $G'$ as described so far will generate encoded string
  graphs which connect all copies of the vertex $v_i$ to all copies of the
  vertex $v_j$ via edges with label $\beta_k$ (the edges with label
  $\beta_i$, with $i <k$ have been established by the same
  construction in previous steps). What remains is to decode the edges
  labelled $\beta_k$. Thus, to get the same language as that of $H$, we need
  to add to $T$ the decoding rule:
  \cstikz{overlap_decode.tikz}
  where the right-hand side contains the same closed-wire (coloured in red
  above) which has been added to $H$.

\end{proof}

Next, we present the main theorem of this section. It combines the
previous lemma with the results of Subsection~\ref{sec:nested} to derive
the result.

\begin{theorem}
  Given a !-graph $H$ such that the only overlap between !-boxes in $H$ is
trivial, then there exists a !-encoded grammar $B=(G,T)$ which generates the
same language as $H.$
  Moreover, this grammar can be effectively constructed.
  \label{thm:overlap}
\end{theorem}

\begin{proof}
The proof is the same as for Theorem~\ref{thm:non-overlap}, with the addition
of an extra case. We have to consider the case when two !-boxes $b_1$
and $b_2$ overlap on the interior of several closed wires:
  \cstikz{thm_overlap.tikz}
In this case, the wire-vertices in the interior of the wire are in both !-boxes
while the endpoints of the wire are node-vertices which are not in the overlap.
Then, several applications of Lemma~\ref{lem:overlap} (one for each wire in the
overlap) can be used to reduce the problem to showing that $H - W$ can be
handled, where $W$ is the set of all wire-vertices in the overlap of $b_1$
and $b_2$. Because $b_1$ and $b_2$ overlap trivially, then in $H-W$ $b_1$
and $b_2$ do not overlap at all and the proof may be finished using the same
arguments as in Theorem~\ref{thm:non-overlap}.

Note, that the proof of Theorem~\ref{thm:non-overlap} (and the accompanying
lemmas) carries through here in the same way as it only manipulates the grammar
$G$. All of the encoding edges are established by Lemma~\ref{lem:overlap} and
only relate to the case of trivial overlap.
\end{proof}

\begin{remark}
This theorem and Lemma~\ref{lem:overlap} were proved in \cite{gam} using
a notion of equality called "wire-encoding". This notion is now superseded
by the encoding/decoding features of encoded B-edNCE grammars and we have
restated both propositions in these terms.
\end{remark}

Combining this theorem with the other results so far, we get the following
relationship between context-free graph grammars, !-graphs and
encoded B-edNCE grammars, where we consider string graph languages only:
\cstikz{inclusion-figure-encoded.tikz}
All regions of the diagram are populated, except for possibly the region
between the dashed line and $\textbf{BG}$. While we have not proven this,
we conjecture that this region is also populated.

\section{B-ESG grammars and their languages}\label{sec:b-esg-languages}

Our goal in this thesis is to model string diagrammatic reasoning. We also wish
to implement machine support for the reasoning process, so we have to identify
decidable conditions on our grammars which imply that they generate string
graphs only. In the previous section we saw that encoded B-edNCE grammars have
satisfactory expressive power.  In this section, we will identify sufficient
conditions on the productions of an encoded B-edNCE grammar which ensure that
it can only generate string graphs. The grammars which satisfy these conditions
are called \emph{B-ESG} grammars.

We begin by introducing a few auxiliary definitions that will help us to
introduce our first major notion, called \emph{wire-consistency}. We will
first define wire-consistency in terms of the structure of a grammar and
show that it is a decidable property. After that, we will provide a
dynamic characterisation -- we will show how the sentential forms of
wire-consistent grammars behave.

\begin{definition}
Given a graph $H$ with a vertex $x \in H$, we say that $x$ is incident to a
$(\sigma, \beta, d)$ edge, if there exists a vertex $y \in H$ with label
$\sigma$
and an edge $(y, \beta, x)$ if $d=in$ or an edge $(x, \beta, y)$ if $d=out$.
\end{definition}

\begin{definition}[Context-cardinality]
Given a production $p$ of a B-edNCE grammar $G$, where $p$ contains a
vertex $x$, we say that its $(\sigma,\beta,d)$
\textit{context-cardinality} is $n$ if the number of $(\sigma,\beta,d)$-edges
incident to $x$ plus the number of connection instructions of the form
$(\sigma,\alpha, \beta, x, d)$ is equal to $n$.
\end{definition}

\begin{definition}[Context-passing]
Given a B-edNCE grammar $G$, terminal vertex label $\sigma \in \Delta$,
edge labels $\alpha, \beta \in \Gamma$ and direction $d \in \{in,\ out\}$, we
define a
binary relation $\mathcal P_{\sigma}^d(\alpha,\beta)$ between nonterminal
vertices $x,y$ in productions $p_1, p_2$ of $G$ respectively. We will refer to
$\mathcal P$ as the \emph{single-step context-passing} relation. We say
$x\ \mathcal P_{\sigma}^d(\alpha, \beta)\ y$, if:
\begin{itemize}
  \item $x$ has context-cardinality $(\sigma,\alpha,d)$ at least one
  \item $p_2$ has the same production label as that of $x$
  \item there exists a connection instruction
        $(\sigma, \alpha, \beta, y,d)$ in $p_2$
\end{itemize}
We define the \emph{multi-step context-passing} relation to be
$\mathcal Q_{\sigma}^d(\alpha,\alpha')$ again between a pair of nonterminal
vertices. We will say
$x\ \mathcal Q_{\sigma}^d(\alpha,\alpha')\ x'$ if there exist
nonterminal vertices $x_1,\ldots,x_n$ and edge labels
$\alpha_1,\ldots,\alpha_n$, such that:
\begin{align*}
x \ \mathcal P_{\sigma}^d(\alpha, \alpha_1)\ x_1\ &\land \\
x_1 \ \mathcal P_{\sigma}^d(\alpha_1, \alpha_2)\ x_2\ &\land \\
\cdots & \\
x_{n-1} \ \mathcal P_{\sigma}^d(\alpha_{n-1}, \alpha_n)\ x_n\ &\land \\
x_n \ \mathcal P_{\sigma}^d(\alpha_n, \alpha')\ x'\ & \\
\end{align*}
\end{definition}

\begin{remark}
To be more precise, both $\mathcal P$ and $\mathcal Q$ are
actually a \emph{family} of binary relations, which are parametrised over
$\sigma \in \Delta, \alpha \in \Gamma, \alpha' \in \Gamma, d \in \{in, out\}$.
Instead of defining them as relations over a 6-tuple of sets, we opt for this
notation as it is more convenient for showing the next proposition.
\end{remark}

\begin{example}\label{ex:shit-relation}
Consider the following grammar:
\cstikz{context-pass-grammar.tikz}
where $x, y, z$ are the nonterminal vertices with labels $X,Y,Z$ respectively.
Then,
$x\ \mathcal P_{\sigma}^{in}(\alpha, \beta)\ y$ and
$y\ \mathcal P_{\sigma}^{in}(\beta, \gamma)\ z.$ This also means 
$x\ \mathcal Q_{\sigma}^{in}(\alpha,\gamma)\ z.$
\end{example}

\begin{proposition}
For a given B-edNCE grammar $G$, the multi-step context-passing relation
$\mathcal Q$ is computable.
\end{proposition}
\begin{proof}
Fix a direction $d$ and vertex label $\sigma$. Define a new relation $Q_0$,
such that for every $\alpha, \beta \in \Gamma$, we have:
\[x\ Q_0(\alpha, \beta) \ y \text{ iff } x\ \mathcal P_{\sigma}^d(\alpha,
\beta)\ y\]
In other words, $Q_0(\alpha, \beta) = \mathcal P_{\sigma}^d(\alpha,\beta)$.
$\mathcal P$ is obviously computable and therefore so is $Q_0$.
Next, having already computed $Q_{n}$, we define:
\[
x\ Q_{n+1}(\alpha, \beta)\ y \text{ iff }
\left( x\ Q_{n}(\alpha, \beta)\ y \right) \lor \left(
\exists z, \gamma.\ x\ Q_n(\alpha,
\gamma)\ z\ \land z\ \mathcal P_{\sigma}^d(\gamma, \beta)\ y
\right)
\]
The computation stops when we find a $k$ such that $Q_{k+1} = Q_k$. Note,
that this procedure is guaranteed to terminate, as we have finitely many
nonterminal vertices, exactly two possible directions and finitely many
labels and each step of the procedure increases the size of the relation,
except for the last one.

We claim that upon termination, $Q_k(\alpha, \beta) = \mathcal Q_{\sigma}^d
(\alpha, \beta).$ In one direction, this is easy to see: if
$x\ Q_k(\alpha, \beta)\ y$, then obviously also $x\ \mathcal
Q_{\sigma}^d(\alpha,\beta)\ y$ holds by construction of $Q_i$.

In the other direction, assume  $x\ \mathcal Q_{\sigma}^d(\alpha,\beta)\ y$.
Thus, there exist
nonterminal vertices $x_1,\ldots,x_n$ and edge labels
$\alpha_1,\ldots,\alpha_n$, such that:
\[
x \ \mathcal P_{\sigma}^d(\alpha, \alpha_1)\ x_1\ \land 
x_1 \ \mathcal P_{\sigma}^d(\alpha_1, \alpha_2)\ x_2\ \land 
\cdots  \land
x_{n-1} \ \mathcal P_{\sigma}^d(\alpha_{n-1}, \alpha_n)\ x_n\ \land
x_n \ \mathcal P_{\sigma}^d(\alpha_n, \beta)\ y
\]
If $n\leq k$, then obviously also $x\ Q_k(\alpha, \beta)\ y$ holds by
construction. If $n>k$, observe that $x\ Q_k(\alpha, \alpha_{k+1})\ x_{k+1}$ by
the same arguments. Because $Q_k = Q_{k+1},$ we get 
$x\ Q_k(\alpha, \alpha_{k+2})\ x_{k+2}$. By iterating this argument, we
eventually get $x\ Q_k(\alpha, \beta)\ y$, as required.
\end{proof}

Now we can introduce wire-consistency. Note that the previous
proposition implies that wire-consistency for a B-edNCE grammar is a decidable
property.

\begin{definition}[Wire-consistent grammar]
\label{def:wire-consistent-crap}
We say a B-edNCE grammar $G$ is \textit{wire-consistent}, if for every
production $p$ of $G$, every $\sigma \in \Delta, \alpha \in \Gamma, \beta \in
\Gamma, d \in \{in, out\}$ and every nonterminal vertex $x$ of $p$ with context
cardinality $(\sigma, \alpha, d)$ at least two, the following holds: every
production with label in the set $\{\lambda(y)\ |\ x\ \mathcal
Q^d_{\sigma}(\alpha,\beta)\ y\}$ cannot have a connection instruction
$(\sigma, \beta, \gamma, z, d)$ where $z$ is a wire-vertex and $\gamma
\in \Gamma$.
\end{definition}

This definition provides a static description of wire-consistency. Its
purpose is to show that this notion is decidable, however this
definition is rather cumbersome to work with. What we will find more useful
is the dynamic description of wire-consistency which is presented in
Lemma~\ref{lem:wire-cons}. Simply put, a grammar is wire-consistent if
each connection instruction attached to a wire-vertex can establish at
most one bridge in any derivation. This is clearly a necessary property,
because string graphs cannot have wire-vertices with in-degree (out-degree)
more than one.

\begin{example}
All of the grammars presented in this thesis so far are wire-consistent,
except for the grammar $G$ from Example~\ref{ex:shit-relation}.
$G$ clearly violates the (static) definition of wire-consistency, so
let's see how its only possible concrete derivation behaves:
\cstikz{wire-cons-derive.tikz}
So, we see that the wire-vertex established at the end has in-degree two,
despite the fact that in production $p_4$ the wire-vertex has a single
connection instruction associated to it.
\end{example}

The next lemma is crucial, because it shows that the scenario from the
previous example is impossible -- wire-consistent grammars ensure that
connection instructions associated to a wire-vertex cannot create more
than one bridge to the same wire-vertex, in any derivation.

\begin{lemma}\label{lem:wire-cons}
Given a wire-consistent grammar $G$ and a sentential form $H = (D_1,C_1)$, then
applying any production $p = X \to (D_2,C_2)$ to $H$ has the following effect:
for any $\sigma \in \Delta,\beta \in \Gamma, d\in \{in, out\}$ and any
wire-vertex $w \in D_2$ the substitution process will establish at most one
$(\sigma, \beta, d)$-bridge between $w$ and $D_1$ using any connection
instruction $(\sigma, \alpha, \beta, w, d) \in C_2$, where $\alpha \in \Gamma$.
\end{lemma}

\begin{proof}
Consider the graph $H'$ which is the result of applying $p = X \to (D_2,C_2)$
to $H$, that is $H \derive_{x,p} H'$. Assume for contradiction that the
wire-vertex $w \in H' - H$ is connected via $(\sigma, \beta, d)$-edges to two
vertices $v, v' \in H - \{x\} \subseteq H'$ and these bridges are established
by a single connection instruction $(\sigma, \alpha,\beta, w, d) \in C_2.$
Then, both $v$ and $v'$ are adjacent to the nonterminal vertex $x$ which is
being replaced in $H$. Moreover, there must exist edges $(v, \alpha, x)$ and
$(v', \alpha, x)$ if $d=in$ or edges $(x, \alpha, v)$ and $(x, \alpha, v')$
otherwise.

$H'$ is a sentential form, then it must be the yield of a derivation tree $T$.
Thus, $H' = \emph{yield}(T).$ Because $G$ is a B-edNCE grammar, observe that
only the productions in $T$ which are a predecessor of $p$ in $T$ may affect
the neighbourhood of $w$. The other productions replace nonterminal vertices
which cannot be connected to $w$. So, without loss of generality, we can
assume:
\[S \derive_{x_0,p_0} H_1 \Longrightarrow_{x_1, p_1} H_2 \Longrightarrow_{x_2,
p_2} \cdots \Longrightarrow_{x_n, p_n} H_n = H \derive_{x,p} H'\]
where all the listed productions are (not necessarily direct) predecessors of
$p$ in $T$ (see Figure~\ref{fig:predecessors}.
\begin{figure}[h]
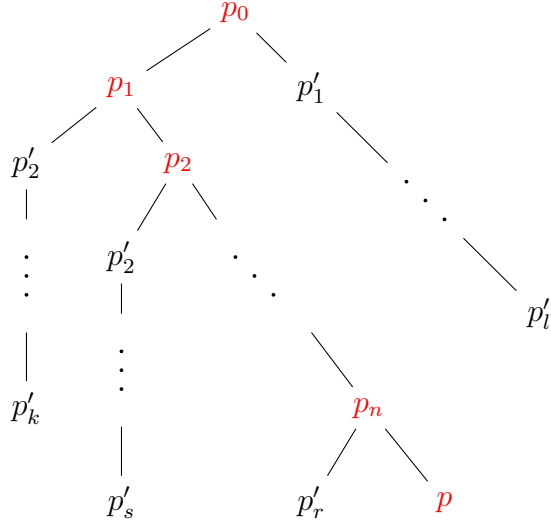

\cstikz{predecessors.tikz}
\caption{Derivation tree for $H'$ and predecessors of $p$ (coloured in red).}
\label{fig:predecessors}
\end{figure}

If both $v, v' \not\in H_{n-1}$ then they are both in the body of production
$p_n$. Hence, the nonterminal node $x$ in $p_n$ must have context-cardinality
$(\sigma, \alpha, d)$ equal to 2 which then contradicts with $G$ being
wire-consistent.

If $v' \in H_{n-1}$ but $v \not\in H_{n-1}$, then $v$ is a vertex in the
body of $p_n$ connected to the nonterminal $x$ (via edge labelled $\alpha$) and
$p_n$ must have a connection instruction of the form $(\sigma, \alpha',\alpha,
x, d)$. Again, the nonterminal $x$ in $p_n$ has context-cardinality $(\sigma,
\alpha, d)$ equal to 2 which contradicts with $G$ being wire-consistent.

Thus, both $v, v' \in H_{n-1}$ and the production $p_n$ must have a connection
instruction $(\sigma, \alpha',\alpha, x ,d)$, for some $\alpha' \in \Gamma$.

By iterating this argument, we conclude that the nonterminal vertices
$x_1,\ldots, x_n$ in productions $p_0, \ldots, p_{n-1}$ satisfy the
context-passing relation $x_i\ \mathcal Q_{\sigma}^d(\alpha,\beta)\ x$ and that
the two vertices $v$ and $v'$ must
be created by the initial production $p_0$ and must be connected to the
nonterminal
$x_1$ via a $(\sigma, \alpha'',d)$-edge. However,
$H_1$ is the body of production $p_0$. Thus, the
nonterminal $x_1$ in the body of $p_0$ has $(\sigma, \alpha'',d)$ context
cardinality equal to at least 2. At the same time we know 
$x_1\ \mathcal Q_{\sigma}^d(\alpha,\beta)\ x$
for productions $p_0$ and $p_n$ and then we get a contradiction with the fact
that $G$ is wire-consistent.
\end{proof}

In other words, during the derivation process, a wire-consistent grammar
can establish at most one bridge between a newly created wire-vertex and the
previously generated graph for a single connection instruction associated
to that wire-vertex. We only need to introduce one additional notion before we
can classify the grammars we are interested in.

\begin{definition}[Production degree]
Given a B-edNCE grammar $G$ and a production $X \to (D,C)$ with vertex $v \in
D$, the \emph{connection instruction in-degree} of $v$ is the number of
connection instructions $(\sigma, \alpha, \beta, v, in)$, for some $\sigma \in
\Sigma, \alpha, \beta \in \Gamma$. Similarly, the \emph{connection instruction
out-degree} of $v$ is the number of connection instructions $(\sigma, \alpha,
\beta, v, out)$, for some $\sigma \in \Sigma, \alpha, \beta \in \Gamma$.  The
\emph{production in-degree} (\emph{production out-degree}) of $v$ is the sum of
the connection instruction in-degree (connection instruction out-degree) and
the in-degree (out-degree) of $v$.
\end{definition}

\begin{example}
The above definition is a simple generalisation of the standard notion
of in-degree (out-degree). For example, in the following grammar:
\cstikz{grammar-degree.tikz}
the $X$-labelled vertex has production in-degree two, the $Y$-labelled
vertex has production in-degree two, the $Z$-labelled vertex has production
in-degree one, the wire-vertex has production in-degree one and the two
$\sigma$-labelled vertices have production out-degrees one.
\end{example}

We can now provide the central definition of this chapter.

\begin{definition}[B-ESG grammar]\label{def:besg} \rm
A \textit{B-ESG} grammar is an encoded B-edNCE grammar $B = (G, T),$
where $G=(\Sigma,\Delta,\Gamma,\Gamma,P,S)$ is wire-consistent
and such that for every production $X \to (D,C) \in P$, the following
conditions are satisfied:
  \begin{enumerate}
    \item[N1:] Any edge connecting two node-vertices must carry an encoding
label.
    \item[N2:] Any connection instruction of the form
$(N, \alpha, \beta, v, d)$ where $N$ is a node-vertex label and $v$ is a
node-vertex, must have $\beta \in \mathcal E$.
    \item[W1:] Every wire-vertex in $D$ has production in-degree at most one
and production out-degree at most one.
    \item[W2:] For $W$ a wire-vertex label and each $\gamma$ and $d$, there is
at most one connection instruction of the form $(W, \gamma,\delta, v, d)$,
where $\delta \in \Gamma$, $v \in V_D$.
  \end{enumerate}
\end{definition}

The conditions N1 and N2 guarantee that node-vertices never become directly
connected by an edge, unless that edge has an encoding label. Conditions W1 and
W2 together with wire-consistency ensure that wires never split, i.e.
wire-vertices always have in-degree (out-degree) at most one.

Before we present the main result of this section, we introduce a helpful
definition which describes the sentential forms of B-ESG grammars.

\begin{definition}[ESG-form \cite{icgt}]
Given a B-edNCE grammar $G$, we call a sentential form an \textit{ESG-form}
(Encoded String Graph form) if it is an encoded string graph, which possibly
has some additional nonterminals that are either connected to node-vertices or
are connected to wire-vertices in such a way that all wire-vertices have
in-degree (out-degree) at most one.
\end{definition}

\begin{example}
Consider the following two graphs:
\cstikz{esg-form-example.tikz}
The graph on the left is an ESG-form, but the graph on the right is not, as
it contains a wire-vertex with in-degree two.
\end{example}

We can now prove the main theorem of this section which shows that
the language of any B-ESG
grammar consists of string graphs only.

\begin{theorem}\label{thm:besg_language}
  Every graph in the language of a B-ESG grammar is a string graph.
\end{theorem}

\begin{proof}
Let $B=(G,T)$ be a B-ESG grammar and let's consider an arbitrary concrete
derivation $S
= H_1 \Longrightarrow_{x_1, p_1}^G H_2 \Longrightarrow_{x_2, p_2}^G \cdots
\Longrightarrow_{x_n, p_n}^G H_n \Longrightarrow_*^T F$. Using
Lemma~\ref{lem:decoding} we can reduce the problem to showing that $H_n$ is an
encoded string graph. We will prove, using
induction, that every graph $H_i$ is an ESG-form.
From this, the theorem follows immediately, because $H_n$ is an ESG-form with
no nonterminal vertices which means it is simply an encoded string graph.

In the base case, $H_1$ contains only the initial nonterminal $S$ and it is
obviously an ESG-form. Assuming that $H_k$ is an ESG-form, let's consider
$H_{k+1}$ which is the result of applying production $p_k$ at vertex $x_k$ in
$H_k$.

Condition N2 guarantees that any newly created node-vertices in $H_{k+1}$ will
be adjacent to the previously established node-vertices in $H_k$ only via
encoding bridges. In addition to this, condition N1 ensures that the newly
created node-vertices in $H_{k+1}$ can be adjacent to other newly created
node-vertices only if the connecting edges have encoding labels. Combining this
with the fact that $H_k$ is an ESG-form means that all node-vertices in
$H_{k+1}$
can only be adjacent to other node-vertices via encoding edges, as required.

Next, we will show how the dynamic behaviour on newly established wire-vertices
is influenced by condition W1 and the wire-consistency of the grammar $G$. For
a newly created wire-vertex $w$ in $H_{k+1}$, we consider two cases and examine
its in-degree. If $w$ is incident to an in-edge with another newly created
vertex in $H_{k+1}$, then condition W1 ensures that this is the only in-edge of
$w$ and that it is not connected to any previously established vertices in
$H_k.$ If $w$ has at least one in-edge with some previously created vertices in
$H_k,$ then condition W1 guarantees that all of the in-edges of $w$ have been
established by consuming a single connection instruction.
Then, using
Lemma~\ref{lem:wire-cons},
we conclude that the in-degree of $w$ is exactly one. Using the same arguments
for out-edges we conclude that any newly created wire-vertex in $H_{k+1}$ can
have in-degree at most one and out-degree at most one. Note, that it is
possible for newly created wire-vertices to have in-degree (out-degree) zero
when their production in-degree (out-degree) is zero.

Condition W2 prevents increasing the in-degree or out-degree of previously
established wire-vertices in $H_k$. Combining this with the previous results
and the induction hypothesis, we get that all wire-vertices in $H_{k+1}$ have
in-degree or out-degree at most one. All of
this together means that $H_{k+1}$ is an ESG-form.
\end{proof}

\section{Expressivity of B-ESG grammars}\label{sec:expres-shit}

In the previous section we introduced B-ESG grammars which are simply
encoded B-edNCE grammars that satisfy some additional conditions. These
conditions ensure that B-ESG grammars generate only string graphs. It is not
immediately obvious whether these conditions are restrictive in the sense that
B-ESG grammars would not be able to simulate some encoded B-edNCE grammars.
In this section we will show that this is not the case -- B-ESG grammars have
the same expressive power as encoded B-edNCE grammars on string graphs.

\begin{theorem}\label{thm:necessary}
Given an encoded B-edNCE grammar $B=(G,T),$ such that $L(B)$ is a language
consisting of string graphs, then there exists a B-ESG grammar $B'=(G',T)$,
such that $L(B') = L(B).$ Moreover, $G'$
can be constructed effectively from $G$ and $L(G') = L(G)$.
\end{theorem}

\begin{proof}
The grammar $G'$ can be constructed effectively from $G$ by using
Theorem~\ref{thm:normal-form}. That is, we consider an equivalent grammar in
CNF form. We
have to show that $G'$ satisfies Definition~\ref{def:besg}.

Since $L(G') = L(G)$, then using the fact that
$G'$ never establishes blocking edges, we conclude that all terminal sentential
forms of $G'$ are encoded string graphs (again, by using Lemma~\ref{lem:decoding}).

Next, we can show that all sentential forms of $G'$ are ESG-forms. We can show
this by contradiction -- if there exists an ESG-form $H$ which contains a
wire-vertex with in-degree (out-degree) more than one, then for any concrete
derivation $H \derive_{*} H'$, $H'$ will also contain a wire-vertex with
in-degree (out-degree) more than one, because $G'$ is neighbourhood-preserving
(the degree of vertices cannot decrease during a derivation). However, this
contradicts with the fact that all terminal sentential forms of $G'$ are
encoded string graphs. The other case to consider is when $H$ contains a
nonencoding edge between a pair of node-vertices. But then, for any concrete
derivation $H \derive_{*} H'$, $H'$ will also contain an edge with nonencoding
label connecting a pair of node-vertices, so we establish a contradiction
again. Thus, all sentential forms of $G'$ are ESG-forms.

Let's consider an arbitrary production $p=X \to (D,C)$ of $G'$. $G'$ is reduced,
therefore there exists a derivation $S \Longrightarrow_* H_1
\Longrightarrow_{x,p} H_2$.

If $p$ violates condition N1, then $p$ establishes two new node-vertices which
are connected via a non-encoding edge. Therefore, $H_2$ is not an ESG-form and
we get a contradiction.

Let's assume $p$ violates condition N2. Therefore, $p$ contains a connection
instruction $(N, \alpha, \beta, v, d),$ where $N \in \mathcal N$, $v$ is a
node-vertex in $rhs(p)$ and $\beta \not \in \mathcal E.$
$G'$ is context consistent
and contains no useless connection instructions, therefore
$(N, \alpha, d) \in cont_{H_1}(x).$ This means that applying production
$p$ to $H_1$ will establish a nonencoding edge between two node-vertices
and we get a contradiction with the fact that $H_2$ must be an ESG-form.

Assume $p$ violates condition W1. Without loss of generality, let's assume the
offending wire-vertex $w$ has connection in-degree $n$ and in-degree $m$ in $D$
with $m+n>1$. As in the previous case, by using context consistency and the
fact that $G'$ doesn't have useless connection instructions, we see that
applying $p$ to $H_1$ will create at least $n$ in-edges to the newly created
wire-vertex $w$ via the embedding process. However, $w$ also has an additional
$m$ in-edges in $D$ and therefore its in-degree in $H_2$ is at least $m+n>1$
which is a contradiction with the fact that $H_2$ is an ESG-form.

Assume $p$ violates condition W2. Therefore, we may assume, without loss of
generality, that there are two connection instructions $(W, \gamma, \delta_1,
x_1, in),\ (W, \gamma, \delta_2, x_2, in),$ where the $\delta_i$ and $x_i$ are
not necessarily distinct. Again, by making use of context consistency and
the fact that $G'$ has no useless connection instructions, we get that the
nonterminal $x$ in $H_1$ must have a $\gamma$-labelled in-edge adjacent to a
$W$-labelled wire-vertex $y$. Then, applying production $p$ to $H_1$ will make
the out-degree of the wire-vertex $y$ to be equal to two in $H_2$. This is
a contradiction with the fact that $H_2$ is an ESG-form.

Finally, let's assume that $G'$ is not wire-consistent. Without loss of
generality, there exists a production $p$, such that it
contains some nonterminal vertex $x$ which has context cardinality
$(\sigma,
\alpha, in)$ equal to two. Because $G'$ is
reduced, there must exist a derivation $S \Longrightarrow_* H_1
\Longrightarrow_{p} H_2$. By making use of the same arguments as in the
previous cases (no useless connection instructions, context consistency) we see
that the nonterminal $x$ in $H_2$ must have two $\alpha-$labelled in-edges
incident to two $\sigma$-labelled vertices. Let's call these two
$\sigma$-labelled
vertices $v$ and $u$. Since we have assumed that $G'$ is not wire-consistent,
there exists a production $p'$ which contains a nonterminal vertex $y$ which is
in the context-passing relation with $x$, that is $x\ \mathcal Q^{in}_{\sigma}
(\alpha, \beta)\ y,$ for some $\beta \in \Gamma$. Then, there must exist a
derivation:
\[S \Longrightarrow_* H_1 \Longrightarrow_{p} H_2
\Longrightarrow_* H_3 \Longrightarrow_{p'} H_4.\]
where the production $p'$ creates a nonterminal vertex $y$, which, thanks to
the context-passing assumption, will
have two $\beta$-labelled in-edges incident to the previously established
vertices $v$ and $u$ (which have the same label).  Finally, by assumption,
there must exist a production $p''$ with production label the same as that of
$y$ which
moreover must have a connection instruction $(\sigma, \beta,\gamma, z, in)$
for some $\gamma \in \Gamma$ and where $z$ is
a wire-vertex. Applying production $p''$ to $H_4$ would therefore result in a
graph where the wire-vertex $z$ has in-degree more than one (it would be
adjacent to both $u$ and $v$) which is a contradiction with the fact that all
sentential forms of $G'$ are ESG-forms.
\end{proof}

Therefore, B-ESG grammars have the same generating power as encoded
B-edNCE grammars (on string graphs). However, unlike (general) encoded B-edNCE
grammars, B-ESG grammars can only generate string graphs. Therefore, for our
purposes, this makes B-ESG grammars the obvious choice over encoded B-edNCE
grammars.  Combining Theorem~\ref{thm:necessary} with the results from the
previous section, we get the following relationship between B-ESG grammars,
!-graphs and context-free graph grammars:
\cstikz{inclusion-figure-bgto.tikz}

\section{Decidability properties of B-ESG grammars}\label{sec:decide}

When working with families of string diagrams, it is necessary to be able
to determine whether a given concrete diagram is an instance of a given
family of diagrams. In terms of string graphs and B-ESG grammars, this
means that we should be able to decide whether a given string (up to
wire-homeomorphism) is in the language of a B-ESG grammar. This is the
membership problem for B-ESG grammars, which we will show is decidable
in this section.

Another problem which is crucial is the ability to enumerate all the matches of
all instances of an equational schema into a concrete diagram, so that we
can rewrite it. In terms of string graphs and B-ESG grammars, this means
that given a B-ESG grammar $B$ and a string graph $H$, we should be able to
enumerate all derivations $S \Longrightarrow_*^B F$, and all matches
$F \to H$ (up to wire-homeomorphism). This is the match enumeration problem
which we will also show is decidable. Note, that decidability of this problem
implies decidability of the membership problem.

To show that these problems are decidable, we will use a logical description
(cf. Subsection~\ref{sub:mso}) of the wire-homeomorphism class of a given
encoded string graph. In particular, we will first show that for given an
encoded string graph $H$, we can effectively construct a closed MSO formula
$\phi_H$, such that $L(\phi_H) = [H]_{\sim},$ that is, the MSO language which
the formula induces is the wire-homeomorphism class of $H$.

\subsection{Wire-homeomorphism as MSO definable language}

We begin by proving a few lemmas which are helpful for establishing the
mentioned logical characterisation. Recall that we assume our labelling
alphabets $\Sigma$ and $\Gamma$ are finite. This is necessary for the
decidability of the problems we are interested in.

\begin{lemma}\label{lem:aux_msol}
There exists MSO formulas for the following properties:
\begin{itemize}
\item edge(u,v) : There exists an edge between u and v.
\item encoding-edge(u,v) : There exists an edge between u and v
which carries an encoding label.
\item non-encoding-edge(u,v) : There exists an edge between u and v
which carries a non-encoding label.
\item wire-vertex(u) : u is a wire-vertex. 
\item node-vertex(u) : u is a node-vertex. 
\item wire-edge(u,v) : u and v are wire-vertices directly connected by an edge
\item one-edge(u,v) : there exists exactly one edge with source u and target v
\end{itemize}
\end{lemma}

\begin{proof}
The required formulas are given by:
\begin{align*}
\emph{edge}(u,v) &:= \bigvee_{\gamma \in \Gamma} \emph{edge}_{\gamma}(u,v) \\
\emph{encoding-edge}(u,v) &:= \bigvee_{\gamma \in
  \mathcal{E}} \emph{edge}_{\gamma}(u,v) \\
\emph{non-encoding-edge}(u,v) &:= \bigvee_{\gamma \in
  \Gamma - \mathcal{E}} \emph{edge}_{\gamma}(u,v) \\
\emph{wire-vertex}(u) &:= \bigvee_{\sigma \in
  \mathcal{W}} \emph{lab}_{\sigma}(u) \\
\emph{node-vertex}(u) &:= \bigvee_{\sigma \in
  \mathcal{N}} \emph{lab}_{\sigma}(u) \\
\emph{wire-edge}(u,v) &:= \emph{wire-vertex}(u) \land \emph{wire-vertex}(v)
\land edge(u,v) \\
\emph{one-edge}(u,v) &:= \bigvee_{\gamma} \left(
  \emph{edge}_{\gamma}(u,v) \land \bigwedge_{\eta \not=
  \gamma} \lnot \emph{edge}_{\eta}(u,v) \right)
\end{align*}
\end{proof}

\begin{lemma}\label{lem:string_graph_msol}
There exists a closed MSO formula $esg$ which
is satisfied by a graph $H$ iff $H$ is an encoded string graph.
\end{lemma}

\begin{proof}
A graph $H$ is an encoded string graph iff the following three conditions
are satisfied.
\begin{itemize}
\item [1.] There exists no non-encoding edge between two node-vertices
\item [2.] The in-degree of any wire-vertex is at most one
\item [3.] The out-degree of any wire-vertex is at most one
\end{itemize}
Each of those three conditions can easily be encoded using MSO logic. The
formula is given by:
\begin{align*}
  esg &:= \left[ \neg \exists u,v.\ \emph{node-vertex}(u) \land
  \emph{node-vertex}(v) \land \emph{non-encoding-edge}(u,v) \right] \\
  &\land \left[ \forall u,v,w.\ \emph{wire-vertex}(w) \land \emph{edge}(u,w)
  \land \emph{edge}(v,w) \implies u=v \land \emph{one-edge}(u,w) \right] \\
  &\land \left[ \forall u,v,w.\ \emph{wire-vertex}(w) \land \emph{edge}(w,u)
  \land \emph{edge}(w,v) \implies u=v \land \emph{one-edge}(w,u) \right] \\
\end{align*}
where each line represents one of the three conditions.
\end{proof}

\begin{lemma}\label{lem:wire-path}
The transitive closure of wire-edge$(u,v)$ can be expressed in MSO logic and we
denote it with wire-edge$^+(u,v).$
\end{lemma}
\begin{proof}
Follows immediately from Lemma 5.2.7, pp. 327 of \cite{msol}. In particular,
$\emph{wire-edge}^+(u,v)$ is given by:
\begin{align*}
&\emph{wire-edge}^+(u,v) := 
\forall X. \{ \forall y.\left(\emph{wire-edge}(u, y) \implies y \in
X \right) \land \\
& \forall y,z.\left( y \in X \land \emph{wire-edge}(y,z) \implies
z \in X \right) \} \implies v \in X \\
\end{align*}
\end{proof}

\begin{corollary}
The reflexive and transitive closure of wire-edge$(u,v)$, denoted
wire-edge$^*(u,v)$ can be expressed in MSO logic.
\end{corollary}
\begin{proof}
It is given by :
\begin{align*}
\emph{wire-edge}^*(u,v) := \emph{wire-edge}^+(u,v) \lor 
\left( u=v \land \emph{wire-vertex}(u) \right)
\end{align*}
\end{proof}

The next definition generalises the notion of wire-homeomorphism from
string graphs to encoded string graphs. The two notions are essentially the
same -- two encoded string graphs are wire-homeomorphic if one can be
obtained from the other by increasing or decreasing the length of some
wires.

\begin{definition}[Wire-homeomorphic encoded string graphs]
\label{def:encoded-wire-homeo}
Two encoded string graphs $H$ and $H'$ are called \textit{wire-homeomorphic},
written
$H \sim H'$ if $H'$ can be obtained from $H$ by either merging two adjacent
wire-vertices (top) or by splitting a wire-vertex into two adjacent
wire-vertices (bottom) any number of times:
    \[ \stikz{two-wires.tikz} \ \ \mapsto \ \ \stikz{one-wire.tikz}\]
    \[\stikz{one-wire.tikz} \ \ \mapsto \ \ \stikz{two-wires.tikz} \]
in the same way as the definition for wire-homeomorphic string graphs
(Definition~\ref{def:wire-homeo}).
\end{definition}

The main theorem of this subsection is proved next. 

\begin{theorem}\label{thm:wire-homeo-mso}
Given an encoded string graph $H$, we can effectively construct a closed
MSO formula $\phi_H$, such that $L(\phi_H) = [H]_{\sim}.$
\end{theorem}

\begin{proof}
$H$ can easily be transformed into the minimal representative of its
wire-homeomorphism class, so we assume without loss of generality that
$H$ is minimal in that sense.

Let $H$ have $n$ node-vertices, given by $u_1,\ldots, u_n$ and $m$ wire-vertices
given by $w_1,\ldots, w_m$. First, we express the fact that $\phi_H$ should not
be
satisfied by graphs with more than $n$ node-vertices. This is given by the
following formula :
\begin{align*}
\emph{nodes} := \lnot \exists u_1, ...,u_{n+1}. \left( \bigwedge_i
  \emph{node-vertex}(u_i) \right) \land \left( \bigwedge_{i \not = j} u_i
  \not =   u_j \right)
\end{align*}
The general form of $\phi_H$ is given by the following :
\begin{align*}
\phi_H := \emph{esg} &\land \emph{nodes} \land \exists u_1,..., u_n,
  w_1,..., w_m. 
  \left( \bigwedge_i \emph{node-vertex}(u_i) \right) \land
  \left( \bigwedge_i \emph{wire-vertex}(w_i) \right) \land \\
  &\land \left( \bigwedge_{i \not = j} u_i \not =   u_j \right)
  \land \left( \bigwedge_{i \not = j} w_i \not =   w_j \right) \land
  \left( \bigwedge_k \psi_k \right)
\end{align*}
where the $\psi_k$ are some MSO formulas (which can depend on the free $u_i$
and $w_j$ variables) that are described below.

So far, we can see that $\phi_H$ will be satisfied by encoded string graphs
which have exactly $n$ node-vertices and at least $m$ wire-vertices.

Next, we add constraints on the edges incident to node-vertices on both ends.
Let the parallel edges between vertices $u_i,$ $u_j \in V_H,$ be indexed by a
set $I_{i,j}.$ Thus, the set of edges with source $u_i$ and target $u_j$ in $H$
is given by $\{(u_i,\alpha, u_j)\ |\ \alpha \in I_{i,j}\}.$ Because $H$ is an
encoded string graph, then we must have $I_{i,j} \subseteq \mathcal{E},$ that
is all these edges must have encoding labels. Therefore, for each $u_i$ and
$u_j,$ with $i \not = j,$ we must have the following formula in
$\phi_H$ (as one of the $\psi_k$) :  
\begin{align*}
\left( \bigwedge_{\alpha \in I_{i,j}} \emph{edge}_{\alpha}(u_i,u_j) \right)
  \land
\left( \bigwedge_{\alpha \not \in I_{i,j}} \lnot
  \emph{edge}_{\alpha}(u_i,u_j) \right)
\end{align*}
which describes precisely what edges are allowed and required between a pair of
node-vertices. Note, that our notion of graph does not allow for loops on
vertices and because of that there's no need to consider the case when $u_i$
and $u_j$ are the same vertex.

Next, we describe the wires and isolated wire-vertices which any graph
satisfying $\phi_H$ is required to have. For each wire-vertex $w_i \in H,$ we
consider all possible cases of its neighbourhood in $H$ and for each case we
add an MSO formula as one of the $\psi_k$ denoted above.

If $w_i$ is an isolated wire-vertex, then we add the following formula:
\begin{align*}
\left( \lnot \exists v.\ \emph{edge}(v,w_i) \right)
\land \left( \lnot \exists v.\ \emph{edge}(w_i,v) \right)
\end{align*}
If $w_i$ is an input and has an out-edge to a wire-vertex $w_j$, then
$w_j$ must be an output and we add the following formula:
\begin{align*}
  \emph{wire-edge}^*(w_i, w_j)
\land \left( \lnot \exists v.\ \emph{edge}(v,w_i) \right)
\land \left( \lnot \exists v.\ \emph{edge}(w_j,v) \right)
\end{align*}
If $w_i$ is an input and has an out-edge to the node-vertex $u_j$, then
we add the following formula:
\begin{align*}
\left(
   \exists w'. \emph{wire-edge}^*(w_i, w') \land \emph{edge}(w',u_j)
\right)
\land \left( \lnot \exists v.\ \emph{edge}(v,w_i) \right)
\end{align*}
If $w_i$ is an output and has an in-edge from the node-vertex $u_j$, then
we add the following formula:
\begin{align*}
\left(
   \exists w'. \emph{wire-edge}^*(w', w_i) \land \emph{edge}(u_j,w')
\right)
\land \left( \lnot \exists v.\ \emph{edge}(w_i,v) \right)
\end{align*}
The next case we have to consider is when
$w_i$ is neither an input, nor an output and it is
adjacent to two node-vertices. Let the in-edge of $w_i$ be from the
vertex $u_j$ and its out-edge to the vertex $u_p.$ Then, we add the following
formula:
\begin{align*}
\left(
  \exists w'. \emph{wire-edge}^*(w_i, w') \land \emph{edge}(u_j,w_i)
  \land \emph{edge}(w', u_p)
\right)
\end{align*}
The final case we have to consider for wire-vertex $w_i$ is when
$w_i$ is neither an input, nor an output and it is
adjacent to a single wire-vertex $w_j$, forming a circle of length two (cf.
Definition~\ref{def:wire}). Then, we add:
\begin{align*}
\emph{wire-edge}^*(w_i, w_j) \land
\emph{wire-edge}^*(w_j, w_i)
\end{align*}

Finally, we restrict the kinds of wires and isolated wire-vertices allowed in
$\phi_H$ to be exactly the ones given by $H$. This is the final formula
which we add as the last of the $\psi_k$ denoted above:
\begin{align*}
\forall w,w'.\ \emph{wire-edge}^*(w,w') \implies
\left( 
  \bigvee_i \emph{wire-edge}^*(w,w_i) \lor \emph{wire-edge}^*(w_i,w)
\right)
\end{align*}
which simply says that any wire-vertex must either be one of the already
described isolated wire-vertices or part of a wire which is already described
in $\phi_H.$
\end{proof}

From the background chapter we know that MSO definable languages have important
decidability properties in relation to context-free languages. We will use this
fact for the proofs in the following subsections.

\subsection{Membership problem}

In this subsection we will show that the membership problem for B-ESG grammars
is decidable. We will do this by showing that for a given string graph, there
are finitely many sentential forms of a B-ESG grammar which need to be
considered. In turn, we can prove this by defining appropriate functions
which measure the size of a graph and its wire-homeomorphic class and
proving some simple properties about them.

\begin{definition}[Graph size]
The \emph{size} of a graph $H$, denoted $\emph{size}(H)$ is the number of
edges of $H$ plus the number of vertices of $H$:
\[\emph{size}(H) = \#V_H + \#E_H \]
\end{definition}

\begin{definition}[Wire-homeomorphic size] The \emph{wire-homeomorphic
size} of an encoded string graph $H$, denoted $\emph{wsize}(H)$ is a tuple
$(n, w, i) \in \mathbb{N}^3,$ where :
\begin{itemize}
\item $n$ is the number of node-vertices of $H$
\item $w$ is the number of wires of $H$
\item $i$ is the number of isolated wire-vertices of $H$
\end{itemize}
Moreover, we define a partial-order when working with \emph{wsize} by
$(n_1, w_1, i_1) \leq (n_2, w_2, i_2)$ iff $n_1 \leq n_2$,
$w_1 \leq w_2$ and $i_1 \leq i_2.$
\end{definition}

\begin{lemma}\label{lem:wsize-equal}
If $H \sim H'$ are two wire-homeomorphic encoded string graphs, then
wsize$(H) = $wsize$(H')$.
\end{lemma}
\begin{proof}
Wire-homeomorphic encoded string graphs only differ in the number of
wire-vertices on each wire. Therefore, the three mentioned counts must be the
same.
\end{proof}

Note, that \emph{wsize}$(H)$ does not count the number of encoding edges
of $H$, even though wire-homeomorphic encoded string graphs also have
equal numbers of encoding edges. This is a deliberate choice as it allows
for more elegant proofs in the next few lemmas and theorems.

\begin{lemma}\label{lem:wsize-increase}
Given a decoding system $T$,
if $H_1 \Longrightarrow_*^T H_2$, then wsize$(H_1) \leq $ wsize$(H_2).$
\end{lemma}
\begin{proof}
$H_2$ is obtained from $H_1$ by replacing encoding edges (which don't
count as wires) by other graphs. Therefore, the numbers of
node-vertices, wires and isolated wire-vertices cannot be decreased.
\end{proof}

\begin{lemma}\label{lem:homeo-decode}
Given a decoding system $T$ and encoded string graphs such that
$H_1 \sim H_2$ and $H_1 \Longrightarrow_*^T H_1'$ and $H_2
\Longrightarrow_*^T H_2'$, then $H_1' \sim H_2'.$
\end{lemma}
\begin{proof}
$H_1$ and $H_2$ only differ in the length of their wires. Therefore, after
decoding both graphs, they will again only differ in the length of the same
wires.
\end{proof}

In our goal to model reasoning with families of string diagrams, we need
to be able to decide when a string diagram is a specific instance of a
string diagram family. Recall, that any two string graphs which are
wire-homeomorphic represent the same string diagram. Because of this, the
membership problem needs to be stated up to wire-homeomorphism.

\begin{prob}[Membership]
Given a string graph $H$ and a B-ESG grammar $B$, does
there exist a string graph $\widetilde H \sim H$, such that $\widetilde H \in L(B)$?
In such a case, construct a derivation sequence $S \Longrightarrow_* \widetilde
H$. In addition, decide if there are finitely many such $\widetilde H$ and
if there are, then construct a concrete derivation for each of them.
\end{prob}

\begin{theorem}\label{thm:member} \rm
  The membership problem for B-ESG grammars is decidable.
\end{theorem}

\begin{proof}
Let $B=(G,T)$ and let $H$ be an arbitrary string graph. For this proof, we will
be using the decidability results from Subsection~\ref{subsec:decide}.

  First, we show that exact membership (i.e. not up to wire-homeomorphism) is
  decidable. That is, we have to decide if $H \in L(B)$.
From Theorem~\ref{thm:besg_language}, we know that any concrete B-ESG
  derivation produces an encoded string graph which is then decoded to a string
  graph. Since the decoding sequence cannot decrease the size of a graph, we
  can limit the problem to considering all graphs of size smaller than $H$.
  However, there are finitely many graphs whose size is smaller than $H$. For
  each such graph $H'$, we can then decide if $H' \in L(G)$ (this is the
  membership problem for B-edNCE grammars). Finally, we check if $H'
  \Longrightarrow_*^T H$ which is also clearly decidable. If no such graph $H'$
  exists, then the answer is no and otherwise the answer is yes.

We now generalise to the wire-homeomorphic case. Using
Theorem~\ref{thm:wire-homeo-mso}, we know that the wire-homeomorphism class of
any encoded string graph $K$ is MSO definable and moreover, we can effectively
construct that MSO formula. Therefore, we can decide if $L(G) \cap \left[ K
\right]_{\sim} = \emptyset$, and moreover, we can also decide if $L(G) \cap
\left[ K \right]_{\sim}$ is finite. This
\replaced{immediately}{immidiately}
implies
that we can decide,
given an encoded string graph $K,$ whether there exists $K'$, such that $K \sim
K' \in L(G)$ and if there are finitely many such $K'$.

Let us consider the encoded string graphs $K$ which could possibly be decoded
into a wire-homeomorphic string graph of $H$. From
Lemma~\ref{lem:wsize-increase}, any such graph $K$ must satisfy
\emph{wsize}$(K) \leq $ \emph{wsize}$(H)$ and therefore we should limit our
search to these kinds of encoded string graphs.

If we could consider all encoded string graphs $K$ with \emph{wsize}$(K) \leq$
\emph{wsize}$(H)$, then we could check if $K \Longrightarrow_*^T K' \sim H$ and
if $K \in L(G)$ and therefore decide the problem. However, that is impossible,
because there are infinitely many such encoded string graphs.  But, observe
that there are finitely many wire-homeomorphism classes $\left[ K
\right]_{\sim}$
where each representative has \emph{wsize} at most that of $H$ (cf.
Lemma~\ref{lem:wsize-equal}). Thus, by using Lemma~\ref{lem:homeo-decode}, we
see it doesn't matter which representative of the wire-homeomorphism class we
choose. Therefore, we can decide the problem by enumerating all minimal
representatives of the wire-homeomorphism classes of encoded string graphs with
\emph{wsize} at most that of $H$. Then, for each minimal representative $K$, we
decide if there exists $K',$ such that $K \sim K' \in L(G)$ and if $K
\Longrightarrow_*^T K'' \sim H.$ If no such minimal representative exists, then
the answer is negative.  If we find at least one such representative, then we
can also construct a derivation $S \Longrightarrow_* K'' \sim H$ because
derivation sequences are recursively enumerable. 

In addition, we can also decide if there are finitely many such $K''$. We have
already pointed out that we can decide if there are finitely many $K'$ such
that $K \sim K' \in L(G).$ Combining this with Lemma~\ref{lem:homeo-decode}
means we only have to check if a single representative decodes to a
wire-homeomorphic graph of $H$ and therefore we can correctly decide the
problem. In the case there are
finitely many $K'$ (and thus $K'')$, then again we can construct a concrete
derivation for each of them because derivation sequences are recursively
enumerable.
\end{proof}

As we have pointed out previously, the match enumeration problem supersedes
the membership problem. However, we will use the decidability of the membership
problem in order to show decidability of the match enumeration problem.

\subsection{Match enumeration problem}

Consider the following problem -- we are given an equational schema
$F_1 = F_2$ between two families of string diagrams, and we are given a
concrete string diagram $D$. We wish to apply an instance of the equational
schema to $D$, so that we can rewrite it. In order to do so, we have to
identify an appropriate instance $I$ of $F_1$ and then a monomorphism $m: I 
\to D$ which is a match. However, it is possible that there are multiple
instances $I$, so ideally we want to enumerate
all of them, so that we can later decide which one is the most appropriate to
use for rewriting.

This is the problem which we will show how to solve in this subsection.
However, we still haven't explained how to represent equational schemas using
B-ESG grammars. This is done in the next chapter. In this subsection we will
show how to enumerate all possible instances $I$ of a family $F_1$ such that
there exists a mono $m: I \to D$. This
\replaced{immediately}{immidiately}
implies we can solve the
match-enumeration problem, because we simply have to check if the mono $m$
satisfies the matching conditions, which as we have shown in
Section~\ref{sec:string} are decidable.

The match enumeration problem is not decidable for arbitrary B-ESG grammars.
However, it is decidable for grammars which satisfy 
some simple conditions.


\begin{definition}[Match-exhaustive B-ESG grammar]\label{def:non-bare}
We say that a B-ESG grammar $B = (G,T)$ is \textit{match-exhaustive}, if (1)
there is a fixed bound on the number of bare wires and isolated wire-vertices
in any graph in $L(B)$ (2) there are no empty productions and (3) there are no
chain productions.
\end{definition}

Conditions (2) and (3) from the above definition can obviously be decided.
In Subsection~\ref{sub:normal_forms} it was furthermore shown that any grammar
can be transformed into an equivalent grammar satisfying conditions (2) and
(3).
For property (1), we provide sufficient static conditions which can be
decided.

\begin{definition}[Iterable production]
For a B-edNCE grammar $G$, we say that a production $p$ is \emph{iterable}, if
for any $n \in \mathbb N$, there exists a concrete derivation $S \derive_* H$,
where $p$ is applied at least $n$ times.
\end{definition}

Whether a production is iterable can be decided in the same way as for
context-free string grammars. We simply create a graph whose vertices are
the productions of the grammar and whose edges indicate possible transitions
between productions. Then, a production is iterable if it is in the same
connected component as one of the initial productions and it is part of a
cycle.

The next proposition presents sufficient decidable conditions for
property (1) from Definition~\ref{def:non-bare}.

\begin{proposition}
Let $B=(G,T)$ be a B-ESG grammar, where $G$ is such that any of its productions
do not have connection instructions attached to wire-vertices and wire-vertices
are not adjacent to nonterminal vertices. If $G$ contains no production which
is iterable and which contains a bare wire or an isolated wire-vertex, then
there is a fixed bound on the number of bare wires and isolated wire-vertices
in any graph in $L(B)$.
\end{proposition}
\begin{proof}
Every wire-vertex in the productions of $G$ is not adjacent to nonterminals
and has no associated connection instructions. Therefore, the neighbourhood
of each wire-vertex in any sentential form of $G$ is the same as the
neighbourhood of the production which created it. Thus, each isolated
wire-vertex and each bare wire are created by a single application of
a production of $G$. The non-iterable productions can clearly produce a finite
number of isolated wire-vertices and bare wires, whereas the iterable
productions cannot produce any bare wires, nor isolated wire-vertices.
Decoding cannot introduce new inputs or outputs and thus the proposition
follows.
\end{proof}

Next, we define the problem which we need to decide in terms of B-ESG
grammars and string graphs.

\begin{prob}[Mono-enumeration]
  Given a string graph $H$ and a B-ESG grammar $B$, enumerate
  all of the B-ESG concrete derivations $S \Longrightarrow_*^{B} K$, such that
  there exists a mono $m : K \to \widetilde H$ for some $\widetilde H \sim H$.
\end{prob}

\begin{lemma}\label{lem:matching-homeo}
Given string graphs $K$ and $H$, there exists a mono $m : K \to
H'$ for some $H' \sim H$ iff for any $K' \sim K$, there
exists a mono $m' : K' \to H'',$ for some
$H'' \sim H.$
\end{lemma}
\begin{proof}
($\Longleftarrow$) This direction is trivial.

($\Longrightarrow$)
Let $K' \sim K$, such that $K'$ only differs from $K$ by a single added
(removed) wire-vertex $w$ associated to the interior of a wire $W \subseteq K$.
Consider a string graph $H''$ which differs from $H'$ by a single added
(removed) wire-vertex associated to $m(W)$ of $H'$. Then, we can define a mono
$m' : K' \to H''$ which acts the same way as $m$ on all vertices not in $W$.
For the wire $W$, its endpoints $v_1$ and $v_2$ can be mapped in the same way
as under $m$. Then, the interior of $W$ is simply mapped onto the wire segment
of $H''$ between $m(v_1)$ to $m(v_2)$, which is guaranteed to have
the same number of elements as $W$, because we have added (removed) a
wire-vertex to (from) $H'$ as required.

In general for any $K' \sim K$ we may construct the required mono $m'$ by
simply applying the above argument multiple times.
\end{proof}

\begin{theorem}\label{thm:match} \rm
  The mono-enumeration problem for a B-ESG grammar $B$ is decidable if
  $B$ is a match-exhaustive grammar.
\end{theorem}
\begin{proof}
Let $w$ be the number of wires in a string graph $H$ and let $n$ be the number
of its node-vertices. Then, for any $\widetilde H \sim H$, we know that
$\widetilde H$ must have the same number of wires and node-vertices as $H$.
However, the number of wire-vertices in $\widetilde H$ may be arbitrarily
large.

Condition (1) of Definition~\ref{def:non-bare} implies that there exists $m \in
\mathbb{N}$, such that, for any $K \in L(B)$, $K$ has at most $m$ bare wires.
Any mono between string graphs will map at most two non-bare wires onto a
single
wire. So, if there exists a mono $m : K \to \widetilde H$, then $K$ can have at
most
$2w$ non-bare wires. Therefore, $K$ can have at most $2w+m$ wires. Moreover,
Condition (1) of Definition~\ref{def:non-bare} also imposes a bound on the
number of isolated wire-vertices in $K$, which we shall assume to be $i \in
\mathbb{N}$. Clearly, the number of node-vertices in $K$ is also bounded by the
number of node-vertices in $H$. Thus, for any $K \in L(B)$ which could possibly
have a mono onto some $\widetilde H \sim H$, we know $\emph{wsize}(K) \leq
(n,2w+m,i).$ Therefore, there are finitely many wire-homeomorphism classes
$[K]_{\sim},$ which we need to consider and which we can enumerate. Using
Lemma~\ref{lem:matching-homeo}, we see that it doesn't matter which
representative of the wire-homeomorphism class we choose to check for the
existence of possible monomorphisms $m': \widetilde K \to \widetilde H$.  So,
for each such wire-homeomorphism class $[K]_{\sim}$, by using
Theorem~\ref{thm:member} we can decide if there exists $\widetilde K \in
[K]_{\sim}$, such that
$\widetilde K \in L(B)$ and if so, we can then check if there exists a mono
$m' : \widetilde K \to \widetilde H,$ for some $\widetilde H \sim H$. If both
checks succeed, we can again
use Theorem~\ref{thm:member} to check if there are finitely many such
$\widetilde K$.
If there aren't, then we clearly cannot enumerate them all. If there are
finitely many such $\widetilde K$ for each $[K]_{\sim}$, then we have to
explain why
we can enumerate the concrete derivations $S \Longrightarrow_*^B \widetilde K.$

Conditions (2) and (3) from Definition~\ref{def:non-bare} imply that
the sentential forms of $B$ can only increase in size and therefore for any
$\widetilde K$ satisfying the above conditions there are finitely many concrete
derivations $S \Longrightarrow_*^{B} \widetilde K$ which we can enumerate.
\end{proof}

So, if we are using a match-exhaustive B-ESG grammar $G$, then we can decide
the mono-enumeration problem, which immediately implies that we can decide the
match-enumeration problem as well.

\section{Related work}
\added{The idea of an encoded B-edNCE grammar is similar to the type of graph
grammar presented in \cite{pair_grammars}. There, the author uses a
node-replacement graph grammar for which he proposes two different extensions.
One of the extensions replaces specially labelled terminal vertices with graphs
in a way specified by a separate rule system. His extension is different from
ours in that the rule system is more complicated and replacement is done on
vertices, instead of edges. Also, the notion of grammar which he uses is less
powerful than ours.}

\added{We introduced B-ESG grammars and encoded B-edNCE grammars, because we
wish to have more expressive power than standard B-edNCE grammars.
\emph{Adaptive star grammars} are introduced in \cite{adaptive-star} with the
motivation that both hyperedge replacement and vertex replacement graph
grammars have limited expressive power. Adaptive star grammars are able to
capture a much larger class of graph languages compared to B-edNCE grammars,
while also retaining important decidability properties such as membership.
However, little is known about their structural properties or normal forms
which makes them difficult to reason about in the context of this thesis.}

\chapter{Rewriting B-ESG grammars}\label{ch:rewriting}

In the previous chapters we introduced B-ESG grammars and showed that
they correctly represent families of string diagrams, we argued they have
sufficient expressive power and we established they have the necessary
decidability properties for rewriting concrete string diagrams. In this
chapter we will build upon these results by showing that B-ESG grammars can
also be used to rewrite B-ESG grammars themselves in a way which correctly
represents equational reasoning on infinite families of string diagrams.

In Section~\ref{sec:partial-besg}, we begin by showing how to do DPO rewriting
on edNCE grammars. Recall that we use DPO rewriting on string graphs in order
to model equational reasoning between string diagrams. Thus, DPO rewriting on
edNCE grammars is the first step towards modelling equational reasoning between
families of string diagrams, in our proposed framework.

Then, in Section~\ref{sec:rewriting}, we show how to restrict DPO rewriting on
B-edNCE grammars such that it is admissible in the sense that it agrees with
the concrete semantics of our grammars (and thus with the concrete semantics of
the families of diagrams we are representing). This is the most crucial
section in this chapter as it entirely relates the two main aspects of
graph transformation which we are using -- DPO rewriting and derivations
in B-edNCE grammars.

In Section~\ref{sec:b-esg-rules} we show how to represent equational schemas
between families of string diagrams using B-ESG grammars. We will show
how to instantiate B-ESG grammars in a meaningful way and we will prove the
resulting instantiations are valid string graph rewrite rules.

Finally, in Section~\ref{sec:final}, we combine all of the results from the
previous sections in order to show how to admissibly rewrite B-ESG grammars
using DPO rewriting. The results from this section show how to correctly
represent equational reasoning between context-free families of string
diagrams, where even the rewrite rules are equational schemas between
context-free families of diagrams.

\section{Partial adhesivity of edNCE grammars}\label{sec:partial-besg}

In Section~\ref{sec:adhesive} and Section~\ref{sec:partial_adhesive} we showed
that the category $\mathbf{Graph}$ is partially adhesive with ambient adhesive
category $\mathbf{MultiGraph}.$ The reason why $\mathbf{Graph}$ is not adhesive
is that its graphs cannot have parallel edges with the same label, while
the multigraphs in $\mathbf{MutliGraph}$ are allowed to have such
parallel edges. In this section, we will show how to DPO rewriting on edNCE
grammars. Observe that edNCE grammars are strictly more general compared to
graphs, because every edNCE grammar with a single production and no
connection instructions is simply a graph. For this reason it is obvious
that edNCE grammars do not form an adhesive category. However, if we define
a category of edNCE grammars where parallel edges with the same labels are
allowed, then we can show that this category is indeed adhesive. This
category will be called $\mathbf{MultiEdNCE}$ and we will show it is adhesive
in Subsection~\ref{sub:multi-ednce}. Then in Subsection~\ref{sub:ednce-partial}
we will show that the category of edNCE grammars is partially adhesive, where
its ambient adhesive category is $\mathbf{MultiEdNCE}$.

In this section the distinction between wire-vertices and node-vertices or
encoding edges and non-encoding edges is irrelevant. All of our constructions
will be over the following labelling alphabets, which are compatible with B-ESG
alphabets.

\begin{definition}[Labelling alphabets]\label{def:labelling_alphabets}
Throughout this section, our constructions will be over a triple of
labelling alphabets $\mathcal A = (\Sigma, \Delta, \Gamma)$, where:
\begin{description}
\item[1.] $\Sigma$ is the \emph{alphabet of all vertex labels}.
\item[2.] $\Delta \subseteq \Sigma$ is the \emph{alphabet of terminal
vertex labels}.
\item[3.] $\Gamma$ is the \emph{alphabet of all edge labels}. We will not
use any non-final edge labels (see Corollary~\ref{cor:nonblocking}).
\end{description}
\end{definition}

\subsection{Multi-edNCE grammars}\label{sub:multi-ednce}

Recall that an extended graph (cf. Definition~\ref{def:extended-graph})
is simply a graph where, in addition, we can attach connection instructions
to its vertices. Parallel edges with the same labels or parallel
connection instructions with the same label are not allowed.
We may generalise the definition of extended graphs to allow such parallel
edges and parallel connection instructions by analogy to the generalisation
of graphs to multigraphs.

\begin{definition}[Unlabelled Extended Multigraphs]
The category of \emph{unlabelled extended multigraphs} is
$\mathbf{UExtMultiGraph}$. This
category is defined as the functor category $[\mathbb{GR},
\mathbf{Set}]$, where $\mathbb {GR}$ is the category given by:
\cstikz{unlabelled_category_extended_graph.tikz}
\added{The only non-trivial morphisms in $\mathbb{GR}$ are the ones shown in
the diagram and we do not require that $s$ and $t$ commute.  $s$ and $t$ should
be seen as the source and target functions which assign a source or target
vertex to edges, $c^0$ should be seen as assigning a vertex to an in-connection
instruction and $c^1$ should be seen as assigning a vertex to an out-connection
instruction.  In most scenarios we would, of course, have $s\not = t.$}
For an extended multigraph $H \in \mathbf{UExtMultiGraph},$ we shall denote with
$V_H$ its set of vertices, with $E_H$ the set of its edges, with $C^0_H$ the
set of its in-connection instructions and with $C^1_H$ the set of its
out-connection instructions.
\end{definition}

\begin{definition}[Extended Multigraphs]\label{def:extended-multi}
The category of \emph{extended multigraphs} over a triple of labelling alphabets
$\mathcal A = (\Sigma, \Delta, \Gamma)$ is the category
$\mathbf{MultiExtGraph}_{\mathcal A}$ whose objects are tuples $(H, l)$
with $H \in \mathbf{UMultiExtGraph}$ an unlabelled
extended multigraph and $l = (l_V, l_E, l_{C^0}, l_{C^1})$ a
4-tuple of labelling functions:
\begin{align*}
l_V     &: V_G   \to \Sigma                             &\text{(the vertex
     labelling function)}\\
l_E     &: E_G   \to \Gamma                             &\text{(the edge
     labelling function)}\\
l_{C^0} &: C^0_G \to \Sigma \times \Gamma \times \Gamma &\text{(the
    in-connection instruction labelling function)} \\
l_{C^1} &: C^1_G \to \Sigma \times \Gamma \times \Gamma &\text{(the
    out-connection instruction labelling function)} \\
\end{align*}
A morphism between two extended multigraphs
$f : (G, l) \to (H, l')$ is a morphism $f: G \to H$ of
$\mathbf{UMultiExtGraph}$, which in addition respects the labelling, that is
the following diagrams commute:
\cstikz{respect_labelling2.tikz}
\end{definition}

Building on top of this, we may generalise edNCE grammars to
\emph{multi-edNCE} grammars. What is left is to partition all the vertices,
edges and connection instructions into productions.

\begin{definition}[Unlabelled multi-edNCE grammar]
\label{def:unlabelled_multiednce}
The category of \emph{unlabelled multi-edNCE grammars} is
$\mathbf{UMultiEdNCE}$. This
category is defined as the functor category $[\mathbb{EDNCE},
\mathbf{Set}]$, where $\mathbb{EDNCE}$ is the category given by:
\cstikz{unlabelled_category_grammar.tikz}
\added{The only non-trivial morphisms in $\mathbb{EDNCE}$ are the ones shown in
the diagram. Again, we do not require that $s$ and $t$ commute, however
we do require that the rest of the diagram commutes, that is:}
\begin{center}
\added{$p^V \circ t = p^E = p^V \circ s$}
\end{center}
For a multi-edNCE grammar $G \in \mathbf{UMultiEdNCE},$ we shall denote with
$V_G$ its set of vertices, with $E_G$ the set of its edges, with $C^0_G$ the
set of its in-connection instructions, with $C^1_G$ the set of its
out-connection instructions and with $P_G$ the set of its productions.
\end{definition}

In the definition above, we shall refer to the morphisms $s$ and $t$
as the \emph{source} and \emph{target} functions. They assign each edge
$e \in E_G$ a
source and target vertex respectively. An edge $e$ is said to be
\emph{incident} to a vertex $v$, if $s(e) = v$ or $t(e) = v.$

The morphism $c^0$ assigns to each in-connection instruction $c \in C^0_G$
a vertex $v \in V_G$. Similarly, the morphism $c^1$ assigns to each
out-connection instruction $c \in C^1_G$ a vertex $v \in V_G$. We will
say that a connection instruction $c \in C^i_G$ is \emph{associated} to a
vertex $v$ if $c^i(c) = v.$

The newly introduced morphisms in $\mathbb{EDNCE}$, $p^E$ and $p^V$ assign
edges and vertices respectively to a production $p \in P_G$ of our grammar. We
will say that a vertex $v$ is \emph{in} a production $p \in P_G$ if $p^V(v)=p$
and we will say that an edge $e$ is \emph{in} a production $p \in P_G$ if
$p^E(e) = p$. The commutativity requirements imposed by the diagram ensure
that if an edge $e$ is assigned to a production $p$, then both its source
vertex $s(e)$ and its target vertex $t(e)$ are also in $p$.

A \emph{component} $X_G$ of $G$ is either its set of vertices $V_G$, its
set of edges $E_G$, its set of production $P_G$, or one of its sets of
connection instructions $C^i_G.$ For a grammar $G$, we will call the morphisms
from $\mathbb{EDNCE}$ its \emph{assigning functions}.

Given a multi-edNCE grammar $G$, a \emph{full subgrammar} of $G$ is a
multi-edNCE
grammar $H$, such that $X_H \subseteq X_G$, for each component of $G$ and each
assigning function of $H$ is the restriction of the corresponding assigning
function of $G$ to the components of $H$. We will denote this with $H \subseteq
G.$

We will say that $x$ is an \emph{element} of a multi-edNCE grammar $G$ and
denote it with $x \in G,$ if $x \in X_G$ for some $X \in \{V, E, P, C^i\}.$
In other words, $x$ is some vertex, edge, production or connection
instruction of $G$. If $f: G \to H$ is a morphism in $\mathbf{MultiEdNCE},$
then we will say that $f(G)$ is the \emph{image} of $f$ and it will refer
to the full subgrammar $H' \subseteq H$ whose components are $f_X(X_G).$

The definition which we have presented for multi-edNCE grammars is clearly
a proper generalization of both multigraphs and extended multigraphs. In
particular, if the set $P$ is a singleton, then we can see a multi-edNCE
grammar as an extended multigraph. If, in addition, both $C^0$ and $C^1$ are
the empty set, then we get a multigraph.

\begin{definition}[MultiEdNCE grammars]\label{def:multi-ednce}
The category of \emph{multi-edNCE grammars} over a triple of labelling alphabets
$\mathcal A = (\Sigma, \Delta, \Gamma)$ is the category
$\mathbf{MultiEdNCE}_{\mathcal A}$ whose objects are tuples $(G, l)$
with $G \in \mathbf{UMultiEdNCE}$ an unlabelled
multi-edNCE grammar and $l = (l_V, l_E, l_{C^0}, l_{C^1}, l_P)$ a
5-tuple of labelling functions:
\begin{align*}
l_V     &: V_G   \to \Sigma                             &\text{(the vertex
     labelling function)}\\
l_E     &: E_G   \to \Gamma                             &\text{(the edge
     labelling function)}\\
l_{C^0} &: C^0_G \to \Sigma \times \Gamma \times \Gamma &\text{(the
    in-connection instruction labelling function)} \\
l_{C^1} &: C^1_G \to \Sigma \times \Gamma \times \Gamma &\text{(the
    out-connection instruction labelling function)} \\
l_P     &: P_G   \to \Sigma - \Delta                    &\text{(the production
    labelling function)}\\
\end{align*}
A morphism between two multi-edNCE grammars
$f : (G, l) \to (H, l')$ is a morphism $f: G \to H$ of
$\mathbf{UMultiEdNCE}$, which in addition respects the labelling, that is
the following diagrams commute:
\cstikz{respect_labelling3.tikz}
\end{definition}

If the triple of labelling alphabets $\mathcal A=(\Sigma, \Delta, \Gamma)$ is
clear from the context, then we will simply refer to this category as
\textbf{MultiEdNCE}.

The above definition of labelled multi-edNCE grammars is a generalisation
of the standard edNCE grammars. What's common is that both edNCE and
multi-edNCE grammars can be seen as a set of productions, where each
production consists of an extended (multi)graph which has an associated
nonterminal label. However, the difference is that extended multigraphs
are strictly more general than extended graphs, because the former allow
for parallel edges with the same label and parallel connection instructions
with the same label, whereas the latter does not. It is because of this
reason that multi-edNCE grammars form an adhesive category, as we show
in the next theorem, but edNCE grammars form only a partially adhesive
category (shown in the next subsection).

\begin{theorem}[\cite{unpublished-crap}]
$\mathbf{MultiEdNCE}$ is adhesive.
\end{theorem}
\begin{proof}
We will use a proof strategy which is similar to the one used in
Lemma~\ref{lem:multigraph-adhesive}, where we showed that the category of
labelled multigraphs is adhesive by showing that it is isomorphic to a
slice category which can easily be shown to be adhesive by the lemmas
in Section~\ref{sec:adhesive}.

Consider an arbitrary object $(G,l) \in$ \textbf{MultiEdNCE}.
$l = (l_V, l_E, l_{C^0}, l_{C^1}, l_P)$ is a 5-tuple of labelling functions,
where
each $l_X$ is a morphism in \textbf{Set}. Thus, setting:
\begin{align*}
\mathcal S &:= (V_G, E_G, C^0_G, C^1_G, P_G)\\
\mathcal L &:= (\Sigma, \Gamma, \Sigma\times \Gamma \times \Gamma,
\Sigma\times \Gamma \times \Gamma, \Sigma - \Delta)
\end{align*}
we can see that $l: \mathcal S \to \mathcal L$ is a morphism in
$\mathbf{Set}^5$. Next, consider the discrete category of five objects, which
we shall denote with $\mathbb{EDNCE}'$. So, $\mathbb{EDNCE}'$ is the same as
$\mathbb{EDNCE}$
from Definition~\ref{def:unlabelled_multiednce}, where all of its
non-trivial morphisms are removed. Let's denote (for brevity)
$\mathbf{C} :=[\mathbb{EDNCE}', \mathbf{Set}]$ and 
$\mathbf{D} :=[\mathbb{EDNCE}, \mathbf{Set}].$

It's easy to see that $\mathbf{C} \cong \mathbf{Set}^5.$ Therefore, we can see
the labelling
$l$ as a morphism in $\mathbf{C}$ and we can see
$\mathcal{S}, \mathcal{L}$ as objects in $\mathbf{C}.$

There's an obvious embedding $E: \mathbb{EDNCE}' \to \mathbb{EDNCE}.$
$E$ can be
used to define a forgetful functor $U: \mathbf{D} \to
\mathbf{C},$ by setting:
\[U(-) := - \circ E\]
Therefore, $U(G) = \mathcal S$ and thus
the labelling $l$ can be seen as a morphism $l: U(G) \to \mathcal L$
in $\mathbf{C}.$ The category $\mathbf{Set}$ has all
finite limits and therefore $U$ has a right-adjoint $U \dashv R$, given by
the right Kan extension $R(-) := Ran_E(-).$ From the adjunction, we get
a natural isomorphism $\Phi$ which provides a family of bijections:
\[
\Phi_{G, \mathcal L}:
hom_{\mathbf{C}} (U(G), \mathcal L)
\to
hom_{\mathbf{D}} (G, R(\mathcal L))
\]
Since the labelling set $\mathcal L$ is fixed by the alphabets, we
get a bijection between objects $(G, l) \in$
\textbf{MultiEdNCE} and objects $(G, \Phi_{G, \mathcal L}(l)) \in
\mathbf{D}/R(\mathcal L).$

Next, for any morphism $f: (G_1, l_1) \to (G_2, l_2) \in$ \textbf{MultiEdNCE}
we know $f$ is a natural transformation from $G_1$ to $G_2$ which is in
addition
subject to the labelling restrictions of Definition~\ref{def:multi-ednce}.  Any
morphism $f: (G_1, \Phi_{G_1, \mathcal L}(l_1)) \to (G_2, \Phi_{G_2, \mathcal
L}(l_2))$ in $\mathbf D/R(\mathcal{L})$ is a natural transformation from
$G_1$ to $
G_2$, subject to the slice restriction, which is equivalent to the labelling
restrictions of Definition~\ref{def:multi-ednce}. Therefore,
$\mathbf{MultiEdNCE}\cong \mathbf{D}/R(\mathcal L) = [\mathbb{EDNCE},
\mathbf{Set}]/R(\mathcal L).$

Finally, it is easy to see that $[\mathbb{EDNCE}, \mathbf{Set}]/R(\mathcal L)$
is adhesive. From Example\ref{ex:set} we know that $\mathbf{Set}$ is adhesive.
Then, Lemma~\ref{lem:adhesive_functor} implies that $[\mathbb{EDNCE},
\mathbf{Set}]$ is adhesive. From Lemma~\ref{lem:adhesive_slice}, we know
that adhesive categories are closed under the slice construction and
therefore $[\mathbb{EDNCE}, \mathbf{Set}]/R(\mathcal L)$ is adhesive,
which completes the proof.
\end{proof}

As we have shown in the background chapter, this means we can do DPO
rewriting in $\mathbf{MultiEdNCE}.$ This is the main result of this
subsection, but before we conclude it, we will present two additional
lemmas which will be helpful for proofs in later sections.
The first lemma will be useful when working with pushout squares in
$\mathbf{MultiEdNCE}.$
The second lemma characterises the matching conditions in
$\mathbf{MultiEdNCE}.$

\begin{lemma}\label{lem:joint_surjection}
Given a commutative square in $\mathbf{MultiEdNCE}$:
\cstikz{commutative_square.tikz}
where all morphisms are monomorphisms, then the square is a pushout iff
$m$ and $s$ are jointly surjective.
\end{lemma}
\begin{proof}
In the first direction, the proof is very simple -- if an element $x \in G_H$
is not in the image of both $m$ and $s$, then we immediately get a
contradiction with the universality of the pushout by considering the
grammar $G_H -x$.

In the other direction, consider an arbitrary grammar $G_H'$ and morphisms
$m': G_L \to G_H'$ and $s': G_K \to G_H'$ with $m' \circ l = s' \circ k$, like
in the diagram below:
\cstikz{pushout_universality.tikz}
Now, consider a morphism $u: G_H \to G_H'$, such that $m' = u \circ m$
and $s' = u \circ s$. This requires
\[
 u(x) =
  \begin{cases} 
   m'\circ l(x) = s' \circ k(x) & \text{if } x \in m\circ l(G_I) = s \circ
     k(G_I) \\
   m'(x)  & \text{if } x \not\in m\circ l(G_I),\ x\in m(G_L) \\
   s'(x)  & \text{if } x \not\in s\circ k(G_I),\ x\in s(G_K) \\
  \end{cases}
\]
Because $m$ and $s$ are jointly surjective, this implies the above definition
of $u$ is unique and therefore the original square is a pushout.
\end{proof}

\begin{lemma}\label{lem:matching-ednce}
Given a pair of monomorphisms $G_H \xleftarrow{m} G_L \xleftarrow{l} G_I$
in $\mathbf{MultiEdNCE}$, their pushout complement exists iff the following
conditions are satisfied:
\begin{description}
\item[No dangling edges:] no edge $e \in E_{G_H},\ e\not\in m_E(E_{G_L})$ is
incident to a vertex $v \in m_V(V_{G_L} - l_V(V_{G_I}))$.
\item[No dangling connection instructions:] no connection instruction $c \in
C^i_{G_H}, c \not\in m_{C^i}(C^i_{G_L})$ is attached to a vertex $v \in
m_V(V_{G_L} - l_V(V_{G_I})).$
\item[No dangling vertices:] no vertex $v\in V_{G_H}, v \not\in m_V(V_{G_L})$
is in a production $p \in m_P(P_{G_L} - l_P(P_{G_I})).$
\end{description}
\end{lemma}
\begin{proof}
Consider the following square:
\cstikz{pushout_complement.tikz}
$(\Longrightarrow)$
Let's assume that the above square is a pushout.  From
Lemma~\ref{lem:pushout_complement_monos} we know that $s$ and $k$ are
monomorphisms. Then, from Lemma~\ref{lem:joint_surjection}  we know that $s$
and $m$ must be jointly surjective.

Assume the no dangling edges condition is violated, that is, 
there exists $e \in E_{G_H},\ e\not\in m_E(E_{G_L})$ and $e$ is incident to a
vertex $v \in m_V(V_{G_L} - l_V(V_{G_I}))$. Joint surjectivity of $m_E$ and
$s_E$ implies
$e \in s_E(E_{G_K})$ and therefore $v \in s_V(V_{G_K})$. However, since $G_H$
is a pushout
and $v$ is in the image of both $m_V$ and $s_V$, it follows $v \in m_V\circ l_V
(V_{G_I})$
which is a contradiction.

The case for connection instructions is similar. Assume the no dangling
connection instructions condition is violated, that is, there exists a
connection instruction $c \in C^i_{G_H}, c \not\in m_{C^i}(C^i_{G_L})$ and $c$
is associated to a
vertex $v \in m_V(V_{G_L} - l_V(V_{G_I}))$. Joint surjectivity of $m_{C^i}$ and
$s_{C^i}$ implies $c
\in s_{C^i}(C^i_{G_K})$ and therefore $v \in s_V(V_{G_K})$. However, since
$G_H$ is a pushout and
$v$ is in the image of both $m_V$ and $s_V$, it follows $v \in m_V\circ l_V
(V_{G_I})$
which is a contradiction. 

For the last case, assume the no dangling vertices condition is violated. That
is, there exists a vertex $v\in V_{G_H}, v \not\in m_V(V_{G_L})$ and $v$ is in
a
production $p \in m_P(P_{G_L} - l_P(P_{G_I})).$ Joint surjectivity of $m_V$ and
$s_V$ implies
$v \in s_V(V_{G_K})$ and therefore $p \in s_P(P_{G_K})$. However, since $G_H$
is a pushout
and $p$ is in the image of both $m_P$ and $s_P$, it follows $p \in m_P\circ l_P
(P_{G_I})$
which is a contradiction.

$(\Longleftarrow)$ Let $G_K$ be the full subgrammar of $G_H$ whose components
are given by $X_{G_K} := X_{G_H} - m_X(X_{G_L} - l_X(X_{G_I})),$ for $X \in
\{V, E, P, C^0, C^1\}.$ The no dangling conditions ensure that $G_K$ is a
well-defined multi-edNCE grammar.
Set the monomorphism $s$ to be simply
the full subgrammar inclusion of $G_K$ into $G_H$. Set $k$ to be
set-theoretically
equal to $m \circ l$. Then, clearly the square commutes and $k$ is also a
monomorphism. Then, by Lemma~\ref{lem:joint_surjection}, we can complete the
proof by showing that $m$ and $s$ are jointly surjective.

Since $s$ is a subgrammar inclusion, then for every component $X$, the image of
$s_X$ is $X_{G_H} - m_X(X_{G_L} - l_X(X_{G_I}))$ by definition. Combining
this with the other monomorphism $m$, we get that $s_X$ and $m_X$ jointly
cover:
\[X_{G_H} - m_X(X_{G_L} - l_X(X_{G_I})) \cup m_X(X_{G_L}) = X_{G_H}\]
Therefore, $s$ and $m$ are jointly surjective.
\end{proof}

The \textbf{no dangling edges} condition is the same as for the case
of multigraphs. The other two conditions are clearly very similar in spirit
to the no dangling edges condition. This is a consequence of the fact
that pushouts are computed component-wise in \textbf{Set} over some constraints
imposed by the structure of our grammars.

\subsection{edNCE grammars}\label{sub:ednce-partial}

In this subsection we will define the category of edNCE graph grammars and show
that it is partially adhesive, where the ambient adhesive category is
$\mathbf{MultiEdNCE}.$ The edNCE graph grammars have been defined in
Section~\ref{sec:graph-grammars} and we begin by first describing a
homomorphism between two edNCE graph grammars. We will not be using any
nonfinal edge labels (see Corollary~\ref{cor:nonblocking}), so our grammars can
be labelled using the same triple of labelling alphabets $\mathcal A = (\Sigma,
\Delta, \Gamma)$ as for $\mathbf{MultiEdNCE}$.

\begin{definition}[Grammar homomorphism]
Given two edNCE grammars
$G_1 = (\Sigma, \Delta, \Gamma,  \Gamma, P_1, S_1)$ and
$G_2 = (\Sigma, \Delta, \Gamma,  \Gamma, P_2, S_2)$, a \emph{grammar
homomorphism} from $G_1$ to $G_2$ is a function $m : P_1 \to P_2$, together
with a collection of extended graph homomorphisms $m_{p_i} : rhs(p_i) \to
rhs(m(p_i))$ one for each production $p_i \in P_1$, such that $lhs(p_i) =
lhs(m(p_i)).$
\end{definition}

\begin{definition}[Category of edNCE grammars]
The category of all edNCE grammars over the triple of labelling alphabets
$\mathcal A = (\Sigma, \Delta, \Gamma)$
is denoted by $\mathbf{edNCE}_{\mathcal A}$ or just by
$\mathbf{edNCE}$
if the alphabets are clear from the context. Its objects are edNCE grammars
$G = (\Sigma, \Delta, \Gamma, \Gamma, P, S)$
and the morphisms of the category are edNCE grammar
homomorphisms.
\end{definition}

The following theorem describes the relationship between edNCE grammars and
their generalized versions -- multi-edNCE grammars.

\begin{theorem}\label{thm:edNCEvsMultiEdNCE}
$\mathbf{edNCE}$ is a partially adhesive category whose ambient
adhesive category is $\mathbf{MultiEdNCE}.$
\end{theorem}
\begin{proof}
We define a functor $\mathcal S: \mathbf{edNCE} \to \mathbf{MultiEdNCE}$ in the
following way. Given an edNCE grammar $G = (\Sigma, \Delta, \Gamma, \Gamma, P,
S),$ we define $\mathcal S(G) := H,$ where $H$ is the multi-edNCE grammar
whose components and assigning functions are given by:
\begin{align*}
P_H &= P & V_H &= \bigcup_{p \in P} V_{rhs(p)}
\quad\quad E_H = \bigcup_{p \in P} E_{rhs(p)} 
\quad\quad C^i_H = \bigcup_{p \in P} C^{i}_{rhs(p)} \\
s &:: (v, \alpha, w) \mapsto v & t &:: (v, \alpha, w) \mapsto w
\quad\quad p^E :: e \mapsto p, \text{ if } e \in E_{rhs(p)}\\
c^i &:: (\sigma, \alpha, \beta, v, d) \mapsto v & p^V &:: v \mapsto p,
\text{ if } v \in V_{rhs(p)}\\
l_P &:: p \mapsto lhs(p) & l_V &:: v \mapsto \lambda_{rhs(p)}(v),
  \text{ if } v \in V_{rhs(p)} \\
l_E &:: (v, \alpha, w) \mapsto \alpha &
l_{C^i} &:: (\sigma, \alpha, \beta, v, d) \mapsto (\sigma, \alpha, \beta)
\end{align*}

Given a morphism $(f, \{f_p\}_{p \in P}) $
between edNCE grammars
$G = (\Sigma, \Delta, \Gamma, \Gamma, P, S)$ and
$G' = (\Sigma, \Delta, \Gamma, \Gamma, P', S'),$
we define its mapping under $\mathcal S$ to be the natural transformation
$\phi$ induced by setting:
\begin{align*}
\phi_P &= f \\
\phi_V &:: v \mapsto f_p(v), \text{ if } v \in V_{rhs(p)}
\end{align*}
The rest of the components of $\phi$ are uniquely determined
from the commutativity conditions the natural transformation has to satisfy.
Nevertheless, we provide them for completeness:
\begin{align*}
\phi_E &:: (v, \alpha, w) \mapsto (f_p(v), \alpha, f_p(w)), &\text{ if } v,w \in
V_{rhs(p)}\\
\phi_{C^i} &:: (\sigma, \alpha, \beta, v, d) \mapsto (\sigma, \alpha, \beta,
f_p(v), d), &\text{ if } v \in V_{rhs(p)}
\end{align*}
It's easy to check that the functor $\mathcal S$ as defined is full and
faithful. Finally, we have to show that $\mathcal S$ preserves monomorphisms.
A morphism $(f, \{f_p\}_{p \in P})$ in $\mathbf{edNCE}$ is a mono iff
$f$ and each $f_p$ are injective functions. Then, its mapping under $\mathcal
S$ is a natural transformation $\phi$, such that all of its components $\phi_X$
are also injective functions, which are precisely the monos in
$\mathbf{MultiEdNCE}.$
\end{proof}

Given a grammar homomorphism $f$ between two edNCE grammars $G_1$ and $G_2$,
if $v$ is some vertex in (the RHS of) a production $p \in P_1$, then through
abuse of notation we will use $f(v)$ to refer to the vertex $f_{p}(v).$
This shouldn't lead to confusion, as each vertex $v$ of $G_1$ is in a unique
production and we can clearly differentiate between vertices of a grammar
and the productions of a grammar. In that sense, we may think of the
functions $f_{p}$ as restrictions of $f$ to the extended graph
$rhs(p).$

Following the example from Section~\ref{sec:partial_adhesive}, we will
characterise the partial adhesive conditions of $\mathbf{edNCE}$ which allow us
to do DPO rewriting in the same way as in $\mathbf{MultiEdNCE}$. We begin by
introducing a lemma which will help us with some of the remaining proofs.
The lemma shows that the functor $\mathcal S : \mathbf{edNCE} \to
\mathbf{MultiEdNCE}$ is essentially surjective.

\begin{lemma}\label{lem:ess-surj}
For any multi-edNCE grammar $G \in \mathbf{MultiEdNCE}$, such that
$G$ does not have any parallel edges, parallel connection instructions
or self-loops,
there exists an
edNCE grammar $H \in \mathbf{edNCE},$ such that $\mathcal S(H)$ is isomorphic
to $G$.
\end{lemma}
\begin{proof}
Let the components and assigning functions of $G$ be described as
in Definition~\ref{def:multi-ednce}.
For every production $p \in P_G$, we define a production
$Y_p \to (D_p,C_p)$ of $H$ by setting:
\begin{align*}
Y_p &:= l_P(p)\\
V_{D_p} &:= \{v \in V_G\ |\ p^V(v) = p\}\\
\lambda_{D_p} &:: v \mapsto l_V(v) \\
E_{D_p} &:= \{(s(e), l_E(e), t(e)) |\ e \in E_G, p^E(e) = p \}\\
C_p &:= \{
(\sigma, \alpha, \beta, c^0(c), in) |\ c \in C^0_G,\ p^V\circ c^0(c) =
p,\ l_{C^0}(c) = (\sigma, \alpha, \beta)
\}\cup \\
&\cup
\{
(\sigma, \alpha, \beta, c^1(c), out) |\ c \in C^1_G,\ p^V\circ c^1(c) =
p,\ l_{C^1}(c) = (\sigma, \alpha, \beta)
\}
\end{align*}
Then, the set of productions of $H$ is given by $\{Y_p \to (D_p, C_p)\ |\ p
\in P_G\}.$ Note, that because $G$ does not contain any self-loops, this
implies that $H$ is a well-defined grammar (in particular, the RHS
of every production is a well-defined graph). Also, note that the definition of
edNCE grammar requires the set of productions to have the above form, that is,
every element must be of the form $Y \to (D,C),$ where $Y$ is a nonterminal
label and $(D,C)$ is an extended graph.

Now, consider the grammar $K := \mathcal S(H)$. It is easy to see that $V_K =
V_G.$ Because $G$ does not contain parallel connection instructions or parallel
edges, it follows that $E_K = E_G, C^i_K = C^i_G$ and that all assigning
functions of $K$, except for $p^V$ and $p^E$ are exactly the same as those of
$G$. However, observe that $P_G$ and $P_K$ are not necessarily the same,
but isomorphic. In particular, every element of $P_K$ is of the form
$X \to (D,C)$, whereas the elements of $P_G$ might be arbitrary. However,
there is an obvious isomorphism $i: P_G \to P_K$ which is simply:
\[i :: p \mapsto (Y_p \to (D_p, C_p))\]
and therefore $K$ and $G$ are isomorphic as required.
\end{proof}

\begin{remark}
Combining the above lemma with Theorem~\ref{thm:edNCEvsMultiEdNCE} implies
that the embedding functor $\mathcal S$ establishes a categorical equivalence
between the category $\mathbf{edNCE}$ and the full subcategory of
$\mathbf{MultiEdNCE}$ whose objects do not contain self-loops, parallel
edges or parallel connection instructions. However, this functor does not
establish an isomorphism between these two categories, because of the reason
mentioned in the proof above -- the elements of the set of productions of
edNCE grammars must be of a specific form, whereas those of multi-edNCE
grammars do not. Of course, this detail is irrelevant for practical
purposes, but we mention it in order to stay formal.
\end{remark}

Next, we characterise the $\mathcal{S}$-spans in \textbf{edNCE}.

\begin{lemma}\label{lem:s-spans}
A span of monomorphisms $G_K \xleftarrow{f} G_I \xrightarrow{g} G_R$ in
$\mathbf{edNCE}$ is an $\mathcal{S}$-span iff the following conditions hold:
\begin{description}
\item[ParEdges:] For any vertices $v,w \in G_I$, if there exist edges $(f(v),
\alpha, f(w)) \in G_K$ and $(g(v), \alpha, g(w)) \in G_R$ then there exists an
edge $(v, \alpha, w) \in G_I$.
\item[ParCI:] For any vertex $v \in G_I$, if there exist connection
instructions $(\sigma, \alpha,\beta, f(v), d) \in G_K$ and
$(\sigma, \alpha,\beta, g(v), d) \in G_R,$ then there exists a connection
instruction $(\sigma, \alpha,\beta, v, d) \in G_I.$
\end{description}
\end{lemma}
\begin{proof}
In the first direction, let's assume that the above conditions hold. Then,
consider the pushout of the span $\mathcal{S}(G_K) \xleftarrow{\mathcal S(f)}
\mathcal{S}(G_I) \xrightarrow{\mathcal S(g)} \mathcal{S}(G_R)$ in
\textbf{MultiEdNCE}:
\cstikz{s-span.tikz}
Because the embedding functor $\mathcal{S}: \mathbf{edNCE} \to
\mathbf{MultiEdNCE}$ reflects pushouts, it is sufficient to show that
$G_M, h$ and $k$ are in the image of $\mathcal S$. From the fullness of
$\mathcal S$ and Lemma~\ref{lem:ess-surj},
it therefore follows that we can complete the proof (in this
direction) by showing that
$G_M$ does not have parallel edges, parallel connection
instructions or self-loops.

Let's consider parallel edges first. $G_K$ and $G_R$ do not have parallel
edges. The pushout in \textbf{MultiEdNCE} is given by component-wise disjoint
union modulo the common components in $G_I$. Therefore if $G_M$ has a pair of
parallel edges, then their source and target vertices must be identified by the
morphisms $f$ and $g$. In other words, a pair of parallel edges in $G_M$ may
only be established if $G_R$ contains an edge $(g(v), \alpha, g(w))$ and $G_K$
contains an edge $(f(v), \alpha, f(w)),$ where $v,w \in G_I.$ But then, the
\textbf{ParEdges} condition requires the edge $(v, \alpha, w) \in G_I$ and thus
the pushout would establish exactly one edge with label $\alpha$ from $h\circ
\mathcal S(f)(v)$ to $h\circ \mathcal S(f)(w)$ in $G_M$.

Next, let's consider parallel connection instructions. The proof is fully
analogous to the case for edges -- $G_K$ and $G_R$ do not have parallel
connection instructions.  Therefore if $G_M$ has a pair of parallel connection
instructions, then their associated vertex must be identified by the morphisms
$f$ and $g$. In other words, a pair of parallel connection instructions in
$G_M$ may only be established if $G_R$ contains a connection instruction
$(\sigma, \alpha, \beta, g(v), d)$ and $G_K$ contains a connection instruction
$(\sigma, \alpha, \beta, f(v), d),$ where $v \in G_I.$ But then, the
\textbf{ParCI} condition requires the connection instruction $(\sigma, \alpha,
\beta, v, d) \in G_I$ and thus the pushout would establish exactly one
connection instruction with label $(\sigma, \alpha, \beta)$ and direction $d$
whose associated vertex is $h\circ \mathcal S(f)(v)$ in $G_M$.

To complete the proof in this direction, we need to show $G_M$ does not contain
self-loops. Neither $G_R$, nor $G_K$ contain self-loops. The pushout is
computed by taking their disjoint union and then identifying certain vertices,
edges, connection instructions and productions between $G_R$ and $G_K$ as being
the same. Taking their disjoint union clearly cannot result in self-loops. A
self-loop can only be established if two vertices $v,w$ from $G_R$ ($G_K$)
connected
by an edge $(v, \alpha, w)$ are identified as the same vertex in $G_K$ ($G_R$),
via the morphism $f$ ($g$).
However this is
impossible because both $f$ and $g$ are monomorphisms.

In the other direction, let's assume that we are given an $\mathcal{S}$-span.
If the \textbf{ParEdges} condition is violated, then the $\mathcal{S}$-pushout
in \textbf{MultiEdNCE} results in a grammar $\mathcal{S}(G_M)$ which contains a
pair of parallel edges, thus $G_M$ is not an object in \textbf{edNCE} which is
a contradiction. Similarly, if the \textbf{ParCI} condition is violated,
then the $\mathcal{S}$-pushout in \textbf{MultiEdNCE} results in a grammar with a
pair of parallel connection instructions and we get a contradiction.
\end{proof}

Next, we characterize the $\mathcal{S}$-matchings in \textbf{edNCE}. It turns
out that an $\mathcal{S}$-pushout complement in \textbf{edNCE} exists when
the same matching conditions are satisfied as in \textbf{MultiEdNCE} as we
can see in the following lemma.

\begin{lemma}\label{lem:s-pushout_complement}
Given a pair of monomorphisms $G_H \xleftarrow{m} G_L \xleftarrow{l} G_I$ in
$\mathbf{edNCE}$, $m$ is an $\mathcal{S}$-matching iff $\mathcal{S}(m)$ is a
matching (in $\mathbf{MultiEdNCE}$).
\end{lemma}
\begin{proof}
One direction is obvious -- if $m$ is an $\mathcal{S}$-matching, then the
$\mathcal{S}$-pushout complement exists and therefore by definition
$\mathcal S(m)$ is a matching in \textbf{MultiEdNCE}.

In the other direction, let's assume $\mathcal{S}(m)$ is a matching. Then, the
pushout complement of $\mathcal S(G_H) \xleftarrow{\mathcal S(m)} \mathcal
S(G_L) \xleftarrow{\mathcal S(l)} \mathcal S(G_I)$ exists and let it be
given by the diagram below:
\cstikz{s-pushout_complement.tikz}
$\mathcal S(G_H)$ is the result of the pushout of grammars $\mathcal S(G_L)$
and $G_K$ in \textbf{MultiEdNCE}. Therefore, since $\mathcal S(G_H)$ does not
contain parallel edges, parallel connection instructions or self-loops, then
neither does $G_K$. Using Lemma~\ref{lem:ess-surj}, we may assume, without
loss of generality, that
$G_K$ is in the image of $\mathcal S$ and therefore
from the fullness of $\mathcal S$, it follows that $s$ and $k$ are also in its
image. Therefore, the reflected pushout square is an $\mathcal S$-pushout.
\end{proof}

DPO rewriting in \textbf{edNCE} is well-defined when the standard DPO diagram
exists and when both pushout squares are preserved by the embedding functor
$\mathcal{S}:~\mathbf{edNCE} \to \mathbf{MultiEdNCE}.$

The following theorem characterizes the conditions under which DPO rewriting
is well-defined in \textbf{edNCE}.

\begin{theorem}\label{thm:dpo-ednce}
In the category $\mathbf{edNCE}$, given a span of monomorphisms $Sp := G_L
\xleftarrow{l} G_I \xrightarrow{r} G_R$ and an $\mathcal{S}$-matching $m: G_L
\to G_H$, then the DPO rewrite induced by $m$ and $Sp$ is well-defined iff
the following conditions are satisfied:
\begin{description}
\item[Edges:] For any two vertices $v,w \in G_I$, if there exist edges
$(m \circ l (v), \alpha, m \circ l (w)) \in G_H$ and
$(r(v), \alpha, r(w)) \in G_R,$ then there must be an edge
$(l(v), \alpha, l(w)) \in G_L.$
\item[CI:] For any vertex $v \in G_I$, if there exist connection instructions
$(V, \alpha, \beta, m\circ l(v), d) \in G_H$ and $(V, \alpha, \beta, r(v), d)
\in G_R,$ then there must be a connection instruction $(V, \alpha,\beta, l(v),
d) \in G_L.$
\end{description}
\end{theorem}
\begin{proof}
Let the following DPO diagram describe the rewrite, it it exists:
\cstikz{sdpo.tikz}
First, let's assume that the DPO rewrite is well-defined. 

Now, let's assume for contradiction that condition \textbf{CI} is violated,
that is, there's a vertex $v \in G_I$, connection instructions $(V, \alpha,
\beta, m\circ l(v), d) \in G_H$ and $(V, \alpha, \beta, r(v), d) \in G_R,$ but
no connection instruction $(V, \alpha,\beta, l(v), d) \in G_L.$ The left square
is a $\mathcal{S}$-pushout square and therefore there exists a connection
instruction $(V, \alpha,\beta, k(v), d) \in G_K.$ 
Because $(V, \alpha, \beta, l(v), d) \not \in G_L$ this implies
$(V, \alpha, \beta, v, d) \not \in G_I.$
But then, applying
Lemma~\ref{lem:s-spans} to the right pushout square yields a contradiction
with the \textbf{ParCI} condition.

For the other condition, the proof is fully analogous.  Let's assume for
contradiction that condition \textbf{Edges} is violated that is, there are
vertices $v,w \in G_I$ and edges $(m \circ l (v), \alpha, m \circ l (w)) \in
G_H$ and $(r(v), \alpha, r(w)) \in G_R,$ but there exists no edge $(l(v),
\alpha, l(w)) \in G_L.$ The left square is a $\mathcal{S}$-pushout square and
therefore there exists an edge $(k(v), \alpha, k(w)) \in G_K.$ 
Because $(l(v), \alpha, l(w)) \not \in G_L$ this implies
$(v, \alpha, w) \not \in G_I.$
But then,
applying Lemma~\ref{lem:s-spans} to the right pushout square yields a
contradiction with the \textbf{Edges} condition. This completes the proof in
one direction.

In the other direction, we know that $m$ is an $\mathcal{S}$-matching and
therefore the left pushout square is an $\mathcal{S}$-pushout. So, we
need to show that the right square exists and is an $\mathcal{S}$-pushout.
We shall show this by using Lemma~\ref{lem:s-spans}.

Let's assume condition \textbf{CI} is satisfied. We shall show that condition
\textbf{ParCI} is satisfied for the span in the right square. Let's assume for
contradiction that is not the case. Thus, there exists vertex $v \in G_I$ and
connection instructions $(V, \alpha,\beta, k(v), d) \in G_K$ and $(V,
\alpha,\beta, r(v), d) \in G_R,$ but there exists no connection instruction
$(V, \alpha,\beta, v, d) \in G_I.$ Since the left square is an
$\mathcal{S}$-pushout, this implies there must be a connection instruction
$(V, \alpha,\beta, m\circ l(v), d) \in G_H.$ Then, condition \textbf{CI}
implies there exists a connection instruction $(V, \alpha,\beta, l(v), d)
\in G_L$ and now we get a contradiction with Lemma~\ref{lem:s-spans}
when applied to the left pushout square.

Next, let's assume condition \textbf{Edges} is satisfied. 
The proof is again fully analogous to the case for connection instructions.
We shall show that
condition \textbf{ParEdges} is satisfied for the span in the right square.
Let's assume for contradiction that is not the case. Thus, there exist vertices
$v,w \in G_I$ and edges $(k(v), \alpha, k(w)) \in G_K$
and $(r(v), \alpha, r(w)) \in G_R,$ but there exists no edge
$(v, \alpha, w) \in G_I.$ Since the left square is an
$\mathcal{S}$-pushout, this implies there must be an edge $(m \circ l(v),
\alpha,m \circ l(w)) \in G_H.$ Then, condition \textbf{Edges} implies
there exists an edge $(l(v), \alpha,l(w)) \in G_L$ and
now we get a contradiction with Lemma~\ref{lem:s-spans} when applied to the
left pushout square.
\end{proof}

We will refer to the two conditions from this theorem, \textbf{Edges} and
\textbf{CI}, as the \emph{partial adhesive conditions} which we will make
use of in later sections of this chapter. Note, that the partial adhesive
conditions alone do not guarantee the existence of a DPO rewrite. For this
to be the case, we also need to combine them with the no dangling conditions,
which we have also described.

In summary, we have shown that edNCE graph grammars form a partial adhesive
category and we have fully characterised the conditions under which
DPO rewriting can be done and is well-behaved.

\section{B-edNCE rewriting}\label{sec:rewriting}

For the rest of the chapter, we will be working with edNCE grammars from
the partially adhesive category \ednce, which was described in the previous
section. From the results established there, we know under what conditions we
may perform DPO rewriting on edNCE grammars. However, these results do not
tell us anything about how the \emph{languages} of the original grammars
relate to the languages of rewritten grammars. Since we are interested
in modelling reasoning on string graphs, it will be necessary to introduce
further restrictions on the rewrite rules and matchings of our grammars if
we want to be able to meaningfully talk about their languages.
The additional restrictions on our matchings and our rewrite rules will ensure
that DPO rewriting in B-edNCE grammars behaves well with respect to the
derivation process and the languages which are generated.

In Subsection~\ref{sub:subst} we will show how to rewrite extended graphs, such
that the rewrites commute with graph substitution. Graph substitution is the
mechanism used to carry out derivations in grammars. DPO rewrites of extended
graphs is how we simulate equational reasoning on a per-production basis for
B-edNCE grammars. By identifying under what conditions these two operations
commute we can significantly simplify the proofs in the following subsection.

In Subsection~\ref{sub:admis} we will show how to build upon the results
from Subsection~\ref{sub:subst} in order to define admissible rewrites
of entire B-edNCE grammars.

\subsection{Extended graph rewrites and graph substitution}\label{sub:subst}

All of the constructions in this subsection will be in the category
$\mathbf{edNCE}$. We will be considering only edNCE grammars which have a
single production (with identical and irrelevant production labels). As we have
pointed out previously, we can see such grammars as extended graphs. For
brevity and in order to avoid notational overhead, we will simply refer
to these objects in $\mathbf{edNCE}$ as extended graphs.

We begin by showing
that the substitution operation behaves well with respect
to monomorphisms of extended graphs.
\begin{lemma}\label{lem:subst_embed}
Given extended graphs $G$ and $D$ where $x \in V_G$ is a (nonterminal) vertex,
and given monomorphisms $m_1 : G \to G'$ and $m_2 : D \to D',$ then there
exists a monomorphism $m : G[x/D] \to G'[m_1(x)/D'],$ which we shall
denote by $SM(m_1,m_2,x)$ and refer to it as the substituted
monomorphism of $m_1$ and $m_2$ over vertex $x$.
\end{lemma}
\begin{proof}
We define $m : G[x/D] \to G'[m_1(x)/D']$ in the following way:
\[ m(v) = \begin{cases} m_1(v) & \text{ if } v \in V_G\\
                        m_2(v) & \text{ if } v \in V_D
          \end{cases}
\]
$m$ is an injection, because both $m_1$ and $m_2$ are injections. It is
also easy to verify that $m$ is an extended graph homomorphism.
\end{proof}

Next, we introduce the concept of an extended graph rewrite rule. This is just
a special case of a rewrite rule in $\multiednce$ when the grammars are
extended graphs, but we provide the definition for completeness.

\begin{definition}[Extended graph rewrite rule]
An \emph{extended graph rewrite rule} is a pair of
monomorphisms $L \xleftarrow{l} I \xrightarrow{r} R,$ where all objects
are extended graphs.
\end{definition}

\begin{example}\label{ex:rewrite}
The following is an extended graph rewrite rule:
\cstikz{rewrite-example.tikz}
The only difference between an extended graph rewrite rule and a graph rewrite
rule is that we also have to map the connection instructions appropriately.
In this case, the monos are obvious as they are uniquely determined by
the labels of the vertices.
\end{example}

Next, we introduce the notion of rewrite rule substitution which essentially
works component-wise and its purpose is to provide notational convenience.

\begin{definition}[Rewrite rule substitution]
Given extended graph rewrite rules $B_1 := L_1 \xleftarrow{l_1} I_1
\xrightarrow{r_1} R_1$ and $B_2 := L_2 \xleftarrow{l_2} I_2 \xrightarrow{r_2}
R_2$, with vertex $v \in I_1$ then the \emph{substitution} of $B_2$ for
$v$ in $B_1$, denoted $B_1[v/B_2]$ is given by the extended graph rewrite rule
$B_3 := L_3 \xleftarrow{l_3} I_3 \xrightarrow{r_3} R_3$, where
$L_3 := L_1[l_1(v)/L_2],$
$I_3 := I_1[v/I_2],$
$R_3 := R_1[r_1(v)/R_2],$
$l_3 := SM(l_1,l_2,v),$
$r_3 := SM(r_1,r_2,v).$
\end{definition}
Crucially, the monomorphisms $l_3$ and $r_3$ are built in a natural way by
using the substituted monomorphism construction.

Our next definition formalizes the notion of \emph{saturated matching} on
extended graphs which we will be using throughout this chapter. A saturated
matching of extended graphs is a restricted form of matching (in the sense of
the previous section). In addition to allowing us to perform DPO rewriting on
extended graphs, these matchings will also later be used in order to rewrite
edNCE grammars in an admissible way. In particular, the saturated matching
conditions ensure that DPO rewrites commute with the substitution operation on
extended graphs (as shown in Theorem~\ref{thm:subst_rewrite}), which is the
basic
operation used to perform grammar derivations. This is not true for arbitrary
DPO rewrites which satisfy only the matching conditions in \ednce, but not the
additional requirements imposed by saturated matchings.

\begin{definition}[Extended graph saturated matching]
Given an extended graph rewrite rule $L \xleftarrow{l} I
\xrightarrow{r} R,$ and an extended graph $H,$ we say that an
\emph{extended graph saturated matching} is an extended graph monomorphism $m :
L \to
H,$ such that:
\begin{itemize}
\item no edge $e \in E_H, e \not \in m(E_L)$ is incident to any vertex in $m(V_L -
l(V_I))$
\item no edge $e \in E_H, e \not \in m(E_L)$ is incident to a nonterminal vertex
in $m(l(V_I))$
\item $m$ is a bijection on the connection instructions of $L$ and $H$, that
is, $C_H = \{(\sigma, \beta, \gamma, m(x), d)\ |\ (\sigma, \beta, \gamma, x,
d) \in C_L\}.$ 
\end{itemize}
\end{definition}
In particular, the first condition is the standard no-dangling edges condition
for (extended) graphs. The third condition concerns connection instructions and
is obviously stricter than the no-dangling connection instructions condition
which is needed in order to ensure that the pushout complement of extended
graphs exists. The additional strictness of the third condition, together with
the second condition (which can be seen as imposing further restrictions on the
no-dangling edges condition) guarantee that rewriting productions of grammars
behave nicely with respect to the derivations of the grammar (which is based on
extended graph substitution).

We can show that DPO rewrites on extended graphs preserve saturated matchings.

\begin{lemma}\label{lem:saturated_preserve}
Given an extended graph rewrite rule $B := L \xleftarrow l I \xrightarrow r R,$
an extended graph $H$ and a saturated matching $m : L \to H,$ then if the
following diagram is an $\mathcal S$-rewrite:
\cstikz{dpo_saturated.tikz}
then $f$ is a saturated matching with respect to the extended graph rewrite
rule $B' := R \xleftarrow r I \xrightarrow l L.$
\end{lemma}
\begin{proof}
By contradiction. Assume there exists an edge $e \in E_M, e \not \in f(E_R)$
such that $e$ is incident to a vertex $v \in f(V_R - r(V_I)).$ Because the
right square is a pushout, then from the joint surjectivity of $f$ and $g$,
we get that $e \in g(E_K).$ Since $g$ is an (extended) graph homomorphism,
it follows immediately that $v \in g(V_K).$ Thus, the vertex $v$ is in
the images of both $g$ and $f$ and because the right square is a pushout,
it follows $v$ is also in the image of $f \circ r = g \circ k$, which is
a contradiction.

Assume there exists an edge $e \in E_M, e \not \in f(E_R)$, such that $e$ is
incident to a nonterminal vertex $v$ in $f(r(V_I)).$ Let $u \in V_I$ be the
nonterminal vertex of $V_I$, such that $f \circ r(u) = v$. As in the previous
case, we see that $e \in g(E_K).$ Let $e' \in E_K$ be the edge such that $g(e')
= e$ and then $e'$ must be incident to the vertex $k(u)$, since
$e = g(e')$ is incident to $v = f \circ r (u) = g \circ k(u)$.
Now, consider the edge $e'' = s(e') \in E_H$. If $e''$ is in the image of
$m$, then because the left square is a pushout, we get that $e''$ must also be
in the image of $m \circ l = s \circ k$ and we get a contradiction with the
fact that $e \not \in f (E_R)$. Thus, $e'' \in E_H - m(E_L)$. Because $s$
is an extended graph homomorphism, it then follows that $e''$ is incident
to $s \circ k(u) = m \circ l (u)$ and we get a contradiction with the fact
that $m$ is a saturated matching (the second condition of the definition
is violated).

Finally, we have to show that $f$ is a bijection on the connection instructions
of $R$ and $M$. We know that $f$ is a mono, so therefore we just have to show
that $f$ is a surjection. Assume that there exists a connection instruction
$c \in C_M$ which is not in the image of $f$. The right square is a pushout,
therefore $f$ and $g$ are jointly surjective. Thus, $c = g(c')$ for some
$c' \in C_K$. Therefore, $s(c') \in C_H$ is a connection instruction of
$H$. Because $m$ is a saturated matching, the third condition of the definition
implies that there exists $c'' \in C_L$, such that $m(c'') = s(c').$ But then,
because the left square is a pushout, this implies there exists $c_0 \in C_I$
such that $c'' = l(c_0)$ and $c' = k(c_0).$ Therefore, $c = g \circ k(c_0) =
f \circ r (c_0)$ and we get a contradiction, because $c$ is in the image of
$f$.
\end{proof}

We say that an extended graph is \emph{boundary} if there are no edges
between nonterminal vertices. The following lemma shows that saturated
matchings carry over the substitution operation in a natural way.

\begin{lemma}\label{lem:subst_matching}
Given boundary extended graphs $H$ and $D,$ 
extended graph rewrite rules
$B_1 := L_1 \xleftarrow{l_1} I_1 \xrightarrow{r_1} R_1$ and
$B_2 := L_2 \xleftarrow{l_2} I_2 \xrightarrow{r_2} R_2$ 
and given extended graph saturated matchings $m_1 : L_1 \to H$ and $m_2: L_2
\to D$, with nonterminal vertex
$v \in V_I,$
then $m_3 := SM(m_1, m_2, l(v))$ is an extended graph saturated
matching from the rewrite rule $B_3 := B_1[v/B_2]$ into the extended graph
$H[m \circ l(v)/D].$
\end{lemma}

\begin{proof}
Without loss of generality, we can take isomorphic
copies of all these graphs, such that all of the monos are simply
subgraph inclusions. So, for simplicity, we shall assume that is the case.
Also, we shall refer to the saturated matching conditions simply as matching
conditions.

Let $B_3 = B_1[v/B_2] =  L_3 \xleftarrow{l_3} I_3 \xrightarrow{r_3} R_3.$
Using Lemma~\ref{lem:subst_embed}, we see that $m_3$ is given by the extended
graph inclusion $L_3 \subseteq H[v/D].$
To show that $m_3$
satisfies the matching conditions, consider an arbitrary edge $e \in
E_{H[v/D]}, e
\not \in E_{L_3}$ and $e$ is incident to a vertex $u \in V_{L_3} - V_{I_3}.$
Also, let's assume that the label of $e$ is $\gamma$.
First, observe
that $u \in V_{L_3} - V_{I_3}$  iff $u \in V_{L_1} - V_{I_1}$ or $u \in V_{L_2}
-V_{I_2}.$ Moreover,
$e \not \in E_{L_3}$ implies $e \not \in E_{L_1}$ and $e \not \in E_{L_2}.$ If $e \in E_H,$
then the fact that $e$ is incident to $u \in V_{L_3} - V_{I_3}$ implies $u \in
V_{L_1}
- V_{I_1}$ which contradicts the matching conditions for $m_1$. If $e \in E_D,$
then
the fact that $e$ is incident to $u \in V_{L_3} - V_{I_3}$ implies that $u \in
V_{L_2} - V_{I_2}$ which contradicts the matching conditions for $m_2.$

Therefore, $e$ must be a bridge, that is, an edge established by the
substitution operation and is
neither in $E_H,$ nor $E_D$. Because $e$ is incident to $u \in V_{L_3} -
V_{I_3},$ there
are two further cases to consider. If $u \in V_{L_1} - V_{I_1},$ then there is
an
edge $e'$ between $u$ and $v$ in $E_H$. From the matching condition for $m_1$ it
then follows that $e' \in E_{L_1}.$ Moreover, there must also be a connection
instruction $c \in C_D,$ which establishes the bridge $e$ between $u$ and $u' \in
V_D$. Then, from the matching condition on $m_2$ it follows $c \in C_{L_2}$ and
therefore $e \in E_{L_3} = E_{L_1[v/L_2]},$ which is a contradiction.

If $u \in V_{L_2} - V_{I_2}$, then let $u' \in V_H$ be the other incident
vertex of
$e$. Because $e$ is a bridge, this implies that there is a connection
instruction $c = (\lambda(u'), \beta, \gamma, u ,d)
\in C_D$ associated to $u.$ Then, the matching condition for $m_2$
implies $c \in C_{L_2}$.  Since
$v$ is a nonterminal vertex in $V_{I_1}$ and there must be an edge $e' \in
E_H$ between $u'$
and $v$ with label $\beta$ and direction $d$, then it follows from the matching condition of
$m_1$ that $u'
\in V_{L_1}$ and $e' \in E_{L_1}$. Therefore, $e \in E_{L_3} = E_{L_1[v/L_2]}$ which is a
contradiction.

Therefore, all edges in $E_{H[v/D]}$ satisfy the first matching condition for
$m_3$.

Next, we show that the second matching condition is satisfied for $m_3$.
Assume there exists $e \in E_{H[v/D]}, e \not \in E_{L_3}$ and $e$ is incident to a
nonterminal vertex $u \in V_{I_3}.$ If $e \in E_H$, then $u \in V_{I_1}$ which violates
the matching condition for $m_1$. If $e \in E_D$, then $u \in V_{I_2}$ which violates
the matching condition for $m_2$. Thus, $e$ must be established by the
substitution operation and is neither in $E_H$, nor $E_D$. If $u \in V_{I_1}$, then
there exists an edge $e' \in E_H$ between $u$ and $v.$ However, this violates the
boundary condition for nonterminal vertices. Finally, if $u \in V_{I_2}$, then
let  $u' \in V_H$ be the other incident vertex of $e$. We must also have
a connection instruction $c=(\lambda(u'), \beta, \gamma,
u, d) \in C_D$ associated to $u.$
Then, the matching condition for $m_2$ implies $c \in C_{L_2}$.  Since $v$ is a nonterminal vertex in $V_{I_1}$
and
there must be an edge $e' \in E_H$ between $u'$ and $v$ with label $\beta$ and
direction $d$, then it follows from
the
matching condition of $m_1$ that $u' \in V_{L_1}$ and $e' \in E_{L_1}$. Therefore, $e
\in E_{L_3} = E_{L_1[v/L_2]}$ which is a contradiction.

Therefore, all edges in $E_{H[v/D]}$ satisfy the second matching condition for
$m_3$.

Finally, let's consider the third matching condition for $m_3.$
The matching condition for $m_1$ and $m_2$ imply that
$L_1$ and $H$ have the same connection instructions and also that
$L_2$ and $D$ have the same connection instructions. Moreover, each of the
two pairs also have the same edges connecting to the nonterminal vertex
$v$. Therefore, after performing the substitutions, we get $L_3 = L_1[v/L_2]$
has the same connection instructions as $H[v/D].$

So, $m_3$ as defined satisfies all
of the saturated matching conditions.
\end{proof}

The previous lemma shows that the substitution operation behaves well with
respect to saturated matchings. The next lemma shows that the substitution
operation preserves $\mathcal S$-pushout squares, when some of the morphisms
are
saturated matchings.

\begin{lemma}\label{lem:subst_pushout}
Given $\mathcal S$-pushout squares
  \[
    \stikz{pushout1.tikz}
    \quad\quad
    \stikz{pushout2.tikz}
  \]
where all objects are extended graphs, all morphisms are monos
and $m_1, m_2$ are saturated matchings, with nonterminal
vertex $v \in I_1,$ then the following diagram is also a $\mathcal S$-pushout
square:
\cstikz{pushout3.tikz}
with monomorphisms 
$i_3 := SM(i_1,i_2,v),$  $l_3 := SM(l_1,l_2,v),$ $m_3 := SM(m_1,m_2,l_1(v)),$
$k_3 := SM(k_1,k_2,i_1(v)).$
\end{lemma}
\begin{proof}
From Lemma~\ref{lem:subst_embed}, we know that $x_3$ is a monomorphism for $x
\in \{l,i,m,k\}.$
Commutativity follows easily by construction of the $x_3$ monos and the
commutativity of the original two pushout squares.
To finish the proof, we need to
show that $k_3$ and $m_3$ are jointly surjective
(Lemma~\ref{lem:joint_surjection}) and that the third square (the
substituted one) satisfies the \textbf{ParCI} and \textbf{ParEdges} conditions
from Lemma~\ref{lem:s-spans}.
For simplicity and without
loss of generality, we shall assume $X_1$ and $X_2$ are disjoint and that
$x_1$ and $x_2$ are subgraph inclusions, where $X \in \{I,L,K,H\}, x \in
\{i,l,m,k\}.$

We will first show that $k_3$ and $m_3$ are jointly surjective.
Now, let's consider an arbitrary vertex $u \in H_3 := H_1[v/H_2]$. By
definition of
substitution, $u \in H_1$ or $u \in H_2$. If $u \in H_1$, then the joint
surjectivity of $k_3$ and $m_3$ follows from the joint surjectivity of
$k_1$ and $m_1$. The other case follows in the same way and therefore $k_3$
and $m_3$ are jointly surjective on vertices.

Taking an arbitrary edge $e \in H_3,$ there are three cases to
consider. If $e \in H_1$ or $e \in H_2$, then the case follows using exactly
the same argument as those for vertices. The remaining case is when
$e$ is a bridge established by the substitution operation. Let the two
endpoints of $e$ be vertices $w_1 \in V_{H_1}$ and $w_2 \in V_{H_2}.$
Then, there must exist an edge $e' \in H_1$ which is adjacent to the
nonterminal vertex $v$, and there must exist a connection instruction
$c \in H_2$ which is compatible with $e'$, that is, $c = (\lambda(w_1),
\alpha, \beta, w_2, d),$ where $\alpha$ is the label of $e'$, $\beta$ is
the label of $e$ and $d$ is the direction $w_1$ is connected to $v$.
$m_2$ is
a saturated matching and then the third matching condition implies $c \in
L_2$. $m_1$ is also a saturated matching and then the second matching
condition implies $e' \in L_1.$ Thus, by definition of substitution, it follows
$e \in L_1[v/L_2]=:L_3.$ Therefore, $k_3$ and $m_3$ are jointly surjective on
edges as well.

From Lemma~\ref{lem:subst_matching}, it follows that $m_3$ is a saturated
matching with respect to the depicted square and therefore, according to the
third matching condition, $m_3$ is a bijection on connection instructions which
immediately implies joint surjectivity of $m_3$ and $k_3$ on connection
instructions.

Next, let's consider the \textbf{ParCI} condition for the third square.
From the left pushout square, we know that $C_{L_1} = C_{H_1}$, because
$m_1$ is a saturated matching. Since the square is a $\mathcal S$-pushout
square, it follows that $C_{K_1} = C_{H_1} - (C_{L_1} - C_{I_1}) = C_{I_1}.$
Similarly, from the second pushout square we get
$C_{K_2} = C_{I_2}.$ Next, observe that any edge
$e \in K_1$ which is incident to $v$, must also be in $I_1$, that is
$e \in I_1$, because otherwise $H_1$ would contain an edge incident to $v$
which is not in the image of $m_1$ and therefore the second condition of
saturated matchings is violated. Therefore, the vertex $v$ has the same set of
adjacent edges (and vertices) in both $K_1$ and $I_1$. Combining this with the
fact that $C_{K_1} = C_{I_1} $ and $C_{K_2} = C_{I_2}$ means that $I_1[v/I_2]
=: I_3$ and $K_1[v/K_2] =: K_3$ have the same connection instructions. This
then immediately implies that the \textbf{ParCI} condition is satisfied for the
third square.

Finally, let's consider the \textbf{ParEdges} condition. Take an arbitrary edge
$e = (w_1, \gamma, w_2) \in K_3,$ where $w_1, w_2 \in I_3$. If $w_1, w_2 \in
I_1$, then $e \in K_1$ and if $e \in L_1$, the \textbf{ParEdges} condition for
the first pushout square implies $e \in I_1$ and therefore $e \in I_3.$ So, in
this case, the \textbf{ParEdges} condition is satisfied. Similarly, if $w_1,
w_2 \in I_2$, the \textbf{ParEdges} condition is again satisfied using the same
argument.

If $w_1 \in I_1$, but $w_2 \in I_2$, then $e$ is a bridge established by the
substitution. Thus, there must be a connection instruction $c = (\lambda(w_1),
\beta, \gamma, w_2, in) \in C_{K_2} = C_{I_2}$ and there must be an edge
$(w_1, \beta, v) \in K_1.$ As we have already pointed out, $v$ has the same
set of adjacent edges and vertices in both $K_1$ and $I_1$. Combining this
with the fact that $c \in C_{I_2}$ means that the edge $e = (w_1, \gamma, w_2)
\in I_3$ and therefore the \textbf{ParEdges} condition is satisfied.

The last case is when $w_1 \in I_2$ and $w_2 \in I_1$. This follows using a
similar argument to the previous case, where we also have to change the
direction of the connection instruction from \emph{in} to \emph{out}.
\end{proof}

Building on the previous two lemmas, we can now show that rewriting of extended
graphs commutes with the substitution operation. We will prove this in the
next theorem, which is the main result of this subsection.
This is crucial for
establishing the admissibility properties of grammar rewriting that we require,
because derivations in edNCE grammars work by repeatedly applying the
substitution operation to nonterminal vertices.
\begin{theorem}\label{thm:subst_rewrite}
Given boundary extended graphs $H, H', D, D',$ such that $H
\leadsto_{B_1}^{m_1} H'$ and $D \leadsto_{B_2}^{m_2} D',$ where $B_1 := L_1
\xleftarrow{l_1} I_1 \xrightarrow{r_1} R_1$, $B_2 := L_2
\xleftarrow{l_2} I_2 \xrightarrow{r_2} R_2,$ $v \in V_{I_1}$
and where $m_1, m_2$ are saturated matchings,
then $H[m_1 \circ l_1(v)/D] \leadsto_{B_3}^{m_3} H'[f_1\circ r_1(v)/D'],$ where
$m_3 := SE(m_1, m_2, v)$ and $B_3 := B_1[v/B_2].$ In terms of diagrams,
given the following two $\mathcal S$-rewrites:
\cstikz{dpo1.tikz}
and
\cstikz{dpo2.tikz}
Then the following diagram is also an $\mathcal S$-rewrite:
\cstikz{dpo3.tikz}
where all morphisms are mono and are given by:
$i_3 := SM(k_1,k_2,v),$  $l_3 := SM(l_1,l_2,v),$ $m_3 := SM(m_1,m_2,l_1(v)),$
$s_3 := SM(s_1,s_2,k_1(v)),$
$f_3 := SM(f_1,f_2,r_1(v)),$
$g_3 := SM(g_1,g_2,k_1(v)),$
$r_3 := SM(r_1,r_2,v).$ Moreover, $m_3$ and $f_3$ are saturated matchings.
\end{theorem}
\begin{proof}
Using Lemma~\ref{lem:saturated_preserve}, we see that $f_1$ and $f_2$ are
saturated matchings. Then, using Lemma~\ref{lem:subst_matching}, we get
that $m_3$ and $f_3$ are saturated matchings. Then, the theorem follows
by two applications of Lemma~\ref{lem:subst_pushout}.
\end{proof}

\subsection{Admissible B-edNCE grammar rewriting}\label{sub:admis}

In the previous subsection we showed how to rewrite extended graphs in a way
which respects graph substitutions. Every production of a B-edNCE grammar is
simply an extended graph which also has an associated nonterminal label. We
will use the results from the previous subsection for rewriting B-edNCE
grammars locally, that is, on a per-production basis. In this subsection,
we will describe the rest of the conditions which we need in order to
rewrite B-edNCE grammars in a way which respects their concrete derivations.

In general, a rewrite rule in $\mathbf{edNCE}$ is simply a span of
monomorphisms $G_L \xleftarrow l G_I \xrightarrow r G_R.$ However, in the
absence of additional structure between these grammars, it is very difficult to
relate their languages in any way. So, we begin by introducing extra structure
which would then allow us to relate their languages by considering
\emph{parallel} derivations in all three grammars.

\begin{definition}[B-edNCE Pattern]
A B-edNCE \emph{pattern} is a triple of B-edNCE grammars $B:= G_L
\xleftarrow l G_I \xrightarrow r G_R$,
where $l$ and $r$ are grammar monomorphisms which are
bijections between the productions of all three grammars. Moreover,
$l$ and $r$ are also label-preserving bijections between the nonterminals in
corresponding productions of the grammars. In addition, all three grammars
have the same initial nonterminal label. If $p$ is a production in $G_L$,
then $B_p$ will refer to the extended graph rewrite rule $rhs(p)
\xleftarrow{l'} rhs(p') \xrightarrow{r'} rhs(p''),$
where $p'$ and $p''$ are the corresponding productions of $p$ in $G_I$ and
$G_R$ respectively,
and $l'$ and $r'$ are the restrictions of $l$ and $r$ to $rhs(p')$.
\end{definition}

\begin{example}\label{ex:b-ednce-rewrite-pattern}
Consider the B-edNCE rewrite pattern
$B:= G_L \xleftarrow l G_I \xrightarrow r G_R$, where the grammars are
given by:
\cstikz[0.9]{b-pattern-example.tikz}
where $l$ and $r$ map each production of $G_I$ vertically up or down
respectively. Note that in this case the mapping on the vertices is uniquely
determined by their labels, because all productions have one node-vertex and
one nonterminal vertex. $G_L$ generates the language of complete graphs
$K_n$ and $G_I$ and $G_L$ (which are identical) generate the language of
all star graphs $K_{1,n}$.
\end{example}

Note that the defining conditions of a B-edNCE pattern imply that given
a derivation sequence for any of the grammars, we can apply the corresponding
productions to corresponding nonterminals in any of the other two grammars.
This is made more precise by the following definition.

Using Lemma~\ref{lem:subst_embed}, we can define how to embed entire
concrete derivations from a grammar $G_1$ which embeds into another grammar
$G_2.$ This works very well for B-edNCE patterns as we can always construct
a pair of parallel derivation sequences, because their productions and
nonterminals within them are in a 1-1 correspondence.

\begin{lemma}\label{lem:induced}
Given a grammar monomorphism $m: G_1 \to G_2$
and a derivation sequence $s_1$ given by:
\begin{align*}
\emph{sn}(S_1,v_1) =: H_0
&\Longrightarrow_{v_1,p_1}^{G_1}& H_1 &\Longrightarrow_{v_2,p_2}^{G_1}& H_2
&\Longrightarrow_{v_3,p_3}^{G_1}& \cdots
&\Longrightarrow_{v_n,p_n}^{G_1}& H_n
\end{align*}
then there exists a derivation sequence $s_2$ in $G_2$, given by:
\begin{align*}
\emph{sn}(S_1,v_1) =: H_0'
&\Longrightarrow_{v_1,m(p_1)}^{G_2}& H_1'
&\Longrightarrow_{m(v_2),m(p_2)}^{G_2}& H_2'
&\Longrightarrow_{m(v_3),m(p_3)}^{G_2}& \cdots
&\Longrightarrow_{m(v_n),m(p_n)}^{G_2}& H_n'
\end{align*}
Moreover, we can
inductively define extended graph monomorphisms $m_i: H_i \to H_i'$ by
setting:
\begin{align*}
m_0 &:= v_1 \mapsto v_1\\
m_i &:= SM(m_{i-1}, m_{p_i},v_i)
\end{align*}
where $m_{p_i}$ is the restriction of $m$ to the extended graph $rhs(p_i)$.
We will refer to $s_2$ as the induced embedding
of $s_1$ into $G_2$ over $m$.
\end{lemma}
\begin{proof}
For each production $p_i$ of $G_1$, we know there exists a corresponding
production $m(p_i)$ of $G_2$ by definition of grammar morphism. The
same definition then also implies that production $m(p_{i-1})$ will create
nonterminal vertex $m(v_i)$ in $G_2$, where nonterminal vertex $v_i$ is
created by production $p_{i-1}$ of $G_1$. Thus, the derivation sequence $s_2$
is well-defined.

Clearly, $m_0$ is a mono and observe that each $m_{p_i}$ is
also a mono on extended graphs (because it is simply the restriction of
$m$ which is a mono). Then, by inductively applying
Lemma~\ref{lem:subst_embed}, it follows that each $m_i$ is a mono.
\end{proof}

Note that in the above definition we implicitly assume that $m$ acts on
disjoint and isomorphic production copies of the productions for simplicity of
the definition. Thus, when there exists a mono between grammars $G_1$ and
$G_2$, then for each concrete derivation $s_1$ in $G_1$, we can consider a
corresponding derivation $s_2$ (not necessarily concrete) in $G_2$ where each
extended graph in $s_1$ can be injectively mapped into a corresponding extended
graph in $s_2$. In particular, this means that the sentential forms of $G_1$
can be injectively mapped into corresponding sentential forms of $G_2$.

\begin{definition}[B-edNCE Pattern Instantiation]
Given a B-edNCE pattern $G_L \xleftarrow{l} G_I \xrightarrow{r} G_R$, a
(concrete) pattern instantiation is a
triple of (concrete) derivation sequences of the following form:
\begin{align*}
sn(S, v_1) &\Longrightarrow_{v_1,l(p_1)}^{G_L} H_1'
\Longrightarrow_{l(v_2),l(p_2)}^{G_L} H_2'
\Longrightarrow_{l(v_3),l(p_3)}^{G_L} \cdots
\Longrightarrow_{l(v_n),l(p_n)}^{G_L} H_n'\\
sn(S, v_1) &\Longrightarrow_{v_1,p_1}^{G_I} H_1
\Longrightarrow_{v_2,p_2}^{G_I} H_2 \Longrightarrow_{v_3,p_3}^{G_I} \cdots
\Longrightarrow_{v_n,p_n}^{G_I} H_n\\
sn(S, v_1) &\Longrightarrow_{v_1,r(p_1)}^{G_R} H_1''
\Longrightarrow_{r(v_2),r(p_2)}^{G_R} H_2''
\Longrightarrow_{r(v_3),r(p_3)}^{G_R} \cdots
\Longrightarrow_{r(v_n),r(p_n)}^{G_R} H_n''
\end{align*}
where the derivation sequence in the grammar $G_L$ ($G_R$) is the
induced embedding from the derivation sequence in $G_I$ over the grammar
embedding $l$ ($r$). The language
of $B,$ denoted $L(B),$ is the set of all graph rewrite rules
$H_n' \xleftarrow{l_n} H_n \xrightarrow{r_n} H_n''$
obtained by
performing concrete parallel derivations, where $l_n$ and $r_n$ are the
induced embeddings (monos) given from Lemma~\ref{lem:induced}.
\end{definition}
Therefore, by performing parallel derivations, a B-edNCE pattern can be used to
specify infinitely many rewrite rules between concrete graphs. In this
sense, a B-edNCE pattern may be seen as representing an equational schema
between families of graphs.

\begin{example}\label{ex:final-parallel-shit}
Consider the B-edNCE pattern from Example~\ref{ex:b-ednce-rewrite-pattern}.
Every parallel instantiation of length $n$ gives us a graph rewrite
rule which may be used for rewriting. For example, a derivation of length
3 results in the following graph rewrite rule:
\cstikz[0.9]{b-pattern-derive2.tikz}
where the concrete graphs form a span which is a graph rewrite
rule. The particular morphisms for the span are induced from the
B-edNCE rewrite rule, as explained in the previous definition. This
particular rule rewrites $K_3$ to $K_{1,3}$ by deleting one edge.
\end{example}

In Section~\ref{sec:partial-besg}, we showed how to compute pushout
complements and we identified the matching conditions under which they exist.
However, arbitrary grammar matchings may result in grammars whose induced
languages are difficult to reason about using the B-edNCE patterns that
we have. So, we shall introduce \emph{saturated grammar matchings} which
are special kinds of grammar matchings. We will later prove that by
restricting ourselves to saturated grammar matchings, we will be able to
rewrite entire B-edNCE grammars using a B-edNCE pattern in a way which
allows us to relate concrete derivations using the graph rewrite rules
induced by the pattern.

\begin{definition}[Saturated grammar matching]
Given a B-edNCE pattern $G_L \xleftarrow{l} G_I \xrightarrow{r} G_R$ with
initial nonterminal label $S$,
and given a
B-edNCE grammar $G_H$, we say
that $m : G_L \to G_H$ is a \emph{saturated grammar matching}, if the following
conditions are satisfied:
\begin{description}
\item[Embedding:] $m : G_L \to G_H$ is a monomorphism
\item[Production saturation:] For every production $p \in G_I,$ let $l(p) :=
p_L, r(p) = p_R$ and let $m\circ l(p) = p_H.$
Then, the restriction of $m$ to $m' : rhs(p_L) \to rhs(p_H)$ is
an extended graph saturated matching with respect to the restrictions of $l$
and $r$ to $rhs(p_L) \xleftarrow{l'} rhs(p) \xrightarrow{r'} rhs(p_R)$
\item[Production branching:] For every production $p \in G_H,$ if there exists
a production $p' \in G_L,$
such that $m(p')=p,$ then for any $q \in G_H$ with $lhs(p)=lhs(q),$ there
exists $q' \in G_L,$ such that $m(q') = q$
\item[Initiality:] The initial nonterminal label of $G_H$ is either $S$ or it
is not used by any vertex in $G_L.$
\item[Nonterminal covering:] For each nonterminal vertex $x \in G_H$, if
$\lambda(x)$ is a label used in $G_L$ and $x \not \in m(G_L),$
then $\lambda(x) = S.$
\end{description}
\end{definition}

Again, as in the case of extended graph saturated matching, these conditions
are stronger
than what is needed in order to perform DPO rewriting on edNCE grammars.
However,
they guarantee that derivations in the grammar are admissible in the
sense shown in Theorem~\ref{thm:main}.

\begin{example}\label{ex:i-am-sick-of-it}
Consider the B-edNCE rewrite pattern from
Example~\ref{ex:b-ednce-rewrite-pattern}.
This rewrite pattern may be matched onto
another B-edNCE grammar $G_H$, given by: 
\cstikz[0.9]{b-example2.tikz}
where again the matching is uniquely determined by the labels of the
productions and vertices. This is an example of a saturated matching.
If we decide to perform the rewrite at this matching, then
the result will be
the B-edNCE grammar $G_M$, given by:
\cstikz[0.9]{b-example3.tikz}
The rewrites are performed locally, on a per-production basis. The language
$G_H$ consists of complete graphs, where in addition, each vertex
of the complete graph has a (grey) line graph of arbitrary length glued onto it.
Similarly, the language of $G_M$ consists of a star graph, where there
are (grey) line graphs glued onto each of the star vertices.
\end{example}

We proceed by showing how a saturated matching between two grammars relates
their derivation trees (cf. Subsection~\ref{sub:derivation-tree}). For this, we
first need to introduce an additional definition.

\begin{definition}[Matched nodes, edges in derivation tree]
Given a saturated grammar matching $m: G_L \to G_H$ and a derivation tree $T$
for $G_H$, we say that a node $n$ of $T$ is \emph{matched}, if there exists
a production $p'
\in G_L$ such that $m(p') = \lambda_T(n)$. We say that an edge $e$ of $T$ is
\emph{matched}, if its identified nonterminal vertex
is in the image of $m$.
\end{definition}

The following lemma describes the effect of the embedding, production
branching,
initiality and nonterminal covering conditions and relates some derivations in
the grammar $G_H$
to those of $G_L$ (and thus also, $G_I$ and $G_R$).

\begin{lemma}\label{lem:tree}
Given a B-edNCE pattern $G_L \xleftarrow{l} G_I \xrightarrow{r} G_R$ with
initial nonterminal label $S$, a saturated grammar matching
$m: G_L \to G_H$
and given a concrete derivation tree $T$ for
$G_H$ then for any maximal subtree $Y$ of $T$, such that $Y$ consists only of
matched nodes and matched edges, there exists an isomorphic concrete derivation
tree $Y'$ for $G_L$.
\end{lemma}

\begin{proof}
Let the root of $Y$ be $n$
and let's assume it has label $p$. If $n$ is also the root of $T$, then $lhs(p)
= S$ by the initiality condition. Otherwise, consider the predecessor of
$n$ which we shall call $k.$ If $k$ is not a matched node, then it follows from
the nonterminal covering condition that $lhs(p) = S.$ If $k$ is a matched node,
then the edge between $k$ and $n$ is not matched as otherwise we would violate
the maximality assumption for $Y$. Then, it follows $lhs(p) = S$ by the
nonterminal covering condition. Therefore, in all cases, $lhs(p) = S$.

From this, by using the production branching condition, we get that there
exists $p' \in G_L$ with $lhs(p') = S$ and $m(p') = p$. The root of $Y'$
is then a node with label $p'$.

Using an inductive argument, we can explain how to construct the rest of $Y'$.
Let $k$ be a node in $Y$ for which we have already constructed its isomorphic
counterpart $k'$ in $Y'$. Let $h$ be any child of $k$ which is also in $Y$ and
let its label be $p$. Then, $h$ is a matched node and therefore we can
construct a node $h'$ in $Y'$ with label $p' \in G_L,$ such that $m(p')=p$, as
required. The parent edge $e$ of $h$ is a matched edge, therefore it identifies
a nonterminal vertex $x$ in the image of $m$. Thus, we can construct an edge
$e'$ from $k'$ to $h'$ in $Y'$ which identifies the pre-image of $x.$

What remains to be shown is that the tree $Y'$ is concrete. We already know
that the root of $Y'$ is an initial production of $G_L$, so the proposition
will follow if we prove that every node in $Y'$ has as many children as it
has nonterminal vertices in its associated production in $G_L$. This is
equivalent to proving that every node in $Y$ has as many children as
nonterminal vertices in the image of $m$ in its associated production in $G_H$
(because nonterminal vertices which aren't in the image of $m$ cannot have
matched edges associated to them). The tree $T$ is concrete and therefore
every node has the full number of children possible in $T$. $Y$
is a maximal subtree where all nodes and edges are matched and by definition
an edge is matched iff its associated nonterminal vertex in the parent
production is in the image of $m$. This completes the proof.
\end{proof}

\begin{example}\label{ex:derivation-shit-final}
Continuing on from Example~\ref{ex:i-am-sick-of-it}, consider the rewrite
described there (which is over a saturated matching). Then, for the derivation
tree for the grammar $G_H$ given by:
\cstikz{derive-tree-laino.tikz}
the maximal subtree consisting of matched nodes and edges is coloured in red.
That is because only these productions are matched from the B-edNCE
pattern. Also, it's easy to see that the corresponding derivation tree
consisting of only the red nodes and edges is concrete for the
B-edNCE pattern.
\end{example}

We have seen that B-edNCE patterns are useful for relating the derivations
between a triple of grammars. However, when we rewrite a target grammar, we
wish to relate its language to the language of the result of the rewrite.
In order to do this, we will introduce a notion which is similar to that
of B-edNCE pattern, but instead it is defined over a pair of grammars.

\begin{definition}[B-edNCE Correspondence]
A B-edNCE \emph{correspondence} is a couple of B-edNCE grammars $B:= (G_L, G_R),$
such that
there is a label-preserving bijection between the productions of both grammars.
Moreover,
there's also a label-preserving bijection between the nonterminals in
corresponding productions of the grammars. In addition, both grammars
have the same initial nonterminal label.
\end{definition}

So, a B-edNCE correspondence is very similar to a B-edNCE pattern, but it
doesn't have a third (boundary) grammar. We shall use a B-edNCE correspondence
to keep track of how DPO rewrite rules modify a B-edNCE grammar and relate the
languages of the original and modified grammar. Similarly to B-edNCE patterns,
we are able to perform parallel instantiations in a similar way.

\begin{definition}[B-edNCE correspondence instantiation]
\label{def:b-ednce-inst}
Given a B-edNCE correspondence $(G_1, G_2)$, an \emph{instantiation} is given
by a pair of concrete derivations:
\begin{align*}
sn(S,v_1)
&\Longrightarrow_{v_1,p_1}^{G_1}& &H_1&
&\Longrightarrow_{v_2,p_2}^{G_1}& &H_2&
&\Longrightarrow_{v_3,p_3}^{G_1}& \cdots
&\Longrightarrow_{v_n,p_n}^{G_1}& &H_n& \\
sn(S,v_1)
&\Longrightarrow_{v_1,F(p_1)}^{G_2}&    &H_1'&
&\Longrightarrow_{F(v_2),F(p_2)}^{G_2}& &H_2'&
&\Longrightarrow_{F(v_3),F(p_3)}^{G_2}& \cdots
&\Longrightarrow_{F(v_n),F(p_n)}^{G_2}& &H_n'&
\end{align*}
where $F$ is the bijection between the productions and nonterminal vertices
of the correspondence.
\end{definition}

Clearly a B-edNCE pattern and a B-edNCE correspondence are very closely
related. In fact, every B-edNCE pattern induces a B-edNCE correspondence
by choosing any two of its grammars (and defining the necessary bijection
appropriately). Our next lemma shows the interplay between B-edNCE patterns,
saturated matchings and B-edNCE correspondences.

\begin{lemma}\label{lem:pattern}
If $G_H$ is a B-edNCE grammar and $G_M$ is an $\mathcal S$-rewrite of $G_H$
over a B-edNCE pattern $B := G_L \xleftarrow{l} G_I \xrightarrow{r} G_R$ at
saturated matching
$m : G_L \to G_H,$ where:
\cstikz{b-ednce-dpo.tikz}
then $(G_H, G_M)$ is a B-edNCE correspondence.
Moreover, the bijection $F : P_{G_H} \to P_{G_M}$ between the productions
of $G_H$ and $G_M$ is given by:
\[
F(p) =
\begin{cases}
f\circ r(p') & \text{ if } p = m\circ l(p') \text{, where } p' \in G_I\\
g(p') & \text{ if }  p = s(p') \text{, where } p' \in G_K, p' \not\in k(G_I)\\
\end{cases}
\]
and the bijection $G$ between the nonterminal vertices of $G_H$ and $G_M$
is given by:
\[
G(x) =
\begin{cases}
f\circ r(x') & \text{ if } x = m\circ l(x') \text{, where } x' \in G_I \text{
is a nonterminal vertex}\\
g(x') & \text{ if }  x = s(x') \text{, where } x' \in G_K, x' \not\in k(G_I)
\text{ and } $x'$ \text{ is a nonterminal vertex}\\
\end{cases}
\]
\end{lemma}
\begin{proof}
An $\mathcal S$-rewrite in \textbf{edNCE} does not modify the initial
nonterminal label of the grammar.
Therefore, both grammars $G_H$ and $G_M$ have the same initial nonterminal.

The fact
that there is a bijection between the productions of $G_H$ and $G_M$
follows after recognizing that only the bodies of the productions in $G_H$ are
rewritten. In particular, notice that each production $p \in G_L\ (G_R)$ is
in the image of $l\ (r).$ Therefore, no productions are added, nor removed
when performing the rewrite. This means that the definition of $F$ above
is complete (in the sense that it is totally defined) and checking that it
is a label-preserving bijection follows immediately.

Showing that $G$ is a bijection on nonterminal vertices follows by a
completely analogous argument.
The two grammars only differ in the bodies of the productions in the image
of $m$
which are modified by a rule $B_p$ from the pattern $B$. But, by definition of
B-edNCE pattern, any such rewrite rule neither deletes, nor adds nonterminal
vertices. This means that the definition of $G$ is complete and checking that
it is a label-preserving bijection follows immediately. If $x$ is a
nonterminal vertex in production $p$ of $G_H$, then $F(x)$ is a nonterminal
vertex in production $G(p)$ of $G_M$. This follows because both bijections
are defined in terms of the same morphisms, which preserve the structure
of the grammars.

We also have to show in addition that $G_M$ is a B-edNCE grammar. This follows
immediately after recognizing that a rule $B_p$ from a B-edNCE pattern cannot
add edges between nonterminal vertices, nor can it add connection instructions
of the form $(X, \alpha, \beta, x, d)$, where $X$ is a nonterminal label, as
this violates the definition of B-edNCE grammars.
\end{proof}

As the above lemma shows, when we are doing grammar rewrites over saturated
matchings, then we get a B-edNCE correspondence. In particular, the bijections
$F$ and $G$ relate nonterminal vertices and productions from $G_H$ to those of
$G_M$. If $x\ (p)$ is a nonterminal vertex (production) of $G_H$
and $x'\ (p')$ is a nonterminal vertex (production) of $G_M$, such that
$G(x) = x'\ (F(p) = p'),$ then we shall say that $x$ and $x'$ ($p$ and $p'$)
are \emph{corresponding} nonterminal vertices (productions). In particular,
if in the above lemma all morphisms are subgraph inclusions, then the
two bijections are simply the identity function.

The following lemma won't be used for the proof of the main theorem in
this section, but it will be useful for multiple proofs in the next section.
It is similar in spirit to Lemma~\ref{lem:saturated_preserve} in that it
shows that saturated matchings are preserved by rewrites.

\begin{lemma}\label{lem:saturated_grammar_matching}
Given a B-edNCE pattern $B := G_L \xleftarrow{l} G_I \xrightarrow{r} G_R$, a
B-edNCE grammar $G_H$ and a saturated matching $m: G_L \to G_H$, then if $G_H
\leadsto_{B,m} G_M$ is an $\mathcal S$-rewrite, with:
\cstikz{b-ednce-dpo.tikz}
then $f$ is a saturated grammar matching with respect to the B-edNCE pattern
$G_R \xleftarrow r G_I \xrightarrow l G_L.$
\end{lemma}
\begin{proof}
We have assumed that $G_M$ is an $\mathcal S$-rewrite of $G_H$ and therefore
from the partially adhesive structure of \textbf{edNCE} we know that all
of the morphisms in the DPO diagram, including $f$ are monomorphisms.
Therefore, the \textbf{Embedding} condition is satisfied by $f$. The
\textbf{Production saturation} condition follows immediately by applying
Lemma~\ref{lem:saturated_preserve} to each rewritten production. 
The remaining conditions follow easily by making use of the results from
Lemma~\ref{lem:pattern}. In particular, the
\textbf{Initiality} condition for $f$ is easily seen to be satisfied
by combining the fact that $(G_H, G_M)$ forms a B-edNCE correspondence
and the fact that the same condition is satisfied
for the saturated matching $m$. The \textbf{Production branching} condition is
satisfied, because the grammar rewrite only modifies the RHS of the
productions, thus each modified production is in the image of $f$ iff its
unmodified corresponding production is in the image of $m$ and therefore the
condition
follows from the \textbf{Production branching} condition which is satisfied for
$m$. Finally, the \textbf{Nonterminal covering} condition is satisfied for
$f$, because $G_L$ and $G_R$ use the same nonterminal labels, each
nonterminal $x \not\in f(G_R)$ iff $x' \not\in m(G_L),$ where $x'$ is
the corresponding nonterminal vertex of $x$ in $G_R$, and because
the \textbf{Nonterminal covering} condition is satisfied for $m$.
\end{proof}

A B-edNCE correspondence between two grammars $G_1$ and $G_2$ allows us to pair
up concrete graphs from $L(G_1)$ and $L(G_2)$. Our final definition from this
section is the notion of an \emph{admissible} B-edNCE correspondence.  In
short, an admissible correspondence with respect to a set of rewrite rules
$\mathcal R$ allows us to rewrite any graph from $L(G_1)$ into its
corresponding graph from $L(G_2)$ using some rewrite rules from $\mathcal R$.
Essentially,
the next definition
formalizes the notion of a sound transformation between two grammars.

\begin{definition}[Admissibility]
\label{def:admissible-cor}
A B-edNCE correspondence $B = (G_L, G_R)$ is \emph{admissible} with respect to
a
set of rewrite rules $\mathcal R,$ if for every instantiation of $B$:
\begin{align*}
sn(S, u_1) &\Longrightarrow_{*}^{G_L} H\\
sn(S, u_1) &\Longrightarrow_{*}^{G_R} H'
\end{align*}
there exists a sequence of rewrite rules $s_1,\ldots,s_n \in \mathcal R,$ such
that $H \leadsto_{s_1} \cdots \leadsto_{s_n} H'.$
\end{definition}

The main result in this section is the next theorem. It states that when
we restrict ourselves to saturated grammar matchings, then the result of the
rewrite is an admissible correspondence with respect to the language of
rewrite rules induced by the rewrite pattern.

\begin{theorem}\label{thm:main}
Given a B-edNCE pattern $B := G_L \xleftarrow{} G_I \xrightarrow{} G_R$, a
B-edNCE grammar $G_H$ and a saturated matching $m: G_L \to G_H$, then if $G_H
\leadsto_{B,m} G_H'$ is an $\mathcal S$-rewrite, then it follows $(G_H, G_H')$
is a B-edNCE correspondence
admissible with respect to $L(B).$
\end{theorem}
\begin{proof}
From Lemma~\ref{lem:pattern}, we know that $(G_H, G_H')$ is indeed
a B-edNCE correspondence. Thus, for every concrete derivation tree $T$ of $G_H$, an
isomorphic concrete derivation tree $T'$ exists for $G_H'$. However, the yield
of $T$ and the yield of $T'$ may differ. In particular, any subtree consisting
entirely of non-matched nodes in $T$ has the same yield as its corresponding
subtree in $T'$. But a subtree consisting of matched nodes in $T$ may have a
different yield compared to its corresponding tree in $T'$.

In the base case, if we consider a single matched node $u$ in $T$ and its
corresponding node $u'$ in $T'$, then if the label of $u$ is a production $p$
of $G_H$, then the label of $u'$ is a production $p'$ obtained by performing
the rewrite $rhs(p) \leadsto_{B_{p_l},m_p} rhs(p'),$ where $p_l \in G_L$ is
the production of $G_L$ such that $m(p_l) = p$ and $m_p$ is the restriction
of $m$ to $rhs(p)$.
Thus, by repeatedly applying
Theorem~\ref{thm:subst_rewrite}, we know there exist $m'$ and $R$, such that
$yield(Y) \leadsto_{R,m'} yield(Y'),$ for any subtree $Y$ of $T$ consisting
of matched nodes and edges and $Y'$ is its corresponding tree in $T'$. If we
choose $Y$ to be maximal with respect to matched nodes and edges,
then by Lemma~\ref{lem:tree} $R$ would be in $L(B)$. This means $R$ is
a standard DPO rewrite rule on terminal graphs.
Then, if $Y_i$ are the maximal subtrees of $T$
consisting of matched nodes and edges and $Y_i'$ are their corresponding
counterparts in $T'$, we can rewrite $yield(T) \leadsto_s yield(T')$ using
a sequence of rewrite rules $s$ consisting of rules $R_i$ obtained as already
explained (where the matchings extend trivially). The rules
$R_i$ can be applied in any order, or in fact all at the same time in parallel,
as they are matched on disjoint parts of the graph $yield(T)$.

What remains to be shown is that the no-dangling
edges condition for DPO rewriting on graphs is satisfied when applying each
rewrite rule
$R_i$ to the larger graph $yield(T)$ and we also need to explain how to extend
the matchings to the larger graph $yield(T)$ from what we have already
described.

Note, that for any maximal matched subtree $Y$ of $T$, $yield(Y)$ is in general
a sentential form of $G_H$. As already explained, we can rewrite such a
sentential form to the corresponding sentential form $yield(Y')$ of $G_H'$
using a graph rewrite rule $R$ at a matching $m'$, that is,
$yield(Y) \leadsto_{R,m'} yield(Y').$ However, in $yield(T)$ some of the
nonterminal vertices of $yield(Y)$ may be expanded, and, in addition,
$yield(Y)$ may itself be substituted for another nonterminal vertex of
another sentential form. In both cases, we can extend the match $m'$ to
$yield(T)$ trivially -- $m'$ maps a terminal graph to a terminal graph of
$yield(Y)$, which is unmodified after all substitutions in $yield(T)$,
therefore the extended
matching is the same as $m'$, where only the
codomain is changed. Because that matching is injective, we need to show the
standard no-dangling edges condition for graphs is satisfied.

Observe, that the root of $Y$ is labelled by the initial
nonterminal label of $G_L.$ Thus, by assumption, this production has no
connection instructions in $G_L$ and therefore neither does the corresponding
production in $G_H$ (follows from the production saturation condition).
Therefore, substituting $yield(Y)$ for any nonterminal vertex in any
sentential form of $G_H$ cannot establish new edges to any of its vertices.
Thus,
dangling edges may only be established by an
expansion of a nonterminal in $yield(Y).$ However, because $Y$ is maximal
with respect to matched nodes and edges, this means that the only non-expanded
nonterminal vertices in $Y$ are outside of the image of the matching $m$. Then,
due to the no-dangling condition which is satisfied locally for each production,
these nonterminal vertices may only be adjacent to either boundary vertices of
$yield(Y)$ (with respect to the rewrite rule $R$) or vertices outside of the
image of $m$. In either case, expanding such a nonterminal vertex cannot create
a dangling edge and therefore the extended matching
satisfies the no-dangling edges condition.
\end{proof}

The results in this section are not directly concerned with string graphs
or string diagrams. They have been stated in general for B-edNCE grammars
which generate languages of (arbitrary) graphs. It might be possible that
these results are applicable in other domains, but we leave this question open
for future work. 

\begin{example}
Continuing the running example from this section, consider the derivation
tree $Y$ from Example~\ref{ex:derivation-shit-final}. Its yield for the grammar
$G_H$ is the graph $H,$ given by:
\cstikz{shit-graph-tree1.tikz}
The same derivation tree may be used for grammar $G_M$ (because the two
grammars form a B-edNCE correspondence) and its yield is the graph $M,$ given
by:
\cstikz{shit-graph-tree2.tikz}
Then, from the main theorem in this section, we know that we can rewrite
$H \leadsto_{m,s} M$, where both the matching $m$ and the rewrite rule
$s$ are induced from the B-edNCE pattern used for the rewrite in combination
with the derivation tree $Y$. In particular, in order to obtain the
rewrite rule $s$, we simply take all the maximal matched subtrees (in
this case it is only one) whose yields are then concrete graph rewrite rules
which we need to apply. In this case, $s$ may be constructed by the
parallel derivation from Example~\ref{ex:final-parallel-shit}, as this is
the yield of the maximal matched subtree of $Y$ when evaluated at the
pattern.
\end{example}

In summary, we have shown how to rewrite B-edNCE grammars using B-edNCE rewrite
patterns in such a way that we can systematically relate their languages using
DPO rewrites induced by the rewrite patterns. In the following sections, we
will build upon these results by showing how to extend them to B-ESG grammars,
so that we may apply these techniques to the languages of string graphs (and
thus string diagrams) we are interested in. Then we will see that the notion
of admissibility we have introduced in this section translates into sound
rewrites of the string diagram families we are representing.

\section{B-ESG rewrite rules}\label{sec:b-esg-rules}

In Chapter~\ref{ch:besg} we introduced B-ESG grammars which generate languages
of
string graphs and which have some important decidability properties.
We can use a single B-ESG grammar to represent a single family of string graphs.
However, we still have not described how to represent equational schemas
between families
of string graphs. 
This is the primary contribution of this section.

We begin by introducing two auxiliary definitions.

\begin{definition}[Production input/output/isolated vertex]
Given a B-ESG grammar $B$, we say that a wire-vertex $w$ is a \emph{production
input} if its production in-degree is zero. $w$ is a \emph{production
output} if its production out-degree is zero. $w$ is a \emph{production
isolated wire-vertex} if its production in-degree and production out-degree are
both zero.
\end{definition}

\begin{definition}[B-ESG normal form]
A B-ESG grammar $B = (G,T)$ is in \emph{B-ESG normal form} if
$G$:
\begin{itemize}
\item is neighbourhood preserving
\item is context-consistent
\item contains no useless connection instructions
\item is reduced
\end{itemize}
Moreover, we will say that $B$ is in \emph{proper B-ESG normal
form}, if $B$ is in B-ESG normal form and also $G$ contains no production
isolated wire-vertices.
\end{definition}

Next, we introduce the notion of \textit{B-ESG rewrite pattern} which extends
the notion of B-edNCE correspondence to B-ESG grammars. As such, we can use
B-ESG rewrite patterns to relate the languages of two B-ESG grammars.

\begin{definition}[B-ESG rewrite pattern]\label{def:besg_rewrite_pattern}
A \textit{B-ESG rewrite pattern} is a pair of B-ESG grammars $B_1 = (G_1, T)$
and $B_2 = (G_2,T)$, where $(G_1, G_2)$ form a B-edNCE correspondence
and such that both grammars $G_1$ and $G_2$ are in proper B-ESG normal form.
Moreover, corresponding pairs of productions $p$ and $p'$ in $G_1$ and $G_2$
satisfy the following condition:
\begin{description}
\item[IO:] There is a label-preserving bijection between the
production inputs (outputs) in $p$ and the production inputs (outputs)
in $p'.$
\end{description}
\end{definition}

In the above definition, grammars $G_1$ and $G_2$ form a B-edNCE correspondence
and therefore we may perform parallel derivation
sequences on both grammars $G_1$ and $G_2$ in the same way as in the previous
section, which would generate a pair of encoded string graphs. If we follow
this by a decoding, then we would get a pair of string graphs. This is
made precise by the next definition.

\begin{definition}[B-ESG pattern instantiation]\label{def:inst}
Given a B-ESG rewrite pattern $(B_1, B_2)$, with $B_1 = (G_1, T)$ and $B_2 =
(G_2, T)$ a B-ESG pattern \emph{instantiation} is given by an
instantiation for $(G_1, G_2),$ followed by a decoding:
\begin{align*}
sn(S,v_1)
&\Longrightarrow_{v_1,p_1}^{G_1}& &H_1&
&\Longrightarrow_{v_2,p_2}^{G_1}& &H_2&
\cdots
&\Longrightarrow_{v_n,p_n}^{G_1}& &H_n& 
&\Longrightarrow_{*}^{T}& &F& \\
sn(S,v_1)
&\Longrightarrow_{v_1,K(p_1)}^{G_2}&    &H_1'&
&\Longrightarrow_{K(v_2),K(p_2)}^{G_2}& &H_2'&
\cdots
&\Longrightarrow_{K(v_n),K(p_n)}^{G_2}& &H_n'&
&\Longrightarrow_{*}^{T}& &F'&
\end{align*}
where $K$ is the bijection between the productions and nonterminal vertices
of the correspondence $(G_1, G_2).$
\end{definition}

In other words, we use an identical derivation sequence in the two B-edNCE grammars
to get two encoded string graphs, which are then uniquely decoded using the
rewrite rules of $T$.

\begin{example}\label{ex:rewrite-pattern-crap}
As a concrete example, our framework allows one to represent the B-ESG
rewrite pattern $B:= (B_L, B_R)$, where the
grammars are given by:
\cstikz[0.9]{pattern-example.tikz}
and where the required bijections map the productions vertically as they appear.
The mapping on the wire-vertices and nonterminals is then obvious, because all
productions have one wire-vertex and one nonterminal vertex. This rewrite
pattern relates a complete string graph (where all node-vertices have one
output) to a star string graph (where all node-vertices have one output).

This example is particularly interesting, because it contains the essential
graph data of the 'Y-$\Delta$ rule' from electrical circuits \cite{Curtis1998}
as well as the \emph{local complementation} rule for which we mentioned in the
introduction that it is very important for \textit{measurement-based quantum
computation} \cite{NestMBQC,euler_necessity}.

Every parallel instantiation of length $n$ thus gives a pair of string
graphs where their inputs and outputs are in 1-1 correspondence. For
example, a parallel derivation of length 3 (not counting decoding), gives us:
\cstikz[0.9]{pattern-derive.tikz}
Observe, that after each derivation step, the inputs/outputs and the
nonterminal vertices are in bijection.
\end{example}

\begin{definition}[Category of B-ESG grammars]
The category of B-ESG grammars over a decoding system $T$, denoted
\textbf{B-ESG}$_T$, or simply \textbf{B-ESG} if $T$ is clear from the context,
has objects B-ESG grammars $B = (G,T)$. A morphism $h$ between two B-ESG
grammars $B_1 = (G_1, T)$ and $B_2 = (G_2, T)$ is an \textbf{edNCE}
morphism $h : G_1 \to G_2$.
\end{definition}

In the above definition, we are using the same symbol to refer to both
a \textbf{B-ESG} morphism and also its underlying \textbf{edNCE} morphism.
We do this for simplicity, because the two are exactly the same. It will
always be clear from context to which notion we are referring to. It is
also obvious that, for any choice of $T$, \textbf{B-ESG}$_T$ is isomorphic to
the full subcategory of \textbf{edNCE} whose objects satisfy the
B-ESG conditions. Also, for brevity, if we denote a B-ESG grammar
as $B_X$, then its underlying edNCE grammar will be denoted $G_X$, so that
$B_X = (G_X, T)$. From now on, we also assume that the decoding system
$T$ is fixed for all B-ESG grammars.

We proceed by defining \emph{B-ESG rewrite rules}. A B-ESG rewrite rule is
simply a B-edNCE rewrite pattern which has been extended to B-ESG grammars.
B-ESG rewrite rules, can be used to relate the languages of three B-ESG
grammars and they are used in order to generate a language consisting of
string graph rewrite rules.

\begin{definition}[B-ESG rewrite rule]
A \emph{B-ESG rewrite rule} is a span of monos $B_L \xleftarrow l B_I
\xrightarrow r B_R,$ in \textbf{B-ESG}, where
$B_L = (G_L,T),
B_I = (G_I, T), B_R = (G_R, T),$ such that $G_L \xleftarrow l G_I \xrightarrow r
G_R$ is a B-edNCE pattern such that for every triple of corresponding
productions $p_L, p_I, p_R$ in $G_L, G_I, G_R$ respectively, we have:
\begin{description}
\item[Boundary:] $p_I$ contains only nonterminal vertices and isolated
wire-vertices and it contains no edges, connection instructions or
node-vertices.
\item[IO1:] $l$ and $r$ are surjections on the production inputs (outputs)
between $p_I$ and $p_L$, $p_I$ and $p_R$ respectively.
\item[IO2:] For every wire-vertex $w \in p_I$, $l(w)$ and $r(w)$ are both a
production input (output) in $p_L$ and $p_R$ respectively.
\end{description}
Moreover, $G_I$ is in B-ESG normal form and grammars $G_L$, $G_R$ are both
in proper B-ESG normal form.
\end{definition}

B-ESG rewrite rules and B-ESG rewrite patterns are very tightly related.
The next two lemmas show that they are interchargable -- we can get a
B-ESG rewrite rule from a B-ESG rewrite pattern and vice-versa. So, a
natural question to ask is why do we introduce two different notions which
are recoverable from one another. This is because a B-ESG rewrite rule
explicitly generates a language consisting of string graph rewrite rules (as we
shall prove) while a B-ESG rewrite pattern is useful for describing the
relationship between grammar rewrites (in this sense it is similar to a B-edNCE
correspondence).

However, note that the analogous notions -- B-edNCE patterns and B-edNCE
correspondences -- are not interchargable. The reason why this is the case is
because string graph rewrite rules have a natural boundary -- it simply
consists of all of the inputs and outputs which must be the same in both the
left-hand side and the right-hand side of each rule (cf.
Section~\ref{sec:string}). However, no such analogous notion exists for
arbitrary graphs.

\begin{lemma}
If $B_L \xleftarrow l B_I \xrightarrow r B_R$ is a B-ESG rewrite rule, then
$(B_L, B_R)$ is a
B-ESG rewrite pattern.
\end{lemma}
\begin{proof}
The only property which isn't immediately obvious is the \textbf{IO} property
from the definition of B-ESG rewrite pattern. However, this follows
by combining properties \textbf{IO1} and \textbf{IO2} from the definition
of B-ESG rewrite rule. In particular, the bijection can be defined
by identifying $l(w)$ with $r(w)$ for each wire-vertex $w \in G_I$. Condition
\textbf{IO2} ensures $l(w)$ and $r(w)$ are both production inputs (outputs)
and condition \textbf{IO1} ensures that every production input (output) in
$G_L$ and $G_R$ is covered by $l$ and $r$ in this way.
\end{proof}

\begin{lemma}
If $(B_L, B_R)$ is a B-ESG rewrite pattern, with $B_L = (G_L, T)$ and $B_R =
(G_R, T),$ then there exists a unique (up to isomorphism) B-ESG grammar $B_I
=(G_I, T)$ and monos $l$ and $r$, such that
$B_L \xleftarrow l  B_I \xrightarrow r B_R$ is a B-ESG rewrite rule.
\end{lemma}
\begin{proof}
The triple $(G_L, G_I, G_R)$ has to be a B-edNCE pattern. This means
that if $G_L$ and $G_R$ have $n$ productions, then $G_I$ must have $n$
productions as well. Moreover $G_L \xleftarrow l G_I \xrightarrow r G_R$
is a span and $l$ and $r$ establish the bijective correspondence between
the productions of the three grammars. So, if the productions of $G_L$
are $\{X_i \to (D_i, C_i)\}_{i \in \mathcal I},$ then we can assume
without loss of generality that the productions of $G_I$ are given by
$\{X_i \to (D_i', C_i')\}_{i \in \mathcal I}.$ Let's consider a triple
of corresponding productions $p_L = X_i \to (D_i, C_i),$
$p_I = X_i \to (D_i', C_i')$ and $p_R =X_i \to (D_i'', C_i'')$
in $G_L, G_I$ and $G_R$ respectively. From the B-edNCE pattern requirement,
it follows that the nonterminal vertices in $D_i'$ are determined by
those of $D_i$ (or $D_i''$). Then, the \textbf{Boundary} condition implies
$C_i' = \emptyset$ and that $D_i'$ may contain only some number of isolated
wire-vertices. However, the number of isolated wire-vertices and their
labels are then completely determined by those of $D_i$ (or $D_i''$) thanks
to the \textbf{IO1} and \textbf{IO2} conditions. Thus,
the grammar $G_I$ is unique, up to isomorphism, if it exists.

The existence of $G_I$ can be demonstrated by simply taking the same
productions like those of $G_L$ (or those of $G_R$) and removing from each
production all:
wire-vertices which are not production inputs or outputs;
node-vertices;
edges; and
connection instructions.
Then, the monomorphism $l : G_I \to G_L$ is simply the grammar
inclusion and the mono $r : G_I \to G_R$ can be defined by the composition
of $l$ with the bijection from the definition of B-ESG rewrite pattern.
\end{proof}

\begin{example}\label{ex:b-esg-rewrite-rule-crap}
Let's consider the rewrite pattern from Example~\ref{ex:rewrite-pattern-crap}.
To get a B-ESG rewrite rule $B:= B_L \xleftarrow l B_I \xrightarrow r B_R$
from it, we simply copy all of the productions
of either $G_L$ or $G_R$ and then we remove everything in the bodies of
the productions, except for the nonterminal vertices and the inputs/outputs:
\cstikz[0.9]{example.tikz}
where $l$ and $r$ map each production of $G_I$ vertically up or down
respectively. Note that in this case the mapping on the vertices is uniquely
determined by their labels, because all productions have one wire-vertex
and one nonterminal vertex. This rewrite rule will rewrite a complete string
graph (where all node-vertices have one output) to a star string graph (where
all node-vertices have one output).

Every parallel instantiation of length $n$ thus gives us a string graph rewrite
rule which may be used for rewriting. For example, a derivation of length
3 (not counting decoding) gives us:
\cstikz[0.9]{pattern-derive2.tikz}
where the concrete string graphs form a span which is a string graph rewrite
rule. The particular morphisms for the span are induced from the
B-ESG rewrite rule, as explained in Definition~\ref{def:b-esg-inst},
which we shall soon provide.
\end{example}

Before we introduce our notion of instantiation for B-ESG rewrite rule, it
is useful to present the next lemma. It shows that if one encoded string graph
embeds into another, then so do their decoded counterparts.

\begin{lemma}\label{lem:decode-embed}
Given encoded string graphs $H, K$ and a monomorphism $m : H \to K$, then
there exists a monomorphism $m' : H' \to K'$, where $H \Longrightarrow_*^T H'$
and $K \Longrightarrow_*^T K'$. We will refer to this monomorphism as
the decoded embedding of $H'$ into $K'$.
\end{lemma}
\begin{proof}
We can define $m'$ by induction on the number of encoding edges of $H$.

If $H$ contains no encoding edges, then $m'=m$. Otherwise, to construct a mono
for a single decoding step  by replacing some edge $e \in H$, consider the
following. Since $m: H \to K$ is a monomorphism, then there exists an edge
$m(e) \in K$ which is also an encoding edge. Then, $m'$ acts exactly as $m$ on
all vertices and edges, except on the newly introduced ones after performing
the decoding $H \Longrightarrow^T_e H'$ and $K \Longrightarrow^T_{m(e)} K',$
where it is undefined. However, by definition, the decoding process will
introduce exactly the same graph (up to isomorphism) in both $H'$ and $K'$, so
$m'$ can be trivially extended.
\end{proof}

Next, we define how to instantiate B-ESG rewrite rules. This is a simple
generalisation of B-edNCE pattern instantiations.

\begin{definition}[B-ESG Rewrite Rule Instantiation]\label{def:b-esg-inst}
Given a B-ESG rewrite rule $B:= B_L \xleftarrow{l} B_I \xrightarrow{r} B_R$, a
(concrete) instantiation is a B-edNCE pattern instantiation for
$G_L \xleftarrow{l} G_I \xrightarrow{r} G_R$ followed by a decoding:
\begin{align*}
sn(S, v_1) &\Longrightarrow_{v_1,l(p_1)}^{G_L} H_1'
\Longrightarrow_{l(v_2),l(p_2)}^{G_L} H_2'
\Longrightarrow_{l(v_3),l(p_3)}^{G_L} \cdots
\Longrightarrow_{l(v_n),l(p_n)}^{G_L} H_n'
\Longrightarrow_*^T F'\\
sn(S, v_1) &\Longrightarrow_{v_1,p_1}^{G_I} H_1
\Longrightarrow_{v_2,p_2}^{G_I} H_2 \Longrightarrow_{v_3,p_3}^{G_I} \cdots
\Longrightarrow_{v_n,p_n}^{G_I} H_n
\Longrightarrow_*^T F\\
sn(S, v_1) &\Longrightarrow_{v_1,r(p_1)}^{G_R} H_1''
\Longrightarrow_{r(v_2),r(p_2)}^{G_R} H_2''
\Longrightarrow_{r(v_3),r(p_3)}^{G_R} \cdots
\Longrightarrow_{r(v_n),r(p_n)}^{G_R} H_n''
\Longrightarrow_*^T F''
\end{align*}
The language of $B,$ denoted $L(B),$ is the set of all rewrite rules $F'
\xleftarrow{l_F} F \xrightarrow{r_F} F''$ obtained by performing concrete
parallel derivations,
where $H_n' \xleftarrow{l_n} H_n \xrightarrow{r_n} H_n''$ is the
induced embedding given by Lemma~\ref{lem:induced} and
$F' \xleftarrow{l_F} F \xrightarrow{r_F} F''$ is the decoded embedding
given by Lemma~\ref{lem:decode-embed} when applied to
$H_n' \xleftarrow{l_n} H_n \xrightarrow{r_n} H_n''$.
\end{definition}

The main result in this section is to show that every B-ESG rewrite rule
induces a language consisting of string graph rewrite rules. From this it
easily follows that every pair of corresponding string graphs in a B-ESG
rewrite pattern have corresponding inputs and outputs. Therefore, B-ESG rewrite
patterns correctly represent equational schemas between families of string
diagrams. We have shown B-ESG rewrite patterns and B-ESG rewrite rules are
interchargable, but in the last section we will mostly be working with B-ESG
rewrite rules. The reason is simple -- B-ESG rewrite rules can be readily used
for rewriting, whereas B-ESG rewrite patterns cannot (they have to be converted
to a B-ESG rewrite rule first).

We will prove the main result by induction. In particular, we will show that
every instantiation (not necessarily concrete) is an \emph{ESG-form rewrite
rule}, which is a straightforward generalisation of string graph rewrite rule
(cf. Definition~\ref{def:string-graph-rewrite-rule}).

\begin{definition}[ESG-form Rewrite Rule]
An \textit{ESG-form rewrite rule} is a span of monomorphisms
$L \stackrel{l}{\longleftarrow} I \stackrel{r}\longrightarrow R$, 
where $L, I, R$ are ESG-forms with the following properties:
  \begin{description}
    \item[P1] $L$ and $R$ do not have any isolated wire vertices
    \item[P2] $In(L) \cong In(R)$ and $Out(L) \cong Out(R)$
    \item[P3] $T(I) \cong In(L) + Out(L) \cong In(R) + Out(R)$, where $T(I)$
is the full subgraph of $I$ consisting of terminal vertices only
    \item[P4] The following diagram commutes :
  \end{description}
  \begin{center}
    \begin{tikzpicture}
  \node (L)  at (-4, 0) {$L$};
  \node (I)  at (0, 0) {$T(I)$};
  \node (R)  at (4, 0) {$R$};
  \node (IL) at (-2, 2) {$In(L)$};
  \node (IR) at (2, 2) {$In(R)$};
  \node (OL) at (-2, -2) {$Out(L)$};
  \node (OR) at (2, -2) {$Out(R)$};

  \draw [->, dashed] (I) to node[above] {$l$} (L);
  \draw [->, dashed] (I) to node[above] {$r$} (R);
  \draw [->] (IL) to node[above] {$i$} (I);
  \draw [->] (IR) to node[above] {$i'$} (I);
  \draw [->] (OL) to node[above] {$j$} (I);
  \draw [->] (OR) to node[above] {$j'$} (I);
  \draw [<->] (IL) to node[above] {$\sim$} (IR);
  \draw [<->] (OL) to node[above] {$\sim$} (OR);
  \draw [left hook->] (IL) to (L);
  \draw [right hook->] (IR) to (R);
  \draw [left hook->] (OL) to (L);
  \draw [right hook->] (OR) to (R);
\end{tikzpicture}
  \end{center}
  where $i, j, i', j'$ are the coproduct inclusions. If $L,I,R$ contain
no nonterminal vertices, then we say that this is an ESG rewrite rule.
\end{definition}

So, in the above definition, we see that if our graphs do not contain any
nonterminal vertices, then we get a rewrite rule between encoded string graphs.
Furthermore, it is easy to see that
a string graph rewrite rule is simply an ESG rewrite rule, where
all three graphs contain no encoding edges.

\begin{lemma}\label{lem:esg-form-inst}
Given a B-ESG rewrite rule, every parallel instantiation:
\begin{align*}
sn(S, v_1) &\Longrightarrow_{v_1,l(p_1)}^{G_L} H_1'
\Longrightarrow_{l(v_2),l(p_2)}^{G_L} H_2'
\Longrightarrow_{l(v_3),l(p_3)}^{G_L} \cdots
\Longrightarrow_{l(v_n),l(p_n)}^{G_L} H_n'\\
sn(S, v_1) &\Longrightarrow_{v_1,p_1}^{G_I} H_1
\Longrightarrow_{v_2,p_2}^{G_I} H_2 \Longrightarrow_{v_3,p_3}^{G_I} \cdots
\Longrightarrow_{v_n,p_n}^{G_I} H_n\\
sn(S, v_1) &\Longrightarrow_{v_1,r(p_1)}^{G_R} H_1''
\Longrightarrow_{r(v_2),r(p_2)}^{G_R} H_2''
\Longrightarrow_{r(v_3),r(p_3)}^{G_R} \cdots
\Longrightarrow_{r(v_n),r(p_n)}^{G_R} H_n''
\end{align*}
is an ESG-form rewrite rule.
\end{lemma}
\begin{proof}
We can prove this by induction on the length of the derivation. If
the length of the derivation is zero, then we have the span
$sn(S,v_1) \xleftarrow{id} sn(S,v_1) \xrightarrow{id} sn(S,v_1)$ which
is obviously an ESG-form rewrite rule as there are no inputs, nor outputs.

Assume we have shown that $H_{n-1}' \xleftarrow{l_{n-1}} H_{n-1}
\xrightarrow{r_{n-1}} H_{n-1}''$ is an ESG-form rewrite rule.
Now, consider the span
$H_n' \xleftarrow{l_n} H_n \xrightarrow{r_n} H_n''$.
By construction (from Lemma~\ref{lem:induced}), both $l_n$ and $r_n$ are
mono and we know that $H_n', H_n, H_n''$ are all ESG-forms from
Theorem~\ref{thm:besg_language}.

We know that all of the grammars are neighbourhood-preserving, this means that
the in-degree (out-degree) of vertices in any sentential form cannot decrease
as we expand nonterminal vertices. Moreover, condition W2 from
Definition~\ref{def:besg} implies that the in-degree (out-degree) of any
wire-vertex cannot increase after applying a production. Therefore, the
inputs, outputs and isolated wire-vertices are preserved when applying
productions to our sentential forms.
Therefore, to prove the proposition, we simply have to show that the newly
introduced wire-vertices in $H_n', H_n, H_n''$ satisfy the required conditions.

We also know that all three grammars are
context-consistent and they contain no useless connection instructions.
Therefore, if some wire-vertex has production in-degree (out-degree) $n$,
then its in-degree (out-degree) in any sentential form of the grammar will
also be $n$, after applying the production which contains it. This fact
is crucial for the proof of this lemma.

Property \textbf{P1} follows for the span
$H_n' \xleftarrow{l_n} H_n \xrightarrow{r_n} H_n''$ from the
induction hypothesis and because
neither $l(p_n)$, nor $r(p_n)$ contain production isolated wire-vertices.

Next, we will show that properties \textbf{P2} and \textbf{P3} hold for inputs
(the case for outputs follows by symmetry). This follows by combining the
induction hypothesis with
conditions \textbf{IO1} and \textbf{IO2} which are satisfied by the
productions $p_n, r(p_n), l(p_n)$.

Finally, property \textbf{P4} holds, because the bijections between
the inputs/outputs are defined in terms of the embedding morphisms
$l_n, r_n$, so they hold by construction.
\end{proof}

\begin{corollary}\label{cor:esg-inst}
Given a B-ESG rewrite rule, every concrete parallel instantiation:
\begin{align*}
sn(S, v_1) &\Longrightarrow_{v_1,l(p_1)}^{G_L} H_1'
\Longrightarrow_{l(v_2),l(p_2)}^{G_L} H_2'
\Longrightarrow_{l(v_3),l(p_3)}^{G_L} \cdots
\Longrightarrow_{l(v_n),l(p_n)}^{G_L} H_n'\\
sn(S, v_1) &\Longrightarrow_{v_1,p_1}^{G_I} H_1
\Longrightarrow_{v_2,p_2}^{G_I} H_2 \Longrightarrow_{v_3,p_3}^{G_I} \cdots
\Longrightarrow_{v_n,p_n}^{G_I} H_n\\
sn(S, v_1) &\Longrightarrow_{v_1,r(p_1)}^{G_R} H_1''
\Longrightarrow_{r(v_2),r(p_2)}^{G_R} H_2''
\Longrightarrow_{r(v_3),r(p_3)}^{G_R} \cdots
\Longrightarrow_{r(v_n),r(p_n)}^{G_R} H_n''
\end{align*}
is an ESG rewrite rule.
\end{corollary}
\begin{proof}
This follows immediately from the previous lemma after recognising that
in a concrete derivation there are no nonterminal vertices left.
\end{proof}

The main result in this section then follows as a simple consequence of
the above results.

\begin{theorem}\label{thm:b-esg-inst}
The language of every B-ESG rewrite rule consists solely of string graph
rewrite rules.
\end{theorem}
\begin{proof}
Consider a concrete instantiation:
\begin{align*}
sn(S, v_1) &\Longrightarrow_{v_1,l(p_1)}^{G_L} H_1'
\Longrightarrow_{l(v_2),l(p_2)}^{G_L} H_2'
\Longrightarrow_{l(v_3),l(p_3)}^{G_L} \cdots
\Longrightarrow_{l(v_n),l(p_n)}^{G_L} H_n'
\Longrightarrow_*^T F'\\
sn(S, v_1) &\Longrightarrow_{v_1,p_1}^{G_I} H_1
\Longrightarrow_{v_2,p_2}^{G_I} H_2 \Longrightarrow_{v_3,p_3}^{G_I} \cdots
\Longrightarrow_{v_n,p_n}^{G_I} H_n
\Longrightarrow_*^T F\\
sn(S, v_1) &\Longrightarrow_{v_1,r(p_1)}^{G_R} H_1''
\Longrightarrow_{r(v_2),r(p_2)}^{G_R} H_2''
\Longrightarrow_{r(v_3),r(p_3)}^{G_R} \cdots
\Longrightarrow_{r(v_n),r(p_n)}^{G_R} H_n''
\Longrightarrow_*^T F''
\end{align*}
where $H_n' \xleftarrow{l_n} H_n \xrightarrow{r_n} H_n''$ is the
induced embedding from the parallel derivation in the context-free grammars
and where $F' \xleftarrow{l_F} F \xrightarrow{r_F} F''$ is the decoded
embedding.

From Theorem~\ref{thm:besg_language}, we know that $F, F'$ and $F''$ are string
graphs. From Corollary~\ref{cor:esg-inst}, we know that $H_n' \xleftarrow{l_n}
H_n \xrightarrow{r_n} H_n''$ is an ESG rewrite rule. From
Lemma~\ref{lem:decode-embed}, we see that the decoded embedding simply extends
the monomorphisms, so $l_F$ ($r_F$) acts in the same way when restricted to
$l_n$ ($r_n$). The decoding process by definition does not establish any new
inputs, nor outputs and therefore, $F' \xleftarrow{l_F} F \xrightarrow{r_F}
F''$ is a string graph rewrite rule.
\end{proof}

\section{B-ESG rewriting}\label{sec:final}

We have shown how to represent families of string diagrams using B-ESG
grammars. We have also shown how to represent equational schemas using B-ESG
rewrite patterns. These rewrite patterns induce B-ESG rewrite rules which may
be used for rewriting other B-ESG grammars.  In the final section of this
thesis, we will combine the results of Section~\ref{sec:b-esg-rules} and
Section~\ref{sec:rewriting} to show how we can rewrite B-ESG grammars in an
admissible way. This would then show that our framework correctly represents
reasoning with infinite families of string diagrams, even when we are rewriting
them using equational schemas of infinite families of string diagrams.

We begin by introducing the most central definition in this chapter which
combines most of our previous constructions. We will be referring to it
multiple times before the end of the chapter. It formally defines what we
mean by rewriting a B-ESG grammar. In particular, we use a B-ESG rewrite
rule and a saturated grammar matching which also satisfies the partially
adhesive conditions to get an $\mathcal S$-rewrite. 

\begin{definition}[B-ESG rewrite]\label{def:b-esg-rewrite}
Given a B-ESG rewrite rule $B= B_L \xleftarrow l B_I \xrightarrow r B_R$
with initial nonterminal label $S$
and a B-ESG grammar $B_H,$ such that $B_H$ is in proper B-ESG normal form,
then we will say that the
\emph{B-ESG rewrite} of $B_H$ using $B$ over a saturated matching $m: G_L \to
G_H$ which satisfies the partially adhesive conditions, is the encoded B-edNCE
grammar $B_M
= (G_M, T),$ denoted by $B_H \leadsto_{B,m} B_M$, where $G_M$ is given by
the $\mathcal S$-rewrite $G_H \leadsto_{B,m} G_M$:
\cstikz{b-esg-dpo.tikz}
\end{definition}

Our first major theorem proves that such a rewrite indeed results in a B-ESG
grammar. This is shown in Subsection~\ref{sub:preserve-shit}. Next, in
Subsection~\ref{sub:pattern-shit}, we will show that B-ESG rewrites form a
B-ESG rewrite pattern. Finally, in Subsection~\ref{sub:fucking-final}, we show
that the B-ESG rewrite pattern which results from the rewrite is admissible.

Before we begin with the proofs, let's consider an example.

\begin{example}\label{ex:big-example}
As a concrete example, consider the B-ESG
rewrite rule $B:= B_L \xleftarrow l B_I \xrightarrow r B_R$ from
Example~\ref{ex:b-esg-rewrite-rule-crap}. Recall, that its grammars are given
by:
\cstikz[0.9]{example.tikz}

This rewrite rule on B-ESG
grammars can be applied to another B-ESG grammar $B_H = (G_H, T)$, where $G_H:$ 
\cstikz[0.9]{example2.tikz}
and the result will be
the B-ESG grammar $B_M = (G_M, T)$, where $G_M$ is shown below:
\cstikz[0.9]{example3.tikz}
The rewrites are performed locally, on a per-production basis. The language
$B_H$ consists of complete string graphs, where in addition, each node-vertex
of the complete graph has a (grey) line graph of arbitrary length glued onto it.
Similarly, the language of $B_M$ consists of a star string graph, where there
are (grey) line graphs glued onto each of the star vertices. The admissibility
of the rewrite ensures that any instantiation
of $B_M$ may be obtained from the parallel instantiation of $B_H$ by applying an
appropriate DPO rewrite rule obtained from an instantiation of the B-ESG
rewrite rule $B$. In particular, the
corresponding instantiation of the rewrite rule consists of only the matched
productions from the instantiation of $B_H$ (or equivalently $B_M)$. For
example, if $H$ is the instantiation of $B_H$ with 3 white node-vertices
and 3 grey node-vertices, the parallel instantiation $M$ of $B_M$ must
also have 3 white node-vertices and 3 grey node-vertices:
\cstikz[0.9]{example4.tikz}
Then from the admissibility result, we know that $H \leadsto_{s,m} M,$
where both the matching $m$ and the rewrite rule $s$ can be effectively
determined from the instantiation. In particular, $s$ is given by the
instantiation of $B$ with length 3 (the number of white node-vertices):
\cstikz[0.9]{example5.tikz}
and the number of grey node-vertices in $H$ (or in $M$) is irrelevant, as they
are generated by productions outside of the image of the matching.

In terms of string diagrams, this B-ESG rewrite is representing the following
rewrite of entire families of diagrams:
\cstikz{family-rewrite-shiiit.tikz}
where the equational schema $(B)$ was used, given by:
\cstikz{family-rewrite-shiiit2.tikz}
where in both of the above examples, $K_n$ is the complete string diagram on
$n$ white node-vertices, $S_n$ is the star string diagram on $n$ node-vertices.
\end{example}
\subsection{Rewrites preserve B-ESG structure}\label{sub:preserve-shit}

A B-ESG grammar has a lot of structure, so we need to prove that the
encoded B-edNCE grammar which results from the rewrite satisfies multiple
conditions. Some of the lemmas in this section, like the ones which establish
normal forms, are not strictly required for this subsection, but they make
the proofs considerably easier.

\begin{lemma}\label{lem:correspond}
Given a B-ESG rewrite $B_H \leadsto_{B,m} B_M$, as in
Definition~\ref{def:b-esg-rewrite}, then $(G_H, G_M)$ is a B-edNCE
correspondence.
\end{lemma}
\begin{proof}
From our assumptions, it follows that the premises of Theorem~\ref{thm:main}
are satisfied (note, in particular, we have assumed the matching $m$ satisfies
the partially adhesive conditions). Therefore, $G_M$ is uniquely defined (up to
isomorphism), it is a B-edNCE grammar and $(G_H, G_M)$ is a B-edNCE
correspondence.
\end{proof}

\begin{corollary}\label{cor:reduced}
Given a B-ESG rewrite $B_H \leadsto_{B,m} B_M$, as in
Definition~\ref{def:b-esg-rewrite}, then $G_M$ is reduced.
\end{corollary}
\begin{proof}
From the previous lemma, we know $(G_H, G_M)$ forms a B-edNCE correspondence.
Therefore, $G_M$ is reduced iff $G_H$ is reduced, which is true by
assumption.
\end{proof}

The next lemma shows the context (cf. Definition~\ref{def:context}) of each
nonterminal vertex in $G_M$ is identical to either the context of its pre-image
or the context of its corresponding nonterminal vertex in $G_H$. This lemma is
crucial for many of the remaining proofs as it easily shows how some of the
relevant structure from the other B-ESG grammars carries over to the rewrite.

\begin{lemma}\label{lem:nonterminal-context}
Given a B-ESG rewrite $B_H \leadsto_{B,m} B_M$, as in
Definition~\ref{def:b-esg-rewrite}, then for each nonterminal vertex $x \in
G_M$, the following is true:
\begin{enumerate}
\item If $x$ is in the image of $f$, that is there exists $x' \in G_R$, such
that $f(x') = x$, then the context of $x$ in $G_M$ is the same as the context
of $x'$ in $G_R$.
\item If $x$ is not in the image of $f$, then the context of $x$ in $G_M$
is the same as the context of $x'$ in $G_H,$ where $x'$ is the corresponding
nonterminal vertex of $x$ in $G_H$.
\end{enumerate}
\end{lemma}
\begin{proof}
As we have already shown, the rewrite does not add, nor delete productions
from $G_H$. Each production is modified on a local level using an
extended graph saturated rewrite.

For the first property, from Lemma~\ref{lem:saturated_grammar_matching}, we
know that $f : G_R \to G_M$ is a saturated grammar matching with respect to the
rewrite rule $G_R \xleftarrow r G_I \xrightarrow l G_L.$ Then, we know from the
\textbf{Production saturation} condition, and in particular the fact that $f$
also acts as a saturated extended graph matching, that $x$ and $x'$ will have
the same connection instructions. The same condition also implies that the edge
neighbourhoods of $x$ and $x'$ will be the same. Therefore $x$ and $x'$ have
the same context.

For the second property, observe that $x$ can only be adjacent to boundary
vertices or to vertices outside of the image of $f$ (that is the standard
no-dangling edges condition). Therefore, its edge-neighbourhood is exactly the
same as the edge-neighbourhood of $x'$. Next, consider an arbitrary connection
instruction associated to $x$. Such a connection instruction is not in the
image of $f$ and therefore it must have been preserved by the rewrite, so $x'$
has the same connection instruction. Therefore, $x$ and $x'$ have the same
context.
\end{proof}

\begin{corollary}\label{cor:non-empty-context}
Given a B-ESG rewrite $B_H \leadsto_{B,m} B_M$, as in
Definition~\ref{def:b-esg-rewrite}, then for each nonterminal vertex $x \in
G_M$, with $\lambda(x) = S$, it follows $x$ has empty context.
\end{corollary}
\begin{proof}
$S$ is the initial nonterminal label in $G_L$ and $G_R$, therefore there exists
at least one production $p$ with label $S$ in both of these grammars. By
assumption, $p$ must have empty connection instructions. Therefore, from the
\textbf{Production saturation} condition which is satisfied for the saturated
matching $m$, it follows $m(p)$ will also have empty connection
instructions in $G_H$. Therefore, both $G_H$ and $G_R$ contain at
least one production with label $S$ and empty connection instructions. This
means that any nonterminal vertex $x$ with label $S$ in $G_H$ or $G_R$ must
have empty context. Otherwise, we get a contradiction with the fact that $G_H$
or $G_R$ is neighbourhood-preserving, context-consistent and contain no useless
connection instructions. Then, from the previous lemma it follows that each
nonterminal vertex with label $S$ in $G_M$ must also have empty context.
\end{proof}

The next lemma is very powerful, as it follows as a corollary that several
normal forms for our rewritten grammar are satisfied. The proof is done
by induction over the length of the derivations which produce the sentential
forms. In particular, the strengthened statements (compared to what is needed
for the normal forms) make the proof considerably simpler compared to
what is required by proving each normal form separately.

\begin{lemma}
Given a B-ESG rewrite $B_H \leadsto_{B,m} B_M$, as in
Definition~\ref{def:b-esg-rewrite}, then the following are true: 
\begin{enumerate}
\item applying any production $p$ of $G_M$
to any sentential form of $G_M$ will use all of the connection
instructions of $p$ to create bridges.
\item expanding any nonterminal vertex $x$
in any sentential form of $G_M$ will use all of its incident
edges to create bridges.
\end{enumerate}
\end{lemma}
\begin{proof}
Let the context functions for $G_R$ and $G_H$ be $\eta_R$ and $\eta_H$
respectively.

Consider an arbitrary nonterminal label $X$. If $X=S$, then from
Corollary~\ref{cor:non-empty-context}, it follows that each nonterminal
vertex $x$ with label $S$ in any sentential form has empty context.
So let us assume that $X \not = S$.

If $X$ is a label which is not used in the grammar $G_R$, then no production
with label $X$ is in the image of $f$ and no nonterminal vertex with label
$X$ is in the image of $f$. Thus, the connection instructions of such
productions in $G_M$ are the same as their counterparts in $G_H$ and the
context of nonterminal vertices with label $X$ in $G_M$ is the same as their
counterparts in $G_H$ (the latter follows from
Lemma~\ref{lem:nonterminal-context}).

If $X$ is a label which is used in the grammar $G_R$, then let's consider an
arbitrary nonterminal vertex $x$ with label $X$ in $G_M$. We know $f$ is a
saturated matching and therefore from the \textbf{Nonterminal covering}
condition it follows $x \in f(G_R)$. Therefore there exists a nonterminal
vertex $x' \in G_R,$ such that $f(x') = x$.  From
Lemma~\ref{lem:nonterminal-context}, we know the context of $x$ in $G_M$ is the
same as the context of $x'$ in $G_R$. From the \textbf{Production branching}
and the \textbf{Production saturation} conditions it follows that every
production with label $X$ in $G_M$ has the same connection instructions as its
pre-image in $G_R$. Therefore, all nonterminal vertices with label $X$
in $G_M$ have the same context as their pre-images in $G_R$ and all
productions with label $X$ in $G_M$ also have the same connection instructions
as their pre-images in $G_R$.

Therefore, in all cases, a nonterminal vertex $x$ with label $X$ has the same
context in $G_M$ as either its counterpart in $G_H$ or its pre-image in $G_R$,
depending solely on whether or not $X$ is used in $G_R$ or not. Because
both grammars $G_R$ and $G_H$ are context-consistent and contain no useless
connection instructions, this means the context of $x$ in $G_M$ is $\eta_R(X)$
or $\eta_H(X)$, if $X$ is used in $G_R$ or not, respectively.

Similarly,
every production $p$ with label $X$ has the same connection instructions in
$G_M$ as either its counterpart in $G_H$ or its pre-image in $G_R$, depending
solely on whether or not $X$ is used in $G_R$ or not. We will use these
two facts to prove that the conclusion holds.

We will prove by induction on the length of the derivation that
each connection instruction is used to establish a bridge every
time a nonterminal vertex is replaced and also that all of its incident
edges are also used to create bridges.

In the base case, the sentential form is just $sn(U,z)$, where $U$ is the
initial nonterminal label for $G_M$, which by definition is the same
as the one for $G_H$. Let's consider an arbitrary expansion
$sn(U,z) \Longrightarrow_{z,p} H$, where $p \in G_M$ with label $U$.
If $U=S$, then as already shown, $p$ must have connection instructions as
its pre-image in $G_R$, which has empty connection instructions by assumption.
If $U \not = S$, then by the \textbf{Initiality} condition $U$ is not used
in $G_R$ and therefore $p$ has the same connection instructions as its
corresponding production in $G_H$, which is initial and therefore has no
connection instructions. In both cases, the first proposition follows trivially
and the second one also follows trivially, because $sn(U,z)$ has no edges.

For the step case, assume that we have proved our propositions for all
sentential forms $sn(S,x_1) \Longrightarrow_{x_1, p_1} H_1
\Longrightarrow_{x_2, p_2} \cdots \Longrightarrow_{x_n, p_n} H_n$, where the
length of the derivation sequence is at most $n$.  Let $x_{n+1} \in H_n$ be an
arbitrary nonterminal vertex with label $X$ and let's assume it was created by
production $p_k$. Then, from the induction hypothesis, it follows the
context of $x_{n+1}$ in $H_n$ is the same as the context of $x_{n+1}$
in $rhs(p_k)$, which is the body of a production in $G_M$.
If $X$ is used in $G_R$, then this context is $\eta_R(X)$ and otherwise it is
$\eta_H(X)$, as already pointed out.

Consider the sentential form $H_n \Longrightarrow_{x_{n+1},
p_{n+1}} H_{n+1},$ where $p_{n+1}$ is an arbitrary production with label $X$.

If $X$ is a label which is not used in $G_R$, then the context of $x_{n+1}$ in
$H_n$ is $\eta_H(X)$. We have also shown that $p_{n+1}$ must have the same
connection instructions as its corresponding production in $G_H$. Therefore,
all of the connection instructions in $p$ must be used to create bridges,
because otherwise we get a contradiction with the fact that $G_H$ is
context-consistent and contains no useless connection instructions. Moreover,
every edge incident to $x_{n+1}$ must also be used in order to create bridges,
because otherwise we get a contradiction with the fact that $G_H$ is
neighbourhood-preserving.

If $X$ is a label which is used in $G_R$, then the context of $x_{n+1}$
in $H_n$ is $\eta_R(X)$. We have also shown that $p_{n+1}$ must have the
same connection instructions as its pre-image in $G_R$. Therefore, all
of the connection instructions in $p$ must be used to create bridges, because
otherwise we get a contradiction with the fact that $G_R$ is context-consistent
and contains no useless connection instructions. Moreover,
every edge incident to $x_{n+1}$ must also be used in order to create
bridges, because otherwise we get a contradiction with the fact that $G_R$ is
neighbourhood-preserving.
\end{proof}

\begin{corollary}\label{cor:normal-forms}
Given a B-ESG rewrite $B_H \leadsto_{B,m} B_M$, as in
Definition~\ref{def:b-esg-rewrite}, then $G_M$ is context-consistent,
neighbourhood-preserving and it contains no useless connection instructions.
\end{corollary}
\begin{proof}
The fact that $G_M$ does not contain useless connection instructions follows
immediately from the first proposition of the previous lemma as its statement
is clearly stronger. The second proposition of the previous lemma immediately
implies that $G_M$ is neighbourhood-preserving, as it is clearly equivalent
to the definition.

For context-consistency, the previous lemma implies that the context
of any nonterminal vertex $x$ with label $X$ in a production $p$ of $G_M$
is the same as the context of $x$ in any sentential form of $G_M$ when
$x$ is created by an application of $p$. However, as we have pointed out
in the proof of the previous lemma, the context of $x$ within the grammar
$G_M$ is determined solely by its label $X$ -- if $X$ is used in $G_R$, then
its context is $\eta_R(X)$ and otherwise it is $\eta_H(X)$. Therefore,
the context function is given by:
\[ \eta_M(X) =
\begin{cases}
\eta_R(X) & \text{ if label } X \text{ used in } G_R\\
\eta_H(X) & \text{ else}
\end{cases}
\]
\end{proof}

Wire-consistency (cf. Definition~\ref{def:wire-consistent-crap}) is one of the
necessary properties of B-ESG grammars. The next lemma proves that the
rewritten grammar is indeed wire-consistent.

\begin{lemma}\label{lem:wire-consistent}
Given a B-ESG rewrite $B_H \leadsto_{B,m} B_M$, as in
Definition~\ref{def:b-esg-rewrite}, then $B_M$ is wire-consistent.
\end{lemma}
\begin{proof}
Let's assume that the initial nonterminal
label of $G_L \xleftarrow l G_I \xrightarrow r G_R$ is $S$.

Assume $G_M$ is not wire-consistent. This means, there exists a production $p
\in G_M$ which contains a nonterminal vertex $x_0 \in rhs(p)$, such that $x_0$
has
context cardinality $(\sigma, \alpha_0, d)$ at least two, for some $\sigma \in
\Sigma, \alpha_0 \in \Gamma, d \in \{in, out\}$. Moreover, there must be
a nonterminal vertex $x_n$ with label $X \in \Sigma$, such that
$x_0\ \mathcal Q_{\sigma}^d(\alpha_0, \alpha_n)\ x_n$ (this is the multi-step
context-passing relation). It must also be the case that there exists
a production $p' = X \to (D,C)$ such that $C$ contains a connection instruction
$(\sigma, \alpha_n, \gamma, z, d),$ where $z \in D$ is a wire-vertex.

First, observe that from Lemma~\ref{lem:saturated_grammar_matching},
we know that $f : G_R \to G_M$ is a saturated grammar matching with respect
to the rewrite rule $G_R \xleftarrow r G_I \xrightarrow l G_L.$

Next, we prove that no nonterminal vertex with label $S$ in $G_M$ can satisfy
the single-step context-passing relation $\mathcal P$ (and thus $\mathcal Q$).
Assume the contrary, that there exist nonterminal vertices $u, u',$ such that
$u\ \mathcal P^d_{\sigma} (\alpha', \alpha'')\ u'$
and one of them has label $S$. If $u$ has label $S$, then, the vertex $u$ must
have non-empty context which contradicts with
Corollary~\ref{cor:non-empty-context}. Otherwise, $u'$ has label $S$, which
means
there exists a production $q$ which has a connection instruction to a
nonterminal vertex with label $S$, which again contradicts
Corollary~\ref{cor:non-empty-context}.

We have assumed that
$x_0\ \mathcal Q_{\sigma}^d(\alpha_0, \alpha_n)\ x_n$, therefore,
there exist
nonterminal vertices $x_1,\ldots,x_{n-1}$ and edge labels
$\alpha_1,\ldots,\alpha_{n-1}$, such that:
\begin{align*}
x_0 \ \mathcal P_{\sigma}^d(\alpha_0, \alpha_1)\ x_1\ &\land \\
x_1 \ \mathcal P_{\sigma}^d(\alpha_1, \alpha_2)\ x_2\ &\land \\
\cdots & \\
x_{n-1} \ \mathcal P_{\sigma}^d(\alpha_{n-1}, \alpha_n)\ x_n\ 
\end{align*}

If $x_0$ is outside of the image of $f$, then this means that all the $x_i$ are
also outside of the image of $f$. To see this, assume the opposite. This means
there exist nonterminal vertices $x_k, x_{k+1}$ in productions $q, q'$ such
that $x_k$ is not in the image of $f$ and $x_{k+1}$ (and thus $q'$) are in the
image of $f$, and moreover 
$x_k\ \mathcal P^d_{\sigma} (\alpha_k, \alpha_{k+1})\ x_{k+1}$.
By definition of the
single-step context-passing relation, this means $\lambda(x_k) = lhs(q')$.
However, $x_k$ is not in the image of $f$, but $q'$ is, and therefore from the
\textbf{Nonterminal covering} condition satisfied for the saturated matching
$f$, it means that $\lambda(x_k) = S ,$ which is a contradiction.

Therefore, if $x_0$ is outside of the image of $f$, then this means that each
$x_i$ has label which is not used in $G_R$, because all of these vertices are
outside of the image of $f$ and their label cannot be $S$ (otherwise we would
violate the \textbf{Nonterminal covering} condition). The productions which
carry such labels have not been modified by the rewrite and moreover, from
Lemma~\ref{lem:nonterminal-context}, we know that the context of each $x_i$
in $G_R$ is the same as its counterpart in $G_H$. Therefore, we get a
contradiction with the wire-consistency of $G_H.$

Thus, it must be the case that $x_0$ is in the image of $f$. Now, we can
prove that all the $x_i$ vertices are also in the image of $f$. Assume
the opposite, therefore there exist $x_k, x_{k+1}$ in productions
$q$ and $q'$, such that
$x_k$ is in the image of $f$, but $x_{k+1}$ isn't. Because
$x_k\ \mathcal P^d_{\sigma} (\alpha_k, \alpha_{k+1})\ x_{k+1}$, it follows
$\lambda(x_k) = lhs(q')$ and also $q'$ has a connection instruction to
$x_{k+1}$.
From the \textbf{Production branching} condition, it follows production $q'$ is
also in the image of $f$. But then, from the \textbf{Production saturation}
condition, it follows that production $q'$ must have the same connection
instructions as those of its pre-image. This results in a contradiction,
because we know $q'$ must also have a connection instruction to $x_{k+1}$ which
is outside of the image of $f$.

Because all $x_i$ are in the image of $f$, then using
Lemma~\ref{lem:nonterminal-context}, we also know that they have the
same context in $G_M$ as their pre-images in $G_R$ and therefore we get
a contradiction with the wire-consistency of $G_R$.
\end{proof}

Building on top of the results from this subsection we may now prove our first
main theorem which states that a B-ESG rewrite results in a B-ESG grammar.

\begin{theorem}\label{thm:rewrite-besg}
Given a B-ESG rewrite $B_H \leadsto_{B,m} B_M$, as in
Definition~\ref{def:b-esg-rewrite}, then $B_M$ is a B-ESG grammar.
\end{theorem}
\begin{proof}
From Lemma~\ref{lem:correspond}, we know that $G_M$ is a B-edNCE grammar.
From Lemma~\ref{lem:wire-consistent}, we know that $G_M$ is wire-consistent.
Next, we have to show that each of the (local) conditions from
Definition~\ref{def:besg} are satisfied.

For condition N1, consider an arbitrary edge $e$ in $G_M$, such that it
connects two node-vertices. If $e$ is outside of the image of $f$, then
$e$ is unmodified by the rewrite and therefore it must carry an encoding
label, because $G_H$ also satisfies condition N1. Otherwise, $e$ is in the
image of $f$ and again we get that $e$ must have an encoding label, because
$G_R$ satisfies condition N1.

Conditions N2 and W2 follow using the same argument. If the production $p \in
G_M$ in question is outside of the image of $f$, then it is unmodified and must
have the same connection instructions as its counterpart in $G_H$ which
satisfies both conditions N2 and W2. Otherwise, $p$ is in the image of $f$ and
therefore by the \textbf{Production saturation} condition we know that $p$ has
the same connection instructions as its pre-image in $G_R$, which satisfies
both conditions N2 and W2.

Finally, let's consider condition W1. Assume for contradiction that there
exists a wire-vertex $w$ in production $p$ of $G_M$, such that $w$ has
production in-degree more than one. Without loss of generality, let's assume
that the production in-degree of $w$ is two. If $w$ is outside of the image
of $f$, then $w$ and its incident edges or connection instructions are
unmodified by the rewrite which means that $G_H$ also contains a wire-vertex
with in-degree two, which is a contradiction. Therefore, $w$ is in the image
of $f$. If $w$ is not a boundary vertex, then from the no-dangling
edges/connection instructions conditions, it follows that the pre-image
of $w$ under $f$ in $G_R$ has production in-degree two, which is a
contradiction. Therefore, $w$ must be a boundary vertex, so let $w = f \circ
r(w')$, where $w' \in G_I$.

The pre-image of $w$ under $f$ is then
$r(w')$. If $r(w')$ has production in-degree zero in $G_R$, this means
that $m \circ l(w')$ must have production in-degree two in $G_H$, which
is a contradiction. Therefore, $r(w')$ has production in-degree equal to
one. However, $B$ is a B-ESG rewrite rule and therefore the inputs and
outputs in corresponding productions of $G_L$ and $G_R$ are in bijection,
specified by the morphisms $l$ and $r$. Thus, we know that $l(w')$
must have production in-degree one as well. Therefore, $m \circ l(w')$
has production in-degree one. But then, in the pushout complement $G_K$, we get
that $k(w')$ has production in-degree zero (because $G_I$ has no edges, nor
connection instructions). After computing the pushout of $k$ and $r$, we see
that $f \circ r(w') = w$ must have production in-degree one, which is a
contradiction.

If we assume that there exists a wire-vertex with production out-degree more
than one, we get a contradiction by symmetry.
\end{proof}

\subsection{Rewrites form a B-ESG pattern}\label{sub:pattern-shit}

We continue by showing that a B-ESG rewrite forms a B-ESG rewrite pattern
together with the original grammar.

\begin{lemma}\label{lem:isolated-shit}
Given a B-ESG rewrite $B_H \leadsto_{B,m} B_M$, as in
Definition~\ref{def:b-esg-rewrite}, then $B_M$ contains no production
isolated wire-vertices.
\end{lemma}
\begin{proof}
Assume for contradiction that $G_M$ contains a wire-vertex $w$ which has
production in-degree zero and production out-degree zero. If $w$ is not in the
image of $f$, then it is unmodified by the rewrite and from the no-dangling
edges/connection instructions condition, it follows that $G_H$ contains a
production isolated wire-vertex, which is a contradiction. Therefore, $w$ must
be in the image of $f$, that is, there exists a wire-vertex $w' \in G_R,$ such
that $f(w') = w$. But then, it follows that $w'$ must be a production isolated
wire-vertex, which is a contradiction.
\end{proof}

\begin{theorem}\label{thm:b-esg-pattern}
Given a B-ESG rewrite $B_H \leadsto_{B,m} B_M$, as in
Definition~\ref{def:b-esg-rewrite}, then $(B_H,B_M)$ is a B-ESG rewrite
pattern.
\end{theorem}
\begin{proof}
From Theorem~\ref{thm:rewrite-besg}, we know that $B_M$ is a B-ESG grammar.
From
Lemma~\ref{lem:correspond}, we know that $(G_H, G_M)$ forms a B-edNCE
correspondence. Corollary~\ref{cor:reduced} combined with
Corollary~\ref{cor:normal-forms} and Lemma~\ref{lem:isolated-shit} show that
$B_M$ is in proper B-ESG normal form. What remains to be shown is that the
\textbf{IO} condition from Definition~\ref{def:besg_rewrite_pattern} is
satisfied.

We define a function $F$ which maps production inputs from $G_H$ into
production inputs of $G_M$:
\[ F(w) =
\begin{cases}
f\circ r(w') & \text{ if } w= m\circ l(w') \text{, where } w' \in G_I
\text{ is a production input}\\
g(w') & \text{ if } w=s(w') \text{ where } w' \in G_K, w' \not\in k(G_I)
\text{ is a production input}
\end{cases}
\]
Note, that the construction is essentially the same as the one from
Lemma~\ref{lem:pattern}.
We claim that $F$ is a bijection between the production inputs of $G_H$
and the production inputs of $G_M$. First, the definition of $F$ is complete
in the sense that it is totally defined, because $m$ and $s$ are jointly
surjective and also because every production input in $G_L$ is in the
image of $l$ by definition.

Next, we need to show $F$ is well-defined, that is $F(w)$ is a production
input. In the first case, $w=m\circ l(w')$ with $w' \in G_I$. Then, $l(w')$ is
also a production input and therefore $r(w')$ is a production input in $G_R$ by
definition of B-ESG rewrite rule. Moreover, $k(w')$ must also be a production
input, because otherwise $s \circ k(w') = m \circ l(w')$ wouldn't be one.
Therefore, $F(w)= f\circ r(w')$ must also be a production input, because $G_M$
is the pushout of $k$ and $r$.
In the other case, $w=s(w')$, with $w' \in G_K$. Then, $w'$ is a production
input. Since $w'$ is not in the image of $k$, it follows $g(w')$ is also a
production input, because $G_M$ is the pushout of $k$ and $r$.

$F$ can easily be seen to be injective, because all of our morphisms in the DPO
diagram are mono. Finally, we have to show that $F$ is a surjection on the
production inputs of $G_M$. Let $w$ be an arbitrary production input of $G_M$.
$f$ and $g$ are jointly surjective on $G_M$ and therefore we have to consider
two cases.

In the first case, $w$ is in the image of $f$.  However, this also
means that $w$ is in the image of $f\circ r$, because all production inputs of
$G_R$ are in the image of $r$. Therefore, $w=f \circ r(w')$ for some $w' \in
G_I$. We know $r(w')$ is a production input and therefore $l(w')$ must be a
production input by definition of B-ESG rewrite rule. Moreover, $k(w')$ must be
a production input, because otherwise $g \circ k(w') = f \circ r(w') = w$
wouldn't be a production input. This means $m\circ l(w')$ must be a production
input, because $G_H$ is the pushout of $l$ and $k$.  But then, it follows by
definition $F(m\circ l(w')) = w.$

In the other case, $w$ is in the image of $g$ and let
$w=g(w')$ for $w' \in G_K$. Clearly, $w'$ is a production input, because it
maps injectively into $w$ which is also a production input. If $w'$ is in the
image of $k$, then it follows $w$ is also in the image of $f\circ r= g \circ
k$. So, we can assume $w' \not\in k(G_I)$. Then, $s(w')$ must be a
production input as well, because $G_H$ is the pushout of $l$ and $k$. It
follows by definition $F(s(w')) = w$, so $F$ is surjective.

For the case of production outputs, the argument follows by symmetry.
\end{proof}

\subsection{Admissibility of the rewrites}\label{sub:fucking-final}

Finally, we will show that the B-ESG rewrite pattern which is formed by a
B-ESG rewrite is admissible. First, we formally define what we mean by that.

\begin{definition}[Admissible pattern]
A B-ESG rewrite pattern $B = (B_1, B_2)$ is \emph{admissible} with respect to a
set of string graph rewrite rules $\mathcal R,$ if for every instantiation of
$B$:
\begin{align*}
sn(S, u_1) &\Longrightarrow_{*}^{B_L} H\\
sn(S, u_1) &\Longrightarrow_{*}^{B_R} H'
\end{align*}
there exists a sequence of rewrite rules $s_1,\ldots,s_n \in \mathcal R,$ such
that $H \leadsto_{s_1} \cdots \leadsto_{s_n} H'.$
\end{definition}

So, an admissible B-ESG rewrite pattern is very similar to an admissible
B-edNCE correspondence (cf. Definition~\ref{def:admissible-cor}). Before
we prove the main result of this chapter, we prove a simple fact which shows
that decoding preserves pushout squares. Then, it follows as a corollary that
decoding also preserves rewrites.

\begin{lemma}\label{lem:final-decode}
Given an $\mathcal S$-pushout:
\cstikz{final_pushout_bitches.tikz}
where all objects are encoded string graphs, then the following square
is also an $\mathcal S$-pushout:
\cstikz{final_pushout_for_real.tikz}
where $X \Longrightarrow_*^T X'$ are string graphs for $X \in \{I, K, R, M\}$
and $x'$ is the decoded embedding given by Lemma~\ref{lem:decode-embed}, for $x
\in \{r, k, f, g\}$.
\end{lemma}
\begin{proof}
From Lemma~\ref{lem:decode-embed}, we know that all of the morphisms in the
bottom square are monos. Commutativity follows trivially by construction of the
morphisms $x'$ and the commutativity of the top square. To complete the proof,
we need to show that $f'$ and $g'$ are jointly surjective.

Consider an arbitrary edge $e' \in M'$. If $e' \in M$, then from the joint
surjectivity of $f$ and $g$ it follows that $e'$ must be in the image of $f'$
or $g'$. If $e' \in M' - M$, then $e'$ has been produced from a decoding step
of $T$ by replacing some edge $e \in M$. But, $f$ and $g$ are jointly
surjective on $M$ and therefore there exists an edge $e'' $ which is the
pre-image of $e$ in $K$ or $R$. Therefore, after decoding edge $e''$ in
$K$ or $R$, the pre-image of $e'$ under $g'$ or $f'$ is produced.

The case for vertices is completely analogous to the case for edges.
\end{proof}

\begin{corollary}\label{cor:final}
Given an $\mathcal S$-rewrite:
\cstikz{dpo_tashak.tikz}
where all objects are encoded string graphs, then the following is also
an $\mathcal S$-rewrite:
\cstikz{dpo_pisna_mi.tikz}
where $X \Longrightarrow_*^T X'$ are string graphs for $X \in \{L, I, R, K, H,
M\}$ and $x'$ is the decoded embedding given by Lemma~\ref{lem:decode-embed},
for $x \in \{l, r, k, m, s, f, g\}$.
\end{corollary}
\begin{proof}
Follows immediately by two applications of the previous lemma.
\end{proof}

The next definition is introduced for notational convenience. Given a rewrite
rule between encoded string graphs, it defines the \emph{decoded rewrite rule}
by simply decoding each graph and using the decoded embeddings induced by the
original span.

\begin{definition}
Given a rewrite rule $s = L \xleftarrow l I \xrightarrow r R,$ where
all objects are encoded string graphs, then we shall denote with
$s \Longrightarrow_*^T s'$, the rewrite rule
$s' := L' \xleftarrow{l'} I' \xrightarrow{r'} R',$ where
$L \Longrightarrow_*^T L'$,
$I \Longrightarrow_*^T I'$,
$R \Longrightarrow_*^T R'$ and where $l'$ and $r'$ are the decoded embeddings
given by Lemma~\ref{lem:decode-embed}.
\end{definition}

The main and final result of this chapter is presented next. It states that
a B-ESG rewrite forms an admissible B-ESG rewrite pattern with respect to the
B-ESG rewrite rule which was used. The proof follows easily by combining
the results which we have established so far.

\begin{theorem}\label{thm:final}
Given a B-ESG rewrite $B_H \leadsto_{B,m} B_M$, as in
Definition~\ref{def:b-esg-rewrite}, then $(B_H, B_M)$ is an admissible
B-ESG rewrite pattern with respect to $L(B).$
\end{theorem}
\begin{proof}
From Theorem~\ref{thm:b-esg-pattern}, we know that $(B_H, B_M)$ is a B-ESG
rewrite pattern. Consider an arbitrary pattern instantiation:
\begin{align*}
sn(S,v_1)
&\Longrightarrow_{v_1,p_1}^{G_H}& &H_1&
&\Longrightarrow_{v_2,p_2}^{G_H}& &H_2&
\cdots
&\Longrightarrow_{v_n,p_n}^{G_H}& &H_n& 
&\Longrightarrow_{*}^{T}& &F& \\
sn(S,v_1)
&\Longrightarrow_{v_1,K(p_1)}^{G_M}&    &H_1'&
&\Longrightarrow_{K(v_2),K(p_2)}^{G_M}& &H_2'&
\cdots
&\Longrightarrow_{K(v_n),K(p_n)}^{G_M}& &H_n'&
&\Longrightarrow_{*}^{T}& &F'&
\end{align*}
where $K$ is the bijection between the productions and nonterminal vertices
of the correspondence $(G_H, G_M).$

Let $G = G_L \xleftarrow l G_I \xrightarrow r G_R.$ Then, from
Theorem~\ref{thm:main}, we know that there exists a sequence of rewrite rules
$s_1,\ldots, s_n \in L(G)$, such that $H_n \leadsto_{s_1} \cdots
\leadsto_{s_n} H_n'$. Using Corollary~\ref{cor:esg-inst}, we get that
every $s_i$ is an ESG rewrite rule. From Theorem~\ref{thm:besg_language}
and Lemma~\ref{lem:decoding}, we know that $H_n$ and $H_n'$ are encoded
string graphs. Then, using Corollary~\ref{cor:final}, we can conclude
that $F \leadsto_{s_1'}\cdots \leadsto_{s_n'} F'$, where $s_i
\Longrightarrow_*^T s_i'$. Finally, we know from Theorem~\ref{thm:b-esg-inst}
that each $s_i'$ is a string graph rewrite rule and that $s_i' \in L(B),$ which
completes the proof.
\end{proof}

This result shows that our framework correctly represents reasoning with
context-free families of string diagrams in the sense that it respects their
concrete semantics.

\subsection{Discussion}

\added{In summary, a B-ESG grammar represents a family of string diagrams.
A B-ESG rewrite pattern (which may be seen as a special kind of span of B-ESG
grammars) represents an equational schema between two families of string
diagrams. A B-ESG rewrite represents rewriting a family of string diagrams
using an equational schema in a sound way.
Another way of looking at this is to consider a B-ESG rewrite as
an equational substitution which modifies a family of subdiagrams which appear
in some larger family (cf. Example~\ref{ex:big-example}). The substitution is
sound in the sense that the result is an equational schema whose every concrete
instance can be derived using concrete instances of the rewrite rule (which is
also an equational schema).}

\added{The ZX-calculus has been used as motivation for many of the
constructions in this thesis, so we shall illustrate how these ideas translate
into it. In the ZX-calculus, a B-ESG grammar can be used to represent a family
of ZX-diagrams, which in turn may represent a quantum algorithm. The equational
axioms and the derived equational rules of the ZX-calculus can be represented
as B-ESG rewrite patterns. Then, a B-ESG rewrite can be used to represent a
sound application of one of its equational rules to some family of ZX-diagrams.
Therefore B-ESG rewrites correspond to sound transformations of families of
ZX-diagrams, that is, the families are equal in the sense that they define the
same linear maps. If, in addition, the initial family of ZX-diagrams models a
quantum algorithm and we choose a sequence of rewrites in such a way that the
resulting family may be translated back into a quantum circuit, then this would
correspond to an equivalent circuit transformation. The sequence of B-ESG
rewrite rules which were applied perform an equivalent circuit transformation
in any context which may embed the sequence of matches that were used.}

\section{Related work}

\added{B-ESG rewrite patterns are similar to the \textit{pair grammars}
approach presented in \cite{pair_grammars}. In that paper the author defines
a pair of graph grammars whose productions are in bijection which moreover
preserves the nonterminals within them. As a result, parallel derivations are
defined in a similar way to our B-ESG rewrite patterns. However, the author
uses a different notion of grammar which is less expressive than ours.}

\added{The pair grammars approach has inspired the development of \emph{triple
graph grammars} \cite{triple_grammars}. In this approach, the author uses a
triple of grammars $(L, C, R),$ which also share a bijective correspondence
between their productions. In this sense, they are similar to our B-ESG rewrite
rules. However, the middle grammar $C$ is used to relate graph elements from
$L$ to graph elements of $R$ in a more powerful way compared to our approach.
We simply use the middle grammar in order to identify the interface and
interior elements for performing DPO rewrites. However, the grammar model used
in \cite{triple_grammars} is based on monotonic single-pushout (SPO)
productions with no notion of nonterminal elements. These grammars are not
expressive enough for our purposes.}

\added{We have shown how to rewrite B-ESG grammars using B-ESG rewrite rules in
a way which allows us to relate the modifications using concrete rules from the
rule grammars. A similar approach is taken in \cite{high-level-high-level},
where the author describes how to transform High Level Replacement Systems
(HLRS) using other HLRS. HLRS are a generalisation of the idea of a DPO graph
grammar and thus this approach can be seen as doing grammar-on-grammar
rewriting. However, the underlying transformation mechanism is based on
(general) DPO rewriting and it looks unlikely that grammars which utilise such
derivations can induce the languages we are interested in.}

\chapter{\label{ch:conclude}Conclusion and future work}
In this thesis we studied the problem of equational reasoning with infinite
families of string diagrams. We started by considering context-free graph
grammars (CFGGs). We showed that both Vertex Replacement and Hyperedge
Replacement grammars, the two dominant classes of graph grammars, have equal
expressive power on string graphs. We identified a large class of !-graph
languages which are context-free, but we also identified important limitations
in the expressive power of !-graph languages which reduce their usefulness in
practical applications. We also compared the expressive power of CFGGs
with that of !-graphs: 
\cstikz{inclusion-figure-bgno2.tikz}
and showed that there are
important languages which they couldn't represent, which we used as a
justification to consider a simple extension of CFGGs that allows us to
overcome this.

Next, we introduced encoded B-edNCE grammars, which are our slightly more
expressive graph grammars, compared to the standard B-edNCE grammars, by
formalising the simple idea of recognising specially labelled edges as fixed
graphs. We also identified sufficient and necessary (up to normal form)
conditions for encoded B-edNCE grammars to generate languages of string
graphs. The grammars which satisfy these conditions are B-ESG grammars which
are the primary objects of study in this thesis as we are only interested
in languages of string graphs and not arbitrary graphs in general. We then
compared their expressive power with that of !-graphs:
\cstikz{inclusion-figure-bgto.tikz}
and
showed that B-ESG grammars are strictly more expressive than !-graphs with
trivial overlap, which are currently the only class of !-graph languages that
has been used in practice. We also showed that B-ESG grammars enjoy important
decidability properties, such as the membership and match enumeration problems
for string graphs, which are necessary properties for computer implementation
of our framework.

After showing that B-ESG grammars can represent context-free families of string
diagrams, we then demonstrated how to represent equational schemas of
context-free families of string diagrams by B-ESG rewrite rules. Finally, we
described how to rewrite B-ESG grammars using B-ESG rewrite rules, such that
the rewrite is sound with respect to the concrete semantics of our grammars.
This corresponds to rewriting a context-free family of string diagrams using an
equational schema between a pair of context-free families of string diagrams in
a sound way with respect to their instantiations. Because string graphs
represent morphisms in traced symmetric monoidal categories, this also means
that we can do equational reasoning on context-free families of morphisms
for these categories. Moreover, all of our constructions have been kept
decidable and may be implemented in software.

\section{Future work: Implementation in software}

All of the theory in this thesis has been designed with the goal of using
machine support for equational reasoning. So an obvious next step is to
actually implement the theory in software and in particular, implement it in
Quantomatic. A computer implementation would obviously benefit from discovering
efficient algorithms for the computational problems that we have described,
such as the membership and match enumeration problems. While we haven't
proved in what complexity class these problems are for the case of B-ESG
grammars, they are NP-complete for B-edNCE grammars which implies that
the complexity for the B-ESG case would be the same or worse. Therefore,
we would most likely need to consider subclasses of B-ESG grammars with
better complexity properties if we wish to implement more efficient
algorithms for these problems.
\added{For example, graph languages consisting of connected graphs
which are also of bounded degree can be parsed in polynomial time.
An overview of efficient parsing algorithms for C-edNCE grammars
is presented in Section 1.5 of \cite{c-ednce} and in Section 2.7 of
\cite{hr-grammars}.}

\added{The software implementation would also benefit from any research
into the visual representation of graph languages. For example, !-graphs
have a convenient graphical representation where we draw blue boxes around
the parts of the !-graph which can be copied, while ignoring the
graph-theoretic details of !-vertices and their adjacent edges. This makes
!-graphs simple to understand and use by domain experts who do not have
any knowledge about the graph-theoretic definition of !-graphs.
Designing similar visual representations for B-ESG families
would also be helpful, because then the practitioners would be able to
manipulate these families without requiring prior knowledge on the
operation of context-free graph grammars.}

\section{Future work: First-order logic for B-ESG grammars}

Another line of future work is the development of a first-order logic for B-ESG
grammars. We have shown how to do equational proofs with B-ESG grammars so our
framework can be seen as an equational logic. By introducing a first-order
logic on top of it we will clearly increase its the usefulness. Such a logic
has already been developed for !-graphs \cite{bang-logic1,bang-logic2} . In
the !-graph case, the logic supports quantification over the !-boxes of a
!-graph, which allows us, among other things, to introduce powerful induction
principles that can derive equational schemas from concrete rewrite rules.
Figure~\ref{fig:david-laina} provides an example.\footnote{Figure credit:
\cite{bang-logic1}, pp. 76}

\begin{figure}[H]
  \centering
      \includegraphics[width=\textwidth]{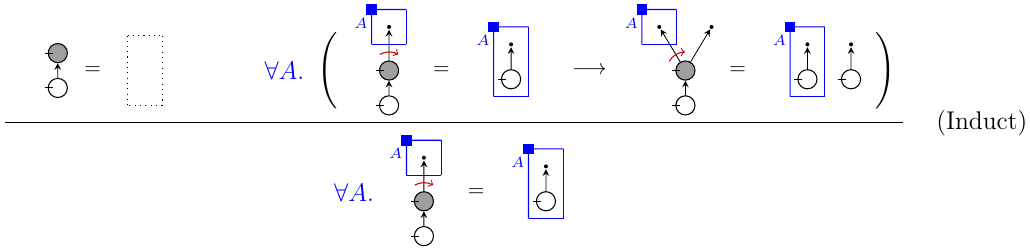}
  \caption{Induction example using !-graphs}
\label{fig:david-laina}
\end{figure}

We believe it should be possible to extend the same ideas to B-ESG grammars
which would then allow us to formally derive more powerful proofs which
go beyond equational rewrites. The key notion of this logic is the universal
quantification. In the !-graph case quantification is done over the !-boxes of
a !-graph and in the B-ESG case the analogous notion of quantification is
likely to be over a suitable subgrammar.

\section{Future work: Concurrency and rule composition}

\added{We have shown how to transform B-ESG grammars using DPO rewriting. A
natural next step is to see under what conditions we may parallelise this
process. The first step towards this goal would be to determine the partial
adhesive conditions under which the Concurrency Theorem
\cite{adhesive_categories} holds (it holds for any adhesive
category), in a similar way to which we identified the partial adhesive
conditions for DPO rewriting. Once that is done, we will have to combine the
results with the ones we have established in this thesis which guarantee the
admissibility of the rewrites. The benefit of introducing parallel rewrites is
obvious -- the rewriting process will be faster on computers with
multiple cores.}

\added{Another line of future work which is similar in spirit is to describe
how B-ESG rewrite rules may be composed. Instead of applying a sequence of
rewrite rules one at a time, this would allow us to compose the rewrite rules
into a single application. This could be beneficial as it would make common
rewrite sequence more compact. Designing support for this is not trivial -- we
would either have to choose standard graph isomorphisms or redesign our theory
to work with graphs with interfaces, where we keep the interfaces concrete.}

\section{Future work: Applications}

B-ESG grammars are an alternative to !-graphs and they can be used to solve similar
problems. So, the usual applications of verifying quantum protocols and
algorithms is still viable. However, a natural question to consider is what
kind of additional problems does the increased expressive power of B-ESG
grammars allow us to solve. Although the framework of B-ESG grammars is
applicable to general string diagrammatic theories, many of the motivating
ideas have been taken from the ZX-calculus in particular.

We have already seen that B-ESG grammars allow us to formally represent
the local complementation rule of the ZX-calculus. This is an important
rule which establishes a decision procedure for large classes of ZX-diagrams.
Once, the B-ESG formalism has been implemented in software, then a natural
next step is to also build support for the decision procedure itself.

Another useful application might be in quantum program optimisation. A
quantum programming language, like Quipper~\cite{quipper} describes (infinite)
families of quantum circuits. Our B-ESG grammars may also generate
infinite families of quantum circuits (represented as ZX-diagrams), but
B-ESG grammars are strictly less expressive compared to a Turing complete
programming language. However, we know how to do equational reasoning
with B-ESG grammars, so if we can choose a suitable orientation of our
rewrite rules and appropriate tactics for controlling derivation sequences,
then we can optimise our B-ESG grammars with respect to some criterion.
Therefore, if we can identify fragments of a quantum programming language,
like Quipper, which produce context-free families of quantum circuits, and
we can describe an adequate bidirectional translation from the programming
code into B-ESG grammars, then it could be possible to use our framework for
program optimisation.

It might also be possible to use B-ESG rewriting to decide some properties
using abstract interpretation of quantum programs. This would, of course,
require identifying a suitable abstract domain.


\addcontentsline{toc}{chapter}{Bibliography}
\printbibliography
\end{document}